\documentclass[12pt,a4paper,oneside]{book}

\newif\ifsetCustomMargin
\setCustomMargintrue

\newif\ifuseCustomBib
\useCustomBibtrue

\ifsetCustomMargin
  \RequirePackage[left=37mm,right=30mm,top=35mm,bottom=30mm]{geometry}
\fi

\usepackage{setspace}
\usepackage{amsmath,amssymb,amsfonts,amsthm,bm}
\newtheorem{prop}{Proposition}
\usepackage{graphicx}
\usepackage{subcaption}
\usepackage{booktabs}
\usepackage{multirow}
\usepackage{siunitx}
\usepackage[labelsep=space,tableposition=top]{caption}

\usepackage{fancyhdr}
\ifuseCustomBib
   \RequirePackage[round, authoryear, sort]{natbib}
\fi

\usepackage{xcolor}
\definecolor{thesislinkgreen}{rgb}{0.131,1.0,0.024}
\definecolor{thesislinkred}{rgb}{0.986,0.007,0.027}
\usepackage[colorlinks,linktocpage,bookmarksnumbered,pdfpagelabels]{hyperref}
\hypersetup{
  pdftitle={Regularized Estimation of Spatial Patterns},
  pdfauthor={Wen-Ting Wang},
  linkcolor=thesislinkgreen,
  citecolor=thesislinkgreen,
  urlcolor=thesislinkgreen,
  filecolor=thesislinkgreen,
}
\makeatletter
\let\thesis@eqref\eqref
\renewcommand{\eqref}[1]{%
  \begingroup
  \hypersetup{linkcolor=thesislinkred}%
  \thesis@eqref{#1}%
  \endgroup
}
\makeatother

\newenvironment{romanpages}{
  \setcounter{page}{1}
  \renewcommand{\thepage}{\roman{page}}}
{\newpage\renewcommand{\thepage}{\arabic{page}}}

\begin{document}


\frontmatter

\begin{titlepage}
  \thispagestyle{empty}
  \centering
  \vspace*{1.0cm}
  {\fontsize{20}{28}\selectfont\bfseries National Chiao Tung University\par}
  \vspace{0.9cm}
  {\fontsize{16}{22}\selectfont Institute of Statistics\par}
  \vspace{2.2cm}
  {\fontsize{18}{26}\selectfont Doctoral Dissertation\par}
  \vspace{2.4cm}
  {\fontsize{18}{26}\selectfont\bfseries Regularized Estimation of Spatial Patterns\par}
  \vspace{3.6cm}
  {\fontsize{13}{20}\selectfont by\par}
  \vspace{0.6cm}
  {\fontsize{16}{22}\selectfont Wen-Ting Wang\par}
  \vspace{1.6cm}
  {\fontsize{13}{20}\selectfont Advisor: Hsin-Cheng Huang, Ph.D.\par}
  \vfill
  {\fontsize{13}{20}\selectfont July 2016\par}
  \vspace*{1.0cm}
\end{titlepage}
\cleardoublepage


\begin{romanpages}

\chapter*{Abstract}

Global warming and the El Niño phenomenon cause various climate changes and anomalies worldwide. Recently, extreme weather events such as heatwaves, droughts, and floods have become increasingly frequent. Researchers have studied atmospheric dynamics to reduce potential damage and increase safety. Statistical methods such as principal component analysis and maximum covariance analysis have been widely used to analyze spatial patterns of atmospheric variables. However, when the signal-to-noise ratio is low, the patterns obtained from these methods are often too noisy to be physically meaningful.

This dissertation proposes regularization methods that simultaneously incorporate smoothness and sparseness penalties to estimate spatial patterns more interpretably. The proposed methods, called SpatPCA and SpatMCA, can be applied to both regularly and irregularly spaced data. An efficient algorithm based on the alternating direction method of multipliers (ADMM) is developed for computation. Through analysis of sea surface temperature data in the Indian Ocean and precipitation data in East Africa, we demonstrate the effectiveness of the proposed methods in revealing the spatial structure and studying how temperature variations in the Indian Ocean influence precipitation in East Africa.



\chapter*{Acknowledgments}

I am grateful to my advisor, Professor Hsin-Cheng Huang, for his guidance and suggestions throughout my doctoral studies. I would also like to thank the examination committee members and all those who have supported my research.


\tableofcontents
\listoffigures

\end{romanpages}


\mainmatter


\chapter{Introduction}\label{ch:ch1}

\section{Motivation}

Climate changes are associated with atmospheric dynamics and inevitably affect human life. Recent developments in atmospheric and oceanographic sciences have shown that these dynamics can be studied through spatial patterns of variables, such as sea level pressure and sea surface temperature (SST) (see e.g., \citealp{vonstorch1999analysis}). Hence, extracting important spatial patterns from these variables has drawn much attention in the past couple of decades. For example, SSTs in the Indian Ocean are known to influence the precipitations in Africa \citep{omondi2013influence} through two spatial patterns called basin and dipole modes \citep{chu,sst}. In fact, the basin mode can be used to predict the El Ni\~{n}o Southern Oscillation, which is associated to the extreme weather \citep{siedler2013ocean}.

There are several methods available to find spatial patterns. For instance, principal component analysis (PCA) has been widely used to identify dominant spatial patterns by utilizing the eigen-decomposition of a spatial covariance matrix \citep{hotelling1933analysis,pearson1901liii}. It is also called empirical orthogonal function analysis in atmospheric science. In addition, maximum covariance analysis (MCA) has been used to capture coupled spatial patterns by applying the singular value decomposition to a cross-covariance matrix between two spatial processes \citep{tucker1958inter}. However, the patterns estimated by both methods may be too noisy to be physically interpretable when the signal-to-noise ratio is not high enough.

In this thesis, we consider a regularization approach by incorporating two forms of regularization through (1) a smoothness penalty and (2) a sparseness penalty, to enhance the interpretability. While smoothness can be achieved by applying smoothing splines to PCA \citep{fda} and MCA \citep{salim2007model,salim2005modelling}, they mainly produce global features rather than local ones. On the other hand, sparseness penalties, such as the $L_1$ penalty \citep{jolliffe2002eofs}, the elastic-net penalty \citep{lee2011sparse}, and the SCAD penalty \citep{lee2011sparse}, can be utilized to capture localized patterns. Nonetheless, they may produce isolated patterns, or not even be applicable to continuous spatial domains with data observed at irregularly spaced locations.

\section{Contributions}

The main contributions of this thesis are summarized below:

\begin{itemize}
  \item We developed two regularization methods that simultaneously incorporate smoothness and localized features in dominant spatial patterns for PCA and MCA. Both methods can be directly applied to data measured irregularly in space. We call these two methods spatial PCA (SpatPCA) and spatial MCA (SpatMCA), respectively.
  
  \item We derived efficient algorithms for computing the resulting SpatPCA and SpatMCA estimates using the alternating direction method of multipliers (ADMM) \citep{admm}. In particular, two R packages SpatPCA and SpatMCA were developed. Both are available on the Comprehensive R Archive Network (CRAN).
\end{itemize}

\section{Thesis Overview}

This thesis is organized as follows. Chapter 2 briefly describes PCA and MCA. Chapter 3 gives an overview of two regularization methods used in this thesis: smoothing splines and Lasso. The proposed regularization approach for PCA is introduced in Chapter 4, and that for MCA is given in Chapter 5, including an application to study how precipitations in East Africa are affected by sea surface temperatures in the Indian Ocean. Finally, the thesis closes with a summary and an outlook on future directions in Chapter 6.


\chapter{Principal Component Analysis and Maximum Covariance Analysis}\label{ch:ch2}

In this chapter, we introduce two commonly used dimension reduction methods in multivariate analysis: principal component analysis (PCA) and maximum covariance analysis (MCA).

\section{Principal Component Analysis}

Principal component analysis was first introduced by \citet{pearson1901liii} and further developed by \citet{hotelling1933analysis}. The key idea behind PCA is to reduce the dimensionality of a dataset for a large number of correlated variables, while preserving as much of the variation as possible. It can be achieved by projecting the original data to a lower dimensional space. Let $\bm{y} = (y_1,\dots,y_p)'$ be a $p$-dimensional zero-mean random vector with covariance matrix $\bm{\Sigma}$. PCA tries to find an orthogonal transformation of $\bm{y}$ as principal components (PCs) such that the resulting components are uncorrelated and have maximum variances. Specifically, the first PC is given by $y_1^* = \bm{\phi}_1'\bm{y}$, where $\bm{\phi}_1 = (\phi_{11},\dots,\phi_{1p})'$ is determined by
\[
\bm{\phi}_1 = \arg\max_{\bm{\phi}\in\mathbb{R}^p} \bm{\phi}'\bm{\Sigma}\bm{\phi},
\]
subject to $\bm{\phi}'\bm{\phi}=1$. The $k$-th PC for $k=2,\dots,K\leq p$, can be successively obtained by $y_k^* = \bm{\phi}_k'\bm{y}$, where
\[
\bm{\phi}_k = \arg\max_{\bm{\phi}\in\mathbb{R}^p} \bm{\phi}'\bm{\Sigma}\bm{\phi},
\]
subject to $\bm{\phi}'\bm{\phi}=1$ and $\bm{\phi}_j'\bm{\phi}=0$, $j=1,\dots,k-1$.

The maximizers, $\bm{\phi}_1,\dots,\bm{\phi}_K$, can be solved by the Lagrange multipliers technique. For $K=1$, the corresponding Lagrange function is
\[
\bm{\phi}'\bm{\Sigma}\bm{\phi}-\lambda(\bm{\phi}'\bm{\phi}-1),
\]
where $\lambda$ is the Lagrange multiplier. Differentiating the above expression, equating it to 0, and evaluating it at $\bm{\phi}=\bm{\phi}_1$, give
\[
\bm{\Sigma}\bm{\phi}_1=\lambda_1\bm{\phi}_1,
\]
for some $\lambda_1\in\mathbb{R}$. It follows that $\bm{\phi}_1$ is an eigenvector of $\bm{\Sigma}$ with the corresponding eigenvalue $\lambda_1$. Since $\bm{\phi}'\bm{\Sigma}\bm{\phi}$ is maximized at $\bm{\phi}_1'\bm{\Sigma}\bm{\phi}_1=\bm{\phi}_1'(\lambda_1\bm{\phi}_1)=\lambda_1$, $\lambda_1$ is the largest eigenvalue of $\bm{\Sigma}$. Analogously, $\bm{\phi}_k$ is the eigenvector of $\bm{\Sigma}$ with the eigenvalue $\lambda_k$, which is the $k$-th largest. Thus, PCs can be directly obtained by the spectral decomposition of $\bm{\Sigma}=\sum_{k=1}^{p}\lambda_k\bm{\phi}_k\bm{\phi}_k'$, where $\bm{\phi}_1,\dots,\bm{\phi}_p$ are the eigenvectors, and $\lambda_1,\dots,\lambda_p$ are the corresponding eigenvalues with $\lambda_1\geq\cdots\geq\lambda_p\geq 0$.

We summarize some properties of PCs as follows. The mean of the $k$-th PC is zero, and its variance is
\[
\mathrm{var}(y_k^*)=\bm{\phi}_k'(\bm{\Sigma}\bm{\phi}_k)=\bm{\phi}_k'(\lambda_k\bm{\phi}_k)=\lambda_k.
\]
The covariance between $k$-th and $j$-th PCs for $j\neq k$ is
\[
\mathrm{cov}(y_k^*,y_j^*)=\bm{\phi}_k'(\bm{\Sigma}\bm{\phi}_j)=\bm{\phi}_k'(\lambda_j\bm{\phi}_j)=0.
\]
In addition, the total population variance is
\[
\sum_{k=1}^{p}\mathrm{var}(y_k)=\mathrm{tr}(\bm{\Sigma})=\lambda_1+\cdots+\lambda_p=\sum_{k=1}^{p}\mathrm{var}(y_k^*).
\]
Consequently, the proportion of total variance explained by the $k$-th PC is
\[
\frac{\lambda_k}{\lambda_1+\cdots+\lambda_p},
\]
for $k=1,\dots,K$.

In practice, the covariance matrix $\bm{\Sigma}$ is usually unknown, and hence the eigenvalues $\{\lambda_1,\dots,\lambda_p\}$ of $\bm{\Sigma}$ and the corresponding eigenvectors $\{\bm{\phi}_1,\dots,\bm{\phi}_p\}$ must be estimated from observed data. Suppose that $\bm{Y}=(\bm{y}_1,\dots,\bm{y}_n)'$ is an $n\times p$ data matrix where $\bm{y}_1,\dots,\bm{y}_n\sim(\bm{0},\bm{\Sigma})$ are a random sample of size $n$ with mean 0 and covariance matrix $\bm{\Sigma}$. Then, we can estimate the covariance matrix $\bm{\Sigma}$ by the sample covariance matrix $\bm{S}=\bm{Y}'\bm{Y}/n$, and use the $k$-th eigenvalue and eigenvector of $\bm{S}$, denoted by $\tilde{\lambda}_k$ and $\tilde{\bm{\phi}}_k$, as the estimates of $\lambda_k$ and $\bm{\phi}_k$, where $\tilde{\lambda}_1\geq\cdots\geq\tilde{\lambda}_p\geq 0$. Thus, the $k$-th PC estimate is $\tilde{y}_k^*=\bm{Y}\tilde{\bm{\phi}}_k$; $k=1,\dots,K$.

\section{Maximum Covariance Analysis}

Maximum covariance analysis has been known as inter-battery factor analysis in psychology \citep{tucker1958inter}. The purpose of MCA is to investigate the relationship between two groups of variables. Specifically, MCA can be defined as the problem of finding a set of ordered pairs of basis vectors such that the covariance between the projections of the variables onto these basis vectors is successively maximized.

Let $\bm{y}_1=(y_{11},\dots,y_{1p_1})'$ and $\bm{y}_2=(y_{21},\dots,y_{2p_2})'$ be $p_1$-dimensional and $p_2$-dimensional zero-mean random vectors with covariance matrices $\bm{\Sigma}_{11}$ and $\bm{\Sigma}_{22}$ and cross-covariance matrix $\bm{\Sigma}_{12}=\mathrm{cov}(\bm{y}_1,\bm{y}_2)$. Without loss of generality, we assume $p_1\leq p_2$. MCA is concerned with finding orthogonal transformations of $\bm{y}_1$ and $\bm{y}_2$ such that the resulting pairs of components have the largest covariances. Specifically, the first pair of components is given by $y_{11}^*=\bm{u}_1'\bm{y}_1$ and $y_{21}^*=\bm{v}_1'\bm{y}_2$, where $\bm{u}_1=(u_{11},\dots,u_{1p_1})'$ and $\bm{v}_1=(v_{11},\dots,v_{1p_2})'$ are determined by
\[
(\bm{u}_1,\bm{v}_1)=\arg\max_{\bm{u}\in\mathbb{R}^{p_1},\,\bm{v}\in\mathbb{R}^{p_2}} \bm{u}'\bm{\Sigma}_{12}\bm{v},
\]
subject to $\bm{u}'\bm{u}=\bm{v}'\bm{v}=1$. The $k$-th pair of components for $k=2,\dots,K\leq p_1$ is successively obtained by $y_{1k}^*=\bm{u}_k'\bm{y}_1$ and $y_{2k}^*=\bm{v}_k'\bm{y}_2$, where
\[
(\bm{u}_k,\bm{v}_k)=\arg\max_{\bm{u}\in\mathbb{R}^{p_1},\,\bm{v}\in\mathbb{R}^{p_2}} \bm{u}'\bm{\Sigma}_{12}\bm{v},
\]
subject to $\bm{u}'\bm{u}=\bm{v}'\bm{v}=1$ and $\bm{u}_j'\bm{u}=\bm{v}_j'\bm{v}=0$, $j=1,\dots,k-1$.

The maximizers $(\bm{u}_1,\bm{v}_1),\dots,(\bm{u}_K,\bm{v}_K)$ can be solved by the Lagrange multipliers technique. For $K=1$, the Lagrange function is
\[
\bm{u}'\bm{\Sigma}_{12}\bm{v}-\lambda(\bm{u}'\bm{u}-1)-\theta(\bm{v}'\bm{v}-1),
\]
where $\lambda$ and $\theta$ are the Lagrange multipliers. Differentiating the above expression, equating it to 0, and evaluating it at $(\bm{u},\bm{v})=(\bm{u}_1,\bm{v}_1)$, give
\begin{align}
\bm{\Sigma}_{12}\bm{v}_1 &= \lambda_1\bm{u}_1, \label{eq:ch2mca1}\\
\bm{\Sigma}_{12}'\bm{u}_1 &= \theta_1\bm{v}_1, \label{eq:ch2mca2}
\end{align}
for some $\lambda_1,\theta_1\in\mathbb{R}$. Multiplying the left-hand side of \eqref{eq:ch2mca1} by $\bm{u}_1'$ and the left-hand side of \eqref{eq:ch2mca2} by $\bm{v}_1'$ gives
\[
\bm{u}_1'\bm{\Sigma}_{12}\bm{v}_1=\lambda_1=\theta_1=d_1,
\]
since $\bm{u}_1'\bm{u}_1=\bm{v}_1'\bm{v}_1=1$. Then, \eqref{eq:ch2mca1} and \eqref{eq:ch2mca2} become
\[
\bm{\Sigma}_{12}\bm{\Sigma}_{12}'\bm{u}_1=d_1^2\bm{u}_1,
\]
\[
\bm{\Sigma}_{12}'\bm{\Sigma}_{12}\bm{v}_1=d_1^2\bm{v}_1.
\]
It follows that $\bm{u}_1$ and $\bm{v}_1$ are the eigenvectors of $\bm{\Sigma}_{12}\bm{\Sigma}_{12}'$ and $\bm{\Sigma}_{12}'\bm{\Sigma}_{12}$ sharing the same eigenvalue $d_1^2$. Since $\bm{u}'\bm{\Sigma}_{12}\bm{v}$ is maximized at $\bm{u}_1'\bm{\Sigma}_{12}\bm{v}_1=d_1$, $d_1^2$ is the largest eigenvalue of $\bm{\Sigma}_{12}\bm{\Sigma}_{12}'$ and $\bm{\Sigma}_{12}'\bm{\Sigma}_{12}$. Similarly, $\bm{u}_k$ and $\bm{v}_k$ are the eigenvectors of $\bm{\Sigma}_{12}\bm{\Sigma}_{12}'$ and $\bm{\Sigma}_{12}'\bm{\Sigma}_{12}$ with the same eigenvalue $d_k^2$, which is the $k$-th largest. Thus, we obtain
\[
\bm{\Sigma}_{12}'\bm{U}=\bm{V}\bm{D},
\]
where $\bm{U}=(\bm{u}_1,\dots,\bm{u}_{p_1})$ is a $p_1\times p_1$ orthogonal matrix formed by the $p_1$ eigenvectors of $\bm{\Sigma}_{12}\bm{\Sigma}_{12}'$, $\bm{V}=(\bm{v}_1,\dots,\bm{v}_{p_1})$ is a $p_2\times p_1$ matrix formed by the first $p_1$ eigenvectors of $\bm{\Sigma}_{12}'\bm{\Sigma}_{12}$, and $\bm{D}=\mathrm{diag}(d_1,\dots,d_{p_1})$ is a diagonal matrix with $d_1\geq\cdots\geq d_{p_1}$. It follows that
\[
\bm{\Sigma}_{12}=\bm{U}\bm{D}\bm{V}'=\sum_{k=1}^{p_1}d_k\bm{u}_k\bm{v}_k',
\]
which is known as the singular value decomposition (SVD) of $\bm{\Sigma}_{12}$, where the triple $\{d_k,\bm{u}_k,\bm{v}_k\}$ consist of the $k$-th singular value, the left singular vector and the right singular vector with $d_1\geq\cdots\geq d_{p_1}\geq 0$.

Therefore, the MCA pairs of components $(y_{11}^*,y_{21}^*),\dots,(y_{1K}^*,y_{2K}^*)$ can be directly found by the SVD of $\bm{\Sigma}_{12}$. The covariance between $y_{1j}^*$ and $y_{2k}^*$ is $\mathrm{cov}(y_{1j}^*,y_{2k}^*)=\bm{u}_j'\bm{\Sigma}_{12}\bm{v}_k=d_k$ if $j=k$, and 0 if $j\neq k$. Note that since maximum covariance analysis can be solved by SVD, it has also been called SVD analysis \citep{bretherton1992intercomparison}.

In practice, the cross-covariance matrix $\bm{\Sigma}_{12}$ is usually unknown, and hence $\bm{U}$, $\bm{D}$, and $\bm{V}$ have to be estimated from observed data. Let $\bm{Y}_j=(\bm{y}_{j1},\dots,\bm{y}_{jn})'$ be $n\times p_j$ data matrices for $j=1,2$, where
\[
\begin{pmatrix}\bm{y}_{11}\\\bm{y}_{21}\end{pmatrix},\dots,
\begin{pmatrix}\bm{y}_{1n}\\\bm{y}_{2n}\end{pmatrix}
\sim
\left(
\begin{pmatrix}\bm{0}\\\bm{0}\end{pmatrix},
\begin{pmatrix}\bm{\Sigma}_{11}&\bm{\Sigma}_{12}\\\bm{\Sigma}_{12}'&\bm{\Sigma}_{22}\end{pmatrix}
\right)
\]
are a random sample of size $n$ with mean 0, covariance matrices $\bm{\Sigma}_{11}$ and $\bm{\Sigma}_{22}$, and cross-covariance matrix $\bm{\Sigma}_{12}$. Then, we can estimate the cross-covariance matrix by the sample cross-covariance matrix $\bm{S}_{12}=\bm{Y}_1'\bm{Y}_2/n$, and the SVD of $\bm{S}_{12}=\tilde{\bm{U}}\tilde{\bm{D}}\tilde{\bm{V}}'$ provides estimates of $\bm{U}$, $\bm{D}$, and $\bm{V}$.


\chapter{Smoothing Splines and Lasso}\label{ch:ch3}

\section{Smoothing Splines}
The smoothing spline method is commonly applied in nonparametric regression problems. The main idea of smoothing splines is to find a smooth function that fits the data appropriately on finite data locations. It can be achieved by using a roughness penalty to promote smoothness.

Consider a nonparametric regression model:
\[
y_i = \varphi(\bm{s}_i)+\epsilon_i; \mbox{~ $i=1,\dots,n$},
\]
where $\varphi(\cdot)$ is an unknown smooth function, $\bm{s}_1,\dots,\bm{s}_n\in D\subset \mathbb{R}^d$ are $n$ given data locations, and $\{\epsilon_i\}$ are zero-mean independent random errors with a common variance $\sigma^2$. The smoothing spline estimate $\hat{\varphi}_{\tau}(\cdot)$ of $\varphi(\cdot)$ is given by minimizing the following penalized least squares:
\[
\sum_{i=1}^n (y_i-\varphi(\bm{s}_i))^2+\tau J(\varphi),
\]
over all twice differentiable functions of $\varphi(\cdot)$, where
\[
J(\varphi)=\sum_{z_1+\cdots+z_d=2}\int_{\mathcal{R}^d}\left(
\frac{\partial^2 \varphi(\bm{s})}{\partial x_1^{z_1}\dots\partial x_d^{z_d}}\right)^2 d\bm{s}
\]
is a roughness penalty, $\bm{s}=(x_1,\dots,x_d)'$, and $\tau\geq 0$ is a smoothness parameter controlling a tradeoff between the goodness of fit and smoothness of $\varphi(\cdot)$. For example, when $\tau$ is larger, $\hat{\varphi}_\tau(\cdot)$ tends to be smoother and vice versa. Two special cases are: (i) when $\tau = 0$, $\hat{\varphi}_\tau(\cdot)$ can be any interpolating function; (ii) when $\tau = \infty$, $\hat{\varphi}_\tau(\cdot)$ is a linear line (or surface) with zero curvature.

According to \citet{nonparametric}, $\hat{\varphi}_\tau(\cdot)$ is a natural cubic spline when $d=1$,
and a thin-plate spline when
$d\in\{2,3\}$ with locations at $\{\bm{s}_1,\dots,\bm{s}_n\}$. Specifically,
\begin{equation}\label{eq:ch3_basis}
\hat{\varphi}_\tau(\bm{s})=\sum_{i=1}^n {a}_i g(\|\bm{s}-\bm{s}_i\|)+b_0+\sum_{j=1}^{d} {b}_j x_j\:,
\end{equation}
\noindent where $\bm{s}=(x_1,\dots,x_d)'$,
\[
g(r) = \left\{
\begin{array}{ll}
\displaystyle\frac{1}{16\pi}r^{2}\log{r};  & \mbox{if $d=2$,}\smallskip\\
\displaystyle\frac{\Gamma(d/2-2)}{16\pi^{d/2}}r^{4-d}; & \mbox{if }d=1,3,\\
\end{array}\right.
\]
and the coefficients ${\bm{a}}=\left({a}_1,\dots,{a}_n\right)'$
and ${\bm{b}}=\left({b}_0,b_1,\dots,{b}_{d}\right)'$ satisfy
\[
{\left(\begin{array}{cc}
	\bm{G} & \bm{E} \\
	\bm{E}' & \bm{0} \\
	\end{array}\right) \left(\begin{array}{c}
	{\bm{a}}\\
	{\bm{b}}
	\end{array}\right)=\left(\begin{array}{c}
	\hat{\bm{\phi}}_\tau\\
	\bm{0}
	\end{array}\right)}.
\]
Here $\bm{G}$ is an $n\times n$ matrix with the $(i,j)$-th element $g(\|\bm{s}_i-\bm{s}_j\|)$,
$\bm{E}$ is an $n\times (d+1)$ matrix with the $i$-th row $(1,\bm{s}'_{i})$, and $\hat{\bm{\phi}}_\tau = (\hat{\varphi}_{\tau}(\bm{s}_1),\dots,\hat{\varphi}_{\tau}(\bm{s}_n))'$.
Consequently, $\hat{\varphi}_\tau(\cdot)$ in (\ref{eq:ch3_basis}) can be expressed in terms of $\hat{\bm{\phi}}_\tau$.
Since the roughness penalty can also be written as
\begin{equation*}
\label{eq:smoothness}
J(\varphi) =\bm{\phi}' \bm\Omega \bm{\phi},
\end{equation*}
\noindent where ${\bm{\phi}} = ({\varphi}(\bm{s}_1),\dots,{\varphi}(\bm{s}_n))'$, and $\bm\Omega = \bm{G}^{-1}- \bm{G}^{-1}\bm{E}'(\bm{E}\bm{G}^{-1}\bm{E}')^{-1}\bm{E}\bm{G}^{-1}$ \citep{nonparametric}, the resulting estimate $\hat{\bm{\phi}}_\tau$ of $\bm{\phi}$ is given by
\begin{equation}
\hat{\bm{\phi}}_\tau=\mathop{\arg\min}_{\bm{\phi} \in\mathbb{R}^n} \sum_{i=1}^n (y_i-\varphi(\bm{s}_i))^2+\tau \bm{\phi}' \bm\Omega \bm{\phi} = {\bm{A}(\tau)}\bm{y},
\label{eq:ch3objective2}
\end{equation}
where {$\bm{A}(\tau) = (\bm{I}+\tau\bm{\Omega})^{-1}$ and} $\bm{y}= (y_1,\dots,y_n)'$.

We introduce two methods to select the smoothness parameter $\tau$: $M$-fold cross-validation (CV) and generalized cross-validation. The $M$-fold CV randomly partitions the index set $\{1,\dots,n\}$ into $M$ parts that are as close to the same size as possible.
For $m=1,\dots,M$, let $\bm{y}^{(m)}$ be the sub-vector of $\bm{y}$ corresponding to the $m$-th part, and let $\hat{\bm{\phi}}^{(-m)}_{\tau}$ be the estimate of $\bm{\phi}$ obtained from \eqref{eq:ch3objective2} based on the remaining data $\bm{y}^{(-m)}$. The CV criterion is
\begin{equation}\label{eq:ch3mfoldcv}
\mathrm{CV}(\tau) = \frac{1}{M}\sum_{m=1}^M\left\|\bm{y}^{(m)}-\hat{\bm{\phi}}^{(-m)}_{\tau}\right\|^2_2,
\end{equation}
which selects $\hat{\tau}=\displaystyle\mathop{\arg\min}_{\tau{\geq 0}}\mathrm{CV}(\tau)$. When $M=n$, the criterion \eqref{eq:ch3mfoldcv} is called leave-one-out CV:
\begin{equation}\label{eq:ch3leaveoneoutcv}
\mathrm{CV}(\tau) = \frac{1}{n}\sum_{m=1}^n\left\|\bm{y}^{(m)}-\hat{\bm{\phi}}^{(-m)}_{\tau}\right\|^2_2 =  \frac{1}{n}\sum_{{i}=1}^n\left(\frac{y_i-{\hat{\varphi}_\tau(\bm{s}_i)}}{1-\bm{A}_{ii}(\tau)} \right)^2,
\end{equation}
where $\bm{A}_{ii}(\tau)$ is the $(i,i)$-th entry of $\bm{A}(\tau)$, and the second equality can be derived by an argument in \citet[page 31-33]{nonparametric} to facilitate computation.

\citet{craven1978smoothing} proposed a modified version of \eqref{eq:ch3leaveoneoutcv} called
generalized cross-validation (GCV). The idea of GCV is to replace the denominators $1-\bm{A}_{ii}(\tau)$ in the last term of \eqref{eq:ch3leaveoneoutcv} with their average  {$1-\mathrm{tr}{(\bm{A}(\tau))}/n$}. Specifically, the GCV criterion is
\begin{equation*}
\mathrm{GCV}(\tau) = \frac{1}{n}\sum_{{i}=1}^n\left(\frac{y_i-{\hat{\varphi}_\tau(\bm{s}_i)}}{1-\mathrm{tr}{(\bm{A}(\tau))}/n} \right)^2,
\end{equation*}
and the smoothness parameter selected by GCV is $\tilde{\tau}=\displaystyle\mathop{\arg\min}_{\tau{\geq 0}}\mathrm{GCV}(\tau)$.

\section{Lasso}
Lasso was developed by \citet{lasso} in a linear regression model. The goal of Lasso is to offer interpretable, stable, and sparse regression estimation, particularly for high-dimensional data. The idea is to reduce the estimation variability by applying the $L_1$ penalty to least squares, forcing certain regression coefficients to be exactly zeros.
Consider a linear regression model:
\begin{equation}\label{eq:ch3lassoreg}
y_i = \bm{x}'_i\bm{\beta}+{\epsilon}_i; \mbox{~$i=1,\dots,n$,}
\end{equation}
where $\bm{x}_i = (x_{i1},\dots,x_{ip})'$ is a $p\times 1$ covariate vector, $\bm{\beta} = (\beta_1,\dots,\beta_p)'$ is a $p\times 1$ unknown regression parameter vector, and $\epsilon_1,\dots,\epsilon_n$ are independent and identically distributed random errors with mean zero and variance $\sigma^2$. Assume that  $y_i$ and $\bm{x}_i$ are centered, so the intercept term is zero. The Lasso estimate $\hat{\bm{\beta}}_\lambda$ of $\bm{\beta}$ is given by
\begin{equation}
\hat{\bm{\beta}}_\lambda=\mathop{\arg\min}_{\bm{\beta}\in \mathbb{R}^p}\left\{\sum_{i=1}^n(y_i-\bm{x}'_i\bm{\beta})^2+\lambda \sum_{j=1}^p |\beta_j|\right\},
\label{eq:ch3lasso}
\end{equation}
where the second term is the $L_1$ penalty, and $\lambda\geq 0$ is a sparseness parameter controlling the size of $\bm{\beta}$. For example, $\hat{\bm{\beta}}_{{0}}$ is the ordinary least squares estimate. For larger $\lambda$, more elements of $\hat{\bm{\beta}}_\lambda$ are shrunk to be exactly zeros, and eventually  $\hat{\bm{\beta}}_{{\infty}} = \bm{0}$.

Let $\bm{y}=(y_1,\dots,y_n)'$ and $\bm{X} = (\bm{x}_1,\dots,\bm{x}_n)'$. Assume that $\bm{X}$ is orthonormal (i.e., $\bm{X}'\bm{X} = \bm{I}_p$). Then \eqref{eq:ch3lasso} has a closed form solution given by soft-thresholding
\[
\hat{\bm{\beta}}_\lambda= \left\{\mathrm{sign}\left(\tilde{\beta}_j\right)\max\left(\left|\tilde{\beta}_j\right|-\lambda,0\right)\right\}_{j=1,\dots,p},
\]
where $\tilde{\bm{\beta}} =\bm{X}'\bm{y} =(\tilde{\beta}_1,\dots,\tilde{\beta}_p)'$ is the ordinary least squares estimate. While $\hat{\bm{\beta}}_\lambda$ generally does not have an explicit form, it can be computed by one of the following two fast algorithms: the LARS algorithm \citep{efron2004least} and the coordinate descent algorithm \citep{coord}.

The $M$-fold CV can be used to select $\lambda$. First, the index set $\{1,\dots,n\}$ is randomly decomposed into $M$ parts of roughly equal size. For $m=1,\dots,M$, let $\bm{y}^{(m)}$ and $\bm{X}^{(m)}$ be the sub-matrices of $\bm{y}$ and $\bm{X}$ corresponding to the $m$-th part, respectively, and let $\hat{\bm{\beta}}^{(-m)}_\lambda$ be the estimate of $\bm{\beta}$ obtained from \eqref{eq:ch3lasso} based on the remaining data $\left(\bm{y}^{(-m)}, \bm{X}^{(-m)}\right)$. The CV criterion is
\begin{equation*}
\mathrm{CV}({\lambda}) = \frac{1}{M}\sum_{m=1}^M\left\|\bm{y}^{(m)}-\bm{X}^{(m)}\hat{\bm{\beta}}^{(-m)}_{\lambda}\right\|^2_2,
\end{equation*}
which selects
$\hat{\lambda}=\displaystyle\mathop{\arg\min}_{\lambda\geq 0}\mathrm{CV}(\lambda)$.

In addition, \citet{zou2007degrees} proposed another method for selecting $\lambda$, which is based on the information criteria with the  degrees of freedom of the Lasso fit  $\bm{X}\hat{\bm{\beta}}_\lambda$ estimated by the Stein's unbiased risk estimate \citep{stein1981estimation}. Specifically, selecting $\lambda$ is given by minimizing the following criterion:
\begin{equation}\label{eq:ch3ic}
\frac{\|\bm{y}-\bm{X}\hat{\bm{\beta}}_\lambda\|^2_2}{n\sigma^2} + \frac{\omega_n}{n}\hat{\mathrm{df}}(\bm{X}\hat{\bm{\beta}}_\lambda),
\end{equation}
where $\hat{\mathrm{df}}(\bm{X}\hat{\bm{\beta}}_\lambda)$ is equal to the number of non-zero elements of $\hat{\bm{\beta}}_\lambda$, and $\epsilon_i\sim N(0,\sigma^2)$ in \eqref{eq:ch3lassoreg} are assumed to follow a normal distribution. When $\omega_n = 2$, the criterion \eqref{eq:ch3ic} is called Akaike's information criterion \citep{akaike1974new}. When $\omega_n = \log(n)$, it is called Bayesian information criterion \citep{schwarz1978estimating}. Both criteria have different asymptotic properties under different model assumptions.
More details of \eqref{eq:ch3ic} can be found in \citet{zou2007degrees}.


\chapter{Regularized Principal Component Analysis for Spatial Data}\label{ch:ch4}

\graphicspath{{Chapter4/plot/}}

\section{Introduction}

In many atmospheric and earth sciences, it is of interest to identify dominant spatial patterns of variation
based on data observed at $p$ locations with $n$ repeated measurements, where $p$ may be larger than $n$.
The dominant patterns are the eigenimages of the underlying (nonstationary) spatial covariance function
with large eigenvalues. A commonly used approach for estimating the eigenimages is the principal component analysis (PCA),
also known as the empirical orthogonal function analysis in atmospheric science.
However, when $p$ is large relative to $n$, the leading eigenimages produced from PCA may be noisy with high estimation variability,
or exhibit some bizarre patterns that are not physically meaningful.
To enhance the interpretability, a few approaches, such as rotation of components according to some criteria (see e.g., \citet{rotate},
\citet{jolliffe1987rotation}, \citet{richman1987rotation}),
have been proposed to form more desirable patterns. However, how to obtain a desired rotation in practice is not completely clear;
see some discussion in \citet{eofreview}.

Another approach to aid interpretation is to seek sparse or spatially localized patterns, which can be done by imposing an $L_1$ constraint
or adding an $L_1$ penalty to an original PCA optimization formulation (\citet{jolliffe2002eofs}, \citet{spca},  \citet{Shen}, \citet{spca2}, and \citet{spca3}).
However, except \citet{jolliffe2002eofs} and \citet{spca3}, the PC estimates produced from these approaches may not produce orthogonal PC loadings.
Consequently, a larger number of patterns may be needed in order to account for the total variation of the signal in a dataset.

For continuous spatial domains, the problem becomes even more challenging.
Instead of looking for eigenimages, we need to find eigenfunctions by essentially solving an infinite dimensional problem
based on data observed at possibly sparse and irregularly spaced locations.
Although some approaches have been developed using functional principal component analysis
(see e.g., \citet{bfda}, \citet{fpca_long} and \citet{fda}), they typically focus on one-dimensional processes, or require data observed at dense locations.
Reviews of PCA on spatial data can be found in \citet{eofreview} and \citet{spatial_pca}.

In this chapter, we propose a regularization approach for estimation of dominant patterns,
taking into account smoothness and localized features that are expected in real-world spatial processes.
The proposed dominant pattern estimates are directly obtained by solving a minimization problem. We call our method SpatPCA,
which not only gives effective estimates of dominant patterns,
but also provides an ideal set of basis functions for estimating the underlying (nonstationary) spatial covariance function,
even when data are irregularly located in space.
In addition, we develop an algorithm for solving the resulting nonconvex optimization problem
using the alternating direction method of multipliers (ADMM) (see \citet{admm}). The algorithm is fast and easy to implement.
An R package called SpatPCA is available on the Comprehensive R Archive Network (CRAN).

The rest of this chapter is organized as follows. In Section~\ref{sec:ch4proposal}, we introduce the proposed SpatPCA method,
including dominant patterns estimation and spatial covariance function estimation.
Our ADMM algorithm for computing the SpatPCA estimates is provided in Section~\ref{sec:ch4algorithm}.
Some simulation experiments that illustrate the superiority of SpatPCA and an application of SpatPCA to a global sea surface
temperature dataset are presented in Section~\ref{sec:ch4numerical}.

\section{The Proposed Method}
\label{sec:ch4proposal}

Consider a sequence of zero-mean $L^2$-continuous spatial processes, $\{\eta_i(\bm{s}); \bm{s} \in D\}$; $i=1,\dots,n$,
defined on a spatial domain $D \subset \mathbb{R}^d$,
which are mutually uncorrelated, and have a common spatial covariance function,
$C_\eta(\bm{s},\bm{s}^*)=\mathrm{cov}(\eta_i(\bm{s}),\eta_i(\bm{s}^*))$.
We consider a rank-$K$ spatial random-effects model for $\eta_i(\cdot)$:
\begin{equation*}
\eta_i(\bm{s})=(\varphi_1(\bm{s}),\dots,\varphi_{K}(\bm{s}))\bm{\xi}_i=
\sum_{k=1}^{K}\xi_{ik}\varphi_k(\bm{s});\quad\bm{s}\in D,\quad i=1,\dots,n,
\end{equation*}
where $\{\varphi_k(\cdot)\}$ are unknown orthonormal basis functions,
$\bm{\xi}_i=(\xi_{i1},\dots,\xi_{iK})'\sim(0,\bm{\Lambda})$; $i=1,\dots,n$, are uncorrelated random variables,
and $\bm{\Lambda}$ is an unknown nonnegative-definite matrix, denoted by $\bm{\Lambda}\succeq\bm{0}$.
A similar model based on given $\{\varphi_k(\cdot)\}$ was introduced by \citet{fixedrank}
and in a Bayesian framework by \citet{bayes}.

Let $\lambda_{kk'}$ be the $(k,k')$-th entry of $\bm{\Lambda}$. Then the spatial covariance function
of $\eta_i(\cdot)$ is:
\begin{equation}
C_\eta(\bm{s},\bm{s}^*)=\mathrm{cov}(\eta_i(\bm{s}),\eta_i(\bm{s}^*))=
\sum_{k=1}^{K}\sum_{k'=1}^{K}	\lambda_{kk'}\varphi_k(\bm{s})\varphi_{k' }(\bm{s}^*).
\label{eq:ch4eta}
\end{equation}

Note that $\bm{\Lambda}$ is not restricted to be a diagonal matrix.

Let $\bm{\Lambda}=\bm{V}\bm{\Lambda^*}\bm{V}'$ be the eigen-decomposition of $\bm{\Lambda}$,
where $\bm{V}$ consists of $K$ orthonormal vectors, and  $\bm{\Lambda}^*=\mathrm{diag}(\lambda^*_1,\dots,\lambda^*_{K})$ and
$\lambda^*_1\geq\cdots\geq\lambda^*_{K}$.
Let $\bm{\xi}^*_i=\bm{V}'\bm{\xi}_i$ and
\[
(\varphi^*_1(\bm{s}),\dots,\varphi^*_{K}(\bm{s}))
=(\varphi_1(\bm{s}),\dots,\varphi_{K}(\bm{s}))\bm{V};\quad\bm{s}\in D.
\]
Then $\{\varphi^*_k\}$ are also orthonormal,
and $\xi^*_{ik}\sim(0,\lambda^*_k)$; $i=1,\dots,n,\,k=1,\dots,K$, are mutually uncorrelated.
Therefore, we can rewrite $\eta_i(\cdot)$ in terms of $\varphi^*_k(\cdot)$'s:
\begin{equation}
\label{eq:ch4kl}
\eta_i(\bm{s})=(\varphi^*_1(\bm{s}),\dots,\varphi^*_{K}(\bm{s}))\bm{\xi}^*_i=
\sum_{k=1}^{K}\xi^*_{ik}\varphi^*_k(\bm{s});\quad\bm{s}\in D.
\end{equation}

The above expansion is known as the Karhunen-Lo\'{e}ve expansion
of $\eta_i(\cdot)$ (\citet{karhunen}; \citet{loeve}) with $K$ nonzero eigenvalues,
where $\varphi^*_k(\cdot)$ is the $k$-th eigenfunction of $C_\eta(\cdot,\cdot)$ with $\lambda^*_k$ the corresponding eigenvalue.

Suppose that we observe data $\bm{Y}_i=(Y_i(\bm{s}_1),\dots,Y_i(\bm{s}_p))'$ with
added white noise $\bm{\epsilon}_i\sim(\bm{0},\sigma^2\bm{I})$ at $p$ spatial locations, $\bm{s}_1,\dots,\bm{s}_p \in D$,  according to
\begin{equation}
\bm{Y}_i=\bm{\eta}_i+\bm{\epsilon}_i=\bm{\Phi}\bm{\xi}_i+\bm{\epsilon}_i;\quad i=1,\dots,n,
\label{eq:ch4measurement}
\end{equation}
where $\bm{\eta}_i=(\eta_{i}(\bm{s}_1),\dots,\eta_{i}(\bm{s}_p))'$,
$\bm{\Phi}=(\bm{\phi}_1,\dots,\bm{\phi}_K)$ is a $p\times K$ matrix with the $(j,k)$-th entry $\varphi_k(\bm{s}_j)$, and
$\bm{\epsilon}_i$'s and $\bm{\xi}_{i}$'s are uncorrelated. Our goal is to identify
the first $L\leq K$ dominant patterns, $\varphi_1(\cdot),\dots,\varphi_{L}(\cdot)$, with relatively large
$\lambda^*_1,\dots,\lambda^*_{L}$.
Additionally,  we are interested in estimating $C_\eta(\cdot,\cdot)$, which is essential for spatial prediction.

Let $\bm{Y}=(\bm{Y}_1,\dots,\bm{Y}_n)'$ be the $n\times p$ data matrix. Throughout the chapter, we assume that the mean of $\bm{Y}$ is known as zero.
So the sample covariance matrix of $\bm{Y}$ is $\bm{S}=\bm{Y}'\bm{Y}/n$.
A popular approach for estimating $\{\varphi^*_k(\cdot)\}$ is PCA, which
estimates $(\varphi^*_k(\bm{s}_1),\dots,\varphi^*_k(\bm{s}_p))'$ by $\tilde{\bm{\phi}}_k$,
the $k$-th eigenvector of $\bm{S}$, for $k=1,\dots,K$.
Let $\tilde{\bm{\Phi}}=\big(\tilde{\bm{\phi}}_1,\dots,\tilde{\bm{\phi}}_K\big)$ be a $p\times K$ matrix
formed by the first $K$ principal component loadings.
Then $\tilde{\bm{\Phi}}$ satisfies the following constrained optimization problem:
\[
\min_{\bm{\Phi}}\|\bm{Y}-\bm{Y}\bm{\Phi}\bm{\Phi}'\|^2_F\quad\text{subject to }
\bm{\Phi}'\bm{\Phi}=\bm{I}_K,
\]
where $\bm{\Phi}=(\bm{\phi}_1,\dots,\bm{\phi}_K)$ and $\|\bm{M}\|_F=
\Big(\displaystyle\sum_{i,j}m^2_{ij}\Big)^{1/2}$ is the Frobenius norm of a matrix $\bm{M}$.
Unfortunately, $\tilde{\bm{\Phi}}$ tends to have high estimation variability when $p$ is large, $n$ is small,
or $\sigma^2$ is large, due to excessive number of parameters.
Consequently, the patterns of $\tilde{\bm{\Phi}}$ may be too noisy to be physically interpretable.
In addition, for a continuous spatial domain $D$,
we also need to estimate $\varphi_k^*(\bm{s})$'s for locations with no data observed (i.e., $\bm{s}\notin\{\bm{s}_1,\dots,\bm{s}_p\}$),
which requires applying some interpolation and extrapolation methods; see e.g., Section 12.4 and 13.6 of \citet{jolliffe2002principal}.

\subsection{Regularized Spatial PCA}

To prevent high estimation variability of PCA, we adopt a regularization approach by minimizing the following objective function:
\begin{equation}\label{eq:ch4objective}
\|\bm{Y}-\bm{Y}\bm{\Phi}\bm{\Phi}'\|^2_F +\tau_1\sum_{k=1}^K
J({\varphi}_k)+\tau_2\sum_{k=1}^K \sum_{j=1}^{p}\big|\varphi_k(\bm{s}_j)\big|,
\end{equation}
over $\varphi_1(\cdot),\dots,\varphi_K(\cdot)$, subject to $\bm{\Phi}'\bm{\Phi}=\bm{I}_K$
and $\bm{\phi}'_1\bm{S}\bm{\phi}_1\geq\bm{\phi}'_2\bm{S}\bm{\phi}_2\geq\cdots\geq\bm{\phi}'_K
\bm{S}\bm{\phi}_K$, where
\[
J(\varphi)=\sum_{z_1+\cdots+z_d=2}\int_{\mathbb{R}^d}\left(
\frac{\partial^2 \varphi(\bm{s})}{\partial x_1^{z_1}\dots\partial x_d^{z_d}}\right)^2 d\bm{s},
\]
is a roughness penalty, $\bm{s}=(x_1,\dots,x_d)'$, $\tau_1\geq 0$ is a smoothness parameter, and $\tau_2\geq 0$ is a sparseness parameter.
The function \eqref{eq:ch4objective} consists of two penalty terms for control of estimation variability.
The first one is designed to enhance smoothness of $\varphi_k(\cdot)$ through the smoothing spline penalty $J(\varphi_k)$,
while the second one is the $L_1$ Lasso penalty (\citet{lasso}), used to promote sparse and localized patterns.
When $\tau_1$ is larger, $\hat{\varphi}_k(\cdot)$'s tend to be smoother and \textit{vice versa}.
When $\tau_2$ is larger, $\hat{\varphi}_k(\cdot)$'s are forced to become more localized, because $\hat{\varphi}_k(\bm{s})$'s are forced
to be zero at some $\bm{s}\in D$.
On the other hand, when both $\tau_1$ and $\tau_2$ are close to zero,
the estimates are very close to those obtained from PCA.
By suitably choosing $\tau_1$ and $\tau_2$, we can obtain a good compromise among
goodness of fit, smoothness of the eigenfunctions, and sparseness of the eigenfunctions, leading to more interpretable results.
Note that the orthogonal constraint, which causes some computational difficulty,
is not considered by many PCA regularization methods (e.g., \citet{spca}, \citet{Shen}, \citet{fused}, \citet{hong}).

Although $J(\varphi)$ involves integration, it is well known from the theory of smoothing splines (\citet{nonparametric}) that
for each $k=1,\dots,K$, $\hat{\varphi}_k(\cdot)$ has to be a natural cubic spline when $d=1$,
and a thin-plate spline when $d= 2,3$, with nodes at $\{\bm{s}_1,\dots,\bm{s}_p\}$. Specifically,
\begin{equation}\label{eq:ch4basis}
\hat{\varphi}_k(\bm{s})=\sum_{i=1}^p {a}_i g(\|\bm{s}-\bm{s}_i\|)+b_0+\sum_{j=1}^{d} {b}_j x_j\:,
\end{equation}
where $\bm{s}=(x_1,\dots,x_d)'$,
\[
g(r) = \left\{
\begin{array}{ll}
\displaystyle\frac{1}{16\pi}r^{2}\log{r};  & \mbox{if $d=2$,}\smallskip\\
\displaystyle\frac{\Gamma(d/2-2)}{16\pi^{d/2}}r^{4-d}; & \mbox{if }d=1,3,\\
\end{array}\right.
\]
and the coefficients ${\bm{a}}=\left({a}_1,\dots,{a}_p\right)'$
and ${\bm{b}}=\left({b}_0,b_1,\dots,{b}_{d}\right)'$ satisfy
\[
\begin{bmatrix}
\bm{G} & \bm{E} \\
\bm{E}' & \bm{0} \\
\end{bmatrix}
\begin{bmatrix}
{\bm{a}}\\
{\bm{b}}
\end{bmatrix}
=\begin{bmatrix}
\hat{\bm{\phi}}_k\\
\bm{0}
\end{bmatrix}.
\]
Here $\bm{G}$ is a $p\times p$ matrix with the $(i,j)$-th element $g(\|\bm{s}_i-\bm{s}_j\|)$,
and $\bm{E}$ is a $p\times (d+1)$ matrix with the $i$-th row $(1,\bm{s}'_{i})$.
Consequently, $\hat{\varphi}_k(\cdot)$ in (\ref{eq:ch4basis}) can be expressed in terms of $\hat{\bm{\phi}}_k$.
Additionally, the roughness penalty can also be written as
\begin{equation}
\label{eq:ch4smoothness}
J(\varphi_k) =\bm{\phi}'_k \bm{\Omega} \bm{\phi}_k,
\end{equation}
with $\bm{\Omega}$ a known $p\times p$ matrix determined only by $\bm{s}_1,\dots,\bm{s}_p$.
The readers are referred to \citet{nonparametric} for more details regarding smoothing splines.

From (\ref{eq:ch4objective}) and (\ref{eq:ch4smoothness}), the proposed SpatPCA estimate of $\bm{\Phi}$ can be written as:
\begin{equation}
\hat{\bm{\Phi}}_{\tau_1,\tau_2}=\mathop{\arg\min}_{\bm{\Phi}:\bm{\Phi}'\bm{\Phi} = \bm{I}_K}
\|\bm{Y} - \bm{Y}\bm{\Phi}\bm{\Phi} '\|^2_F +
\tau_1 \sum_{k=1}^K \bm{\phi}'_k\bm{\Omega}\bm{\phi}_k+\tau_2 \sum_{k=1}^K\sum_{j=1}^p\left|{\phi}_{jk}\right|,
\label{eq:ch4objective2}
\end{equation}
subject to $\bm{\phi}'_1\bm{S}\bm{\phi}_1\geq\bm{\phi}'_2\bm{S}\bm{\phi}_2\geq\cdots\geq\bm{\phi}'_K
\bm{S}\bm{\phi}_K$.
The resulting estimates of $\varphi_1(\cdot),\dots,\varphi_K(\cdot)$ then follow directly from (\ref{eq:ch4basis}).
When no confusion may arise, we shall simply write $\hat{\bm{\Phi}}_{\tau_1,\tau_2}$ as $\hat{\bm{\Phi}}$.
Note that the SpatPCA estimate of (\ref{eq:ch4objective2}) reduces to a sparse PCA estimate of \citet{spca}
if the orthogonal constraint is dropped and $\bm{\Omega}=\bm{I}$, where no spatial structure is considered.

The tuning parameters $\tau_1$ and $\tau_2$ are selected using $M$-fold cross-validation (CV).
First, we partition $\{1,\dots,n\}$ into $M$ parts with as close to the same size as possible.
Let $\bm{Y}^{(m)}$ be the sub-matrix of $\bm{Y}$ corresponding to the $m$-th part, for $m=1,\dots,M$.
For each part, we treat $\bm{Y}^{(m)}$ as the validation data,
and obtain the estimate $\hat{\bm{\Phi}}^{(-m)}_{\tau_1,\tau_2}$ of $\bm{\Phi}$ for $(\tau_1,\tau_2)\in\mathcal{A}$
based on the remaining data $\bm{Y}^{(-m)}$ using the proposed method,
where $\mathcal{A}\subset[0,\infty)^2$ is a candidate index set. The proposed CV criterion based on the
residual sum of squares is:
\begin{equation}
\label{eq:ch4cv}
\mathrm{CV}_1(\tau_1,\tau_2)=\frac{1}{M}\sum_{m=1}^M\big\|\bm{Y}^{(m)}-\bm{Y}^{(m)}\hat{\bm{\Phi}}_{\tau_1, \tau_2}^{(-m)}
(\hat{\bm{\Phi}}^{(-m)}_{\tau_1, \tau_2})'  \big\|^2_F\:,
\end{equation}
where $\bm{Y}^{(m)}\hat{\bm{\Phi}}_{\tau_1, \tau_2}^{(-m)}
(\hat{\bm{\Phi}}^{(-m)}_{\tau_1, \tau_2})' $ is the projection of $\bm{Y}^{(m)}$
onto the column space of $\hat{\bm{\Phi}}^{(-m)}_{\tau_1,\tau_2}$. The final $\tau_1$ and $\tau_2$ values are
$(\hat{\tau}_1, \hat{\tau}_2)=\displaystyle\mathop{\arg\min}_{(\tau_1,\tau_2)\in\mathcal{A}}\mathrm{CV}_1(\tau_1,\tau_2)$.

\subsection{Estimation of Spatial Covariance Function}
\label{subsec:ch4cov_fn}

For estimation of $C_\eta(\cdot,\cdot)$ in (\ref{eq:ch4eta}), we also need to estimate the spatial covariance parameters,
$\sigma^2$ and $\bm{\Lambda}$.
We apply the regularized least squares method of \citet{regularized_covariance}:
\begin{equation}
\big(\hat{\sigma}^2,\hat{\bm{\Lambda}}\big)=\mathop{\arg\min}_{(\sigma^2,\bm{\Lambda}):  \sigma^2\geq 0,\,\bm{\Lambda}\succeq\bm{0}}
\bigg\{\frac{1}{2}\big\|\bm{S}-\hat{\bm{\Phi}}\bm{\Lambda}\hat{\bm{\Phi}}'-\sigma^2\bm{I}\big\|^2_F+\gamma\|\hat{\bm{\Phi}}\bm{\Lambda}
\hat{\bm{\Phi}}'\|_{*}\bigg\},
\label{eq:ch4covariance.estimate}
\end{equation}
where $\gamma \geq 0$ is a tuning parameter, and  $\|\bm{M}\|_{*}=\mathrm{tr}((\bm{M}'\bm{M})^{1/2})$
is the nuclear norm of $\bm{M}$.
The first term of (\ref{eq:ch4covariance.estimate}) is a goodness-of-fit term, which is based on
$\mathrm{var}(\bm{Y}_i)=\bm{\Phi\Lambda\Phi}'+\sigma^2\bm{I}$. The second term of (\ref{eq:ch4covariance.estimate}) is a penalty term,
penalizing the eigenvalues of $\hat{\bm{\Phi}}\bm{\Lambda}\hat{\bm{\Phi}}'$ to avoid the eigenvalues being overestimated. By suitably choosing
a tuning parameter $\gamma$, we can control the bias, while reducing the estimation variability. This is particularly
effective when $K$ is large.

\citet{regularized_covariance} provides a closed-form solution for $\hat{\bm{\Lambda}}$, but requires an iterative procedure
for solving $\hat{\sigma}^2$.
We are able to derive closed-form expressions for both $\hat{\sigma}^2$ and $\hat{\bm{\Lambda}}$, which are given in the following proposition
with its proof given in Appendix~\ref{app:appendix_A}.
\begin{prop}
The solutions of (\ref{eq:ch4covariance.estimate}) are given by
\begin{align}
\hat{\bm{\Lambda}}
=&~ \hat{\bm{V}}\mathrm{diag}\big(\hat{\lambda}_1^*,\dots,\hat{\lambda}_K^*\big)\hat{\bm{V}}',
\label{eq:ch4Lambda.hat}\\
\hat{\sigma}^2
=&~	\left\{
 	\begin{array}{ll}
 	\displaystyle\frac{1}{p-\hat{L}} \bigg(\mathrm{tr}(\bm{S})-\sum_{k=1}^{\hat{L}} \big(\hat{d}_k-\gamma\big)\bigg);
	& \mbox{if $\hat{d}_1 > \gamma$,}\\
 	\displaystyle \frac{1}{p} \left(\mathrm{tr}(\bm{S})\right);
	& \mbox{if $\hat{d}_1 \leq \gamma$ ,}\\
 	\end{array}\right.
	\label{eq:ch4sigma.hat}
\end{align}
where $\hat{\bm{V}}\mathrm{diag}(\hat{d}_1,\dots,\hat{d}_K)\hat{\bm{V}}'$ is the
eigen-decomposition of $\hat{\bm{\Phi}}'\bm{S}\hat{\bm{\Phi}}$ with $\hat{d}_1\geq\cdots\geq\hat{d}_K$,
 \[
 \hat{L} =\max \bigg\{ L:  \hat{d}_L-\gamma >\frac{1}{p-L} \bigg(\mathrm{tr}(\bm{S})-\sum_{k=1}^L (\hat{d}_k-\gamma)\bigg), L=1,\dots,K \bigg\},
 \]
 and $\hat{\lambda}^*_k=\max(\hat{d}_k-\hat{\sigma}^2-\gamma,0)$; $k=1,\dots,K$.
\label{prop:sigma_Lambda_hat}
\end{prop}

With $\bm{\Lambda}$ estimated by $\hat{\bm{\Lambda}}=\big(\hat{\lambda}_{kk'}\big)_{K\times K}$ in (\ref{eq:ch4covariance.estimate}),
the proposed estimate of $C_\eta(\bm{s},\bm{s}^*)$ is
\begin{equation}
\hat{C}_\eta(\bm{s},\bm{s}^*)
=\sum_{k=1}^K\sum_{k'=1}^K \hat{\lambda}_{kk'}\,\hat{\varphi}_k(\bm{s})\hat{\varphi}_{k'}(\bm{s}^*),
\label{eq:ch4covariance.hat}
\end{equation}
where $\hat{\varphi}_k(\bm{s})$ is given in \eqref{eq:ch4basis}, and
the proposed estimate of $(\varphi^*_1(\bm{s}),\dots,\varphi^*_K(\bm{s}))$ is
\[
(\hat{\varphi}^*_1(\bm{s}),\dots,\hat{\varphi}^*_K(\bm{s}))=
(\hat{\varphi}_1(\bm{s}),\dots,\hat{\varphi}_K(\bm{s}))\hat{\bm{V}};\quad \bm{s}\in D.
\]
We consider $M$-fold CV to select $\gamma$.
As in the previous section, we partition the data
into $M$ parts, $\bm{Y}^{(1)},\dots,\bm{Y}^{(M)}$. For $m=1,\dots,M$, we estimate
$\mathrm{var}\big(\bm{Y}^{(-m)}\big)$ by $\hat{\bm{\Sigma}}^{(-m)}=\hat{\bm{\Phi}}^{(-m)}\hat{\bm{\Lambda}}^{(-m)}
\big(\hat{\bm{\Phi}}^{(-m)}\big)'+\big(\hat{\sigma}^2\big)^{(-m)}\bm{I}$ based on the remaining data $\bm{Y}^{(-m)}$ by removing $\bm{Y}^{(m)}$
from $\bm{Y}$, where $\hat{\bm{\Lambda}}^{(-m)}$, $\big(\hat{\sigma}^2\big)^{(-m)}$ and $\hat{\bm{\Phi}}^{(-m)}$ are the estimates of $\bm{\Lambda}$, $\sigma^2$ and $\bm{\Phi}$ based on $\bm{Y}^{(-m)}$,
and for notational simplicity, their dependences on the selected $(\tau_1,\tau_2)$ and $K$ are suppressed.
The proposed CV criterion is given by
\begin{equation}
\label{eq:ch4cv.gamma}
\mathrm{CV}_2(K,\gamma)= \frac{1}{M}\sum_{m=1}^{M}\big\|\bm{S}^{(m)} -\hat{\bm{\Phi}}^{(-m)}\hat{\bm{\Lambda}}^{(-m)}\big(\hat{\bm{\Phi}}^{(-m)}\big)'-(\hat{\sigma}^2)^{(-m)}\bm{I} \big\|^2_F\:,
\end{equation}
where $\bm{S}^{(m)}=(\bm{Y}^{(m)})'\bm{Y}^{(m)}/n$. Then the selected $\gamma$ by $\mathrm{CV}_2$ based on $K$ is
$\hat{\gamma}_K=\displaystyle\mathop{\arg\min}_{\gamma\geq 0}\mathrm{CV}_2(K,\gamma)$.

The dimension of eigen-space $K$, corresponding to the maximum rank of $\bm{\Lambda}\succeq\bm{0}$, could be
selected by traditional approaches based on a given proportion of total variation explained or
the scree plot of the sample eigenvalues. However, these approaches tend to be more subjective
and may not be effective for the covariance estimation purpose. We propose selecting $K$ using
$\mathrm{CV}_2$ of \eqref{eq:ch4cv.gamma} by subsequently increasing the value of $K$ from $K=1,2,\dots$, until no further
reduction of the $\mathrm{CV}_2$ value. Specifically, we select
\begin{equation}
\label{eq:ch4K.hat}
\hat{K}=\min\big\{K:\mathrm{CV}_2(K,\hat{\gamma}_K)\leq\mathrm{CV}_2(K+1,\hat{\gamma}_{K+1}),\; K=1,2,\dots\big\}.
\end{equation}

\section{Computation Algorithm}
\label{sec:ch4algorithm}

Solving \eqref{eq:ch4objective2} is a challenging problem especially when both the orthogonal constraint and
the $L_1$ penalty appear simultaneously.
Consequently, many regularized PCA approaches, such as sparse PCA \citep{spca}, do not cope with the orthogonal constraint.
We adopt the ADMM algorithm by decomposing the original constrained optimization problem
into small subproblems that can be easily and efficiently handled through an iterative procedure.
This type of algorithm was developed early in \citet{admm_2}, and was systematically studied by \citet{admm} more recently.

First, the optimization problem of \eqref{eq:ch4objective2} is transferred into the following equivalent problem by adding a $p\times K$
parameter matrix $\bm{Q}$:
\begin{align}\label{eq:ch4ADMM1}
& \min_{\bm{\Phi},\bm{Q}\in\mathbb{R}^{p\times K}} \|\bm{Y} - \bm{Y}\bm{\Phi}\bm{\Phi} '\|^2_F +\tau_1 \sum_{k=1}^K \bm{\phi}_{k}'\bm{\Omega}\bm{\phi}_k+
\tau_2 \sum_{k=1}^K\sum_{j=1}^{p}\left|\phi_{jk}\right|,
\end{align}
subject to $\bm{Q}'\bm{Q}=\bm{I}_K$, $\bm{\phi}'_1\bm{S}\bm{\phi}_1\geq\bm{\phi}'_2\bm{S}\bm{\phi}_2
\geq\cdots\geq\bm{\phi}'_K\bm{S}\bm{\phi}_K$, and a new constraint, $\bm{\Phi}=\bm{Q}$. Then the resulting
constrained optimization problem of \eqref{eq:ch4ADMM1} is solved using the augmented Lagrangian method with its Lagrangian given by
\begin{align*}
L(\bm{\Phi}, \bm{Q},\bm{\Gamma})
=&~ \|\bm{Y}-\bm{Y}\bm{\Phi}\bm{\Phi}'\|^2_F+\tau_1\sum_{k=1}^K \bm{\phi}_k'\bm{\Omega}\bm{\phi}_k
+\tau_2 \sum_{k=1}^K\sum_{j=1}^{p}|\phi_{jk}|\\
&~ + \mathrm{tr}(\bm{\Gamma}'(\bm{\Phi}-\bm{Q}))+\frac{\rho}{2}\|\bm{\Phi}-\bm{Q}\|^2_F\:,
\end{align*}
subject to $\bm{Q}'\bm{Q}=\bm{I}_K$ and
$\bm{\phi}'_1\bm{S}\bm{\phi}_1\geq\bm{\phi}'_2\bm{S}\bm{\phi}_2\geq\cdots\geq\bm{\phi}'_K\bm{S}\bm{\phi}_K$,
where $\bm{\Gamma}$ is a $p\times K$ matrix of the Lagrange multipliers, and
$\rho>0$ is a penalty parameter to facilitate convergence. Note that the value of $\rho$ does not affect the original optimization problem.
The ADMM algorithm iteratively updates one group of parameters at a time in both the primal and the dual spaces until convergence. Given the initial estimates, $\bm{Q}^{(0)}$ and $\bm{\Gamma}^{(0)}$ of $\bm{Q}$ and $\bm{\Gamma}$,
our ADMM algorithm consists of the following steps at the $\ell$-th iteration:
\begin{align}
\bm{\Phi}^{(\ell+1)}
=&~ \mathop{\arg\min}_{\bm{\Phi}}L\big(\bm{\Phi}, \bm{Q}^{(\ell)},\bm{\Gamma}^{(\ell)}\big)
\notag\\
=&~ \mathop{\arg\min}_{\bm{\Phi}} \sum_{k=1}^K \bigg\{\|\bm{z}_k^{(\ell)} - \bm{X}\bm{\phi}_k \|^2 + \sum_{j=1}^p \tau_2|\phi_{jk}|\bigg\}
\label{eq:ch4phi_1},\\
\bm{Q}^{(\ell+1)}
=&~ \mathop{\arg\min}_{\bm{Q}: \bm{Q}'\bm{Q}=\bm{I}_K}L\big(\bm{\Phi}^{(\ell+1)},\bm{Q},
\bm{\Gamma}^{(\ell)}\big)\,=\,\bm{U}^{(\ell)}\big(\bm{V}^{(\ell)}\big)',
\label{eq:ch4Q_1}\\
\bm{\Gamma}^{(\ell+1)}
=&~ \bm{\Gamma}^{(\ell) }+ \rho\left(\bm{\Phi}^{(\ell+1)}-\bm{Q}^{(\ell+1)}\right), \label{eq:ch4Gamma}
\end{align}
where $\bm{X} =(\tau_1\bm{\Omega}-\bm{Y}'\bm{Y}+\rho\bm{I}_p/2)^{1/2}$, $\bm{z}^{(\ell)}_k$ is the $k$-th column of $\bm{X}^{-1}(\rho\bm{Q}^{(\ell)}-\bm{\Gamma}^{(\ell)})/2$, $\bm{U}^{(\ell)}\bm{D}^{(\ell)}\big(\bm{V}^{(\ell)}\big)'$ is the singular value decomposition of
$\bm{\Phi}^{(\ell+1)}+\rho^{-1}\bm{\Gamma}^{(\ell)}$, and
$\rho$ must be chosen large enough (e.g., twice the maximum eigenvalue of $\bm{Y}'\bm{Y}$) to ensure that $\bm{X}$ is positive-definite.
Note that \eqref{eq:ch4phi_1} is simply a
Lasso problem (\citet{lasso}), which can be solved effectively using the coordinate descent algorithm \citep{coord}.

Except (\ref{eq:ch4phi_1}), the ADMM steps given by (\ref{eq:ch4phi_1})--(\ref{eq:ch4Gamma}) have closed-form solutions.
In fact, it is possible to develop an ADMM algorithm with complete closed-form updates by further incorporating (\ref{eq:ch4phi_1}) into another ADMM step.
Specifically, we can replace $\phi_{jk}$'s in the last term of \eqref{eq:ch4ADMM1} by new parameters $r_{jk}$'s, and then
add the constraint, $\phi_{jk}=r_{jk}$ for $j=1,\dots,p$ and $k=1,\dots,K$, to form an equivalent problem:
\begin{equation*}
\min_{\bm{\Phi},\bm{Q},\bm{R}} \|\bm{Y} - \bm{Y}\bm{\Phi}\bm{\Phi} '\|^2_F +\tau_1 \sum_{k=1}^K \bm{\phi}_k'\bm{\Omega}\bm{\phi}_k+
\tau_2 \sum_{k=1}^K\sum_{j=1}^{p}\left|r_{jk}\right|,
\end{equation*}
subject to $\bm{Q}'\bm{Q}=\bm{I}_K$, $\bm{\Phi}=\bm{Q}=\bm{R}$, and $\bm{\phi}'_1\bm{S}\bm{\phi}_1\geq\bm{\phi}'_2\bm{S}\bm{\phi}_2\geq\cdots\geq\bm{\phi}'_K\bm{S}\bm{\phi}_K$,  where
$r_{jk}$ is the $(j,k)$-th element of $\bm{R}$. Then the corresponding augmented Lagrangian is
\begin{align*}
L(\bm{\Phi}, \bm{Q},\bm{R},\bm{\Gamma}_1,\bm{\Gamma}_2)
=&~ \|\bm{Y}-\bm{Y}\bm{\Phi}\bm{\Phi}'\|^2_F+\tau_1\sum_{k=1}^K \bm{\phi}_k'\bm{\Omega}\bm{\phi}_k+
\tau_2 \sum_{k=1}^K\sum_{j=1}^{p}|r_{jk}|\nonumber\\
&~ + \mathrm{tr}(\bm{\Gamma}_1'(\bm{\Phi}-\bm{Q}))+ \mathrm{tr}(\bm{\Gamma}_2'(\bm{\Phi}-\bm{R}))\nonumber\\
&~ +\frac{\rho}{2}(\|\bm{\Phi}-\bm{Q}\|^2_F+\|\bm{\Phi}-\bm{R}\|^2_F),
\end{align*}
subject to $\bm{Q}'\bm{Q}=\bm{I}_K$ and
$\bm{\phi}'_1\bm{S}\bm{\phi}_1\geq\bm{\phi}'_2\bm{S}\bm{\phi}_2\geq\cdots\geq\bm{\phi}'_K\bm{S}\bm{\phi}_K$,
where $\bm{\Gamma}_1$ and $\bm{\Gamma}_2$ are $p\times K$ matrices of the Lagrange multipliers.
Then the ADMM steps at the $\ell$-th iteration are given by
\begin{align}
\bm{\Phi}^{(\ell+1)}
=&~ \mathop{\arg\min}_{\bm{\Phi}}L\big(\bm{\Phi}, \bm{Q}^{(\ell)},\bm{R}^{(\ell)},\bm{\Gamma}^{(\ell)}_1,\bm{\Gamma}^{(\ell)}_2\big)
\notag\\
=&~ \frac{1}{2}(\tau_1\bm{\Omega}+\rho\bm{I}_p-\bm{Y}'\bm{Y})^{-1}\big\{\rho\big(\bm{Q}^{(\ell)}+\bm{R}^{(\ell)}\big)-
\bm{\Gamma}_1^{(\ell)}-\bm{\Gamma}_2^{(\ell)}\big\},\label{eq:ch4phi}\\
\bm{Q}^{(\ell+1)}
=&~ \mathop{\arg\min}_{\bm{Q}: \bm{Q}'\bm{Q}=\bm{I}_K}L\big(\bm{\Phi}^{(\ell+1)},\bm{Q},\bm{R}^{(\ell)},
\bm{\Gamma}^{(\ell)}_1,\bm{\Gamma}^{(\ell)}_2\big)\notag\\
=&~ \bm{U}^{(\ell)}\big(\bm{V}^{(\ell)}\big)',
\label{eq:ch4Q}\\
\bm{R}^{(\ell+1)}
=&~ \mathop{\arg\min}_{\bm{R}} L\big(\bm{\Phi}^{(\ell+1)},\bm{Q}^{(\ell+1)},\bm{R},\bm{\Gamma}^{(\ell)}_1,\bm{\Gamma}^{(\ell)}_2\big)
\notag\\
=&~ \frac{1}{\rho}\mathcal{S}_{\tau_2}\big(\rho\bm{\Phi}^{(\ell+1)}+\bm{\Gamma}_{2}^{(\ell)}\big),
\label{eq:ch4R}\\
\bm{\Gamma}^{(\ell+1)}_1
=&~ \bm{\Gamma}^{(\ell) }_1+ \rho\left(\bm{\Phi}^{(\ell+1)}-\bm{Q}^{(\ell+1)}\right), \label{eq:ch4Gamma1}\\
\bm{\Gamma}^{(\ell+1)}_2
=&~ \bm{\Gamma}^{(\ell) }_2+ \rho\left(\bm{\Phi}^{(\ell+1)}-\bm{R}^{(\ell+1)}\right),\label{eq:ch4Gamma2}
\end{align}
where $\bm{R}^{(0)}$, $\bm{\Gamma}_1^{(0)}$ and $\bm{\Gamma}_2^{(0)}$ are initial
estimates of $\bm{R}$, $\bm{\Gamma}_1$ and $\bm{\Gamma}_2$, respectively,
$\bm{U}^{(\ell)}\bm{D}^{(\ell)}\big(\bm{V}^{(\ell)}\big)'$ is the singular value decomposition of
$\bm{\Phi}^{(\ell+1)}+\rho^{-1}\bm{\Gamma}_{1}^{(\ell)}$, and
$\mathcal{S}_{\tau_2}(\cdot)$ is the element-wise soft-thresholding operator with a threshold $\tau_2$
(i.e., the $(j,k)$-th element of $\mathcal{S}_{\tau_2}(\bm{M})$ is
$\mathrm{sign}(m_{jk})\max(|m_{jk}|-\tau_2,0)$, where $m_{jk}$ is the $(j,k)$-th element of $\bm{M}$).
Similarly to (\ref{eq:ch4phi_1}), $\rho$ must be chosen large enough to ensure that $\tau_1\bm{\Omega}+\rho\bm{I}_p-\bm{Y}'\bm{Y}$ in \eqref{eq:ch4phi}
is positive definite.

\section{Numerical Examples}
\label{sec:ch4numerical}

In this section, we conducted some simulation experiments in one-dimensional and two-dimensional spatial domains,
and applied SpatPCA to a real-world dataset. We compared the proposed SpatPCA with three methods: (1) PCA ($\tau_1=\tau_2=0$);
(2) SpatPCA with the smoothness penalty only ($\tau_2=0$);
(3) SpatPCA with the sparseness penalty only ($\tau_1=0$), based on
two loss functions. The first one measures the prediction ability in terms of an average squared
prediction error:
\begin{equation}
\label{eq:ch4loss_sim}
\mathrm{Loss}(\hat{\bm{\Phi}}) = \frac{1}{n}\sum_{i=1}^n\big\|\hat{\bm{\xi}}_i - \bm{\xi}_i\big\|^2,
\end{equation}
where $\bm{\Phi}$ is the true eigenvector matrix formed by the first $K$ eigenvectors and
\begin{equation}
\label{eq:ch4xi.hat}
\hat{\bm{\xi}}_i=\hat{\bm{V}}\,\mathrm{diag}\!\left(\frac{\hat{\lambda}^*_1}{\hat{\lambda}^*_1+\hat{\sigma}^2},\dots,\frac{\hat{\lambda}^*_K}{\hat{\lambda}^*_K+\hat{\sigma}^2}\right)\hat{\bm{V}}'\hat{\bm{\Phi}}'\bm{Y}_i
\end{equation}
is the empirical best linear unbiased predictor of $\bm{\xi}_i$ with the estimated parameters plugged in.
The second one concerns the goodness of covariance function estimation in terms of an average
squared estimation error:
\begin{equation}
\label{eq:ch4loss2_sim}
\mathrm{Loss}(\hat{C}_\eta) =\frac{1}{p^2} \sum_{i=1}^{p}\sum_{j=1}^{p}\big(\hat{C}_\eta(\bm{s}_i,\bm{s}_j)-C_\eta(\bm{s}_i,\bm{s}_j)\big)^2\:.
\end{equation}

We applied the ADMM algorithm given by (\ref{eq:ch4phi})--(\ref{eq:ch4Gamma2}) to compute the SpatPCA estimates. We chose $\rho$
equal to ten times the maximum eigenvalue of $\bm{Y}'\bm{Y}$.
The stopping criterion for the ADMM algorithm is
\[
\frac{1}{\sqrt{p}} \max\left( \|\bm{\Phi}^{(\ell+1)}-\bm{\Phi}^{(\ell)}\|_F,\|\bm{\Phi}^{(\ell+1)}-\bm{R}^{(\ell+1)}\|_F,
\|\bm{\Phi}^{(\ell+1)}-\bm{Q}^{(\ell+1)}\|_F \right)\leq 10^{-4}\:.
\]
An R package to carry out SpatPCA is available on the Comprehensive R Archive Network (CRAN).

\subsection{One-Dimensional Experiment}

In the first experiment,  we generated data according to (\ref{eq:ch4measurement}) with $K=2$,
$\bm{\xi}_i\sim N(\bm{0}, \mathrm{diag}(\lambda_1,\lambda_2))$, $\bm{\epsilon}_{i}\sim N(\bm{0}, \bm{I})$,
$n=100$, $p=50$, $\bm{s}_1,\dots,\bm{s}_{50}$ equally spaced in $D=[-5,5]$,
and
\begin{align}
\varphi_1(\bm{s})
=&~ \frac{1}{c_1}\exp(-(x_1^2+\cdots+x_d^2)),
\label{eq:ch4phi1_sim}\\
\varphi_2(\bm{s})
=&~ \frac{1}{c_2}x_1\cdots x_d\exp(-(x_1^2+\cdots+x_d^2)),
\label{eq:ch4phi2_sim}
\end{align}
where $\bm{s}=(x_1,\dots,x_d)'$, $c_1$ and $c_2$ are normalization constants such that
$\|\bm{\phi}_1\|_2=\|\bm{\phi}_2\|_2=1$, and $d=1$.
We considered three pairs of $(\lambda_1,\lambda_2)\in\{(9,0),(1,0),(9,4)\}$
with different strengths of signals, and applied the proposed
SpatPCA with $K\in\{1,2,5\}$ and $\hat{K}$ selected from \eqref{eq:ch4K.hat}, resulting in 12 different combinations.
For each combination, we considered $11$ values of $\tau_1$
(including $0$, and 10 values from $1$ to $10^3$ equally spaced on the log scale) and $31$ values of $\tau_2$
(including $0$, and 30 values from $1$ to $10^3$ equally spaced on the log scale).
But instead of performing a two-dimensional optimization by selecting among all possible pairs of $(\tau_1,\tau_2)$,
we applied a more efficient two-step procedure involving only one-dimensional optimization.
First, we selected among $11$ values of $\tau_1$ by fixing $\tau_2 = 0$ using 5-fold CV of \eqref{eq:ch4cv}
with the initial estimate of $\hat{\bm{\Phi}}_{\tau_1,0}$ given by the first $K$ eigenvectors of $\bm{Y}'\bm{Y}-\tau_1\bm{\Omega}$ as its columns.
Note that this initial estimate is actually the true estimate $\hat{\bm{\Phi}}_{\tau_1,0}$ when $\bm{Y}'\bm{Y}-\tau_1\bm{\Omega}\succeq\bm{0}$.
Then we selected among $31$ values of $\tau_2$ with the selected $\tau_1$ using 5-fold CV of \eqref{eq:ch4cv}.
For covariance function estimation, we selected the tuning parameter $\gamma$
among $11$ values of $\gamma$ using 5-fold CV of \eqref{eq:ch4cv.gamma}, including $0$, and 10 values
from $1$ to $\hat{d}_1$ equally spaced on the log scale, where $\hat{d}_1$ is the largest eigenvalue of $\hat{\bm{\Phi}}'\bm{S}\hat{\bm{\Phi}}$.

Figure~\ref{fig:ch4est_d1} shows the estimates of $\varphi_1(\cdot)$ and $\varphi_2(\cdot)$ for the four methods
based on three different combinations of eigenvalues. Each case contains four estimated functions based on four randomly generated datasets.
As expected, the PCA estimates, which consider no spatial structure, are very noisy, particularly when the signal-to-noise ratio is small.
By considering only the smoothness penalty with $\tau_2=0$,
the resulting estimates are much less noisy. But we can still see some bias, particularly around the two ends,
at which the true values are approximate zeros.
On the other hand, by considering only the sparseness penalty with $\tau_1=0$,
the resulting estimates, while not as noisy as those from PCA, are still noisy despite that
the estimates are shrunk to zeros at some locations.
Overall, our SpatPCA estimates are very close to the targets for all cases
even when the signal-to-noise ratio is small, indicating the effectiveness of regularization.

\begin{figure}[p]
\centering
\small $\hat{\varphi}_1(\cdot)$ based on $(\lambda_1,\lambda_2)=(9,0)$ and $K=1$\\[0.25em]
\includegraphics[width=\textwidth]{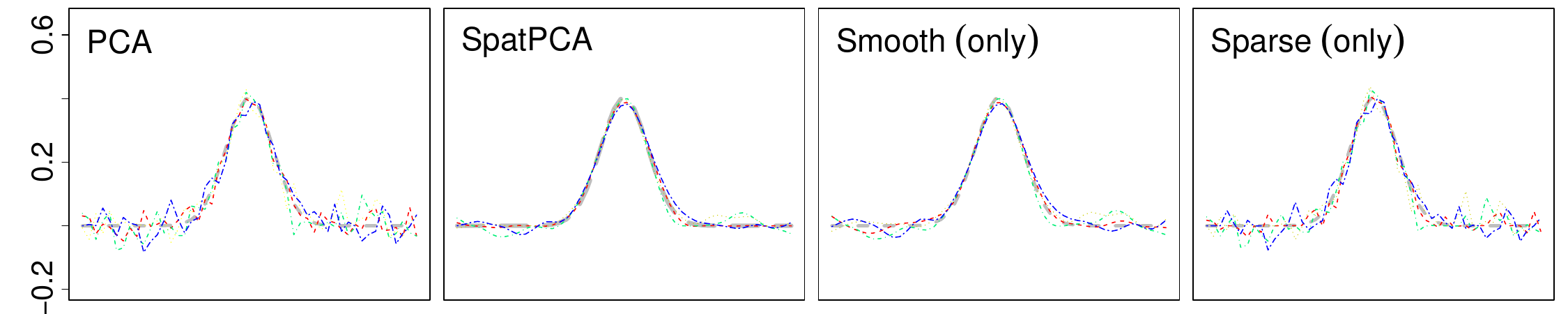}\\[1.0em]
\small $\hat{\varphi}_1(\cdot)$ based on $(\lambda_1,\lambda_2)=(1,0)$ and $K=1$\\[0.25em]
\includegraphics[width=\textwidth]{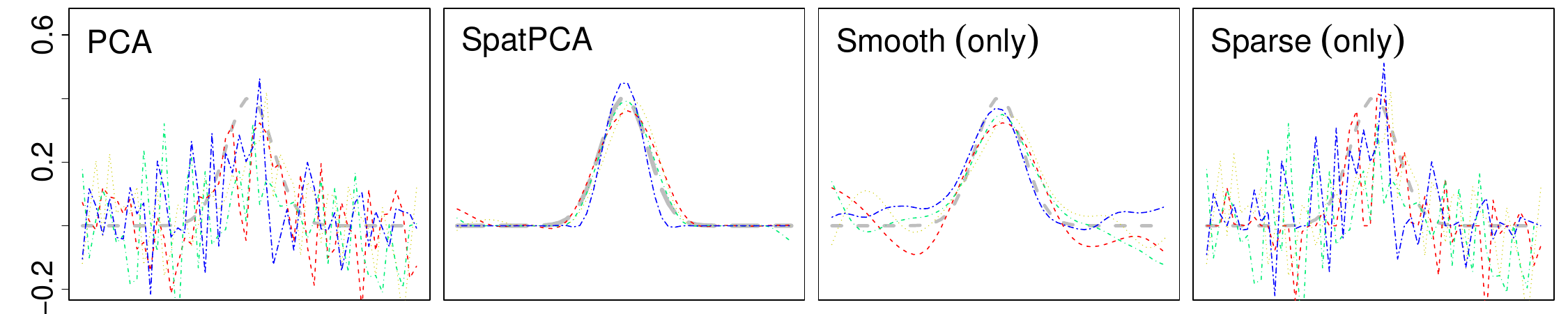}\\[1.0em]
\small $\hat{\varphi}_1(\cdot)$ based on $(\lambda_1,\lambda_2)=(9,4)$ and $K=2$\\[0.25em]
\includegraphics[width=\textwidth]{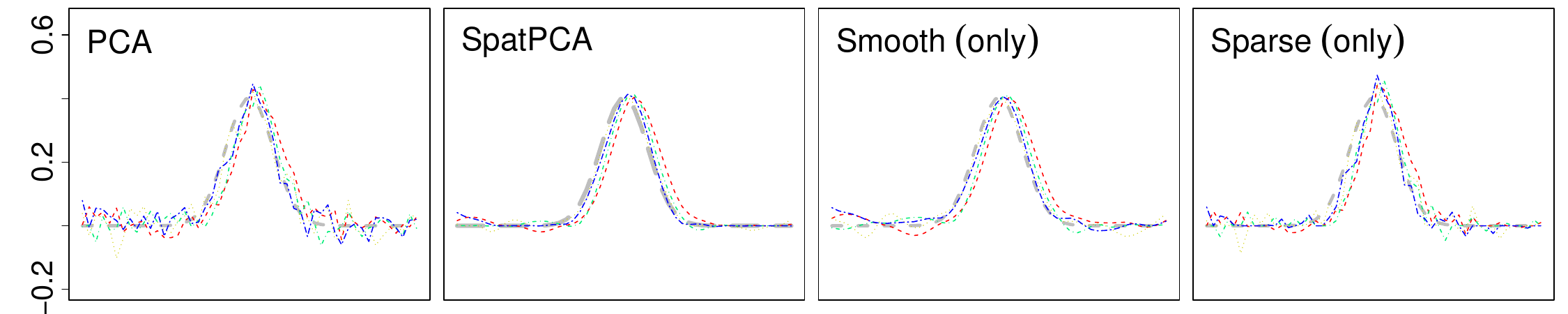}\\[1.0em]
\small $\hat{\varphi}_2(\cdot)$ based on $(\lambda_1,\lambda_2)=(9,4)$ and $K=2$\\[0.25em]
\includegraphics[width=\textwidth]{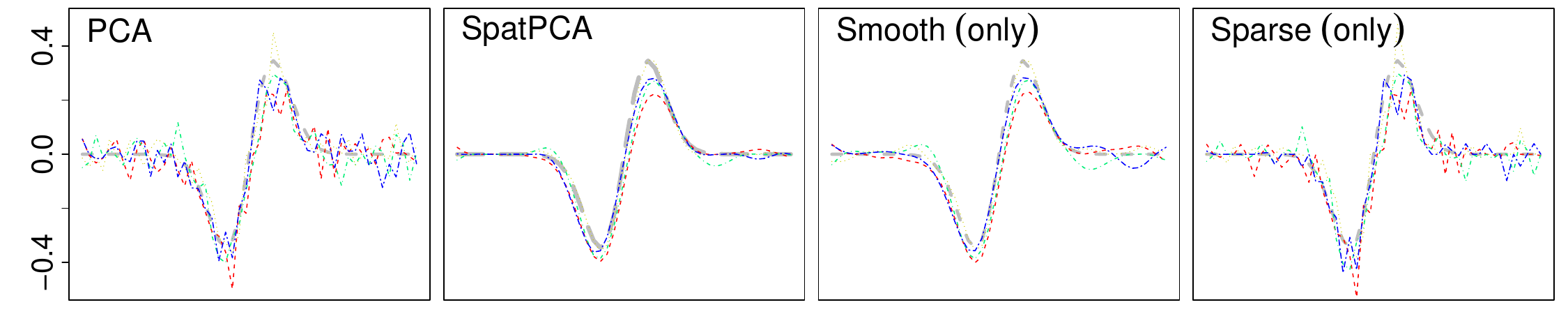}
\caption{Estimates of $\varphi_1(\cdot)$ and $\varphi_2(\cdot)$ from various methods for three different combinations of eigenvalues. Each panel consists of four estimates (in four different colors) corresponding to four randomly generated datasets, where the dotted grey line is the true eigenfunction.}
\label{fig:ch4est_d1}
\end{figure}

Figure~\ref{fig:ch4cov_d1} shows the covariance function estimates for the four methods based on a randomly generated dataset.
The proposed SpatPCA can be seen to perform considerably better than the other methods for all cases by being able to
reconstruct the underlying nonstationary spatial covariance functions without having noticeable visual artifacts.

\begin{figure}[tbhp]
\centering
\begin{tabular}{@{}c@{}}
\small $(\lambda_1,\lambda_2)=(9,0)$ and $K=1$ \\
\includegraphics[width=0.95\textwidth]{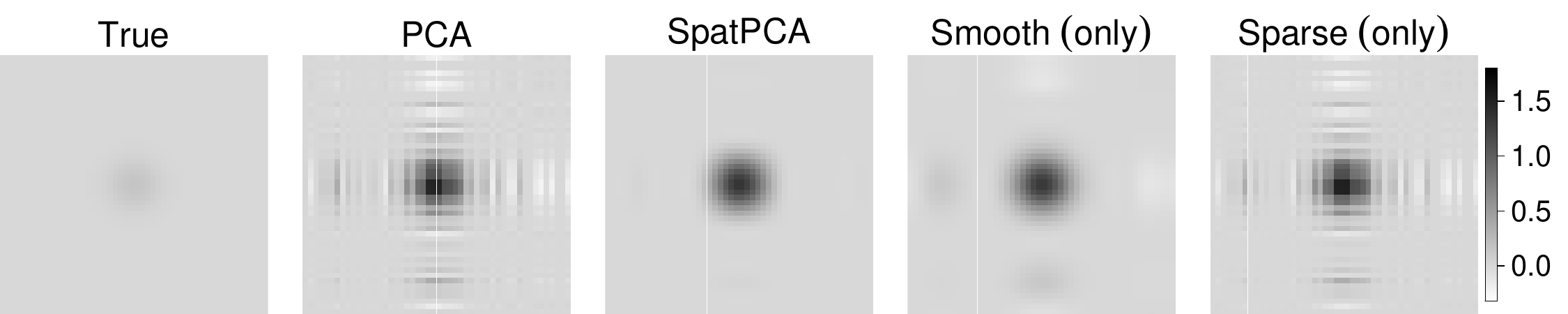} \\[0.4em]
\small $(\lambda_1,\lambda_2)=(1,0)$ and $K=1$ \\
\includegraphics[width=0.95\textwidth]{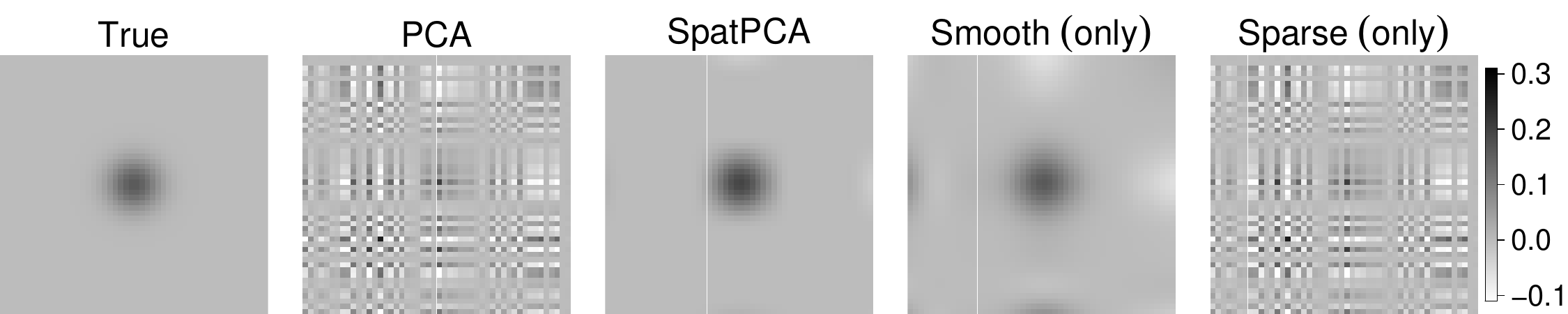} \\[0.4em]
\small $(\lambda_1,\lambda_2)=(9,4)$ and $K=2$ \\
\includegraphics[width=0.95\textwidth]{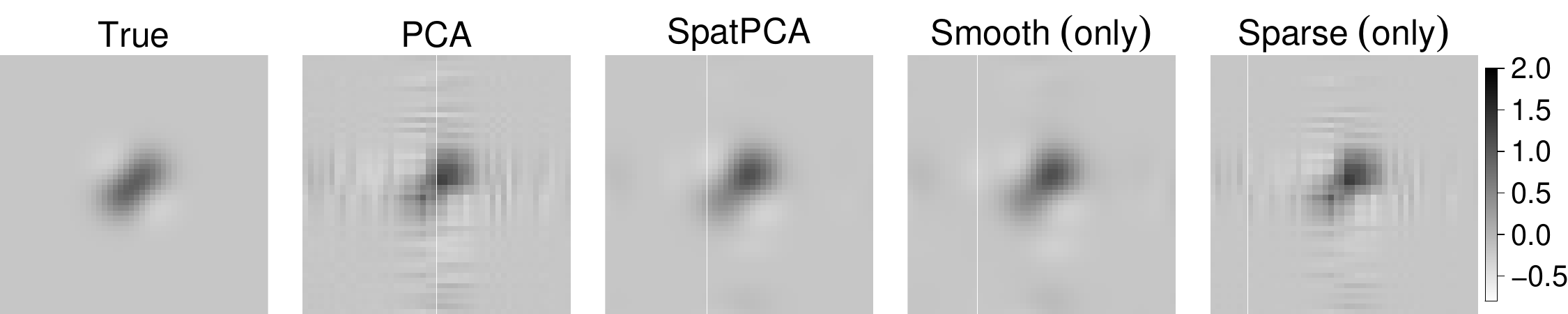}
\end{tabular}
\caption{True covariance functions and their estimates from various methods for three different combinations of eigenvalues.}
\label{fig:ch4cov_d1}
\end{figure}

Finally, the performance of the four methods in terms of the loss functions \eqref{eq:ch4loss_sim} and \eqref{eq:ch4loss2_sim} is shown in Figures~\ref{fig:ch4box_d1_loss1} and \ref{fig:ch4box_d1} respectively based on $50$
simulation replicates. Once again, SpatPCA outperforms all the other methods in all cases.

\begin{figure}[tbhp]
\centering
{\footnotesize
\begin{tabular}{@{}ccc@{}}
$(\lambda_1,\lambda_2)=(9,0)$, $K=1$ & $(\lambda_1,\lambda_2)=(9,0)$, $K=2$ & $(\lambda_1,\lambda_2)=(9,0)$, $K=5$ \\
\includegraphics[width=0.32\textwidth]{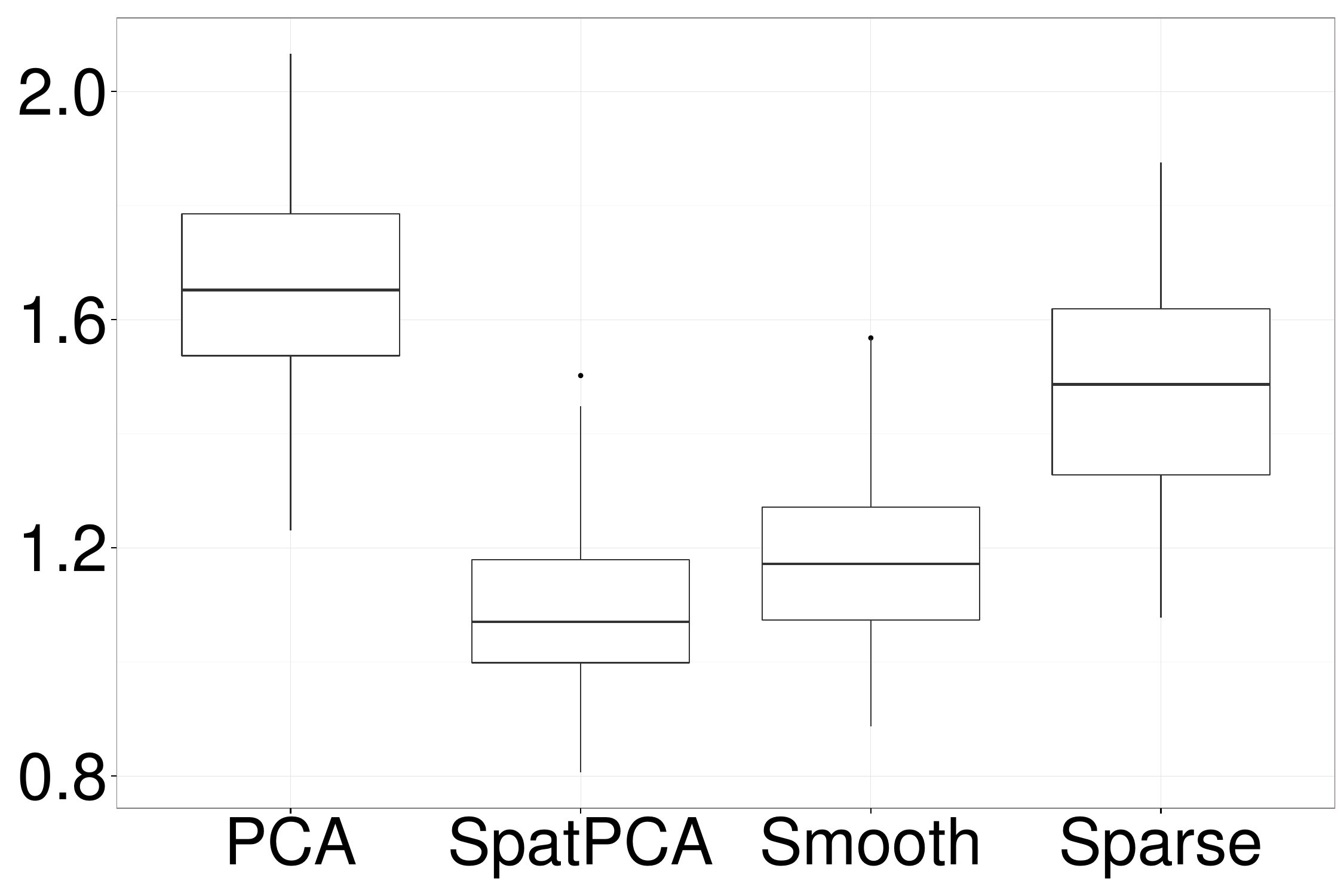} &
\includegraphics[width=0.32\textwidth]{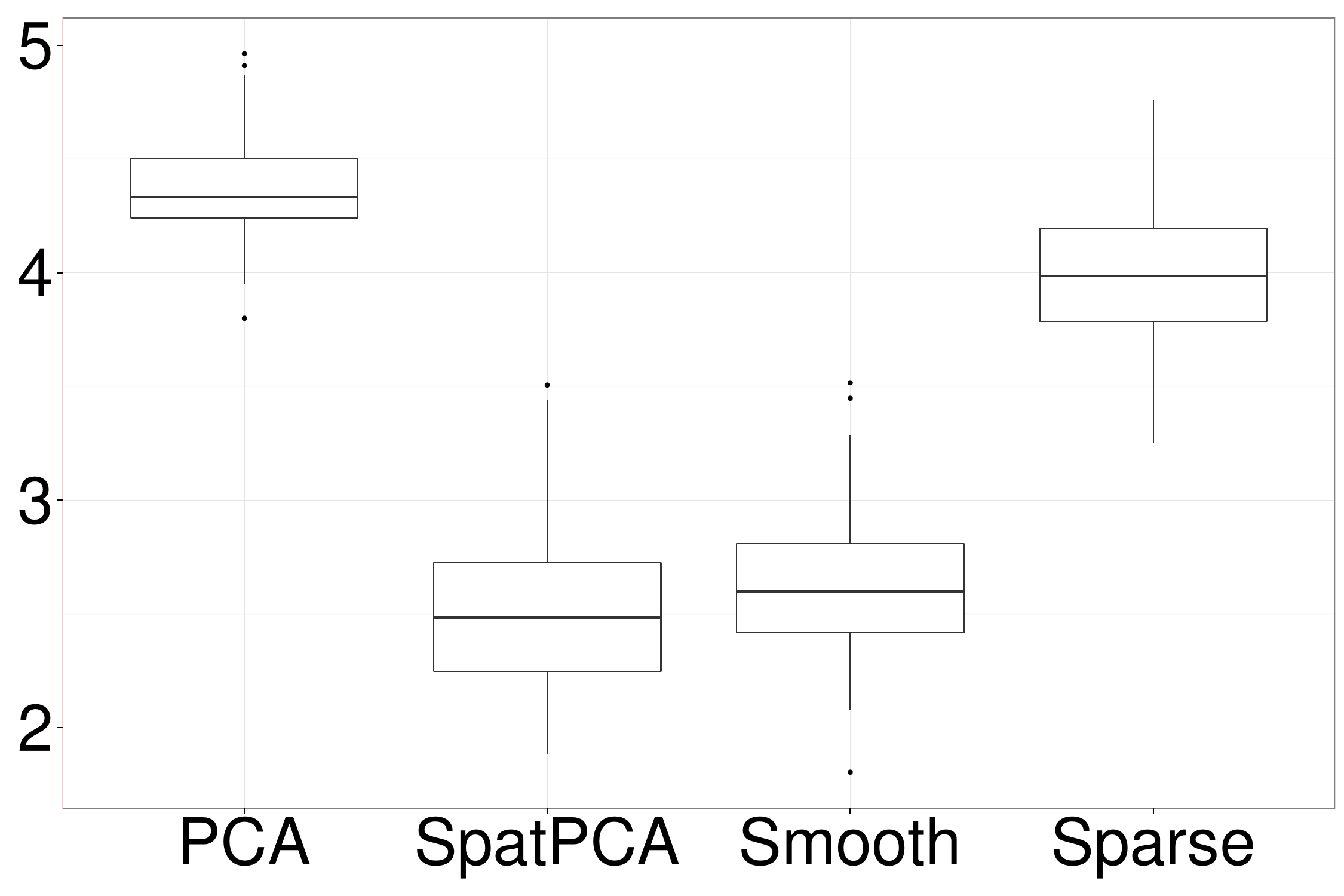} &
\includegraphics[width=0.32\textwidth]{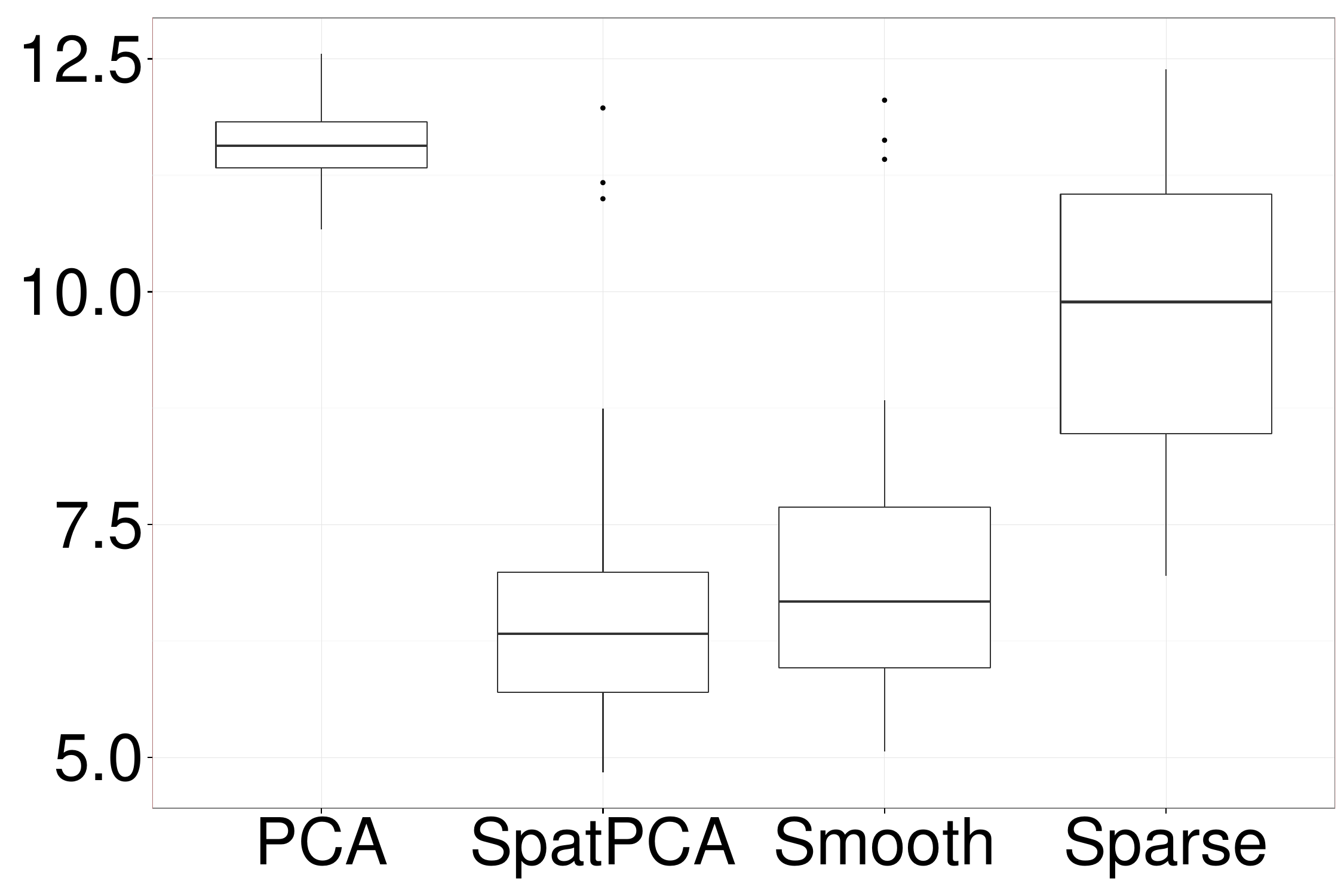} \\[0.4em]
$(\lambda_1,\lambda_2)=(1,0)$, $K=1$ & $(\lambda_1,\lambda_2)=(1,0)$, $K=2$ & $(\lambda_1,\lambda_2)=(1,0)$, $K=5$ \\
\includegraphics[width=0.32\textwidth]{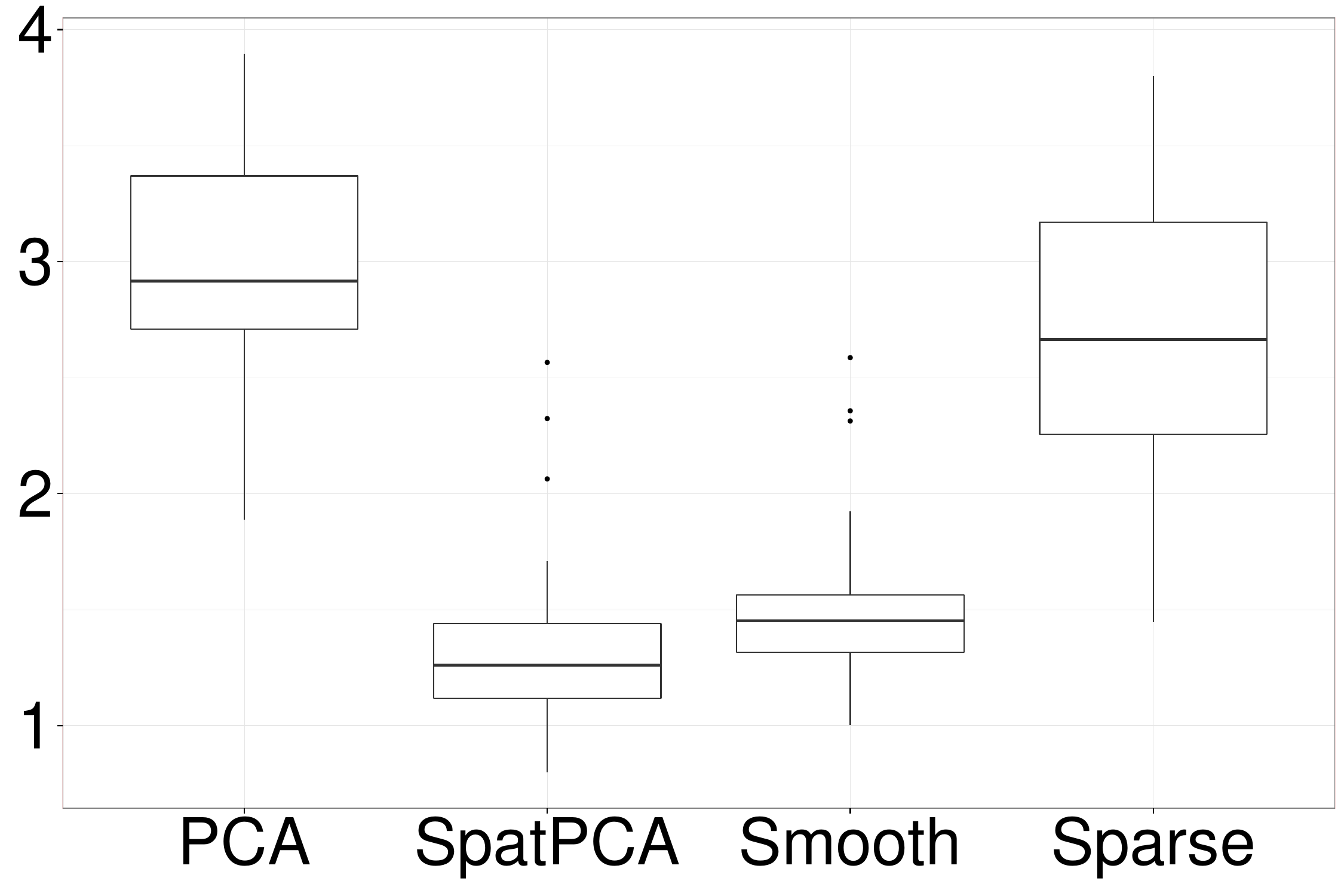} &
\includegraphics[width=0.32\textwidth]{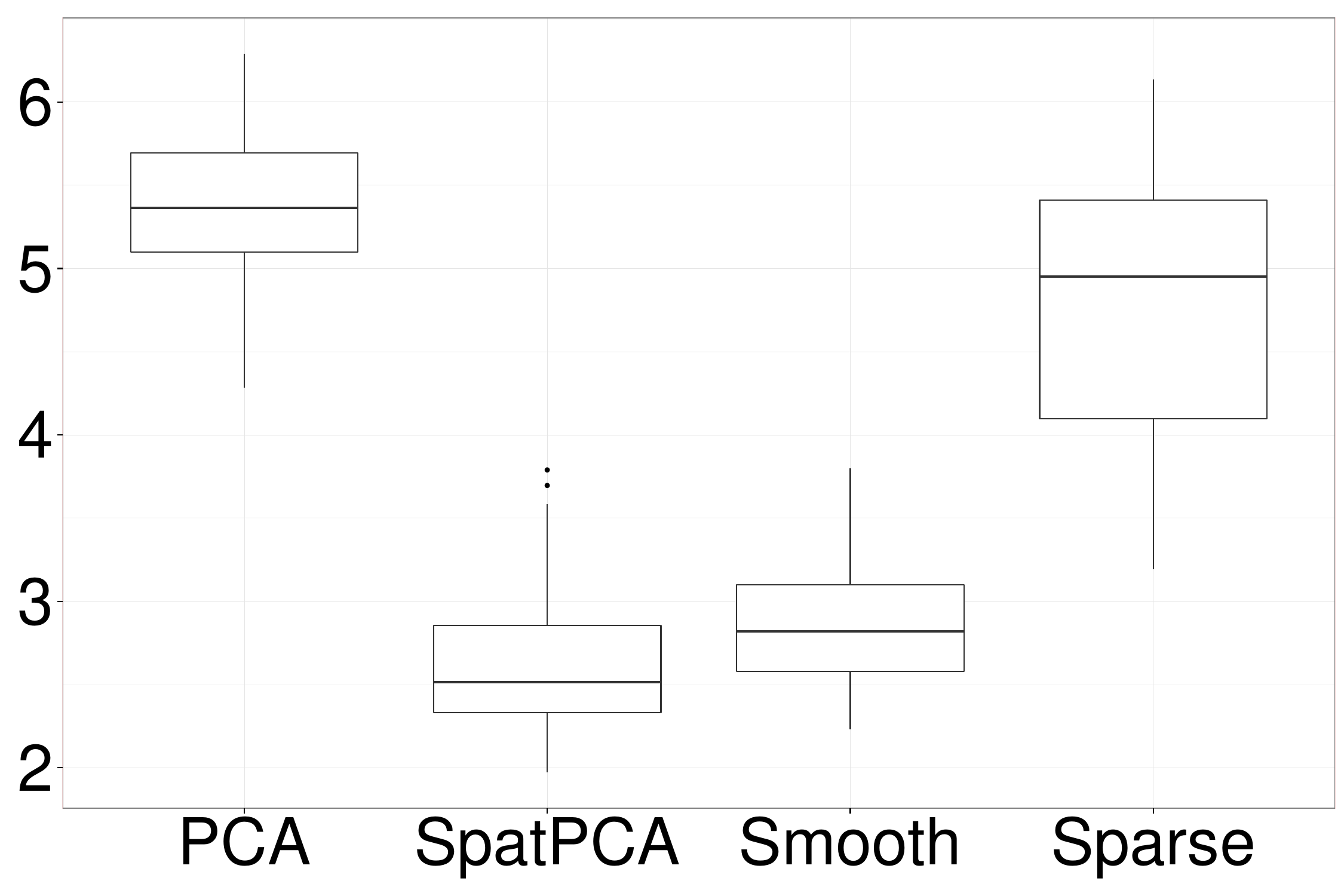} &
\includegraphics[width=0.32\textwidth]{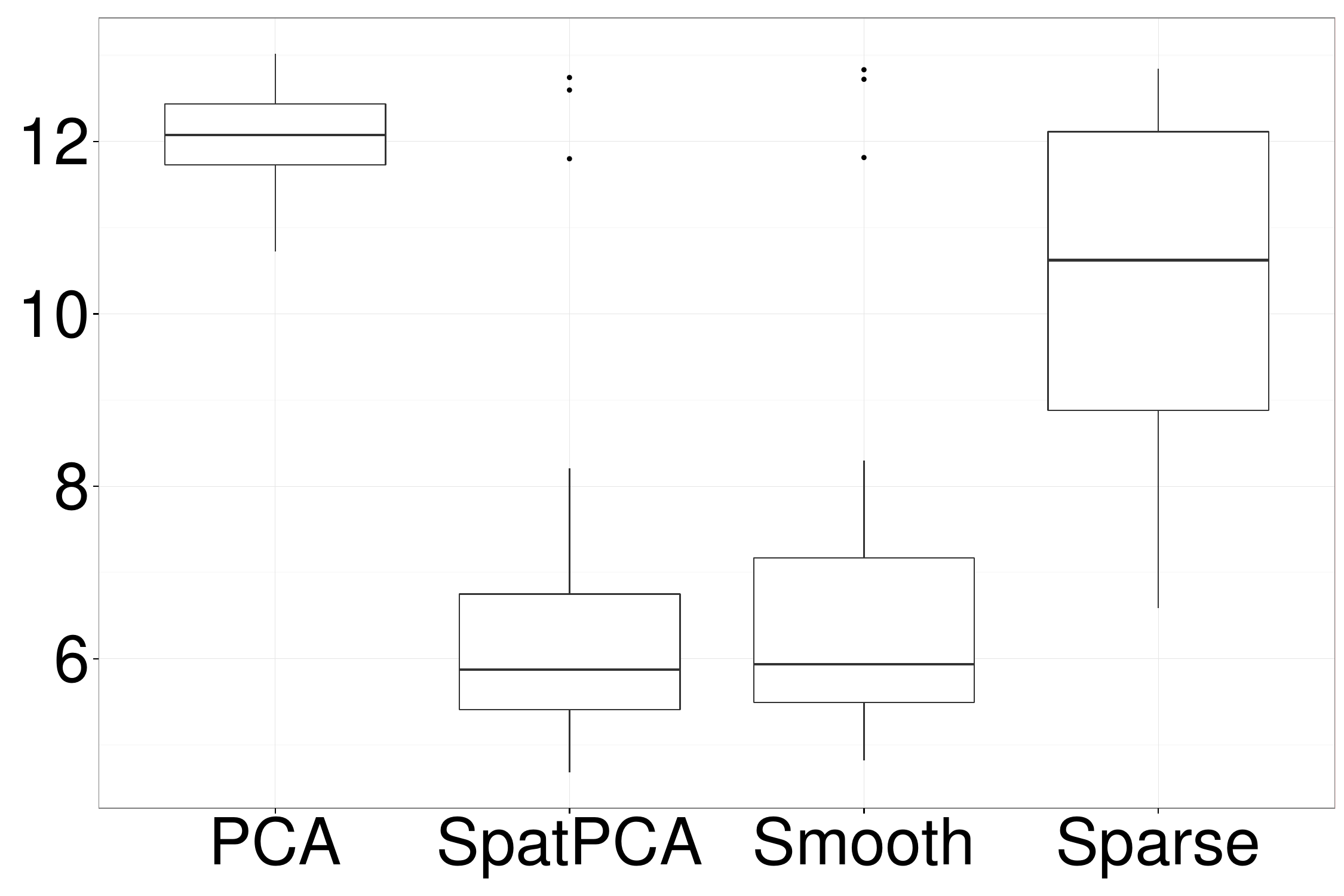} \\[0.4em]
$(\lambda_1,\lambda_2)=(9,4)$, $K=1$ & $(\lambda_1,\lambda_2)=(9,4)$, $K=2$ & $(\lambda_1,\lambda_2)=(9,4)$, $K=5$ \\
\includegraphics[width=0.32\textwidth]{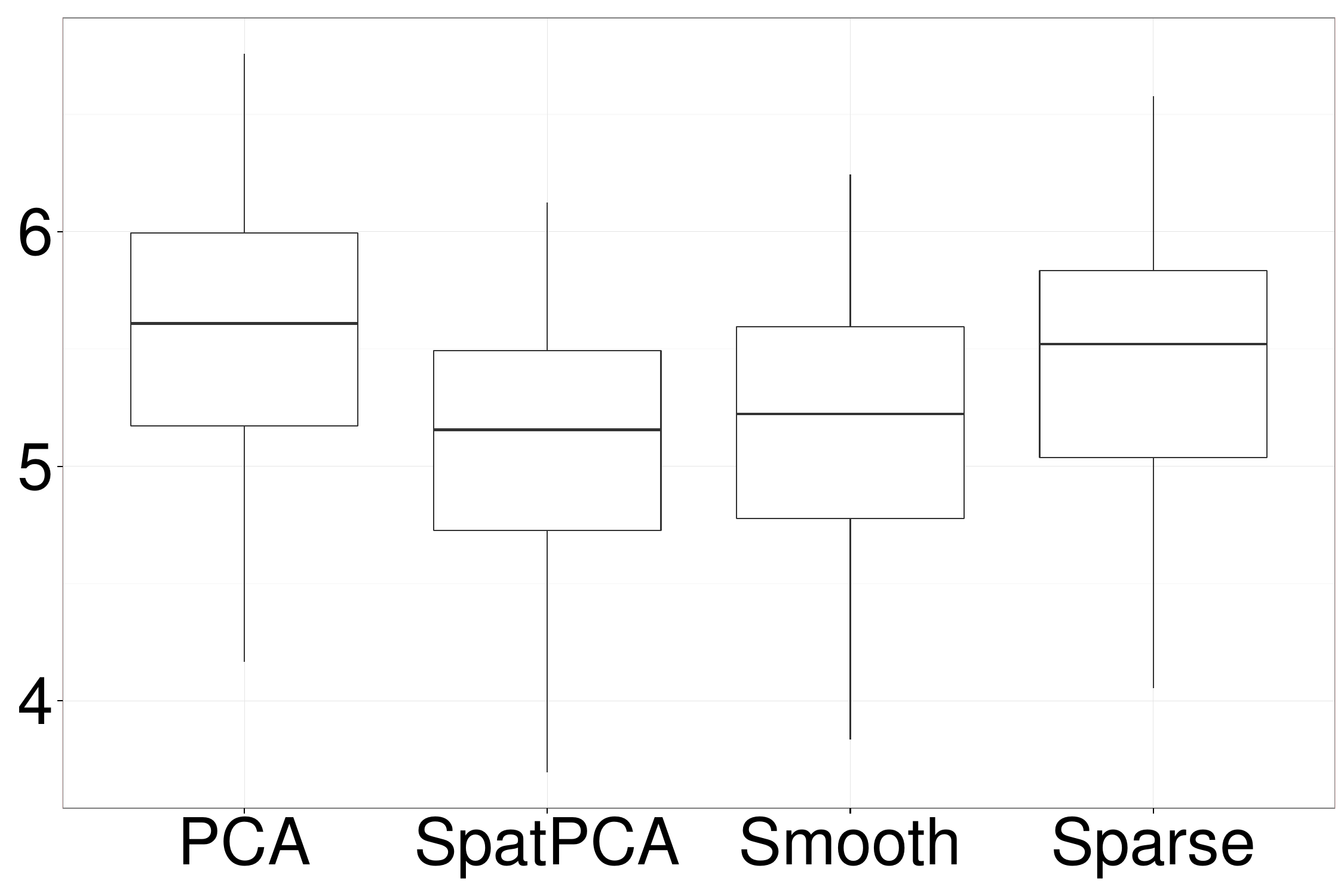} &
\includegraphics[width=0.32\textwidth]{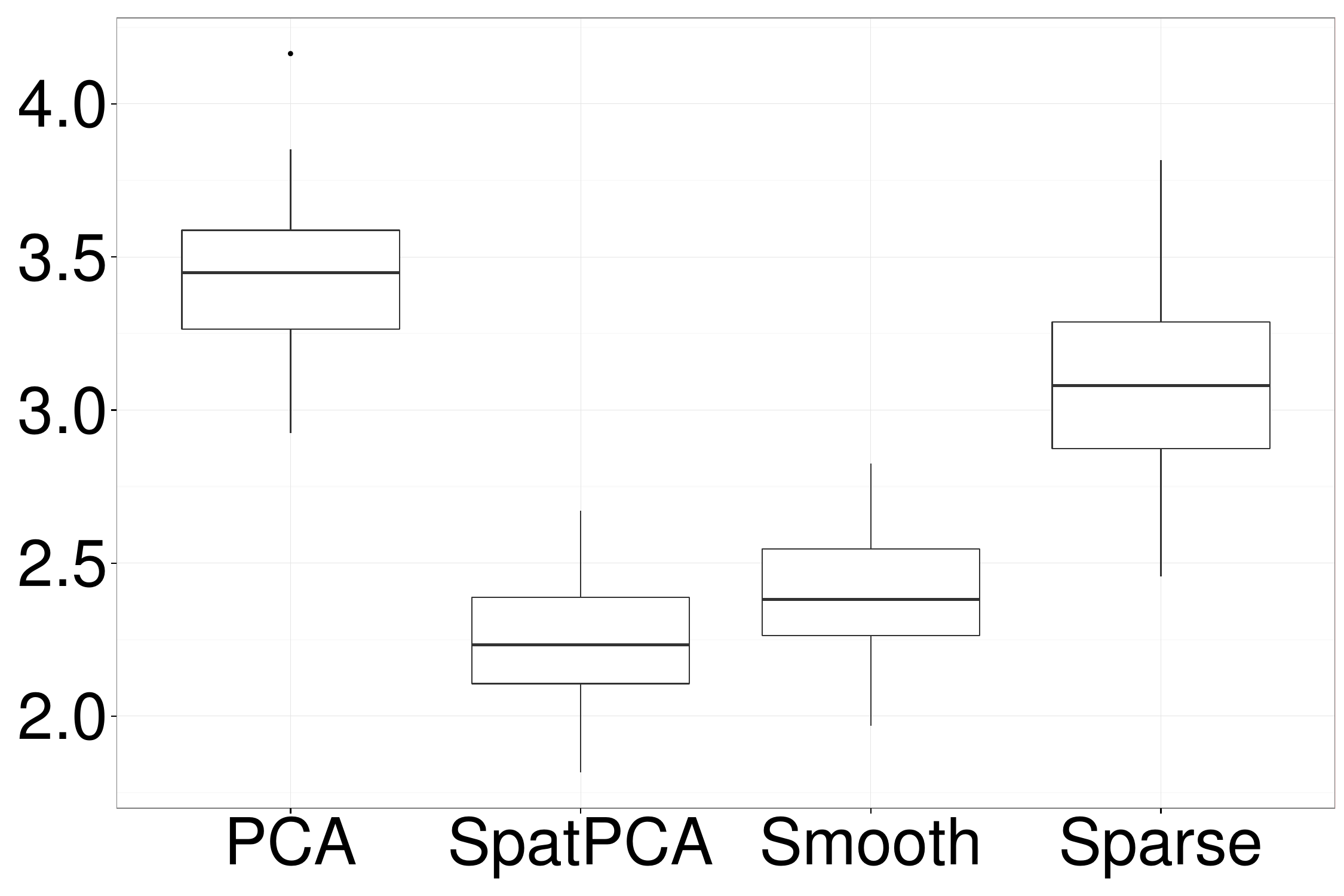} &
\includegraphics[width=0.32\textwidth]{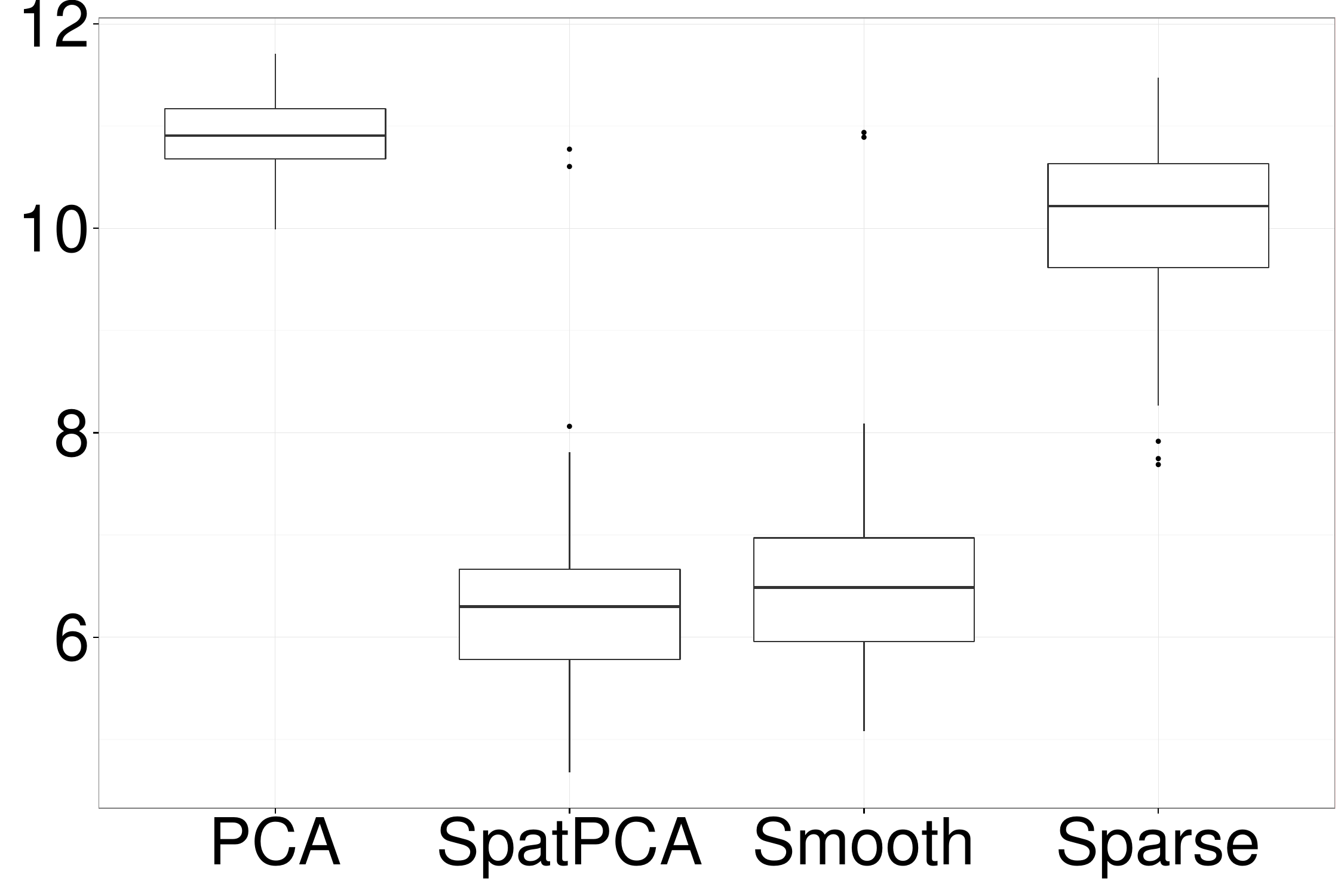}
\end{tabular}}
\caption{Boxplots of average squared prediction errors of (\ref{eq:ch4loss_sim}) for various methods based on 50 simulation replicates.}
\label{fig:ch4box_d1_loss1}
\end{figure}

\begin{figure}[tbhp]
\centering
{\footnotesize
\begin{tabular}{@{}ccc@{}}
$(\lambda_1,\lambda_2)=(9,0)$, $K=1$ & $(\lambda_1,\lambda_2)=(9,0)$, $K=2$ & $(\lambda_1,\lambda_2)=(9,0)$, $K=5$ \\
\includegraphics[width=0.32\textwidth]{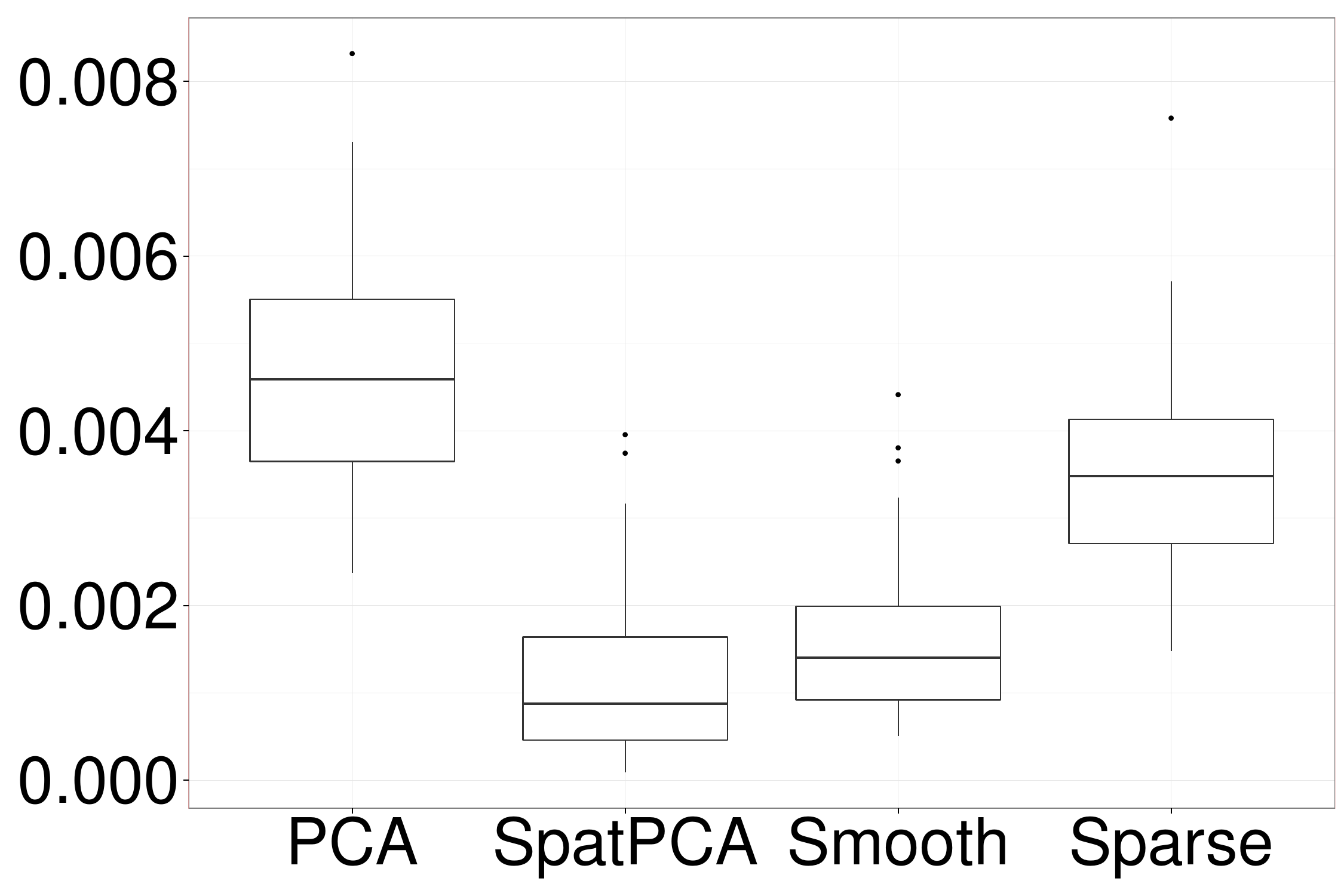} &
\includegraphics[width=0.32\textwidth]{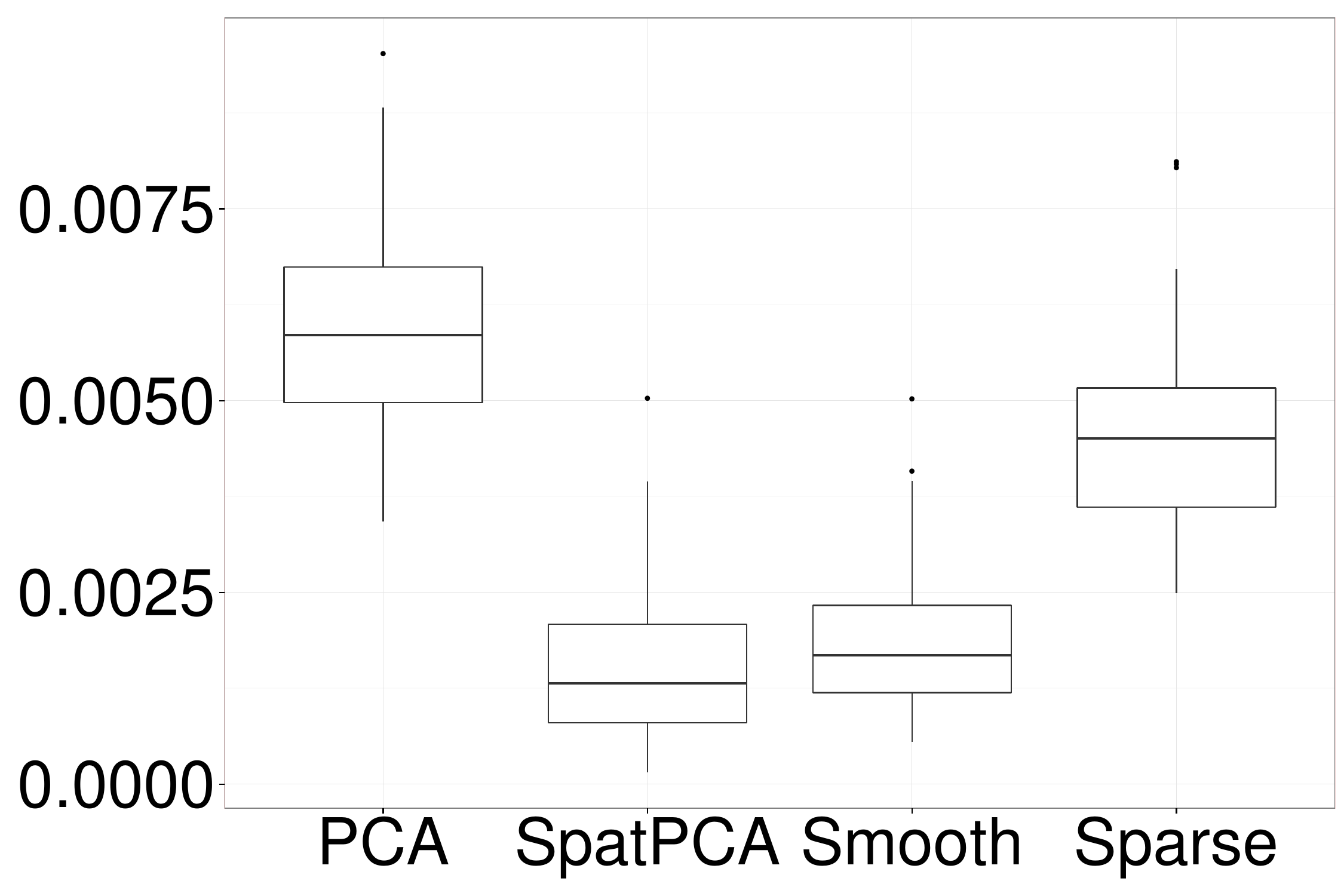} &
\includegraphics[width=0.32\textwidth]{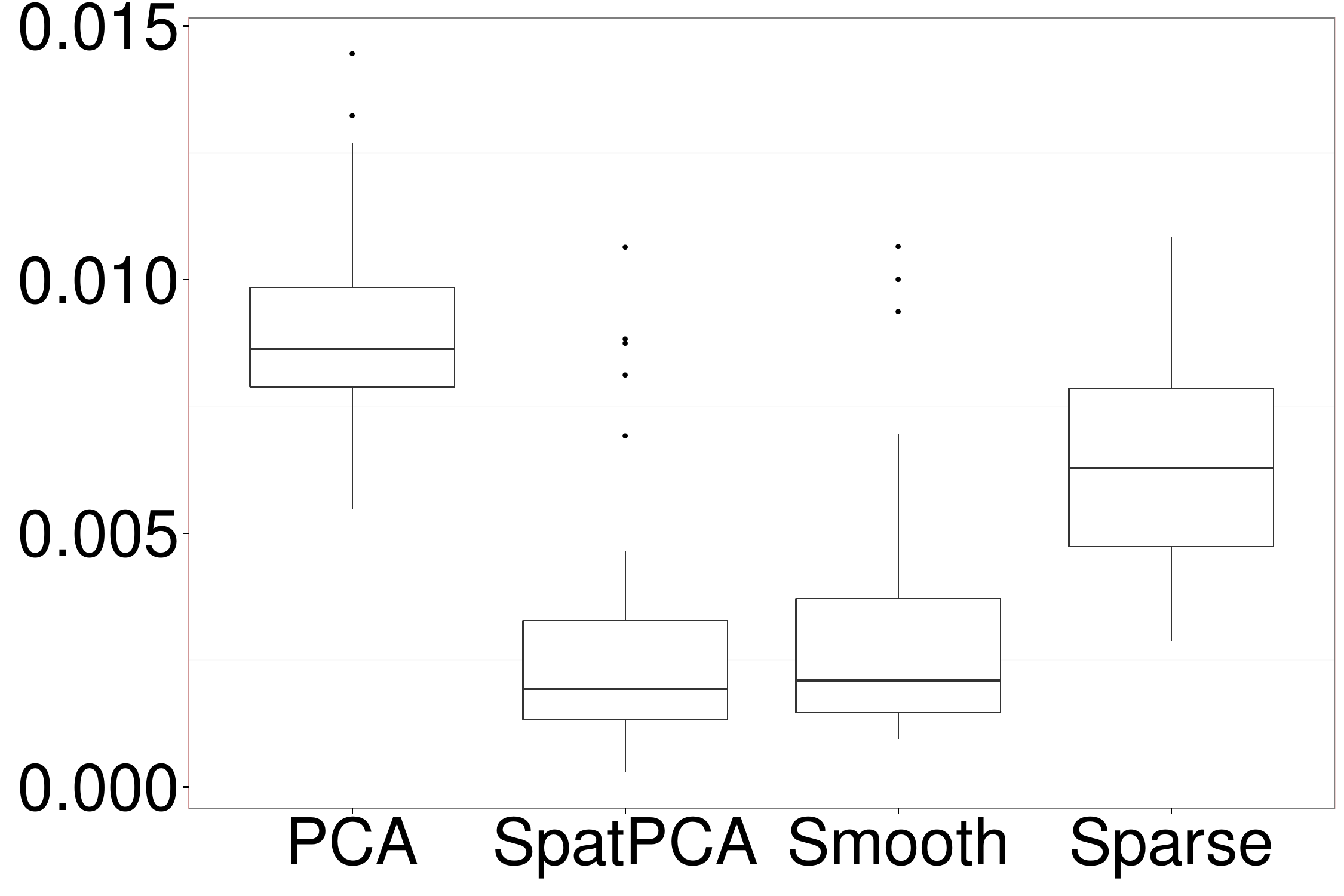} \\[0.4em]
$(\lambda_1,\lambda_2)=(1,0)$, $K=1$ & $(\lambda_1,\lambda_2)=(1,0)$, $K=2$ & $(\lambda_1,\lambda_2)=(1,0)$, $K=5$ \\
\includegraphics[width=0.32\textwidth]{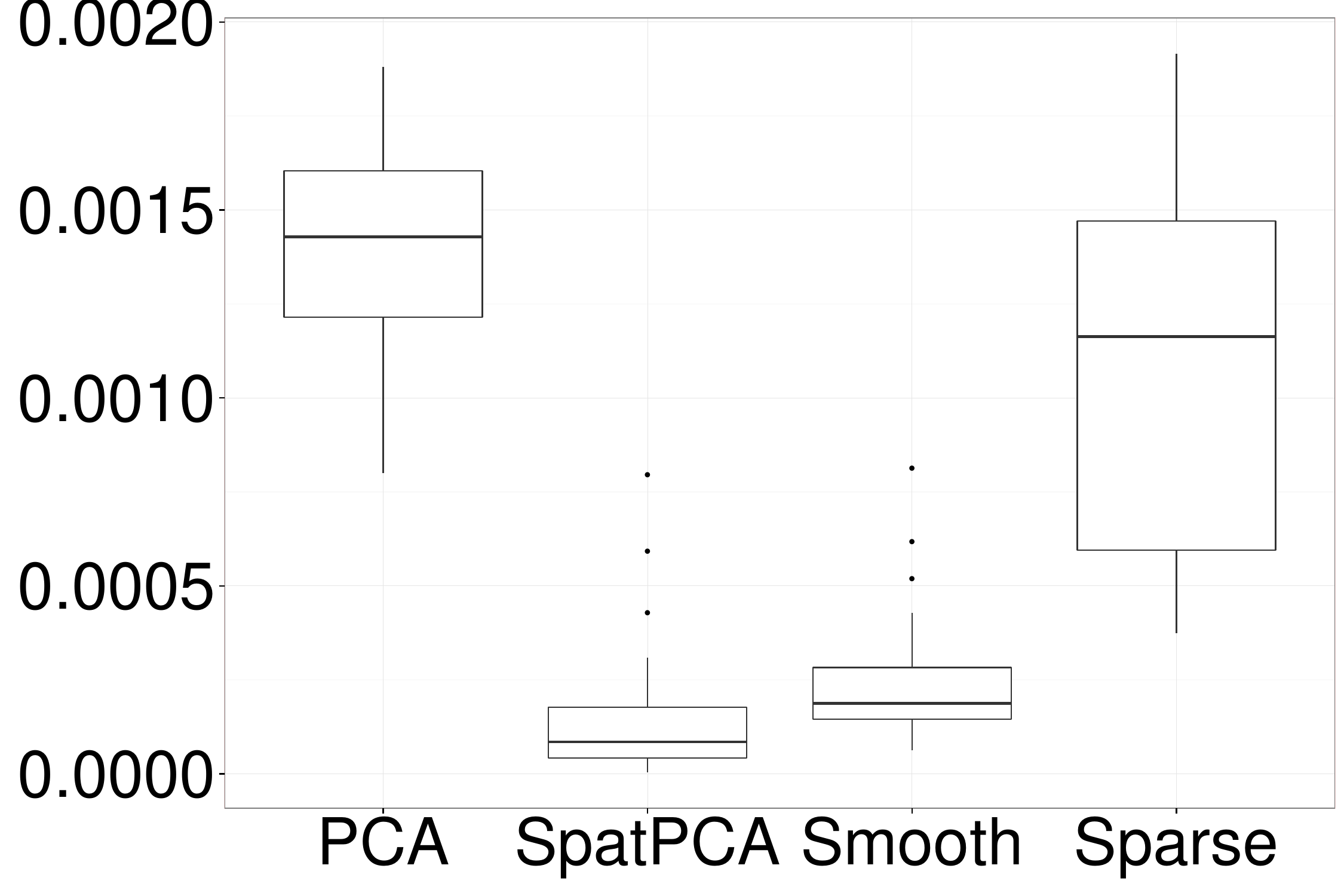} &
\includegraphics[width=0.32\textwidth]{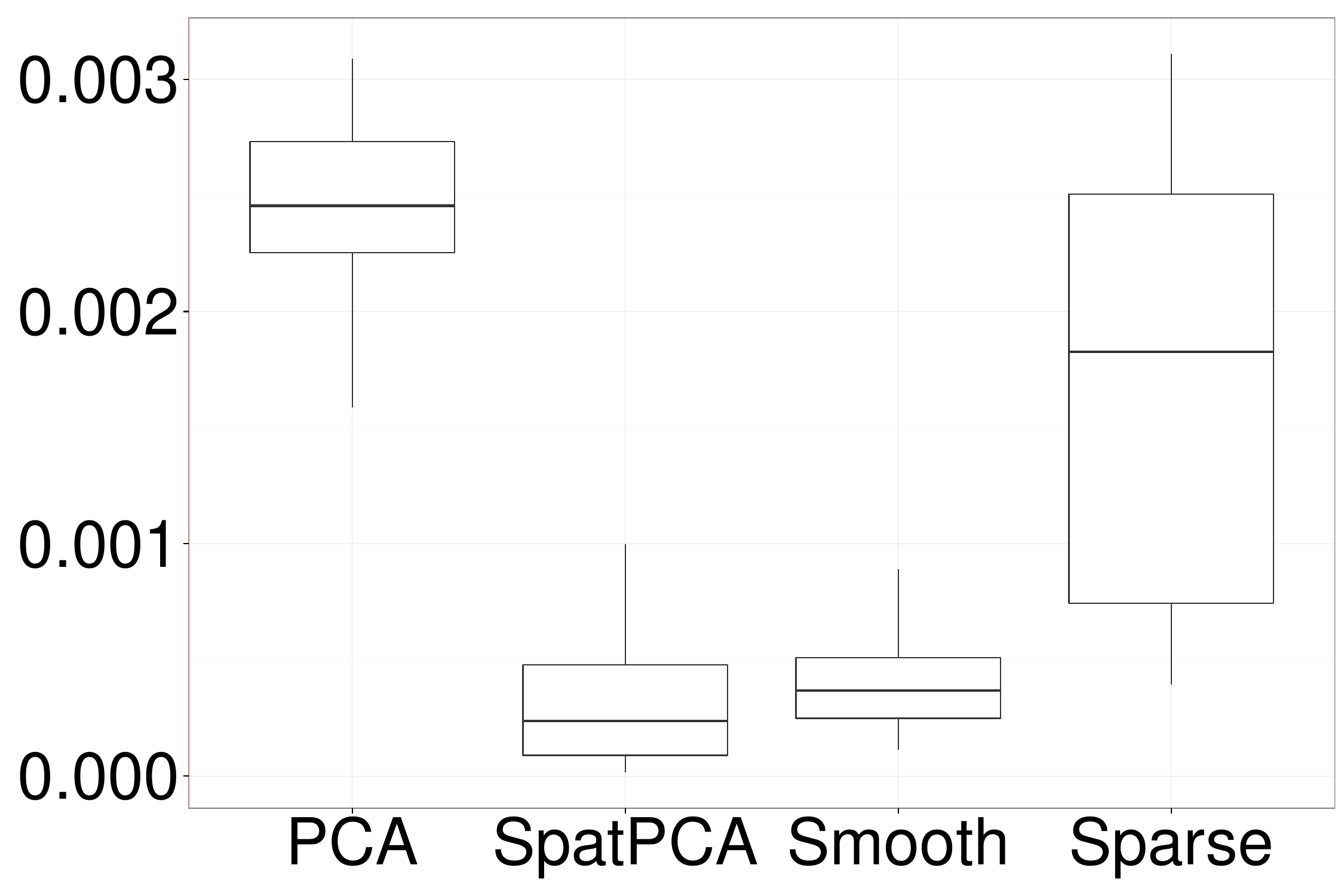} &
\includegraphics[width=0.32\textwidth]{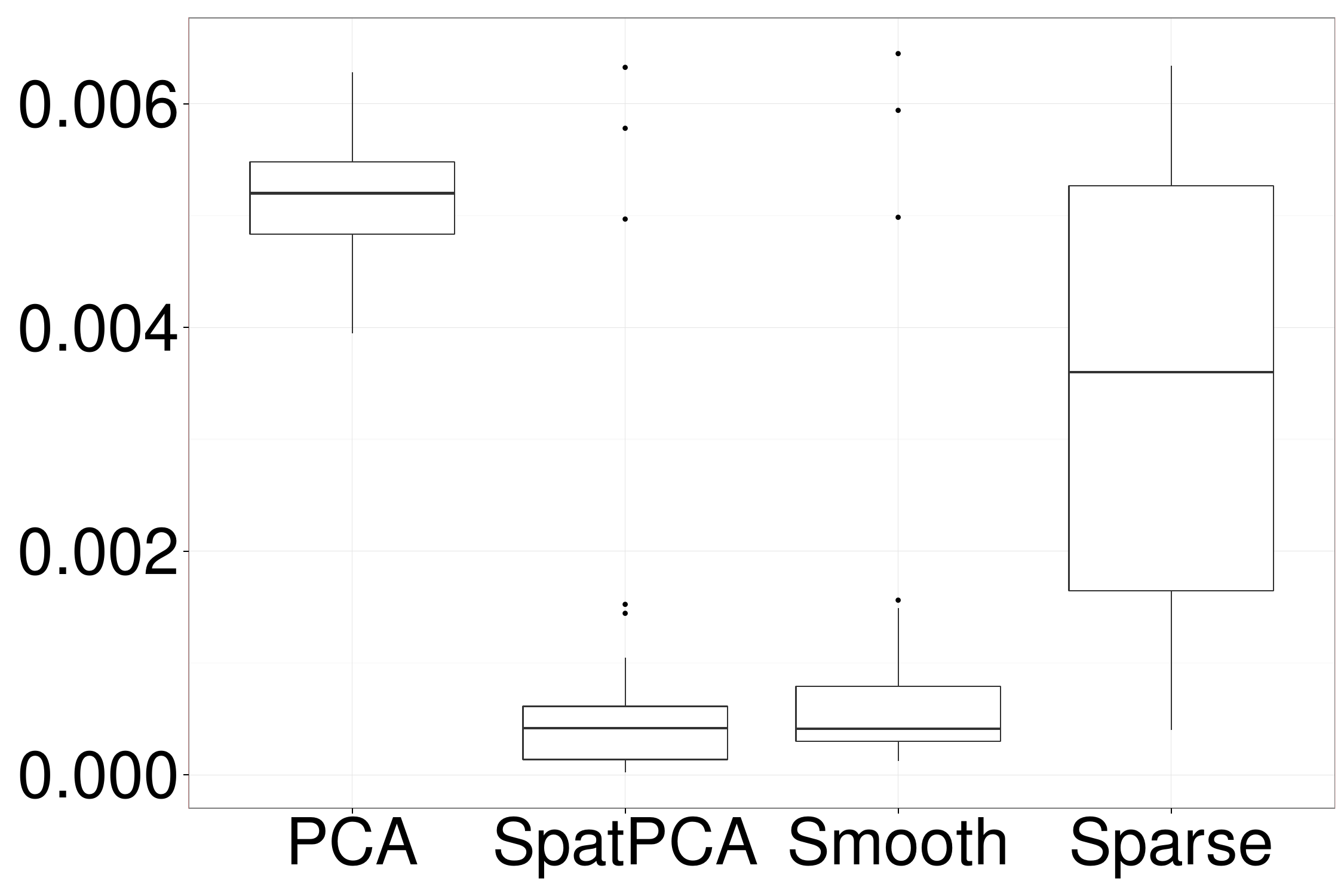} \\[0.4em]
$(\lambda_1,\lambda_2)=(9,4)$, $K=1$ & $(\lambda_1,\lambda_2)=(9,4)$, $K=2$ & $(\lambda_1,\lambda_2)=(9,4)$, $K=5$ \\
\includegraphics[width=0.32\textwidth]{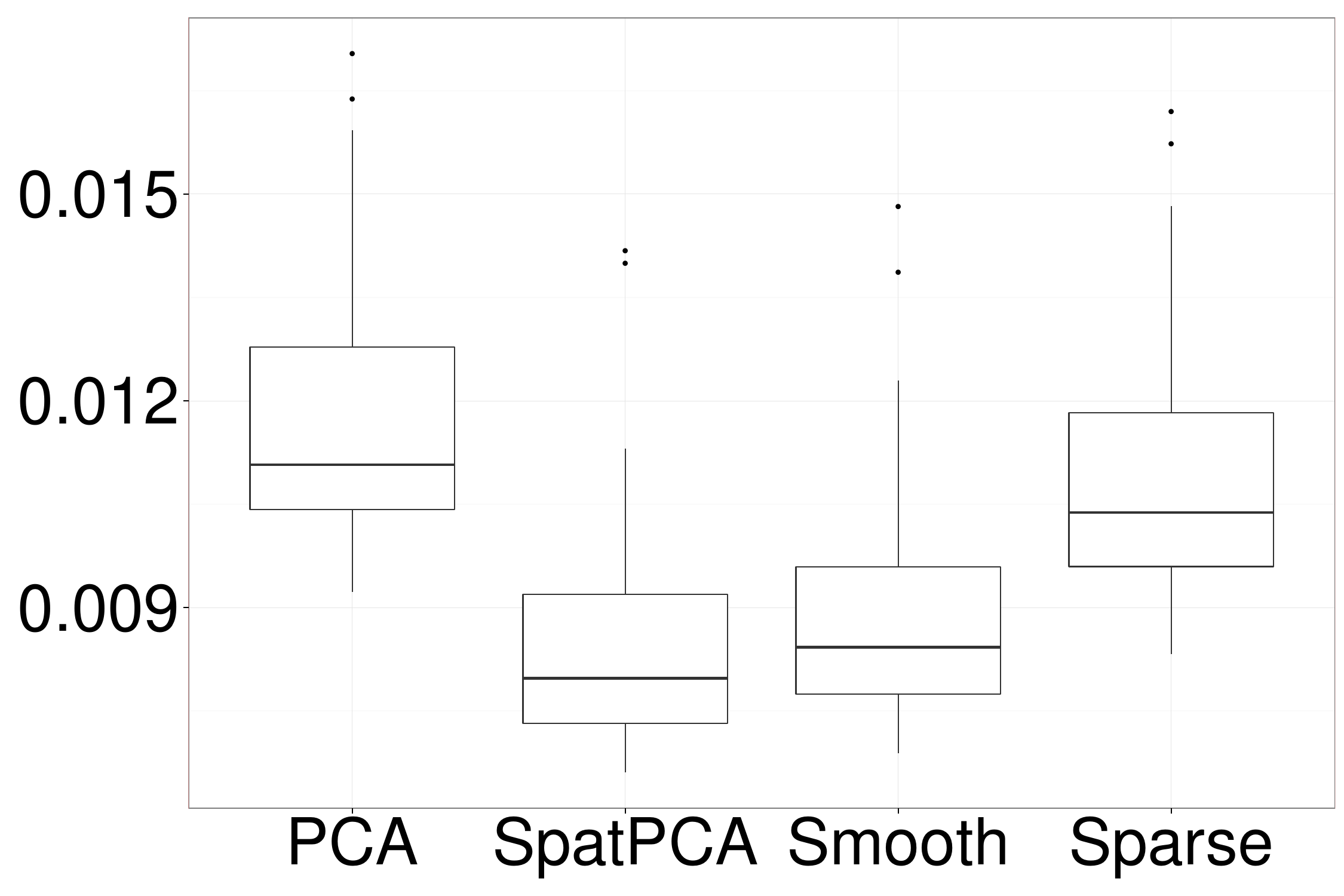} &
\includegraphics[width=0.32\textwidth]{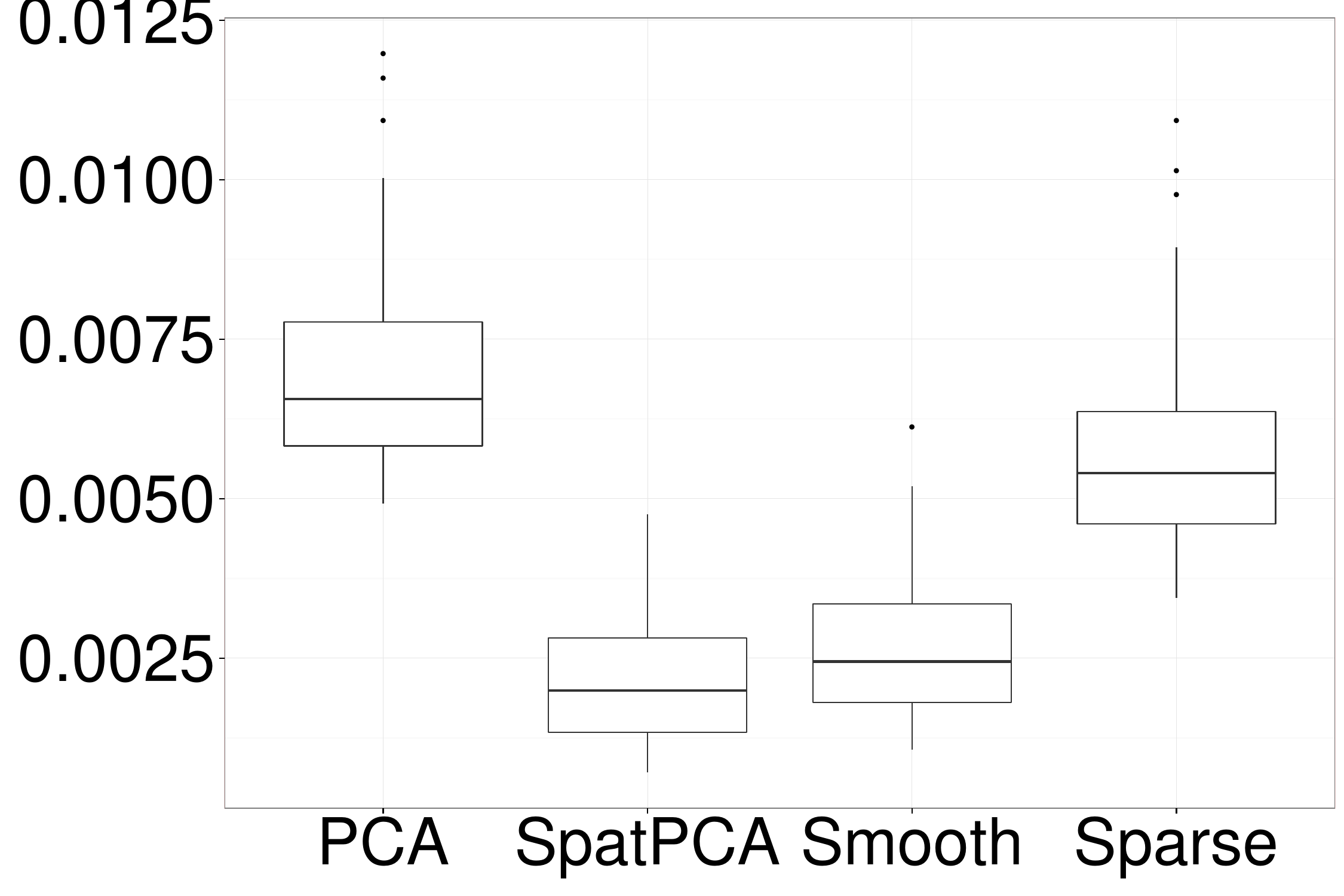} &
\includegraphics[width=0.32\textwidth]{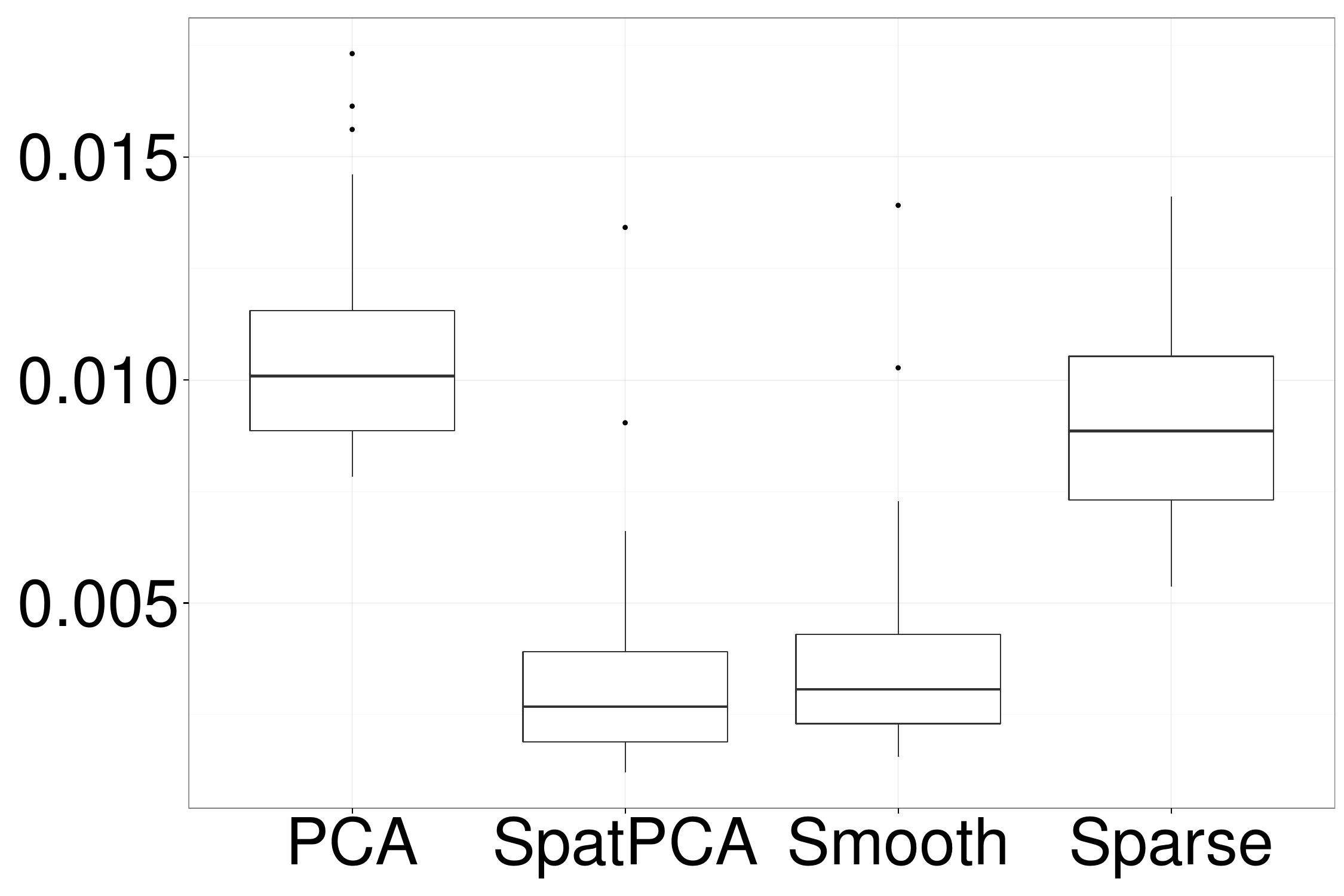}
\end{tabular}}
\caption{Boxplots of average squared estimation errors of (\ref{eq:ch4loss2_sim}) for various methods based on 50 simulation replicates.}
\label{fig:ch4box_d1}
\end{figure}

\subsection{Two-Dimensional Experiment I}

We considered a two-dimensional experiment by generating data according to (\ref{eq:ch4measurement}) with $K=2$,
$\bm{\xi}_i\sim N(\bm{0}, \mathrm{diag}(\lambda_1,\lambda_2))$, $\bm{\epsilon}_i\sim N(\bm{0},\bm{I})$, $n=500$, $\bm{s}_1,\dots,\bm{s}_p$
regularly spaced at $p=20^2$ locations in $D=[-5,5]^2$.
Here $\varphi_1(\cdot)$ and $\varphi_2(\cdot)$ are given by (\ref{eq:ch4phi1_sim}) and (\ref{eq:ch4phi2_sim}) with $d=2$
(see the left panels in Figure~\ref{fig:ch4est_d2}).

As in the one-dimensional experiment,
we considered three pairs of $(\lambda_1,\lambda_2)\in\{(9,0),(1,0),(9,4)\}$, and applied the proposed
SpatPCA with $K\in\{1,2,5\}$ and $\hat{K}$ selected from \eqref{eq:ch4K.hat}.
We used 5-fold CV of \eqref{eq:ch4cv} and a two-step procedure to select among the same 11 values of $\tau_1$
and $31$ values of $\tau_2$.
Similarly, we used 5-fold CV of \eqref{eq:ch4cv.gamma} to select among the same 11 values of $\gamma$
for covariance function estimation.

Figure~\ref{fig:ch4est_d2} shows the estimates of $\varphi_1(\cdot)$ and $\varphi_2(\cdot)$ from the four methods
for various cases based on a randomly generated dataset.
The performance of the four methods in terms of the loss functions \eqref{eq:ch4loss_sim} and \eqref{eq:ch4loss2_sim} is summarized in
Figure~\ref{fig:ch4box_d2_loss1} and Figure~\ref{fig:ch4box_d2} respectively.
Similarly to the one-dimensional examples, SpatPCA performs significantly better than all the other methods in all cases.

\begin{figure}[tbhp]
\centering
\begin{tabular}{@{}c@{}}
\small $\hat{\varphi}_1(\cdot)$ based on $(\lambda_1,\lambda_2)=(9,0)$ and $K=1$ \\
\includegraphics[width=0.95\textwidth]{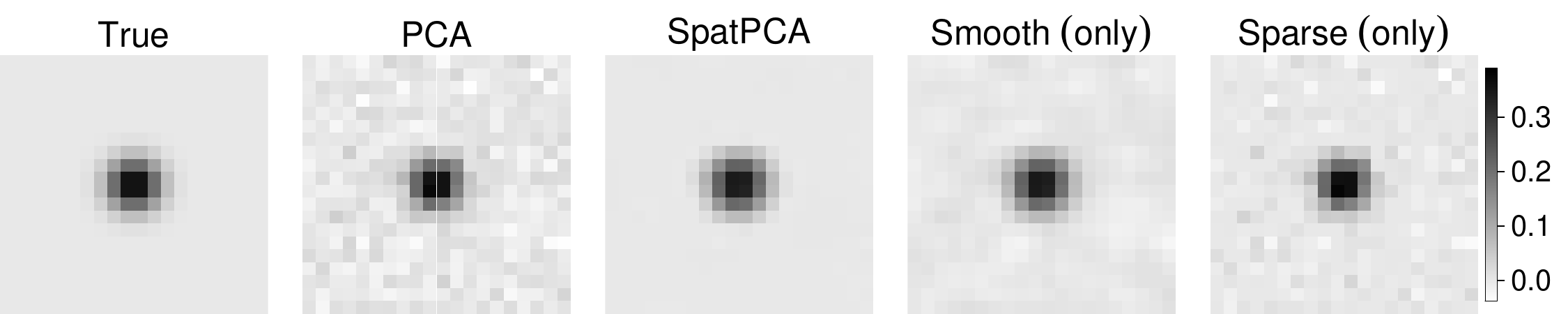} \\[0.4em]
\small $\hat{\varphi}_1(\cdot)$ based on $(\lambda_1,\lambda_2)=(1,0)$ and $K=1$ \\
\includegraphics[width=0.95\textwidth]{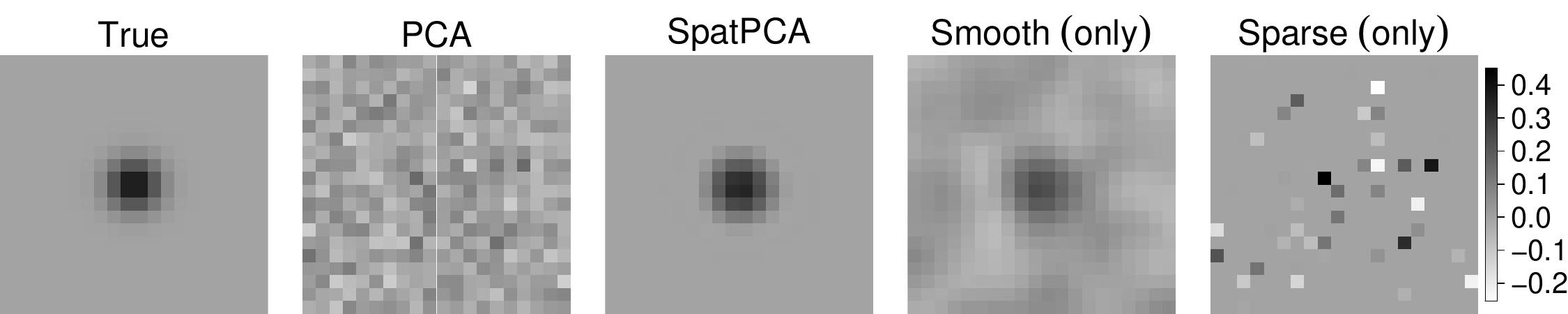} \\[0.4em]
\small $\hat{\varphi}_1(\cdot)$ based on $(\lambda_1,\lambda_2)=(9,4)$ and $K=2$ \\
\includegraphics[width=0.95\textwidth]{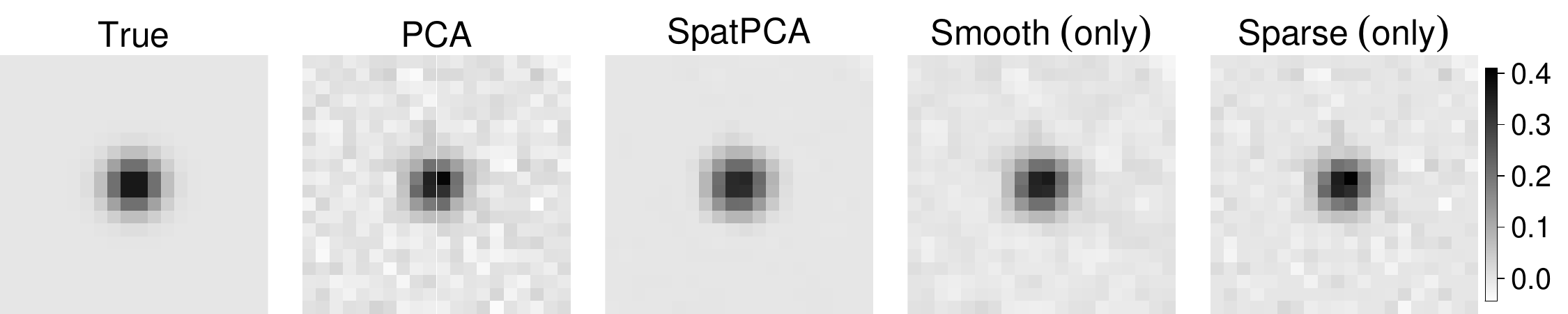} \\[0.4em]
\small $\hat{\varphi}_2(\cdot)$ based on $(\lambda_1,\lambda_2)=(9,4)$ and $K=2$ \\
\includegraphics[width=0.95\textwidth]{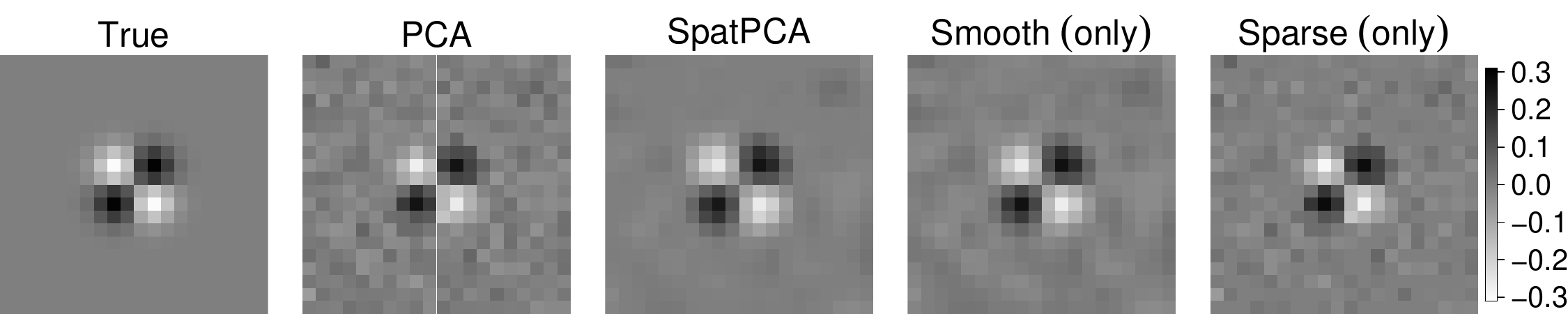}
\end{tabular}
\caption{Estimates of $\varphi_1(\cdot)$ and $\varphi_2(\cdot)$ from various methods for three different combinations of eigenvalues.}
\label{fig:ch4est_d2}
\end{figure}

\begin{figure}[tbhp]
\centering
{\footnotesize
\begin{tabular}{@{}ccc@{}}
$(\lambda_1,\lambda_2)=(9,0)$, $K=1$ & $(\lambda_1,\lambda_2)=(9,0)$, $K=2$ & $(\lambda_1,\lambda_2)=(9,0)$, $K=5$ \\
\includegraphics[width=0.32\textwidth]{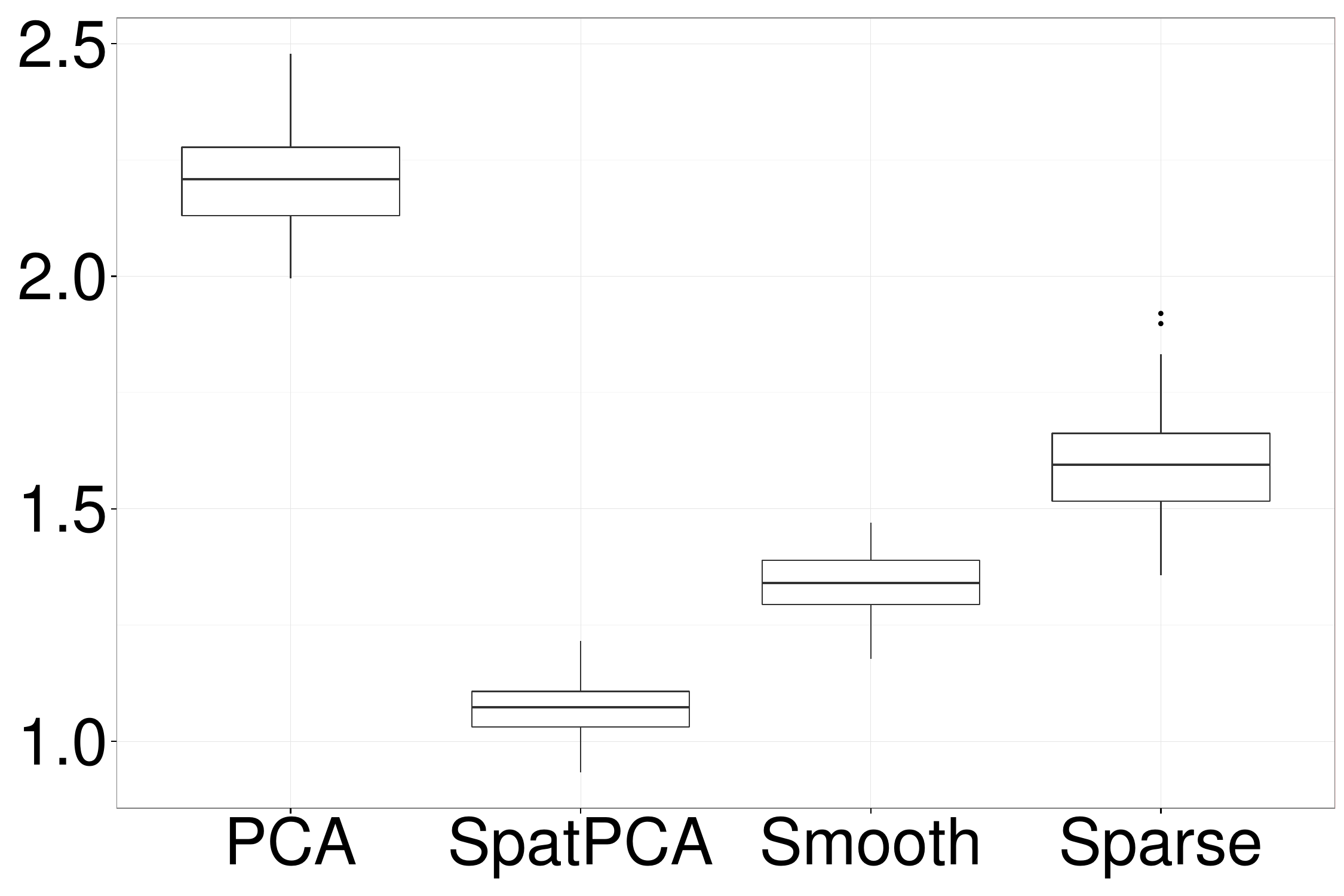} &
\includegraphics[width=0.32\textwidth]{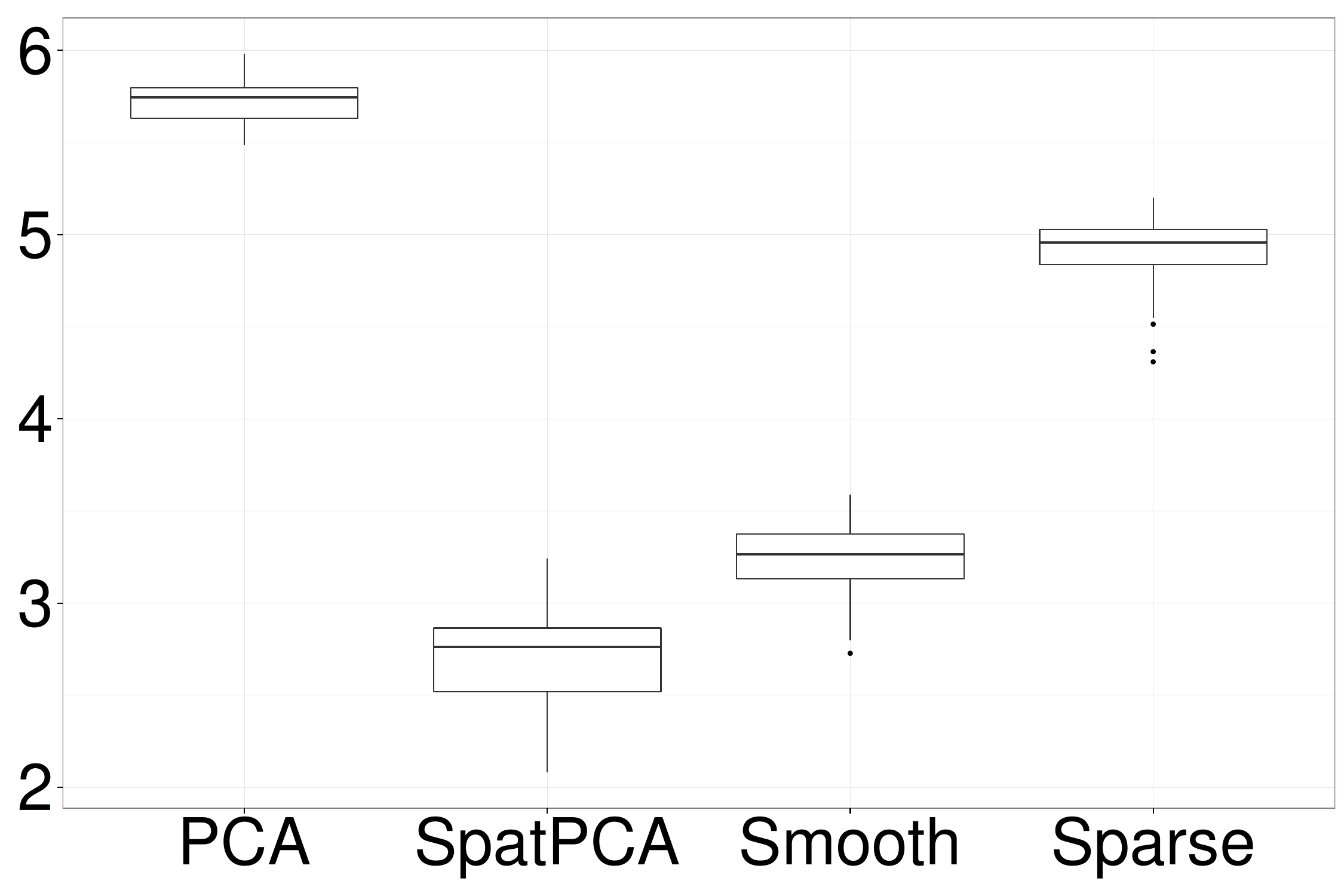} &
\includegraphics[width=0.32\textwidth]{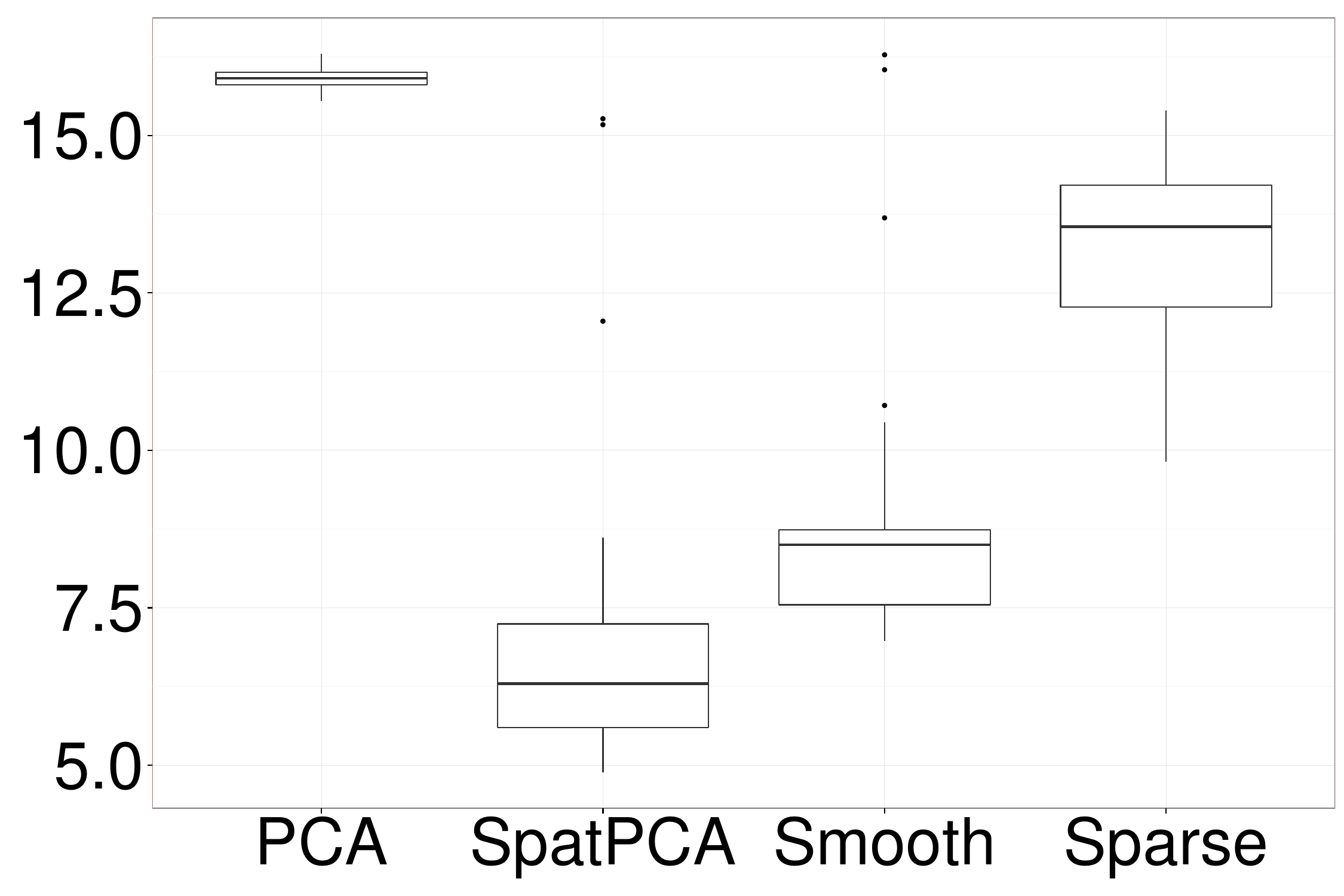} \\[0.4em]
$(\lambda_1,\lambda_2)=(1,0)$, $K=1$ & $(\lambda_1,\lambda_2)=(1,0)$, $K=2$ & $(\lambda_1,\lambda_2)=(1,0)$, $K=5$ \\
\includegraphics[width=0.32\textwidth]{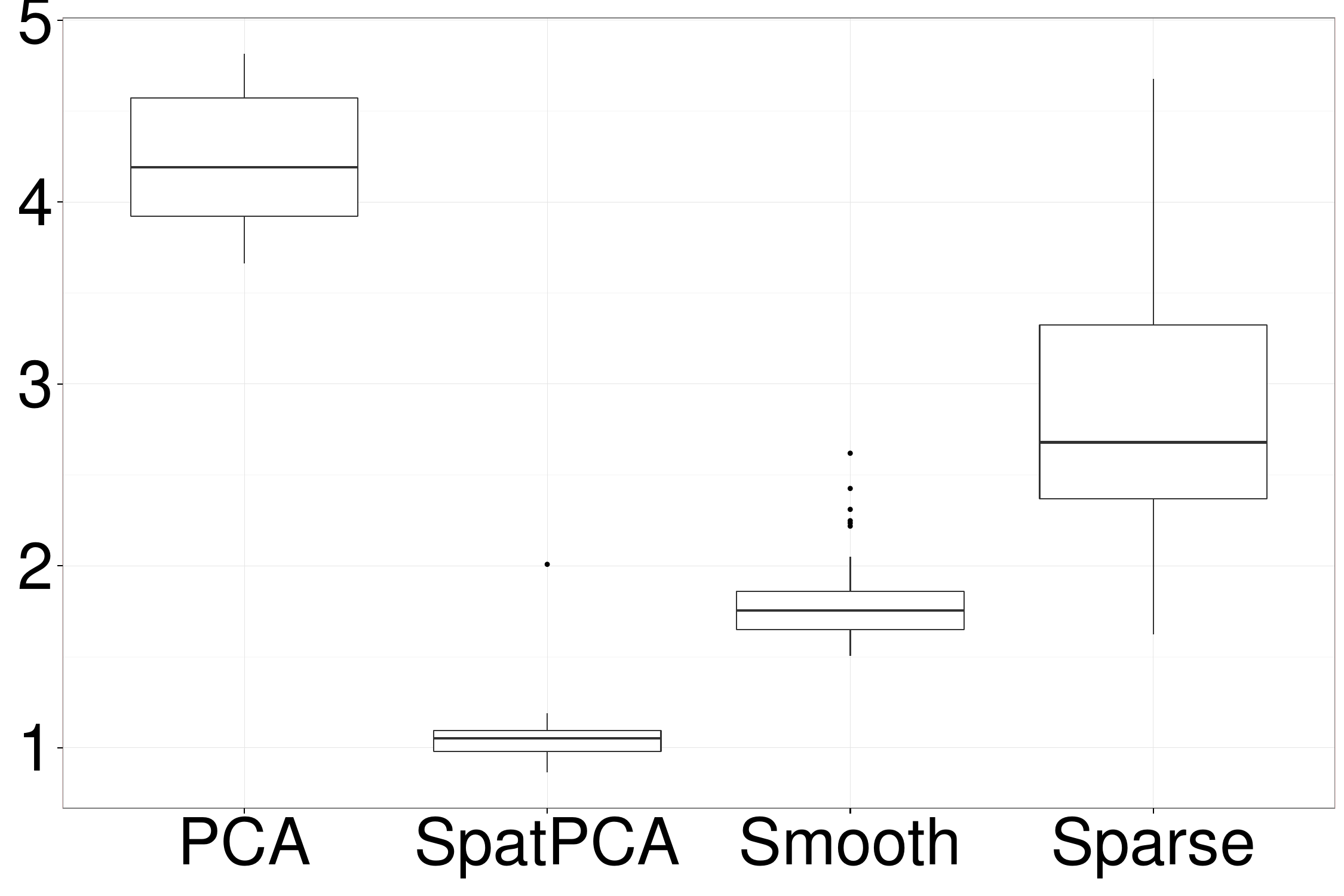} &
\includegraphics[width=0.32\textwidth]{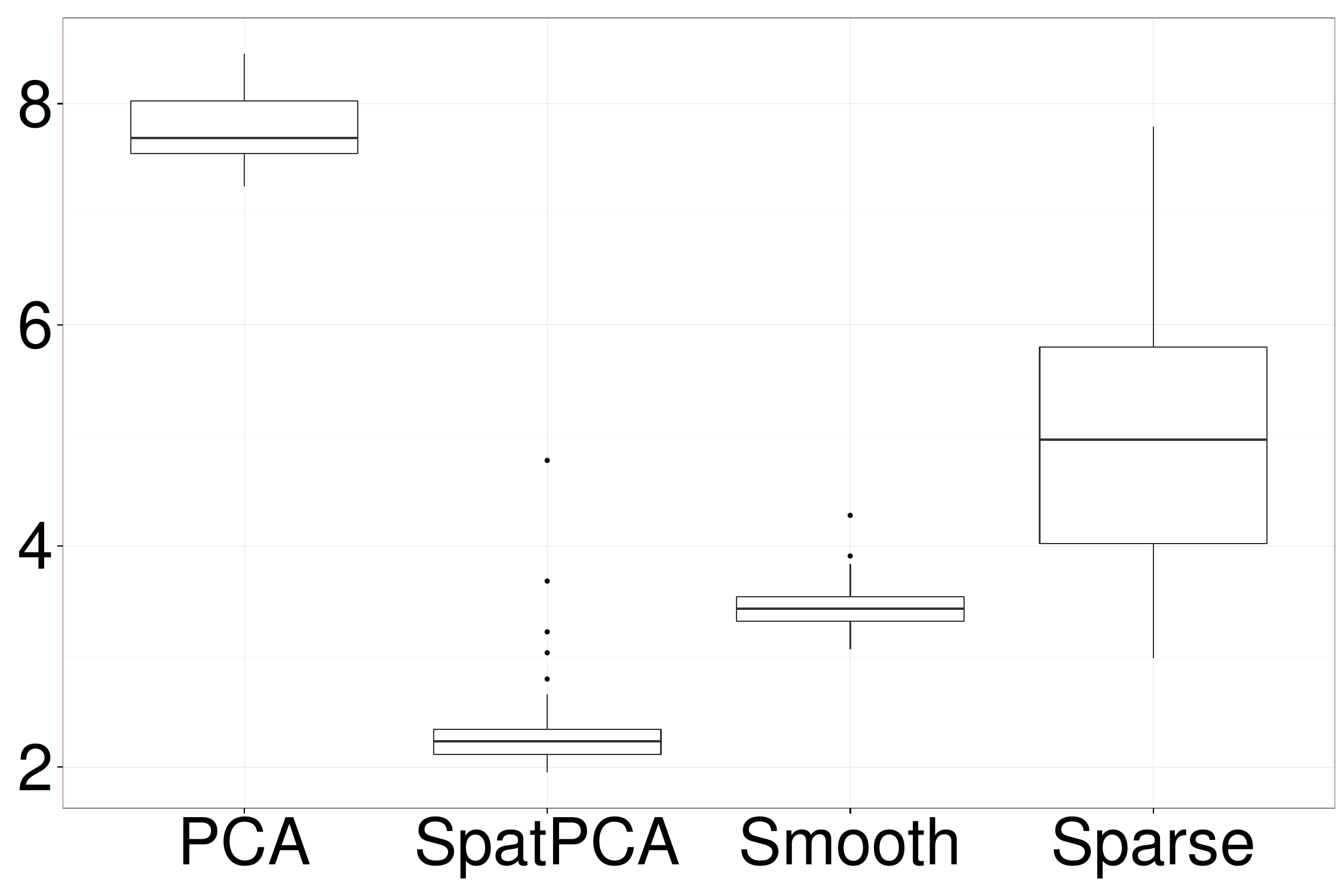} &
\includegraphics[width=0.32\textwidth]{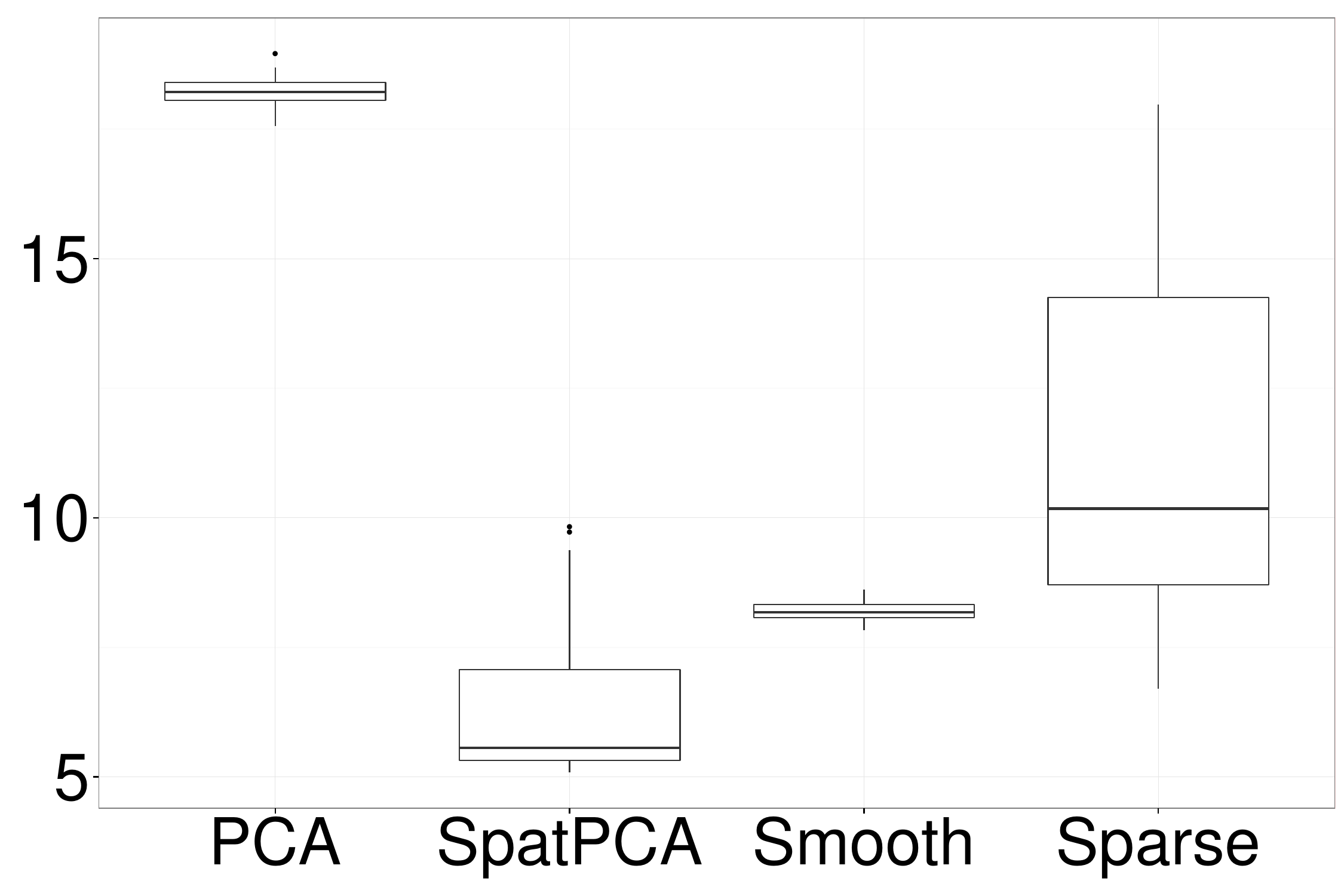} \\[0.4em]
$(\lambda_1,\lambda_2)=(9,4)$, $K=1$ & $(\lambda_1,\lambda_2)=(9,4)$, $K=2$ & $(\lambda_1,\lambda_2)=(9,4)$, $K=5$ \\
\includegraphics[width=0.32\textwidth]{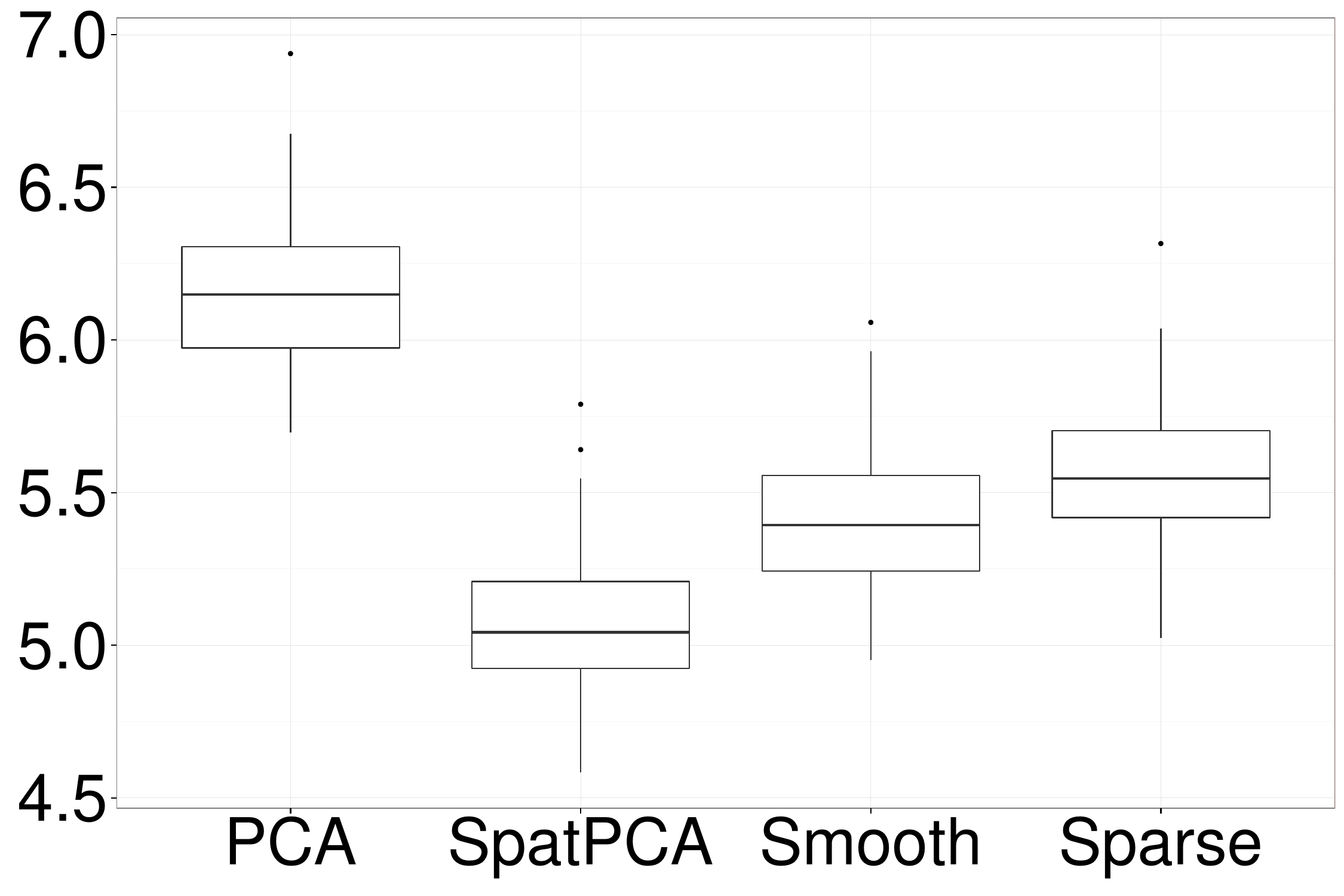} &
\includegraphics[width=0.32\textwidth]{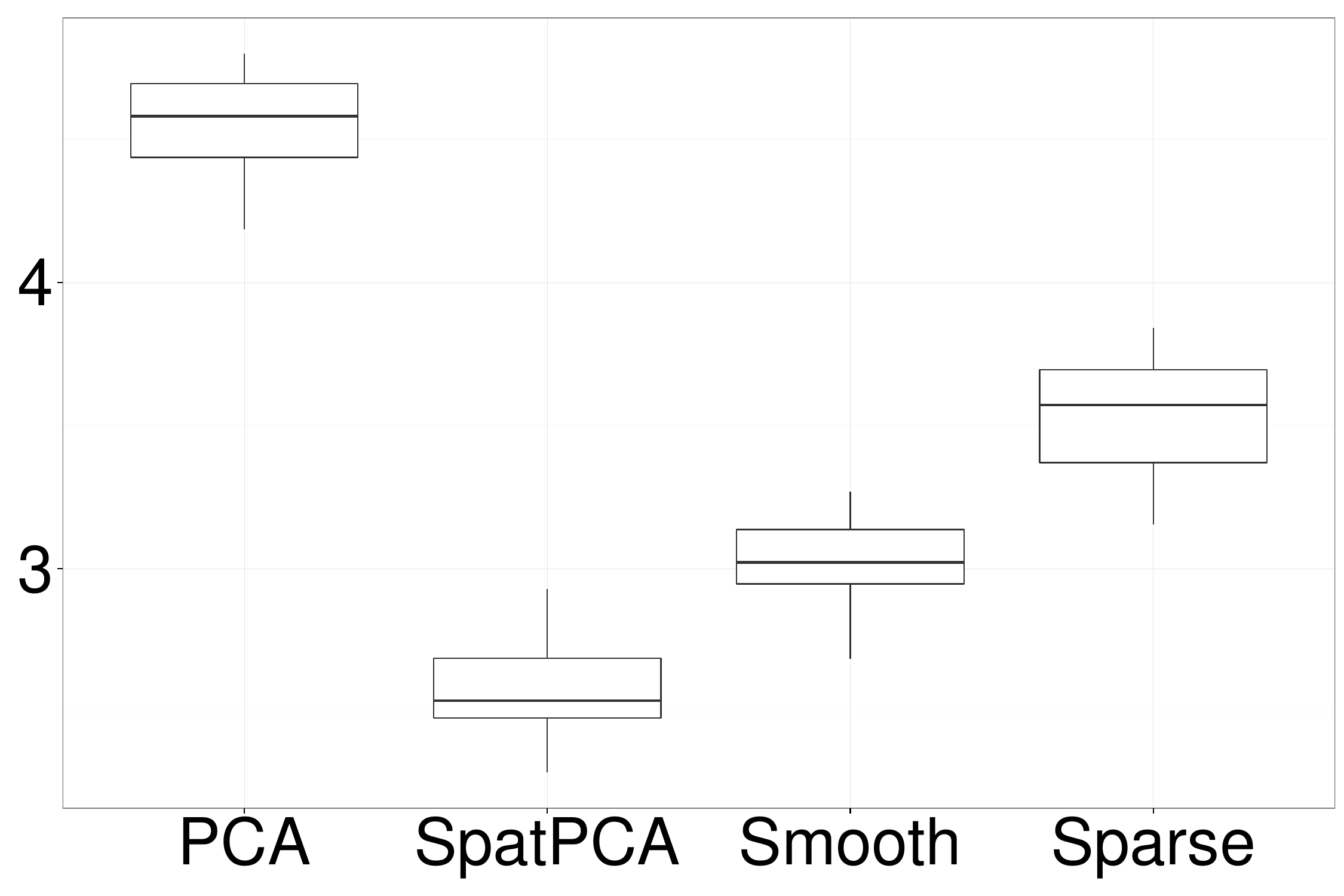} &
\includegraphics[width=0.32\textwidth]{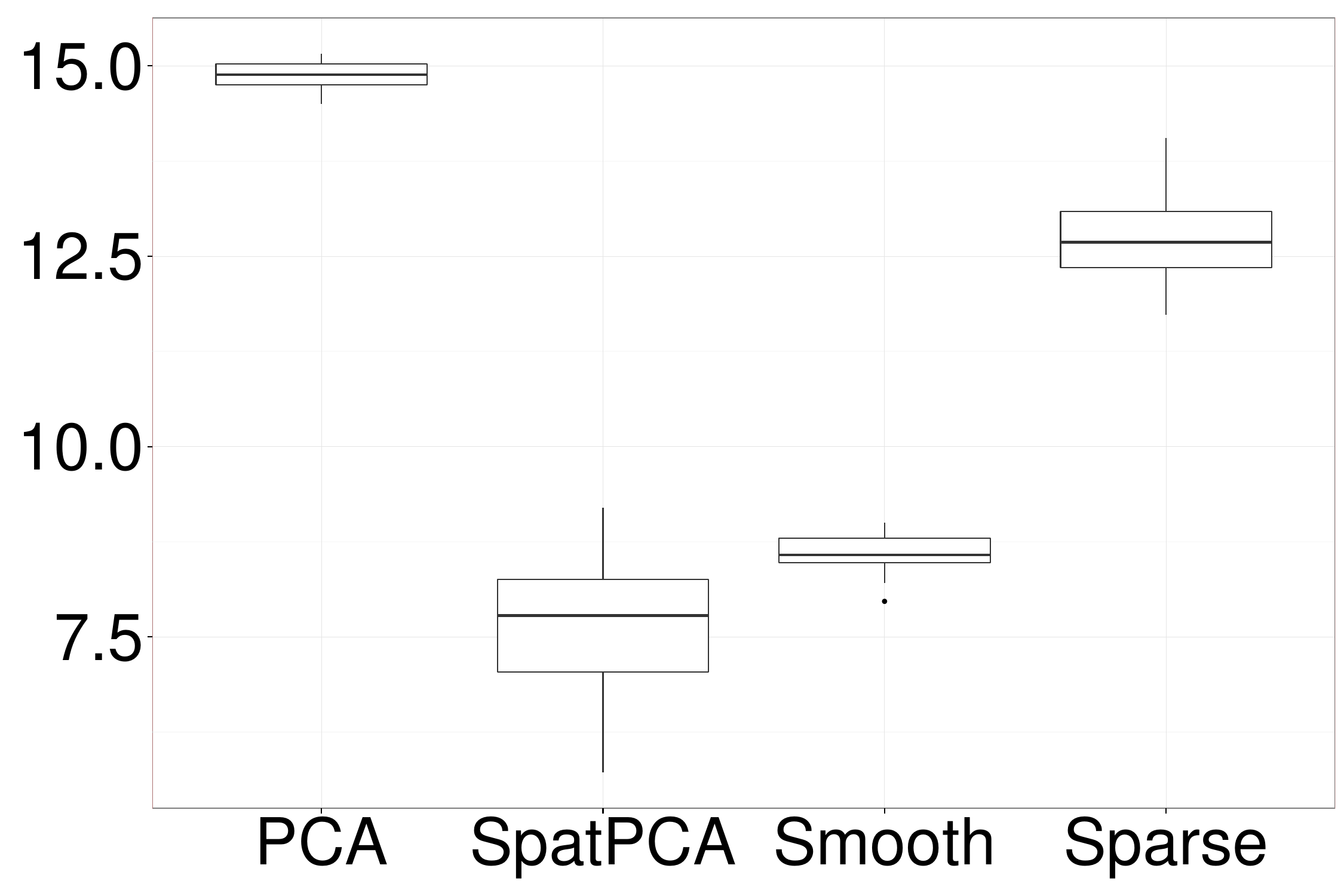}
\end{tabular}}
\caption{Boxplots of average squared prediction errors of (\ref{eq:ch4loss_sim}) for various methods based on 50 simulation replicates.}
\label{fig:ch4box_d2_loss1}
\end{figure}

\begin{figure}[tbhp]
\centering
{\footnotesize
\begin{tabular}{@{}ccc@{}}
$(\lambda_1,\lambda_2)=(9,0)$, $K=1$ & $(\lambda_1,\lambda_2)=(9,0)$, $K=2$ & $(\lambda_1,\lambda_2)=(9,0)$, $K=5$ \\
\includegraphics[width=0.32\textwidth]{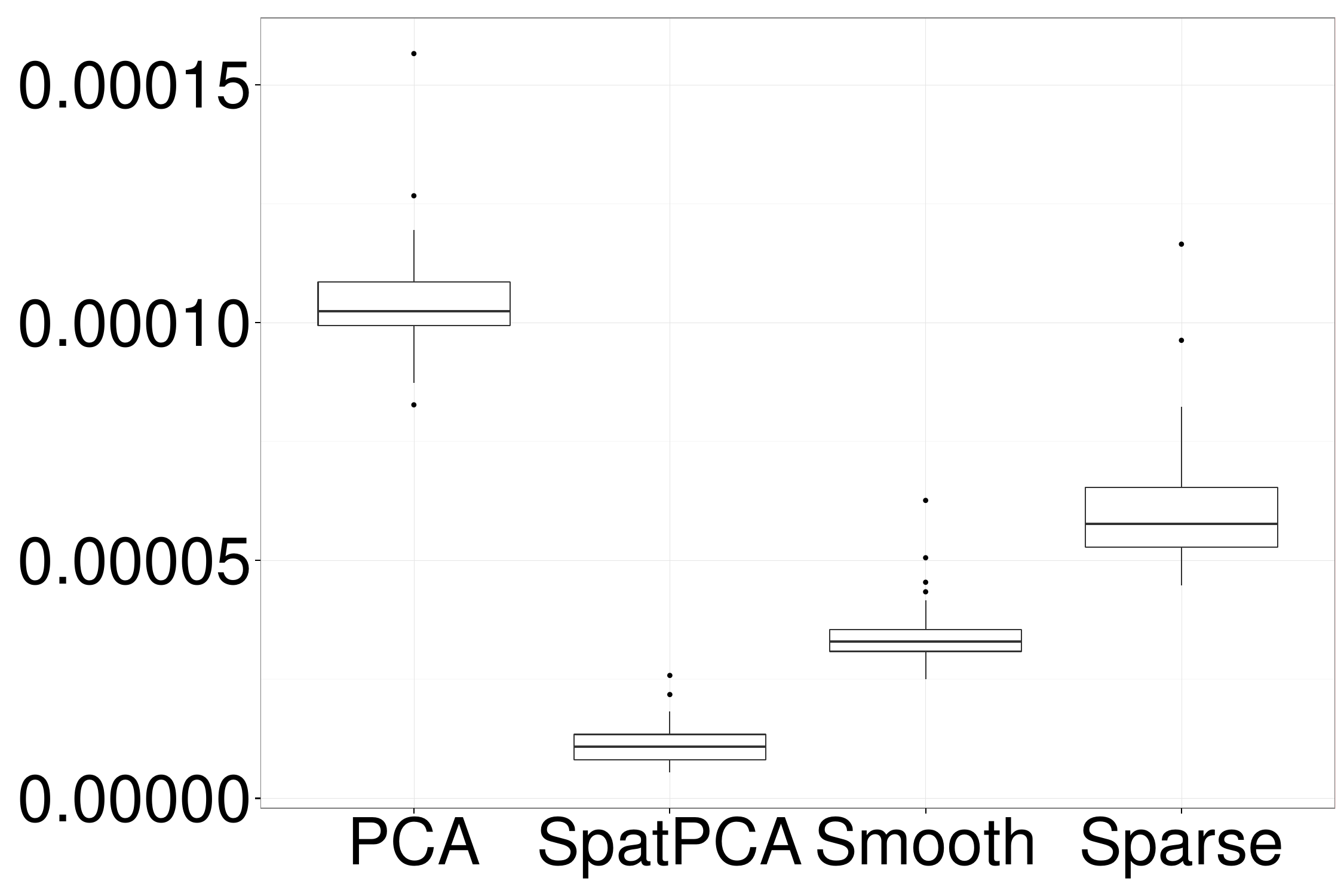} &
\includegraphics[width=0.32\textwidth]{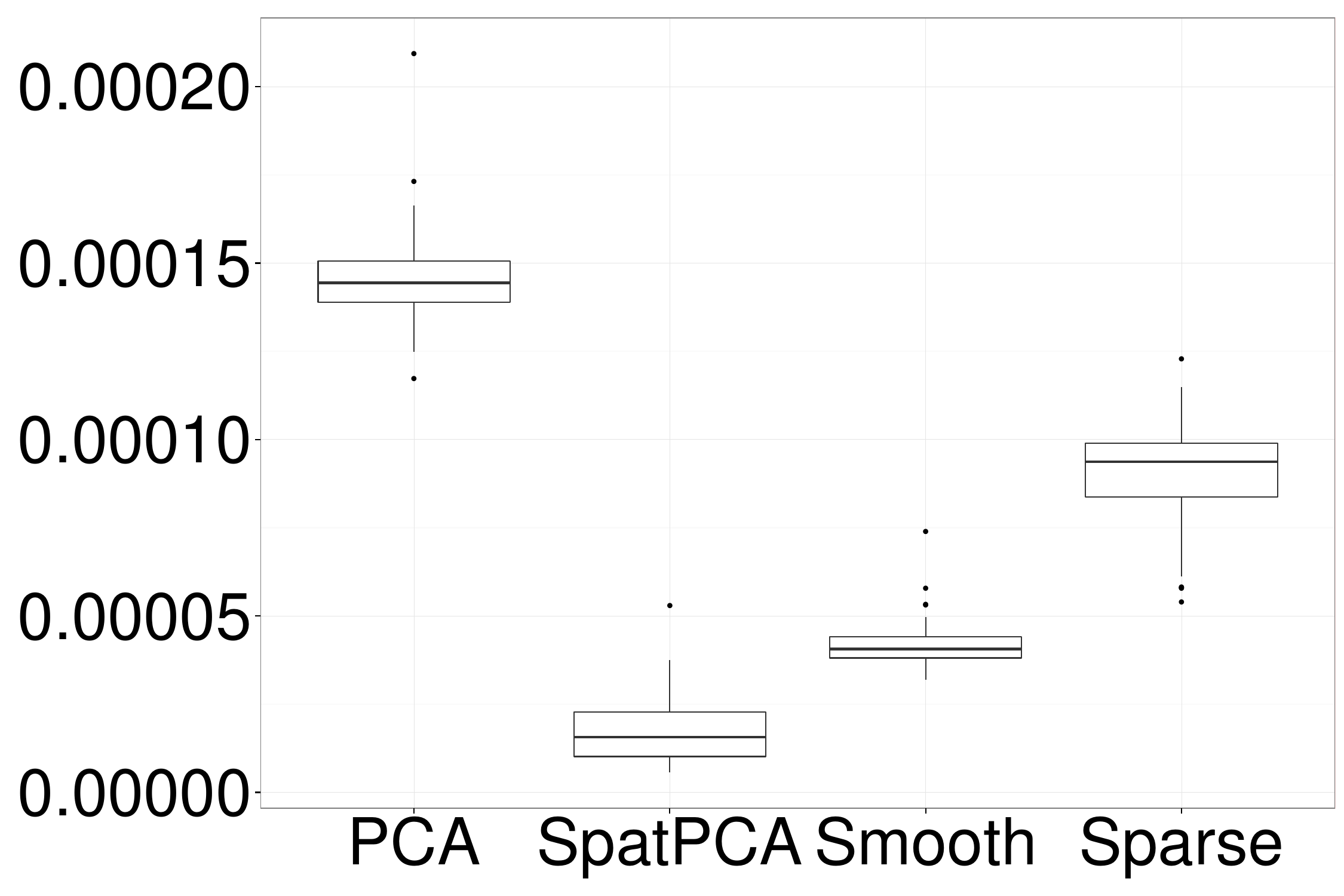} &
\includegraphics[width=0.32\textwidth]{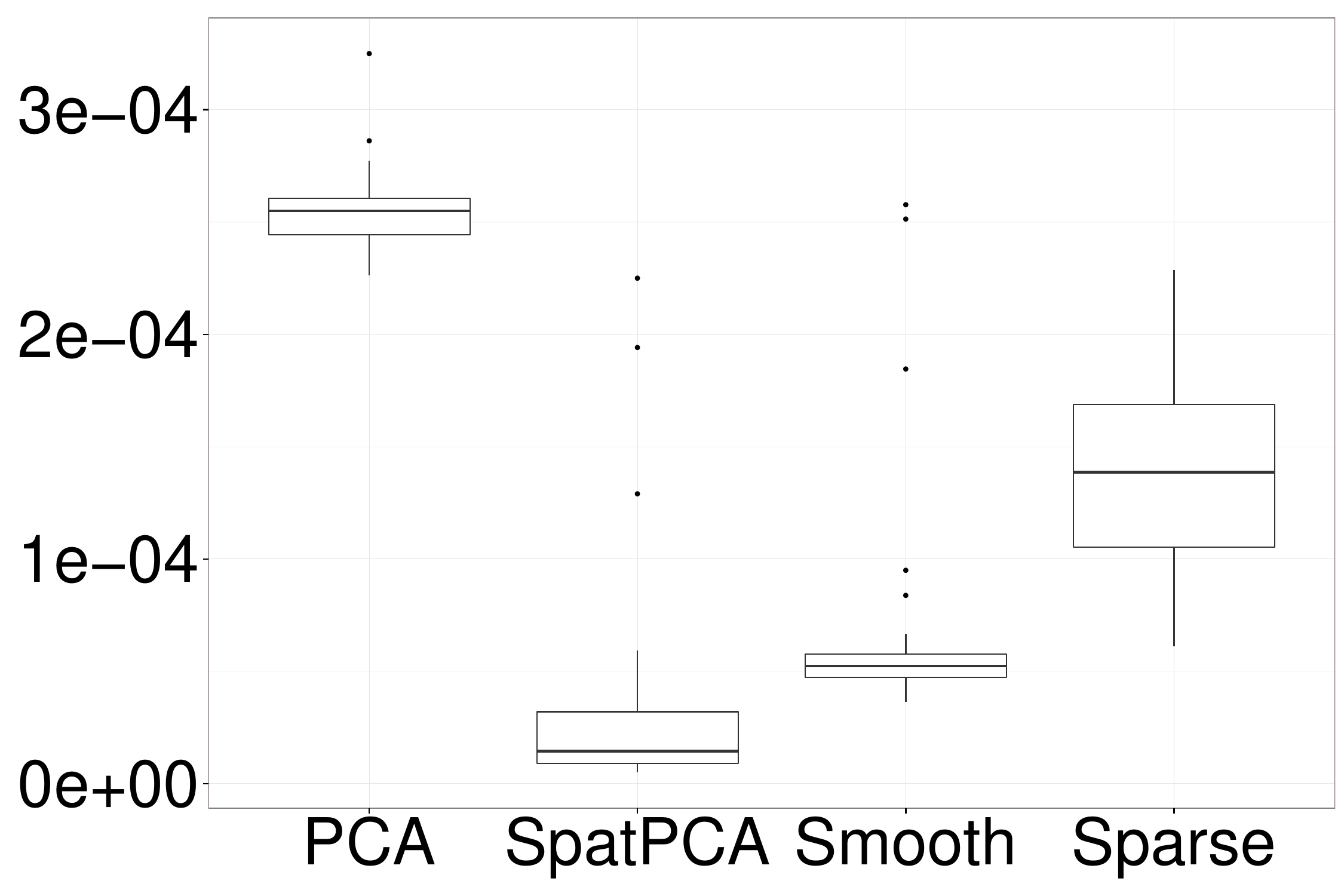} \\[0.4em]
$(\lambda_1,\lambda_2)=(1,0)$, $K=1$ & $(\lambda_1,\lambda_2)=(1,0)$, $K=2$ & $(\lambda_1,\lambda_2)=(1,0)$, $K=5$ \\
\includegraphics[width=0.32\textwidth]{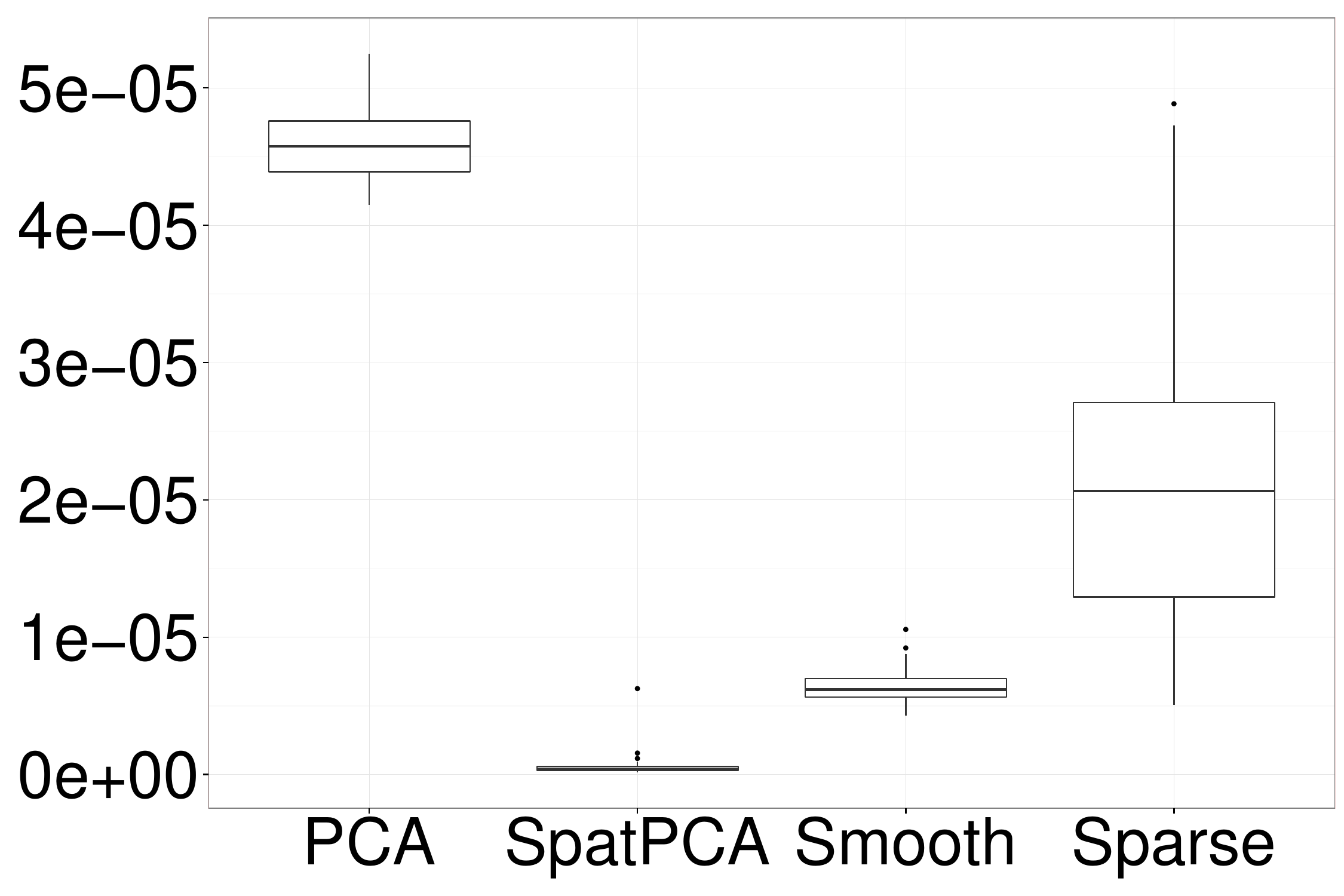} &
\includegraphics[width=0.32\textwidth]{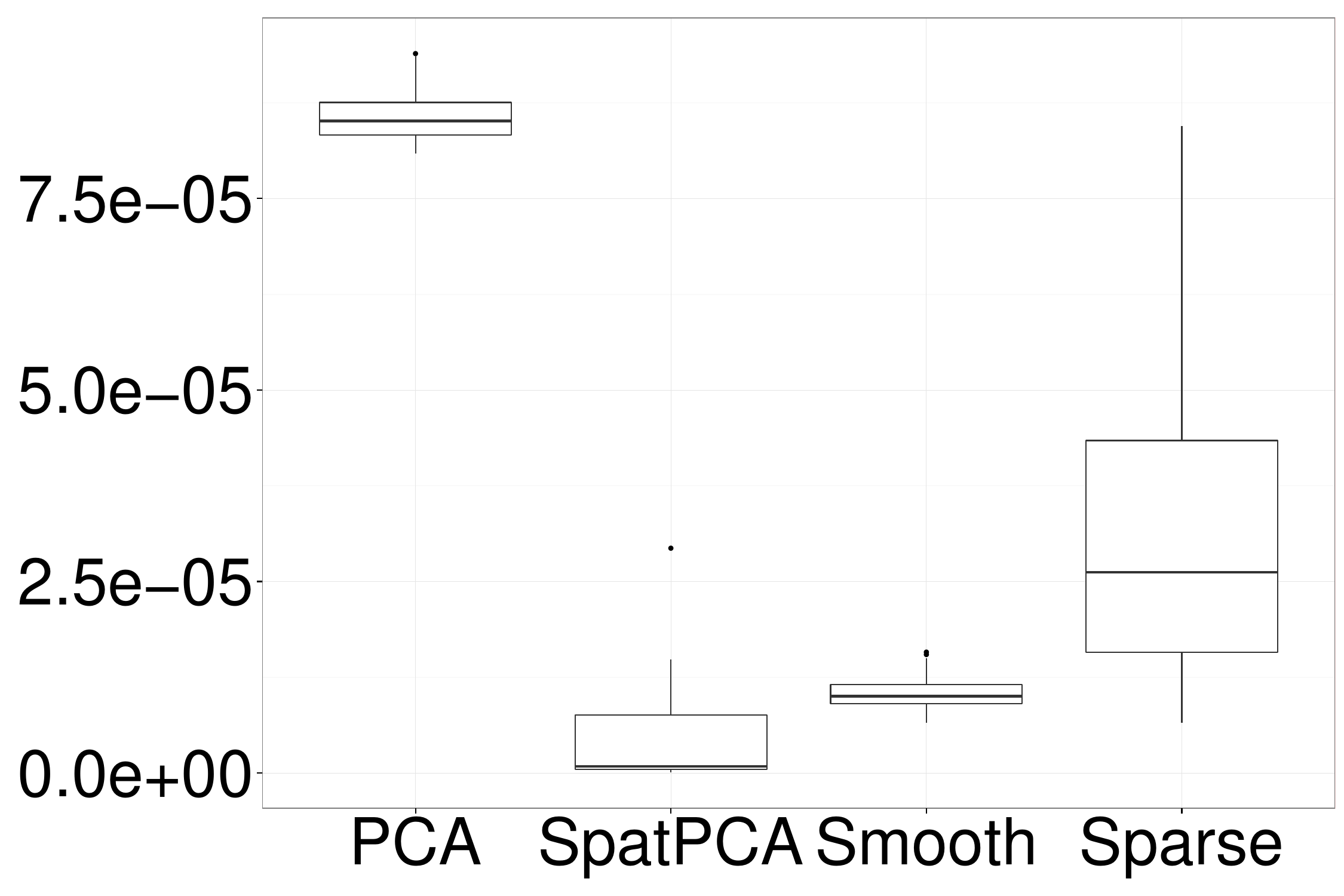} &
\includegraphics[width=0.32\textwidth]{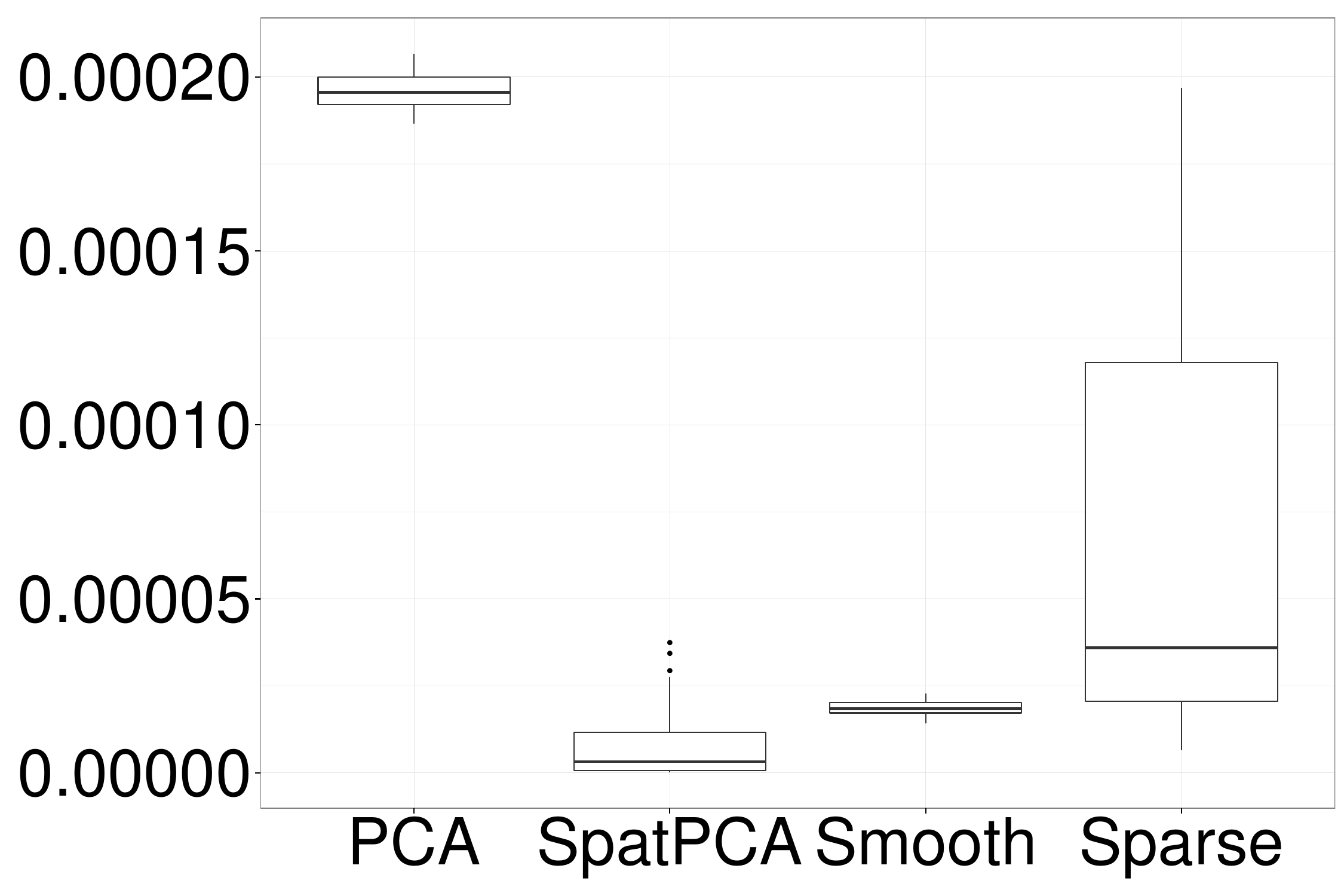} \\[0.4em]
$(\lambda_1,\lambda_2)=(9,4)$, $K=1$ & $(\lambda_1,\lambda_2)=(9,4)$, $K=2$ & $(\lambda_1,\lambda_2)=(9,4)$, $K=5$ \\
\includegraphics[width=0.32\textwidth]{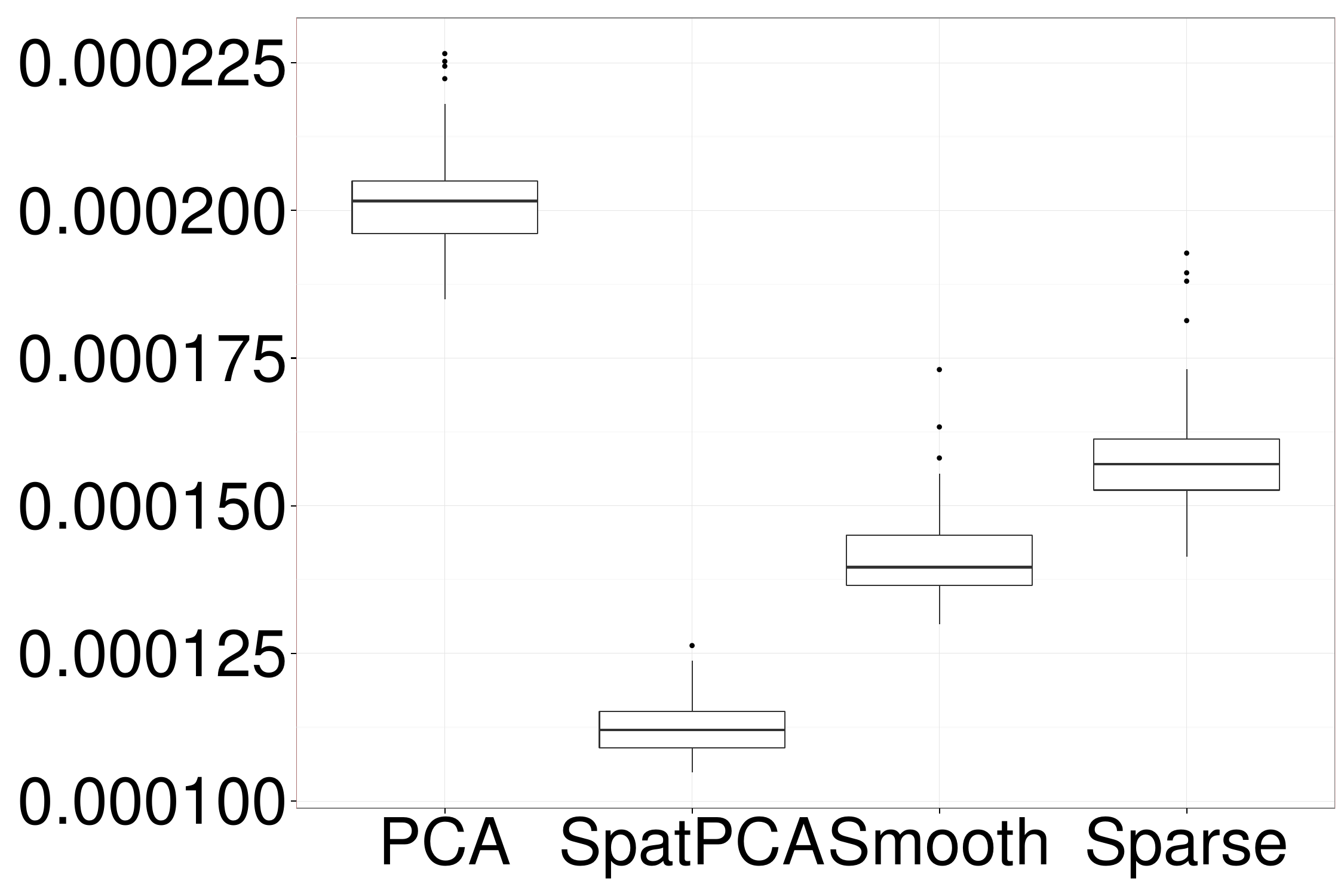} &
\includegraphics[width=0.32\textwidth]{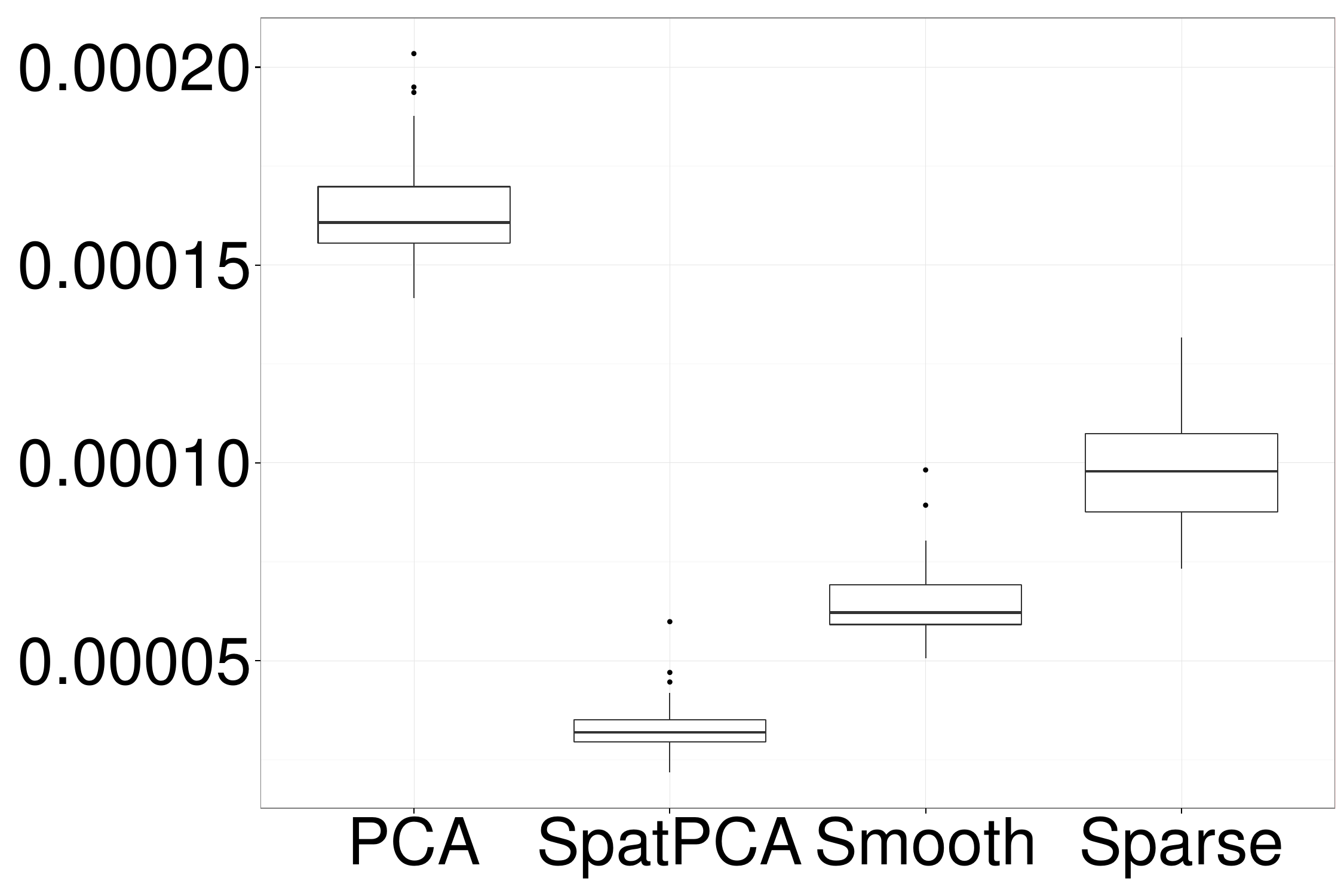} &
\includegraphics[width=0.32\textwidth]{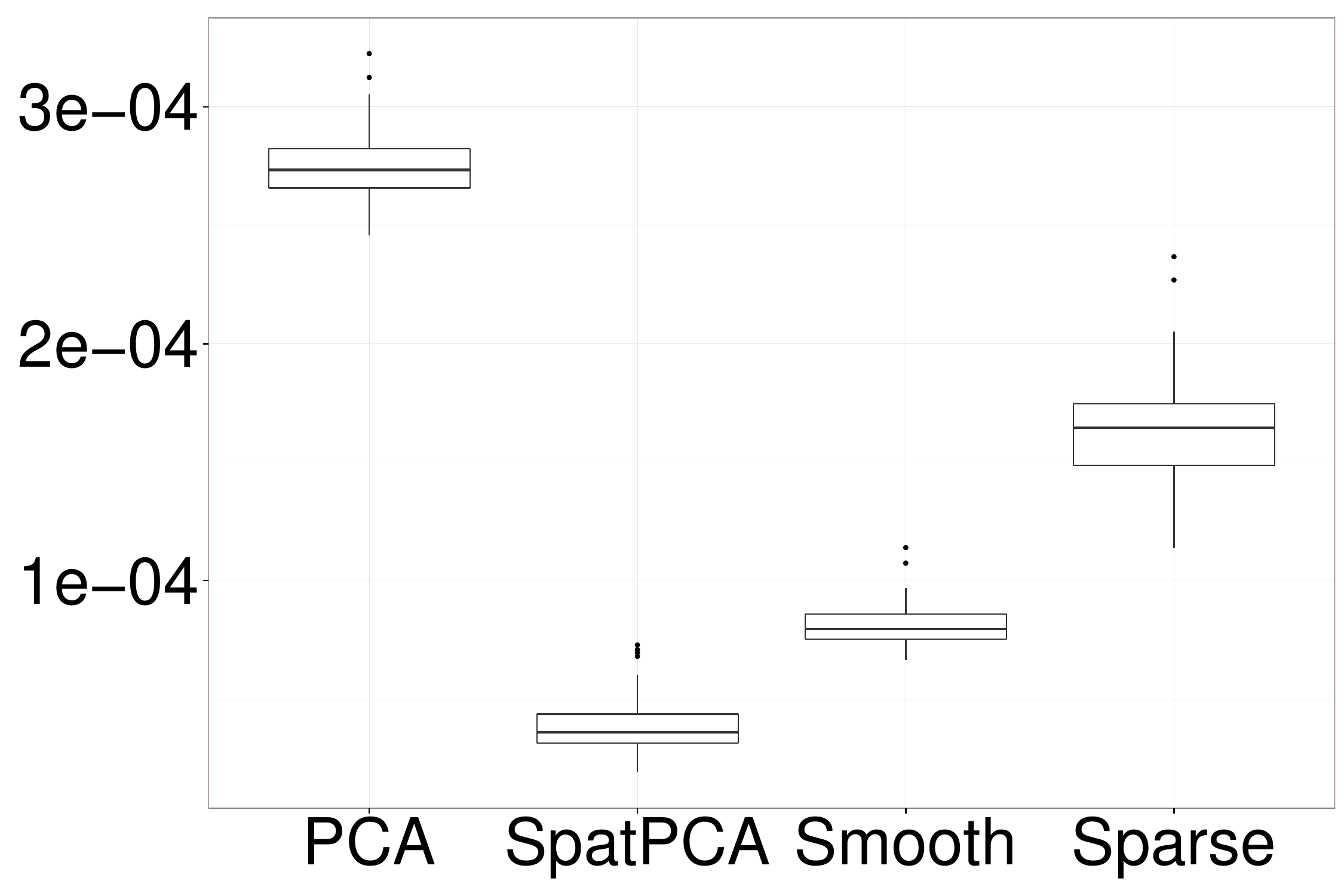}
\end{tabular}}
\caption{Boxplots of average squared estimation errors of (\ref{eq:ch4loss2_sim}) for various methods based on 50 simulation replicates.}
\label{fig:ch4box_d2}
\end{figure}

\subsection{An Application to a Sea Surface Temperature Dataset}
\label{sec:ch4application}

Since the proposed SpatPCA works better with both smoothness and sparseness penalties according to the simulation experiments,
we applied only the proposed SpatPCA with both penalty terms to a sea surface temperature (SST) dataset observed in the Indian Ocean.
The data are monthly averages of SST obtained from the Met Office Marine Data Bank
(available at \url{http://www.metoffice.gov.uk/hadobs/hadisst/}) on $1$ degree latitude by 1 degree longitude ($1 ^{\circ}\times 1 ^{\circ}$)
equiangular grid cells
from January 2001 to December 2010 in a region between latitudes $20^\circ N$ and $20^\circ S$ and between longitudes $39^\circ E$ and $120^\circ E$.
Out of $40\times 81=3{,}240$ grid cells, there are 420 cells on the land where no data are available.
Hence the data we used are observed at $p=2{,}820$ cells and $120$ time points.
We first detrended the SST data by subtracting the SST for a given cell and a given month by the average SST for that cell and that month
over the whole period. Then we randomly decomposed the data into two parts with each part consisting of $60$ time points.
One part was used for training data, while the other part was used for validation purpose.

We applied SpatPCA on the training data with $K$ selected by $\hat{K}$ of \eqref{eq:ch4K.hat}.
Similar to the two-step method described in the one-dimensional experiment,
we selected among 11 values of $\tau_1$ (including $0$, and 10 values from $10^3$
to $10^8$ equally spaced on the log scale) and 31 values of $\tau_2$ (including $0$, and $30$ values from $1$ to $10^3$
equally spaced on the log scale) using 5-fold CV of \eqref{eq:ch4cv}.
Figure~\ref{fig:ch4scree_plot} shows the scree plot of the sample eigenvalues based on the training data.

\begin{figure}[tbhp]
\centering
\includegraphics[width=0.55\textwidth]{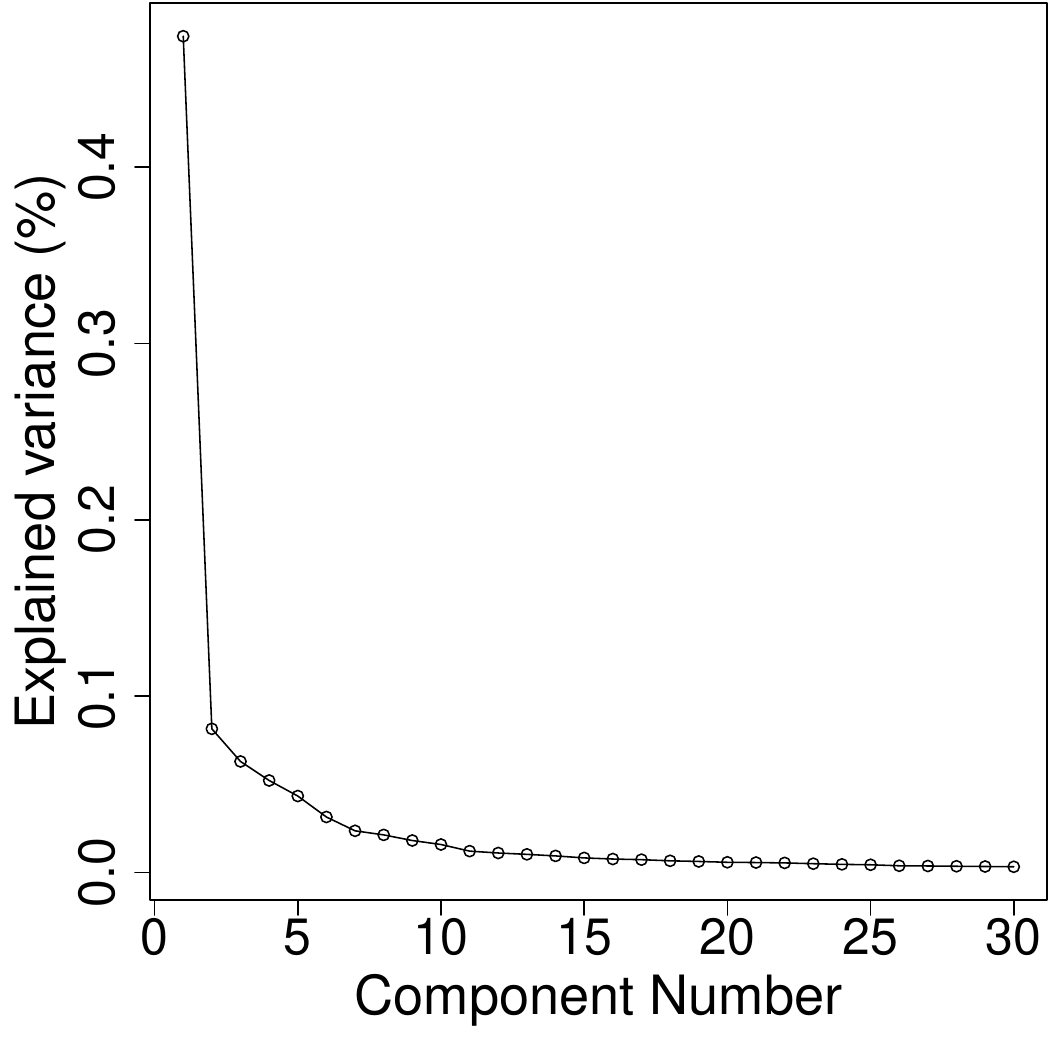}
\caption{Scree plot of sample eigenvalues based on the training data.}
\label{fig:ch4scree_plot}
\end{figure}

\begin{figure}[tbhp]
\centering
\begin{tabular}{@{}c@{}}
\small $\hat{\varphi}_1(\cdot)$ from PCA\hspace{0.35\textwidth}$\hat{\varphi}_1(\cdot)$ from SpatPCA \\
\includegraphics[width=0.95\textwidth]{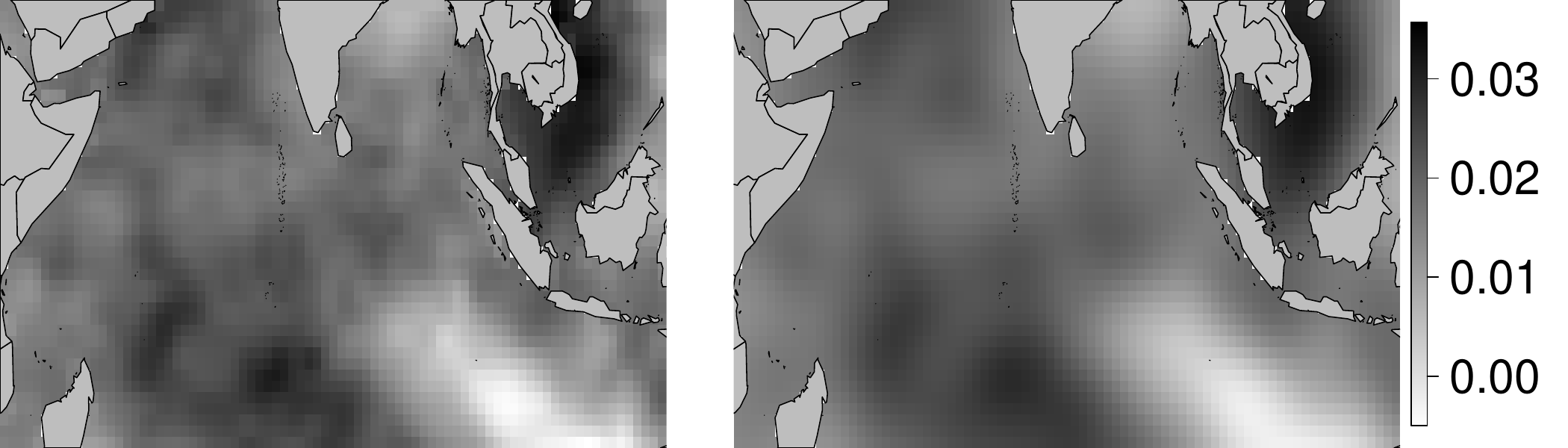} \\[0.6em]
\small $\hat{\varphi}_2(\cdot)$ from PCA\hspace{0.35\textwidth}$\hat{\varphi}_2(\cdot)$ from SpatPCA \\
\includegraphics[width=0.95\textwidth]{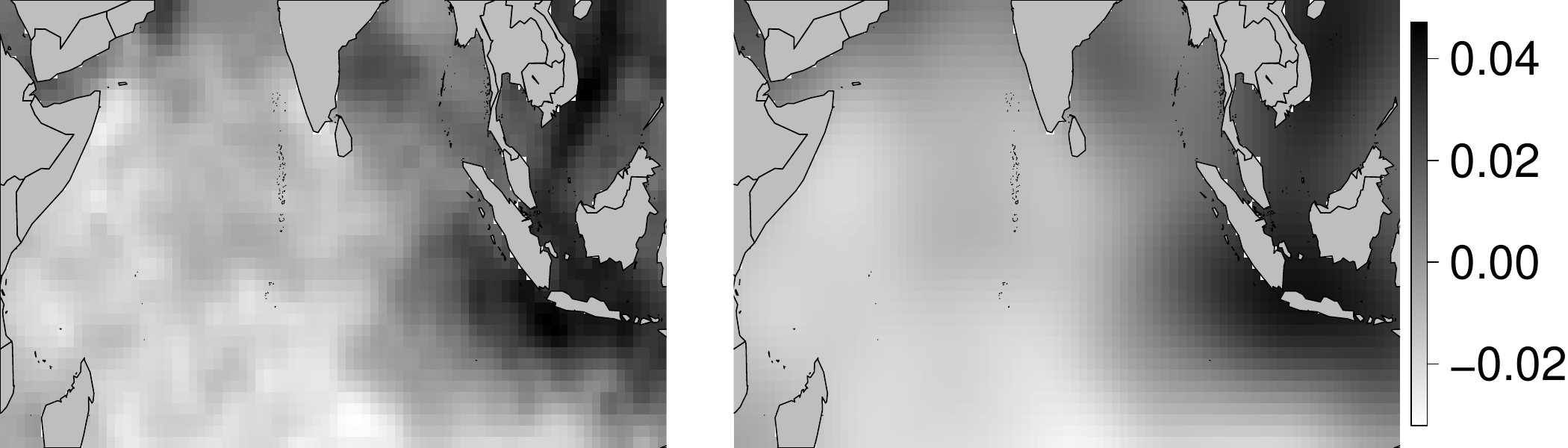} \\[0.6em]
\small $\hat{\varphi}_3(\cdot)$ from PCA\hspace{0.35\textwidth}$\hat{\varphi}_3(\cdot)$ from SpatPCA \\
\includegraphics[width=0.95\textwidth]{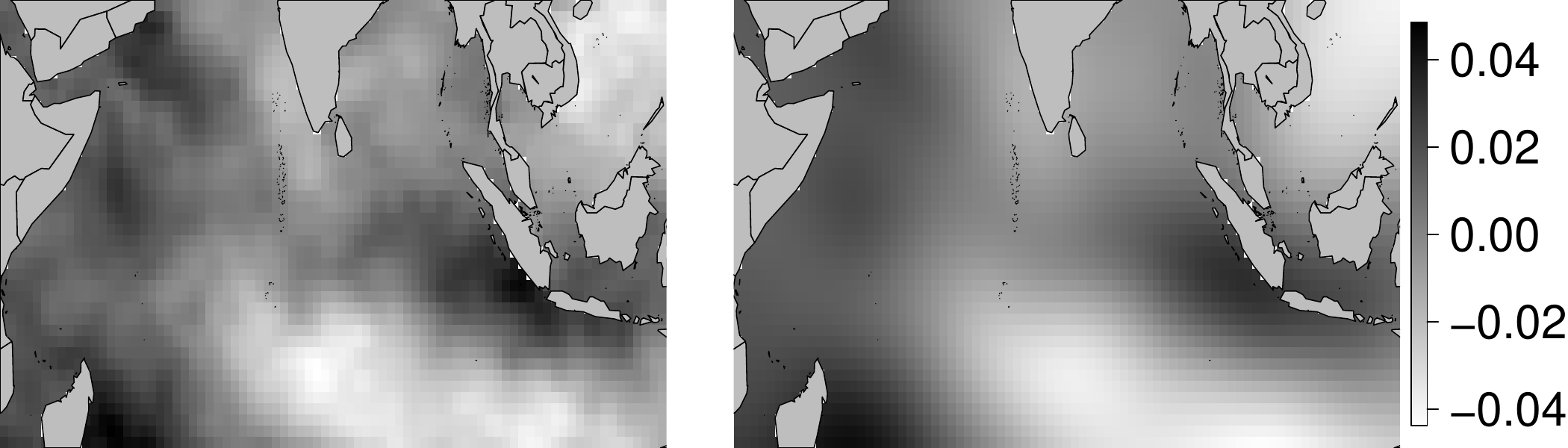}
\end{tabular}
\caption{Estimated eigenimages for PCA and SpatPCA over a region of the Indian Ocean, where the gray regions correspond to the land.}
\label{fig:ch4sst_real}
\end{figure}

As shown in Figure~\ref{fig:ch4cv_plot}, the smallest CV value occurs at $(\tau_1,\tau_2)=(46416,6.7)$.
The first three patterns estimated from PCA and SpatPCA are shown in Figure~\ref{fig:ch4sst_real}.
Both methods identify similar patterns with the ones estimated from SpatPCA
being a bit smoother than those estimated from PCA. The first pattern
is basically the basin-wide mode and the second pattern corresponds to the east-west dipole mode (\citet{sst}).

\begin{figure}[tbp]
\centering
\includegraphics[width=\textwidth]{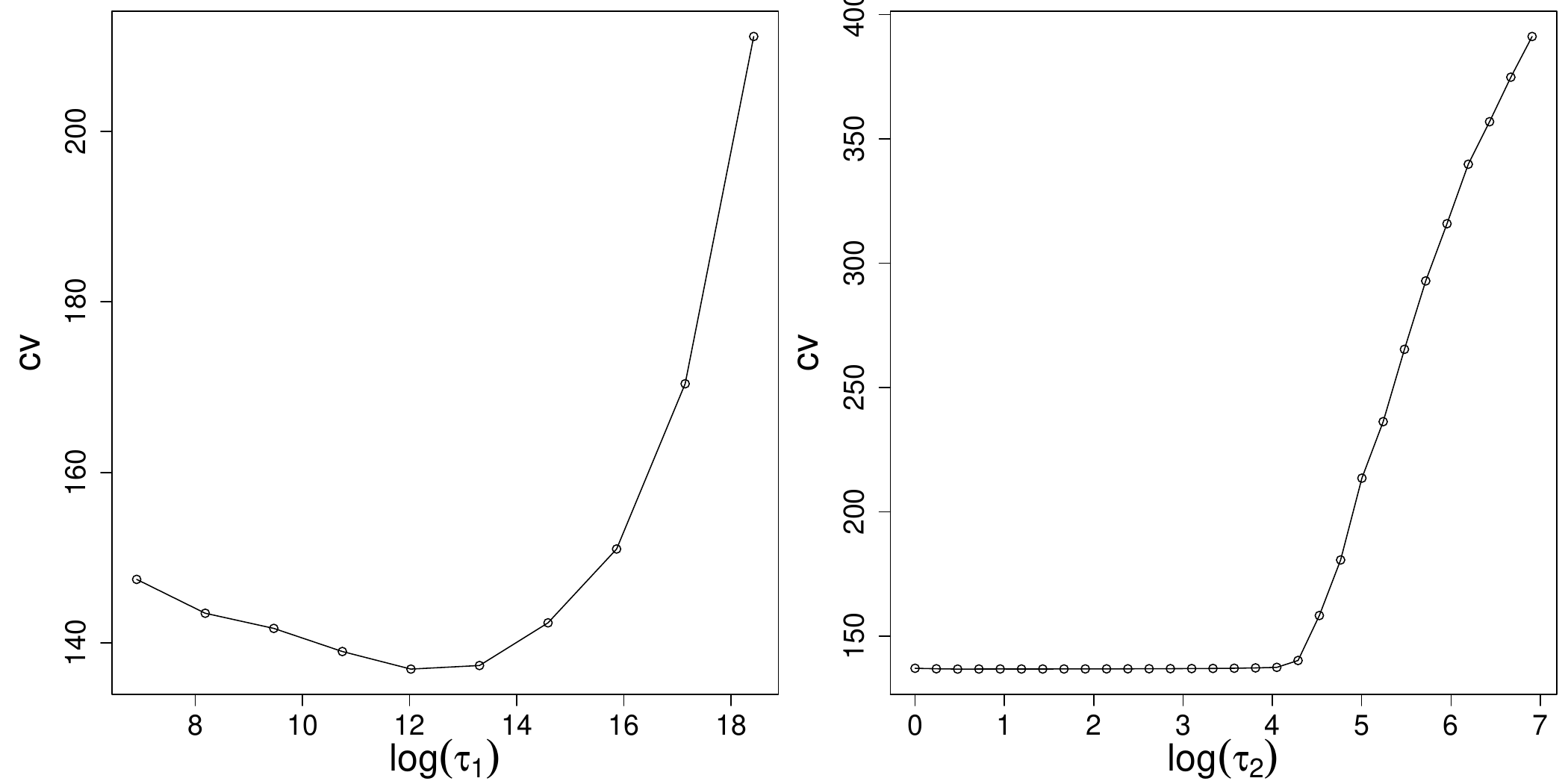}\\
\small (a)\hspace{0.45\textwidth}(b)
\caption{(a) Cross-validation values with respect to $\tau_1$ on the log scale given $\tau_2 = 0$. (b) Cross-validation values with respect to $\tau_2$ on the log scale given selected $\tau_1$ based on (a).}
\label{fig:ch4cv_plot}
\end{figure}

We used the validation data to evaluate the performance between PCA and SpatPCA in terms of the mean squared error (MSE),
$\|\hat{\bm{\Sigma}}-\bm{S}_v\|_F^2/p^2$,
where $\hat{\bm{\Sigma}}$ is a generic estimate of $\bm{\Sigma}$ based on the training data, and
$\bm{S}_v$ is the sample covariance matrix based on the validation data.
For both PCA and SpatPCA, we applied $5$-fold CV of \eqref{eq:ch4cv.gamma} to select  among
11 values of $\gamma$, including $0$, and ten values from $\hat{d}_1/10^3$ to $\hat{d}_1$ equally spaced on the log scale, where $\hat{d}_1$ is the largest eigenvalue of $\hat{\bm{\Phi}}'\bm{S}_t\hat{\bm{\Phi}}$, and $\bm{S}_t$ is the sample covariance matrix obtained from the training data.
The resulting MSE for PCA is $1.05\times 10^{-4}$, which is slightly larger than $1.02\times 10^{-4}$ for SpatPCA.
Figure~\ref{fig:ch4sse} shows the MSEs with respect to various $K$ values for both PCA and SpatPCA.
The results indicate that SpatPCA is not sensitive to the choice of $K$
as long as $K$ is sufficiently large. Our choice of $\hat{K}=6$ for SpatPCA based on \eqref{eq:ch4K.hat} appears to be effective,
and is smaller than $\hat{K}=15$ for PCA.

\begin{figure}[tbp]
\centering
\includegraphics[width=0.55\textwidth]{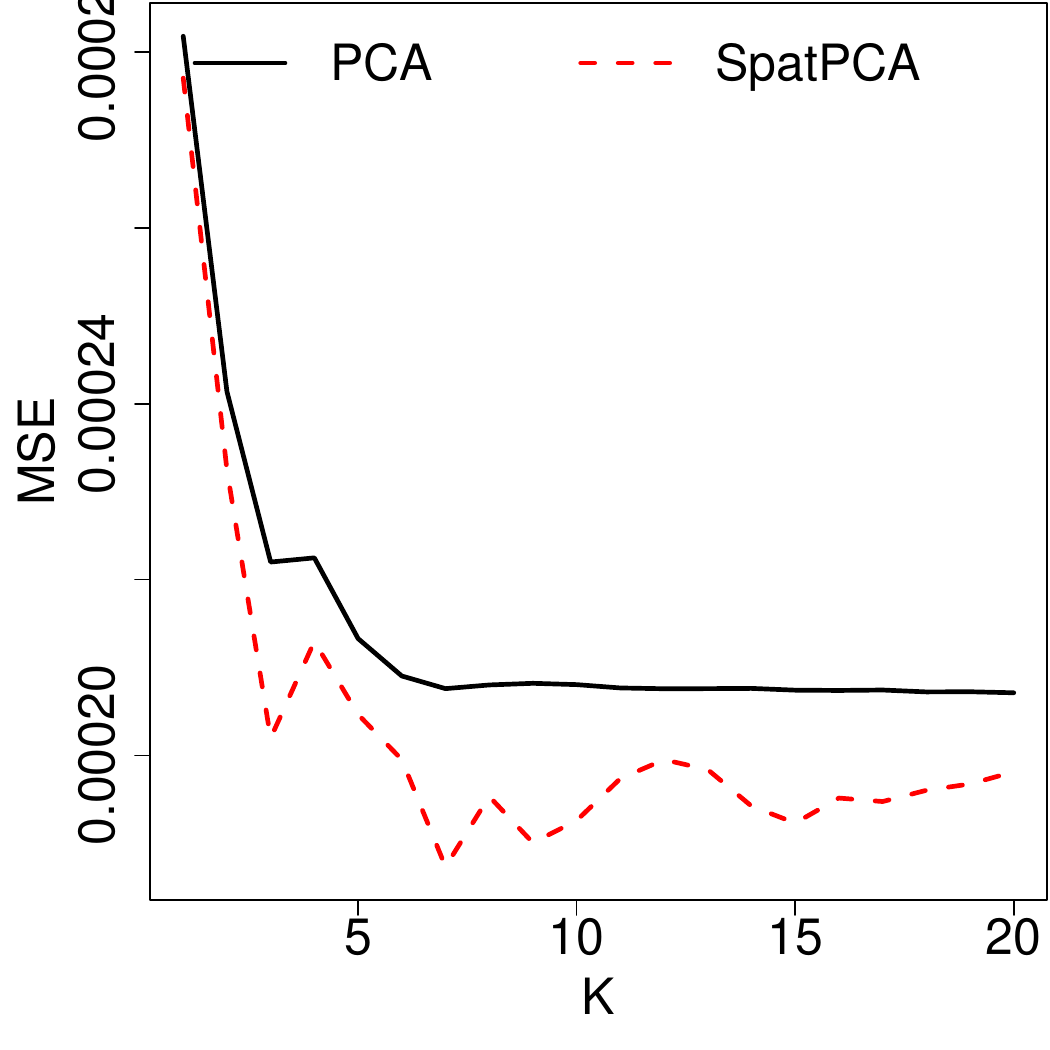}
\caption{Mean of squared errors with respect to $K$ for covariance matrices estimated from PCA and SpatPCA.}
\label{fig:ch4sse}
\end{figure}

\subsection{Two-Dimensional Experiment II}
\label{sec:ch4realsim}

To reflect a real-world situation, we generated data by mimicking the SST dataset
analyzed in the previous subsection, except we applied a larger noise variance. Specifically, we generated data according to (\ref{eq:ch4measurement})
with $K=2$, $\bm{\xi}_i\sim N(\bm{0}, \mathrm{diag}(101.3,16.8))$, $\bm{\epsilon}_i\sim N(\bm{0},\bm{I})$, $n=60$,
and at the same $2{,}820$ locations.
Here $\varphi_1(\cdot)$ and $\varphi_2(\cdot)$ are given by $\hat{\varphi}_1(\cdot)$ and $\hat{\varphi}_2(\cdot)$ estimated by SpatPCA in
Figure~\ref{fig:ch4sst_real}.

We applied the 5-fold CVs of \eqref{eq:ch4cv}, \eqref{eq:ch4cv.gamma} and \eqref{eq:ch4K.hat}
to select the tuning parameters $(\tau_1,\tau_2)$, $\gamma$ and $K$ in the same way as in the previous subsection.
Figure~\ref{fig:ch4sst_realsim} shows the estimates of $\varphi_1(\cdot)$ and $\varphi_2(\cdot)$ for PCA and SpatPCA
based on a randomly generated dataset.
Because we consider a larger noise variance than those in the previous subsection,
the first two patterns estimated from PCA turn out to be very noisy.
In contrast, SpatPCA can still reconstruct the first two patterns very well with little noise.
The results in terms of the loss functions of \eqref{eq:ch4loss_sim} and \eqref{eq:ch4loss2_sim} are summarized in Figure~\ref{fig:ch4loss2_realsim}.
Once again, SpatPCA outperforms PCA by a large margin.

\begin{figure}[tbp]
\centering
\begin{tabular}{@{}c@{}}
\small $\hat{\varphi}_1(\cdot)$ from PCA\hspace{0.35\textwidth}$\hat{\varphi}_1(\cdot)$ from SpatPCA \\
\includegraphics[width=0.95\textwidth]{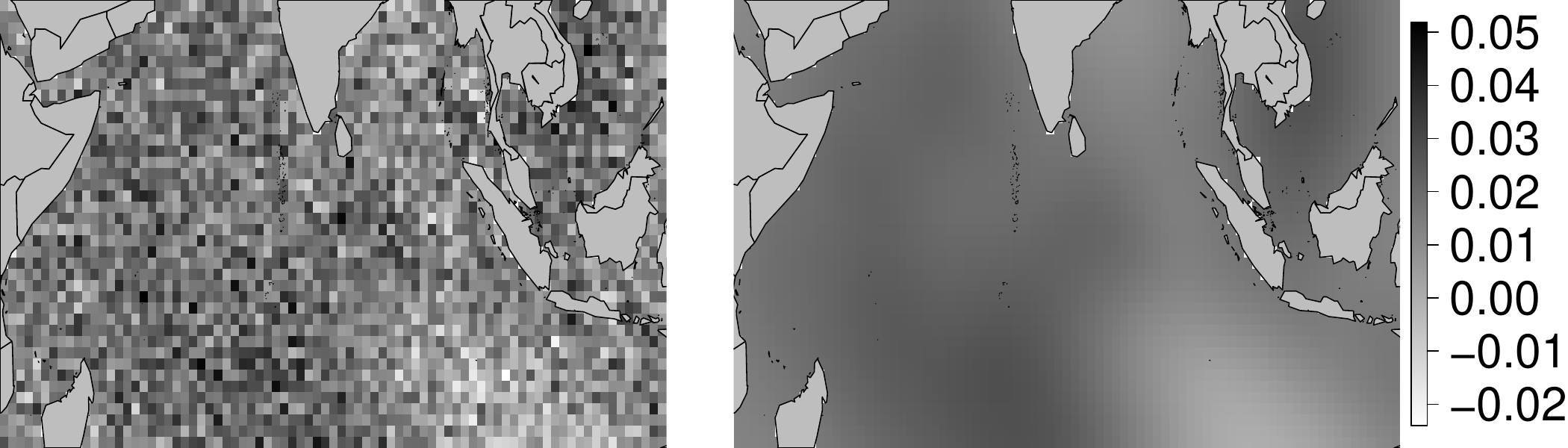} \\[0.6em]
\small $\hat{\varphi}_2(\cdot)$ from PCA\hspace{0.35\textwidth}$\hat{\varphi}_2(\cdot)$ from SpatPCA \\
\includegraphics[width=0.95\textwidth]{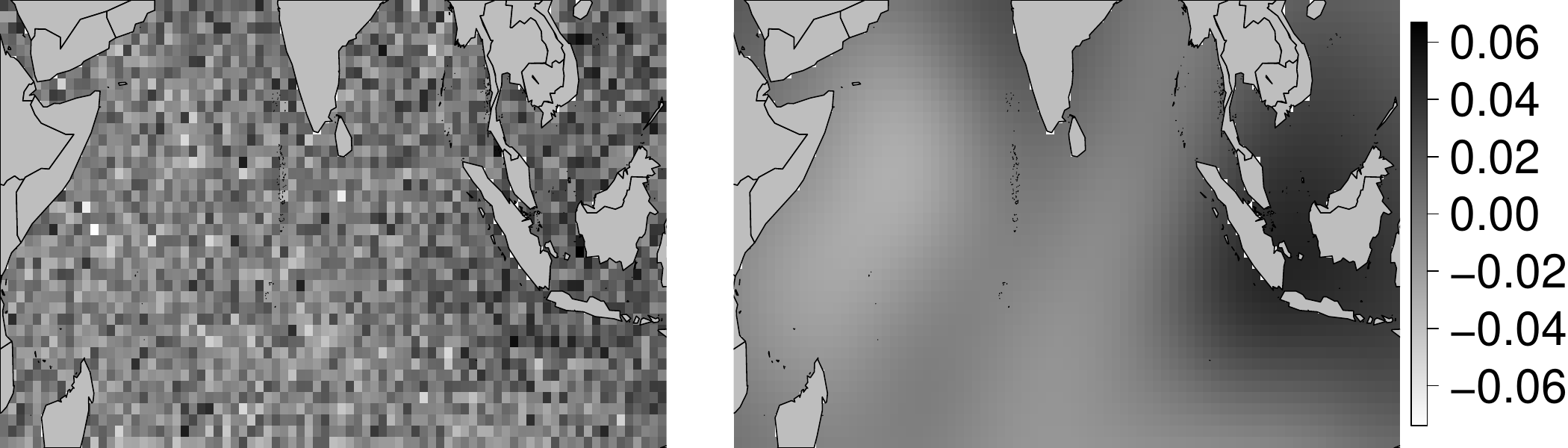}
\end{tabular}
\caption{Estimates of $\varphi_1(\cdot)$ and $\varphi_2(\cdot)$ from
PCA and SpatPCA in the two-dimensional experiment II based on a randomly simulated dataset,
where the areas in gray are the land.}
\label{fig:ch4sst_realsim}
\end{figure}

\begin{figure}[tbhp]
\centering
\includegraphics[width=0.48\textwidth]{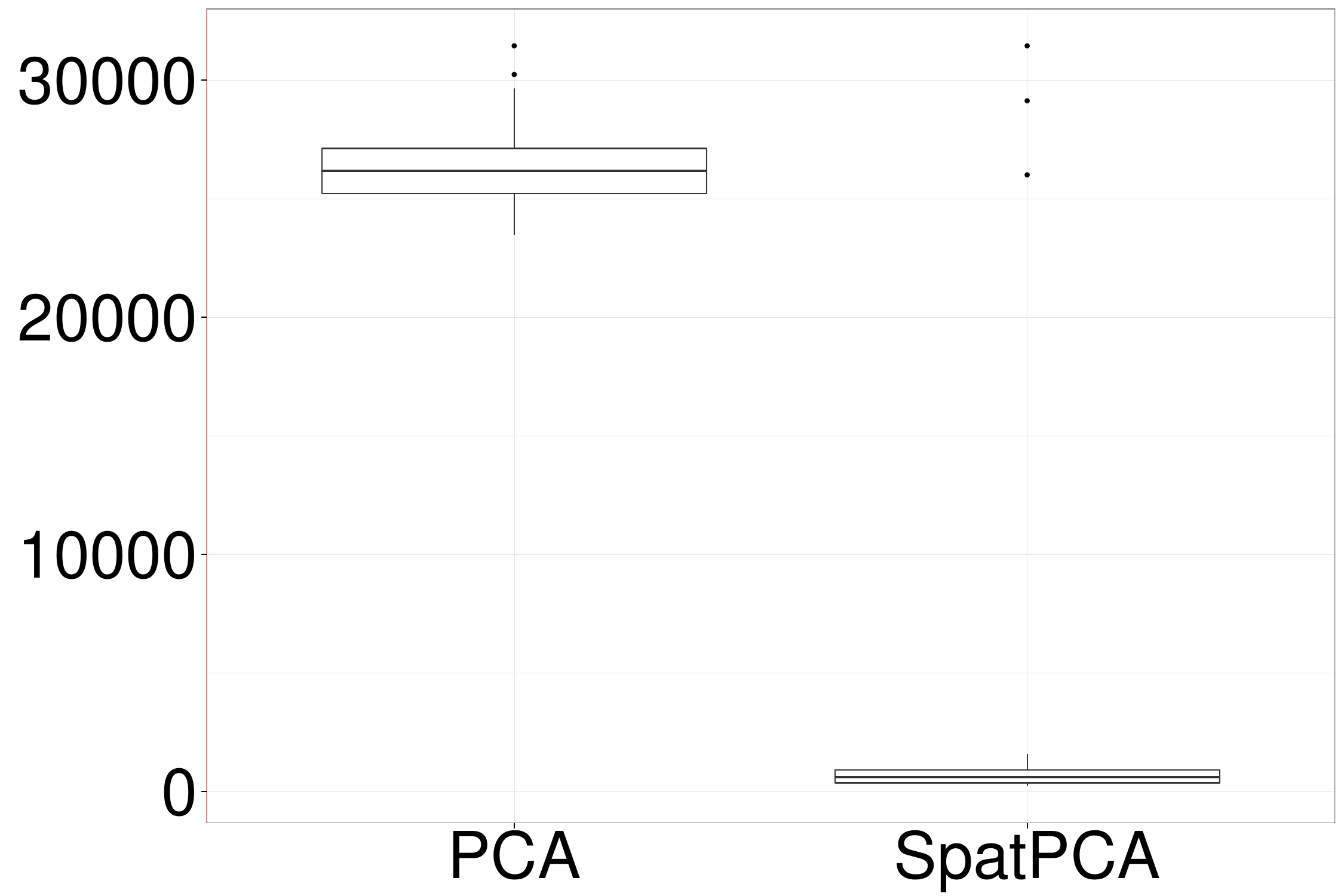}
\includegraphics[width=0.48\textwidth]{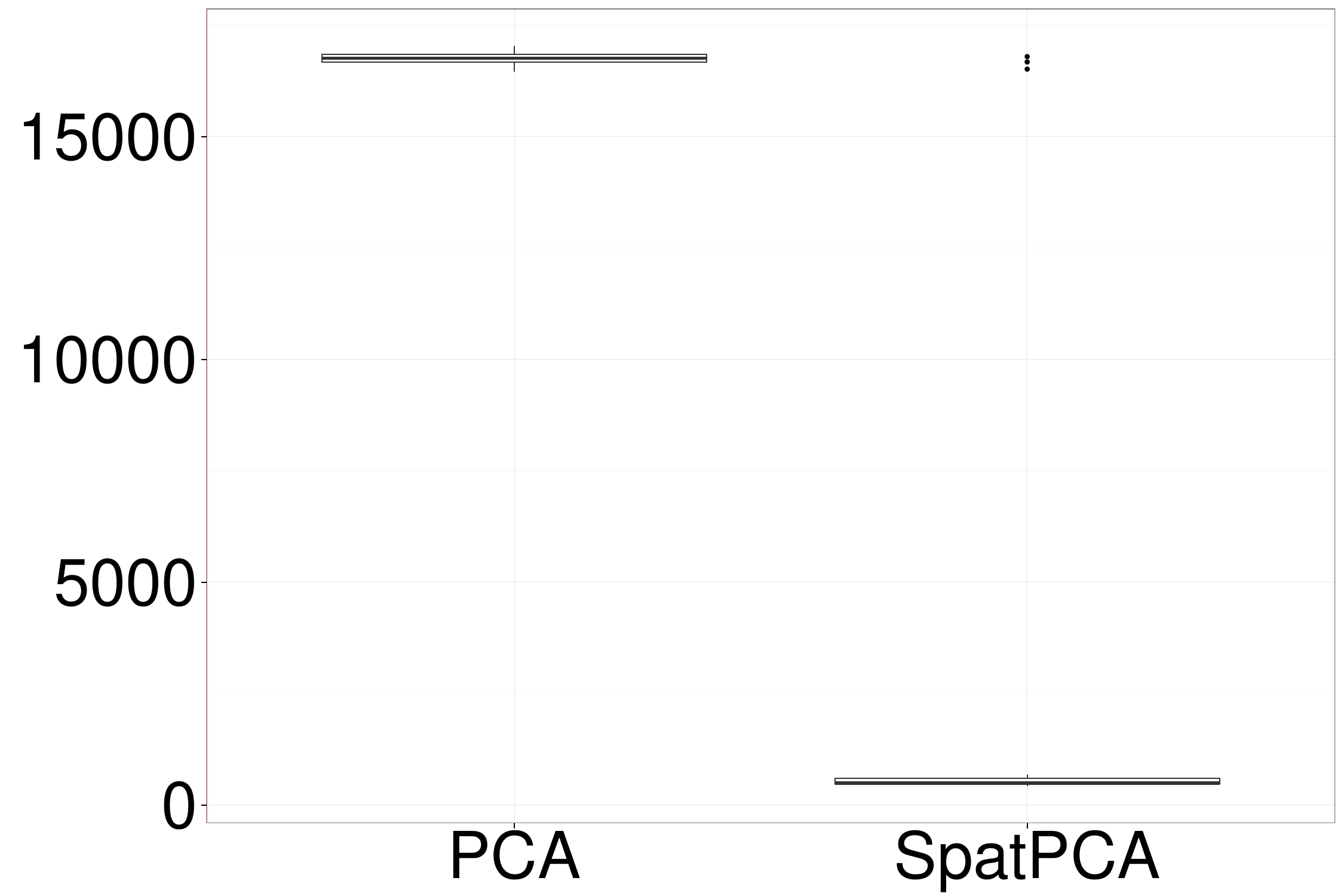}\\
\small{(a)\hspace{0.42\textwidth}(b)}
\caption{Boxplots of loss function values for PCA and SpatPCA in the two-dimensional
experiment II based on 50 simulation replicates: (a) average squared prediction errors of (\ref{eq:ch4loss_sim}); (b) average squared estimation errors of (\ref{eq:ch4loss2_sim}).}
\label{fig:ch4loss2_realsim}
\end{figure}


\chapter{Regularized Spatial Maximum Covariance Analysis}\label{ch:ch5}

\section{Introduction}
Many climate and atmospheric phenomena involve more than one meteorological spatial processes covarying in space. It is of interest to find dominant coupled patterns among these processes. For example, variations of sea surface temperatures (SSTs) in the Indian Ocean may affect precipitations in nearby countries in Africa, particularly over sensitive agricultural regions, and hence threaten the economies and livelihoods of these countries. Consequently, many studies have been conducted on the relationship between the SST and precipitation by analyzing their coupled patterns \citep[e.g.][]{reason2002sensitivity,morioka2012subtropical,omondi2013influence}. A commonly used method is maximum covariance analysis (MCA), which seeks important spatial patterns that explain the maximum amount of covariance between the two processes by using the singular value decomposition (SVD) of the cross-covariance matrix \citep{tucker1958inter}. 

However, the leading coupled patterns obtained by MCA may sometimes be too noisy to be physically interpretable when the signal-to-noise ratio is low. Many approaches have been proposed to improve MCA. For example, \citet{salim2005modelling} and \citet{salim2007model} proposed penalized likelihood approaches using roughness penalties to promote smoothness of the leading coupled patterns in space. However, these methods tend to capture global features but not localized ones. On the other hand, \citet{witten2009penalized} considered canonical correlation analysis with an $L_1$ constraint, and \citet{lee2011sparse} proposed a penalized likelihood method with the SCAD penalty \citep{fan2001variable} to facilitate sparse patterns. However, these methods cannot be applied to continuous spatial domains with data observed at irregularly spaced locations. Additionally, 
all these methods ignore the orthogonal constraints in MCA patterns.

In this chapter, we propose a regularization approach of MCA that incorporates smoothness and localized features in dominant coupled patterns. The proposed method, called spatial MCA (abbreviated as SpatMCA), is applicable to data measured irregularly in space. 
In addition, the resulting estimates can be effectively computed using the alternating direction method of multipliers (ADMM) \citep{admm}.
An R package to carry out SpatMCA is available on the Comprehensive R Archive Network (CRAN).

The remainder of this chapter is organized as follows. In Section~\ref{sec:ch5proposal}, we introduce the proposed SpatMCA method,
including dominant coupled patterns estimation and spatial cross-covariance function estimation.
Our ADMM algorithm for computing the SpatMCA estimate is provided in Section~\ref{sec:ch5algorithm}.
Numerical experiments that illustrate the superiority of SpatMCA and an application to study the relationship between sea surface temperature and precipitation datasets are presented in Section~\ref{sec:ch5numerical}.

\section{The Proposed Method}
\label{sec:ch5proposal}
Consider a sequence of uncorrelated, zero-mean, bivariate $L^2$-continuous spatial processes on spatial domains $D_1 \subset \mathbb{R}^d$ and $D_2 \subset \mathbb{R}^d$,
\[\{(\eta_{1i}(\bm{s}_1),\eta_{2i}(\bm{s}_2)) :\bm{s}_1\in D_1, \bm{s}_2\in D_2\};\quad  {i=1,\dots,n},\]
which have a common spatial covariance function $C_{jk}(\bm{s}_j,\bm{s}_k)=\mathrm{cov}(\eta_{ji}(\bm{s}_j),\eta_{ki}(\bm{s}_k))$; for $j,k=1,2$.
According to \cite{azaiez2015karhunen}, $C_{12}(\bm{s}_1,\bm{s}_2)$ can be decomposed as $C_{12}(\bm{s}_1,\bm{s}_2) = \sum_{k=1}^\infty d_ku_k(\bm{s}_1)v_k(\bm{s}_2)$, where  $\{d_k\}$ are nonnegative singular values with $d_1\geq d_2\geq\cdots$,  and $\{u_k(\cdot)\}$ and $\{v_k(\cdot)\}$ are the two corresponding sets of orthonormal basis functions. The decomposition is similar to the Karhunen-Lo\'{e}ve expansion \citep{karhunen,loeve} for a univariate spatial process. Suppose we observe data  $\bm{Y}_{ji} = (Y_{ji}(\bm{s}_{j1}),\dots,Y_{ji}(\bm{s}_{jp_j}))'$  with added noise $\bm{\epsilon}_{ji} \sim (\bm{0}, \sigma_j\bm{I})$ at the $p_j$ spatial locations $\bm{s}_{j1},\dots,\bm{s}_{jp_j}\in D_j$ for $j=1,2$, according to 
\begin{equation}
\left(\begin{array}{c}
\bm{Y}_{1i}\\ 
\bm{Y}_{2i}
\end{array}\right) =\left(\begin{array}{c}
\bm{\eta}_{1i}\\ 
\bm{\eta}_{2i}
\end{array}\right)+\left(\begin{array}{c}
\bm{\epsilon}_{1i}\\ 
\bm{\epsilon}_{2i}
\end{array}\right);\quad  {i=1,\dots,n},
\label{eq:ch5measurement}
\end{equation}
where $\bm{\eta}_{ji} = (\eta_{ji}(\bm{s}_{j1}),\dots,\eta_{ji}(\bm{s}_{jp_j}))'$, and $\bm{\epsilon}_{1i}, 
\bm{\epsilon}_{2i}$ and $(\bm{\eta}_{1i},\bm{\eta}_{2i})$ are mutually uncorrelated. Assume $d_{K+1} = 0$, and denote the cross-covariance matrix between $\bm{\eta}_{1i}$ and $\bm{\eta}_{2i}$ by $\bm{\Sigma}_{12} =\mathrm{cov}(\bm{\eta}_{1i},\bm{\eta}_{2i})$. Let $\bm{\Sigma}_{12} = \bm{U} \bm{D}\bm{V}'$ be the SVD of $\bm{\Sigma}_{12}$, where  $\bm{D}=\mathrm{diag}(d_1,\dots,d_{K})$, $\bm{U} = {(}\bm{u}_1,\dots,\bm{u}_{K}) $ is a $ p_1\times K$  matrix with the $(k,{i})$-th element $u_k({\bm{s}_{1{i}}})$, and $\bm{V}=(\bm{v}_1,\dots,\bm{v}_{K})$ is a $p_2\times K$ matrix with the $(k,{i})$-th element $v_k({\bm{s}_{2{i}}})$. We aim to identify the first $L\leq K$ dominant spatial coupled patterns $(u_1(\cdot),\dots,u_{{L}}(\cdot))$ and $(v_1(\cdot),\dots,v_{L}(\cdot))$ with large $d_1,\dots,d_{L}$ for processes $\eta_1({\cdot})$ and $\eta_2({\cdot})$, as well as to estimate $C_{12}(\cdot,\cdot)$.

Let $\bm{Y}_j = (\bm{Y}_{j1}, \dots,\bm{Y}_{jn})'$ for $j=1,2$. The sample cross-covariance matrix of $\bm{Y}_1$ and $\bm{Y}_2$ is $\bm{S}_{12} = \bm{Y}'_1\bm{Y}_2/n$. Then the MCA estimates of $\bm{u}_k$ and $\bm{v}_k$ {obtained by the SVD of} $\bm{S}_{12}$ are $\tilde{\bm{u}}_k$ and $\tilde{\bm{v}}_k$, the $k$-th left and right singular vectors of $\bm{S}_{12}$, for $k=1,\dots,K$. Let $\tilde{\bm{U}} = (\tilde{\bm{u}}_1,\dots,\tilde{\bm{u}}_K)$ and $\tilde{\bm{V}} = (\tilde{\bm{v}}_1,\dots,\tilde{\bm{v}}_K)$ be  $p_1\times K$ and $p_2 \times K$ matrices formed by the first $K$ left and right singular vectors {of} $\bm{S}_{12}$. Then $(\tilde{\bm{U}},\tilde{\bm{V}})$ solves the following constrained optimization problem {\citep{svd_trace}}:
\[
\max_{\bm{U},\bm{V}} \mathrm{tr}(\bm{U}' \bm{S}_{12} \bm{V})\quad \mbox{subject to $\bm{U}'\bm{U}=\bm{V}'\bm{V}=\bm{I}_K$}, 
\]
where $\bm{U}=(\bm{u}_1,\dots,\bm{u}_K)$ and $\bm{V}=(\bm{v}_1,\dots,\bm{v}_K)$. However, $(\tilde{\bm{U}},\tilde{\bm{V}})$ may suffer from high estimation variability when $p_1$ or $p_2$ is large, $n$ is small, or $\sigma^2_1$ or $\sigma^2_2$ is large.
Consequently, the patterns of $(\tilde{\bm{U}},\tilde{\bm{V}})$ may be too noisy to be physically interpretable. 
Additionally, for continuous spatial domains $D_1$ and $D_2$, we also need to estimate $(u_k(\bm{s}^*_1), v_k(\bm{s}^*_2))$ at locations $\bm{s}^*_1{\in D_1}$ and $\bm{s}^*_2{\in D_2}$, where data may be {unavailable}.

\subsection{Regularized Spatial MCA}
To reduce high estimation variability of MCA while controlling bias, our main idea is to introduce some spatial structure. We propose a regularization approach by maximizing the following objective function:
\begin{eqnarray}
\mathrm{tr}(\bm{U}' \bm{S}_{12} \bm{V})- \sum_{k=1}^K \left\{{\tau_{1u}J(u_k)+\tau_{2u} \|\bm{u}_k\|_1+\tau_{1v}J(v_k) +\tau_{2v}\|\bm{v}_k\|_1}\right\},
\label{eq:obj1}
\end{eqnarray}	
over $u_1(\cdot),\dots,u_{K}(\cdot)$ and $v_1(\cdot),\dots,v_{K}(\cdot)$, subject to $\bm{U}'\bm{U}=\bm{V}'\bm{V}=\bm{I}_K$ and $\bm{u}'_1\bm{S}_{12}\bm{v}_1\geq\dots\geq\bm{u}'_K\bm{S}_{12}\bm{v}_K$, where
\[
J(u)=\sum_{z_1+\cdots+z_d=2}\int_{\mathcal{R}^d}\left(
\frac{\partial^2 u(\bm{s})}{\partial x_1^{z_1}\dots\partial x_d^{z_d}}\right)^2 d\bm{s},
\]			
is a roughness penalty, $\|\bm{u}_k\|_1 = \sum_{{i}=1}^{p_1}u_k(\bm{s}_{1{i}})$, $\|\bm{v}_k\|_1 = \sum_{{i}=1}^{p_2}v_k(\bm{s}_{2{i}})$, $\bm{s}=(x_1,\dots,x_d)'$, $\tau_{1u}$ and $\tau_{1v }$ are nonnegative smoothness parameters, and $\tau_{2u}$ and $\tau_{2v}$ are nonnegative sparseness parameters. Since the patterns of $u_k(\cdot)$ and $v_k(\cdot)$ could be very different, we allow $\tau_{1u} \neq \tau_{1v}$ and $\tau_{2u} \neq \tau_{2v}$. Note that $J(\cdot)$ is the smoothing spline penalty,  designed to enhance smoothness of $u_k(\cdot)$ and $v_k(\cdot)$, and the $L_1$ Lasso penalty {\citep{lasso}} is applied to seek sparse patterns by shrinking $u_k(\cdot)$ and $v_k(\cdot)$ toward zero. The combination of the smoothness and sparseness penalties was shown by \citet{spatpca} to be effective {in obtaining} smooth and localized patterns for a univariate spatial process. Denote $\hat{u}_1(\cdot),\dots,\hat{u}_K(\cdot)$ and $\hat{v}_1(\cdot),\dots,\hat{v}_K(\cdot)$ as the maximizers of \eqref{eq:obj1}. When $\tau_{1u}$ is larger, $\{\hat{u}_k(\cdot)\}$ become smoother, and {vice versa}. 
When $\tau_{2u}$ is larger, {$\{\hat{u}_k(\cdot)\}$} become more localized by forcing more elements of $\bm{u}_k$ to be zero. Similar results can be applied to $\tau_{1v}$ and $\tau_{2v}$ for {$\{\hat{v}_k(\cdot)\}$}. On the other hand, when $\tau_{1u}=\tau_{2u}=\tau_{1v}=\tau_{2v}=0$, the estimates reduce to the MCA estimates. 

According to the smoothing spline theory {\citep{nonparametric}}, $\hat{u}(\cdot)$ and $\hat{v}(\cdot)$ are natural cubic splines and thin-plate splines for $d=1$ and $d\in{\{2,3\}}$ {with knots at $\{\bm{s}_{11},\dots, \bm{s}_{1p_1}\}$ and $\{\bm{s}_{21},\dots, \bm{s}_{2p_2}\}$, respectively}. Specifically,
\begin{eqnarray}
\hat{u}_k(\bm{s}_1)&={\displaystyle\sum_{i=1}^{p_1}} {a}_{1i} g(\|\bm{s}_1-\bm{s}_{1i}\|)+b_{10}+{\displaystyle\sum_{j=1}^{d}} {b}_{{1j}} x_{1j},\label{eq:basis_u}\\
\hat{v}_k(\bm{s}_2)&={\displaystyle\sum_{i=1}^{p_2}} {a}_{2i} g(\|\bm{s}_2-\bm{s}_{2i}\|)+b_{20}+{\displaystyle\sum_{j=1}^{d}} {b}_{{2j}}x_{2j},\:
\label{eq:basis_v}
\end{eqnarray}

\noindent where $\bm{s}_{{j}}=(x_{{{j}}1},\dots,x_{{{j}}d})'$ for ${{j}}=1,2$,
\[
g(r) = \left\{
\begin{array}{ll}
\displaystyle\frac{1}{16\pi}r^{2}\log{r};  & \mbox{if $d=2$,}\smallskip\\
\displaystyle\frac{\Gamma(d/2-2)}{16\pi^{d/2}}r^{4-d}; & \mbox{if }d=1,3,\\
\end{array}\right.	 
\]
and the coefficients ${\bm{a}_j}=\left({a}_{j1},\dots,{a}_{jp_j}\right)'$
and ${\bm{b}_j}=\left({b}_{j0},b_{j1},\dots,{b}_{jd}\right)'$ for $j=1,2$ satisfy
\[
{\left(\begin{array}{cc}
	\bm{G}_1 & \bm{E}_1 \\
	\bm{E}'_1 & \bm{0} \\
	\end{array}\right) \left(\begin{array}{c}
	{\bm{a}_1}\\
	{\bm{b}_1}
	\end{array}\right)=\left(\begin{array}{c}
	\hat{\bm{u}}_k\\
	\bm{0}
	\end{array}\right)}
\quad
\mbox{and}\quad{
	\left(\begin{array}{cc}
	\bm{G}_2 & \bm{E}_2 \\
	\bm{E}'_2 & \bm{0} \\
	\end{array}\right) \left(\begin{array}{c}
	{\bm{a}_2}\\
	{\bm{b}_2}
	\end{array}\right)=\left(\begin{array}{c}
	\hat{\bm{v}}_k\\
	\bm{0}
	\end{array}\right){.}}
\]
Here $\hat{\bm{u}}_k = (\hat{u}_k(\bm{s}_{11}),\dots,\hat{u}_k(\bm{s}_{1p_1}))'$, $\hat{\bm{v}}_k = (\hat{v}_k(\bm{s}_{21}),\dots,\hat{v}_k(\bm{s}_{2p_2}))'$, $\bm{G}_j$ is a $p_j\times p_j$ matrix with the $(i,{i'})$-th element $g(\|\bm{s}_{{ji}}-\bm{s}_{j{i'}}\|)$,
and $\bm{E}_j$ is a $p_j\times (d+1)$ matrix with the $i$-th row $(1,\bm{s}'_{ji})$ for $j=1,2$.
Therefore, $\hat{u}_k(\cdot)$ and $\hat{v}_k(\cdot)$ in \eqref{eq:basis_u} and \eqref{eq:basis_v} can be expressed in terms of $\hat{\bm{u}}_k$ and $\hat{\bm{v}}_k${,} respectively.

The roughness penalties of $u_k(\cdot)$ and $v_k(\cdot)$ can also be written as
\begin{equation}
\label{eq:ch5smoothness}
J(u_k) =\bm{u}'_k \bm\Omega_1 \bm{u}_k \quad \mbox{and} \quad J(v_k) =\bm{v}'_k \bm\Omega_2 \bm{v}_k,
\end{equation}
where $\bm\Omega_j$ is a known $p_j\times p_j$ matrix determined only by $\bm{s}_{j1},\dots,\bm{s}_{jp_j}$ for $j=1,2$ {\citep{nonparametric}}. Therefore, from (\ref{eq:obj1}) and (\ref{eq:ch5smoothness}), the proposed estimate $(\hat{\bm{U}}_{K,\tau_{1u},\tau_{2u}}, \hat{\bm{V}}_{K,\tau_{1v},\tau_{2v}})$ of $(\bm{U},\bm{V})$ can be simplified by maximizing the following objective function:
\begin{eqnarray}
\label{eq:obj2}
\mathrm{tr}(\bm{U}' \bm{S}_{12} \bm{V})- \sum_{k=1}^K \left\{{\tau_{1u}\bm{u}'_k\bm\Omega_1\bm{u}_k+\tau_{2u} \|\bm{u}_k\|_1+ \tau_{1v}\bm{v}'_k\bm\Omega_2\bm{v}_k+ \tau_{2v} \|\bm{v}_k\|_1}\right\}, 
\end{eqnarray}	
subject to $\bm{U}'\bm{U}=\bm{V}'\bm{V}=\bm{I}_K$ and $\bm{u}'_1\bm{S}_{12}\bm{v}_1\geq\dots\geq\bm{u}'_K\bm{S}_{12}\bm{v}_K$. We call the proposed method based on \eqref{eq:obj2} SpatMCA. Given $(\hat{\bm{U}}_{K,\tau_{1u},\tau_{2u}}$,$\hat{\bm{V}}_{K,\tau_{1v},\tau_{2v}})$, the estimates of $(u_1(\cdot),v_1(\cdot)),\dots,(u_K(\cdot),v_K(\cdot))$ can be directly calculated by \eqref{eq:basis_u} and \eqref{eq:basis_v}.
Note that the SpatMCA estimate of (\ref{eq:obj2}) reduces to a sparse CCA estimate of  \cite{witten2009penalized} 
if {$\mathrm{var}(\bm{Y}_j)=\bm{I}_{p_j}$}, $\bm{\Omega}_j=\bm{I}_{p_j}$ for $j=1,2$, and the orthogonal constraints of $\bm{U}$ and $\bm{V}$ are dropped.
\subsection{Estimation of Cross-Covariance Function}

To estimate $C_{12}(\cdot,\cdot)$, we also have to estimate $\bm{D}$. Given $(\hat{\bm{U}},\hat{\bm{V}})=(\hat{\bm{U}}_{K,\tau_{1u},\tau_{2u}},\hat{\bm{V}}_{K,\tau_{1v},\tau_{2v}})$   with $\hat{\bm{U}} = (\hat{\bm{u}}_1,\dots,\hat{\bm{u}}_K)$ and $\hat{\bm{V}} = (\hat{\bm{v}}_1,\dots,\hat{\bm{v}}_K)$, the proposed estimate of $\bm{D}$ is 
\begin{equation}\label{eq:estimate_d}
\hat{\bm{D}} = \mathop{arg\min}_{d_1,\dots,d_K\geq 0}\|\bm{S}_{12} - \hat{\bm{U}}\bm{D}{\hat{\bm{V}}'}\|^2_F =\mathrm{diag}(\hat{d}_1,\dots,\hat{d}_K),
\end{equation}
where $\hat{d}_k = \max\{\hat{\bm{u}}'_k\bm{S}_{12}\hat{\bm{v}}_k,0\};$ $k=1,\dots,K$, and $\|\bm{M}\|_F=
\Big(\displaystyle\sum_{i,j}m^2_{ij}\Big)^{1/2}$ is the Frobenius norm of a matrix $\bm{M}$. Then, the proposed estimate of $C_{12}(\cdot,\cdot)$ is 
\begin{equation}\label{eq:estimate_crosscov}
\hat{C}_{12}(\bm{s}_1,\bm{s}_2)=\sum_{k=1}^K \hat{d}_k \hat{u}_k(\bm{s}_1)\hat{v}_k(\bm{s}_2).
\end{equation}
\subsection{Tuning Parameter Selection}
An $M$-fold cross-validation (CV) is applied to select the tuning parameters $\tau_{1u}$, $\tau_{2u}$, $\tau_{1v}$ and $\tau_{2v}$. First, we randomly decompose the index set $\{1,\dots,n\}$ into $M$ parts {that are} as close to the same size, $n_M$, as possible. Let $(\bm{Y}^{(m)}_1,\bm{Y}^{(m)}_2)$ be the sub-matrix of $(\bm{Y}_1,\bm{Y}_2)$ corresponding to the $m$-th part.
For $m=1,\dots, M$, we treat $(\bm{Y}^{(m)}_1,\bm{Y}^{(m)}_2)$  as the validation data,
and {we} obtain the estimate $(\hat{\bm{U}}^{(-m)}_{K,\tau_{1u},\tau_{2u}},\hat{\bm{V}}^{(-m)}_{K,\tau_{1v},\tau_{2v}})$ of $(\bm{U},\bm{V})$ for ${\{}\tau_{1u}, \tau_{2u},\tau_{1v},\tau_{2v}{\}}\in\mathcal{A}$
based on the remaining data  $(\bm{Y}^{(-m)}_1,\bm{Y}^{(-m)}_2)$ using the proposed method \eqref{eq:obj2}, where $\mathcal{A}\subset[0,\infty)^4$ is a candidate index set. Then the proposed CV criterion is
\begin{equation}
\label{eq:cv_criterion}
\mathrm{CV}({K,}\tau_{1u},\tau_{2u},\tau_{1v},\tau_{2v})=\frac{1}{M}\sum_{m=1}^M\|\bm{S}^{(m)}_{12} - \hat{\bm{U}}^{(-m)}_{K,\tau_{1u},\tau_{2u}}\hat{\bm{D}}^{(-m)}_{{K,\tau_{1u},\tau_{2u},\tau_{1v},\tau_{2v}}}(\hat{\bm{V}}^{(-m)}_{K,\tau_{1{v}},\tau_{2{v}}})'\|^2_F\:,
\end{equation}
where $\bm{S}^{(m)}_{12} = \left(\bm{X}^{(m)}\right)'\bm{Y}^{(m)}/n_M$, and $\hat{\bm{D}}^{(-m)}_{{K,\tau_{1u},\tau_{2u},\tau_{1v},\tau_{2v}}}$ is the estimate of $\bm{D}$ from \eqref{eq:estimate_d} with $({\hat{\bm{U}},\hat{\bm{V}}})$ replaced by $(\hat{\bm{U}}^{(-m)}_{K,\tau_{1u},\tau_{2u}},\hat{\bm{V}}^{(-m)}_{K,\tau_{1v},\tau_{2v}})$. 

Owing to the high computation cost to select {$\{\tau_{1u}, \tau_{2u},\tau_{1v},\tau_{2v}\}$} simultaneously {for each $K$}, 
we recommend an effective two-step procedure for selecting them. Specifically, we first select $\tau_{1u}$ and $\tau_{1v}$ with $\tau_{2u}=\tau_{2v}=0$ {by}
\begin{equation}\label{eq:cv_max1}
(\hat{\tau}_{1u}(K),\hat{\tau}_{1v}(K))=\displaystyle\mathop{\arg\min}_{\{\tau_{1u}, \tau_{1v}\}\subset[0,\infty)^2}\mathrm{CV}({K,}\tau_{1u},0,\tau_{1v},0),
\end{equation}
and then select $\tau_{2u}$ and $\tau_{2v}$ by
\begin{equation}\label{eq:cv_max2}
(\hat{\tau}_{2u}(K),\hat{\tau}_{2v}(K))=\displaystyle\mathop{\arg\min}_{\{\tau_{2u}, \tau_{2v}\}\subset[0,\infty)^2}\mathrm{CV}(K,\hat{\tau}_{1u}(K),\tau_{2u},\hat{\tau}_{1v}(K),\tau_{2v}).
\end{equation}

{Finally, we select the }rank $K$ of $\bm{U}\bm{D}\bm{V}'$ by computing the CV values of \eqref{eq:cv_criterion} for $K=1,2,\dots$, evaluated at the {four} selected tuning {parameter values} until no further reduction of the CV value is obtained. That is,
\begin{align}
\hat{K}=\min\{&K: \mathrm{CV}\left(K,\hat{\tau}_{1u}(K),\hat{\tau}_{2u}(K),\hat{\tau}_{1v}(K),\hat{\tau}_{2v}(K)\right)\leq\notag\\
~&\mathrm{CV}\left(K+1,\hat{\tau}_{1u}(K+1),\hat{\tau}_{2u}(K+1),\hat{\tau}_{1v}(K+1),\hat{\tau}_{2v}(K+1)\right);K=1,2,\dots\}.
\label{eq:ch5khat}
\end{align}

\section{Computation Algorithm}
\label{sec:ch5algorithm}
Let $\bm{G} = (\bm{U}', \bm{V}')'$ be a $(p_1+ p_2)\times K$ matrix with the $({i},k)$-th element $g_{{i}k}$. The objective function \eqref{eq:obj2} can be rewritten as
\begin{equation}\label{eq:obj3}
\mathrm{tr}( \bm{G}'\bm{\Theta}\bm{G})-\sum_{k=1}^K \left(\tau_{2u} \sum_{{i}=1}^{p_1}|{g}_{{i}k}|+\tau_{2v} \sum_{{i}=p_1+1}^{p_1+p_2}|g_{{i}k}|\right),
\end{equation}
subject to $\bm{U}'\bm{U} = \bm{V}'\bm{V}=\bm{I}_{K}$, where $\bm{\Theta} ={\left(\begin{array}{cc}
	-\tau_{1u}\bm{\Omega}_{1} & \bm{S}_{12}/2 \\
	\bm{S}'_{12}/2 & -\tau_{1v}\bm{\Omega}_{2} \\
	\end{array}\right)}$.  The maximizer of \eqref{eq:obj3}, consisting of the orthogonal constraint and the Lasso penalty, is too complex to solve directly. We adopt the ADMM algorithm \citep[originated by][]{admm_2} by decomposing the constrained optimization problems into small subproblems that can be efficiently handled. The readers are referred to \citet{admm} for more details regarding ADMM.

First, we transform \eqref{eq:obj3} into the following equivalent {form} by adding $(p_1+p_2) \times K$ parameter matrices $\bm{Q}$ and  $\bm{R}${:}
\begin{equation*}\label{admm_opt}
\mathrm{tr}( \bm{G}'\bm{\Theta}\bm{G})-\sum_{k=1}^K \left(\tau_{2u} \sum_{i=1}^{p_1}|{r}_{ik}|+\tau_{2v} \sum_{i=p_1+1}^{p_1+p_2}|r_{ik}|\right),
\end{equation*}
subject to $\bm{Q}_1'\bm{Q}_1 = \bm{Q}_2'\bm{Q}_2=\bm{I}_{K}$, and a new constraint $\bm{G}=\bm{Q}=\bm{R}$, where $r_{ik}$ is the $(i,k)$-th element of $\bm{R}$, $\bm{Q}=(\bm{Q}'_1,\bm{Q}'_2)'$, and $\bm{Q}_1$ and $\bm{Q}_2$ {are} $p_1\times K$ and $p_2\times K$ sub-matrices of $\bm{Q}$ formed by the first $p_1$ and the last $p_2$ rows of $\bm{Q}$, respectively. The resulting augmented Lagrange function is		
\begin{align*}
L(\bm{G}, \bm{R}, \bm{Q}, \bm{\Gamma}_{1},\bm{\Gamma}_{2})=&~\mathrm{tr}( \bm{G}'\bm{\Theta}\bm{G})-\sum_{k=1}^K \left(\tau_{2u} \sum_{{i}=1}^{p_1}|{r}_{{i}k}|+\tau_{2v} \sum_{{i}=p_1+1}^{p_1+p_2}|{r}_{{i}k}|\right)\\
&~-\mathrm{tr}(\bm{\Gamma}'_{1}(\bm{G}-\bm{R}))-\mathrm{tr}(\bm{\Gamma}'_{2}(\bm{G}-\bm{Q}))\\&~-\frac{\zeta}{2}(\|\bm{G}-\bm{R}\|^2_F+\|\bm{G}-\bm{Q}\|^2_F),
\end{align*}subject to $\bm{Q}_1'\bm{Q}_1 = \bm{Q}_2'\bm{Q}_2=\bm{I}_{K}$, where $\bm{\Gamma}_{1}$ and $\bm{\Gamma}_{2}$ are $(p_1+p_2)\times K$ matrices of Lagrange multipliers, and $\zeta \geq 0$ is a penalty parameter to promote convergence. Then the ADMM steps at the $(\ell+1)$-th iteration {have the following} closed formed expressions:
\begin{align}
\bm{G}^{(\ell+1)}
=&~ \mathop{\arg\max}_{\bm{G}}L(\bm{G}, \bm{R}^{(\ell)}, \bm{Q}^{(\ell)}, \bm{\Gamma}_{1}^{(\ell)},\bm{\Gamma}_{2}^{(\ell)})
\notag\\
=&~ \frac{1}{2}(\zeta\bm{I}-\bm{\Theta})^{{-}1}\left(\zeta(\bm{R}^{(\ell)}{+}\bm{Q}^{(\ell)}){-}\bm{\Gamma}^{(\ell)}_{1}{-}\bm{\Gamma}^{(\ell)}_{2}\right),\label{eq:G}\\
\bm{R}^{(\ell+1)}
=&~\mathop{\arg\max}_{\bm{R}}L(\bm{G}^{(\ell+1)}, \bm{R}, \bm{Q}^{(\ell)}, \bm{\Gamma}_{1}^{(\ell)},\bm{\Gamma}_{2}^{(\ell)})
\notag\\
=&~ \left(\frac{1}{\zeta}\mathcal{S}_{\tau_{2}}\left(\zeta{g}^{(\ell+1)}_{{i}k}+{\gamma}_{1{i}k}^{(\ell)}\right)\right)_{(p_1+p_2)\times K}, \label{eq:ch5R}\\
\bm{Q}^{(\ell+1)}
=&~ \mathop{\arg\max}_{\bm{Q:\bm{Q}'\bm{Q}=\bm{I}}}L(\bm{G}^{(\ell+1)}, \bm{R}^{(\ell+1)}, \bm{Q}, \bm{\Gamma}_{1}^{(\ell)},\bm{\Gamma}_{2}^{(\ell)})\notag\\
=&~\left(\bm{F}_1^{(\ell)}\left(\bm{E}_1^{(\ell)}\right)',\bm{F}_2^{(\ell)}\left(\bm{E}_2^{(\ell)}\right)'\right)' ,\label{eq:ch5Q}\\
\bm{\Gamma}^{(\ell+1)}_{1}
=&~ \bm{\Gamma}^{(\ell)}_{1}+ \zeta\left(\bm{G}^{(\ell+1)}-\bm{Q}^{(\ell+1)}\right), \label{eq:ch5Gamma1}\\
\bm{\Gamma}^{(\ell+1)}_{2}
=&~ \bm{\Gamma}^{(\ell) }_{2}+ \zeta\left(\bm{G}^{(\ell+1)}-\bm{R}^{(\ell+1)}\right),\label{eq:ch5Gamma2}
\end{align}
where
\[\displaystyle \mathcal{S}_{\tau_{2}}(\gamma_{1jk})= \left\{
\begin{array}{ll}
\mathrm{sign}(\gamma_{1{i}k})\max(|\gamma_{1{i}k}|-\tau_{2u},0);  & \mbox{if ${i} \leq p_1$,}\smallskip\\
\mathrm{sign}(\gamma_{1{i}k})\max(|\gamma_{1{i}k}|-\tau_{2v},0); & \mbox{{if $i > p_1$}},\\
\end{array}\right.\]
$\gamma_{1{i}k}$ is the $({i},k)$-th element of $\bm{\Gamma}_1$, $\bm{E}_j^{(\ell)}\bm{\Lambda}_j^{(\ell)}\left(\bm{F}_j^{(\ell)}\right)'$ is the SVD of $\zeta\bm{G}_j^{(\ell+1)}+ \bm{\Gamma}^{(\ell)}_{2j}$ for $j=1,2$,  $\bm{G}^{(\ell+1)}_1$ and $\bm{G}^{(\ell+1)}_{21}$ are $p_1\times K$ and $p_2\times K$ sub-matrices of $\bm{G}^{(\ell+1)}$ corresponding to $\bm{U}$ and $\bm{V}$, and  $\bm{\Gamma}^{(\ell+1)}_{21}$ and $\bm{\Gamma}^{(\ell+1)}_{22}$ are $p_1\times K$ and $p_2\times K$ sub-matrices of $\bm{\Gamma}^{(\ell+1)}_2$ corresponding to $\bm{U}$ and $\bm{V}$. Note that $\zeta$ must be chosen large enough to ensure that $\zeta\bm{I}-\bm{\Theta}$ in \eqref{eq:G}
is positive-definite.

\section{Numerical Examples}
\label{sec:ch5numerical}
This section contains several simulation examples in one-dimensional and two-dimensional spatial domains and an application of SpatMCA to a real dataset. We compared the performance of the proposed SpatMCA with three other methods: (1) MCA ($\tau_{1u}=\tau_{1v}=\tau_{2u}=\tau_{2v}=0$);
(2) SpatMCA with the smoothness penalties only ($\tau_{2u}=\tau_{2v}=0$);
(3) SpatMCA with the sparseness penalties only ($\tau_{1u}=\tau_{1v}=0$),
in terms of the following loss function:
\begin{equation}
\label{eq:loss1} 
\mathrm{Loss}(\hat{C}_{12}) =\frac{1}{p_1p_2} \sum_{i=1}^{p_1}\sum_{j=1}^{p_2}\big(\hat{C}_{12}(\bm{s}_{1i},\bm{s}_{2j})-C_{12}(\bm{s}_{1i},\bm{s}_{2j})\big)^2\:.
\end{equation}
Throughout this section, we applied the proposed SpatMCA method and the ADMM algorithm given by \eqref{eq:G}–\eqref{eq:ch5Gamma2} to compute the SpatMCA estimates
with $\zeta$ being ten times the maximum singular value of $\bm{S}_{12}$. Additionally, the stopping criterion for the ADMM algorithm is
\[
\frac{1}{\sqrt{p_1p_2}} \max\left( \|\bm{G}^{(\ell+1)}-\bm{G}^{(\ell)}\|_F,\|\bm{G}^{(\ell+1)}-\bm{R}^{(\ell+1)}\|_F,
\|\bm{G}^{(\ell+1)}-\bm{Q}^{(\ell+1)}\|_F \right)\leq 10^{-4}\:.
\]
\subsection{A One-Dimensional Experiment}\label{sec:1d}
We generated data from \eqref{eq:ch5measurement} with $K=2$, {$d=1$, $n=1000,$}
\[\left(\begin{array}{c}
\bm{\eta}_{1i} \\
\bm{\eta}_{2i} \\
\end{array}\right)\sim N\left(\bm{0}, {\left(\begin{array}{cc}
	\bm{I} & \bm{U}\mathrm{diag}(d_1,d_2)\bm{V}' \\
	\bm{V}\mathrm{diag}(d_1,d_2)\bm{U}' & \bm{I}\\
	\end{array}\right)}\right),\] 
$\bm{\epsilon}_{ji}\sim N(\bm{0},\bm{I})$, ${p_j}=50$, $(\bm{s}_{j1},  \dots,\bm{s}_{jp_j})$ equally spaced in $[-7,7]$, and
\begin{align}
u_1(\bm{s}_1)
=&~ \frac{1}{c_1}\exp(-(x_{11}^2+\cdots+x_{1d}^2)),
\label{eq:u1_sim}\\
v_1(\bm{s}_2)
=&~ \frac{1}{c_2}\exp(-((x_{21}-2)^2+\cdots+(x_{2d}-2)^2)/2),
\label{eq:v1_sim}\\
u_2(\bm{s}_1)
=&~ \frac{1}{c_3}x_{11}\cdots x_{1d}\exp(-(x_1^2+\cdots+x_{1d}^2)),
\label{eq:u2_sim}\\
v_2(\bm{s}_2)
=&~ \frac{1}{c_4}(x_{21}-2)\cdots (x_{2d}-2)\exp(-((x_{21}-2)^2+\cdots+(x_{2d}-2)^2)/2),
\label{eq:v2_sim}
\end{align}
where $\bm{s}_j=(x_{j1},\dots,x_{jd})'$, $c_1$, $c_2$, $c_3$ and $c_4$ are normalization constants such that
$\|\bm{u}_j\|_2=\|\bm{v}_j\|_2=1$ for $j=1,2$. We considered three pairs of $(d_1,d_2)\in\{(1,0), (0.5,0), (1,0.7)\}$, and applied the proposed SpatMCA with $K=\{1,2,5\}$ and $\hat{K}$ selected by \eqref{eq:ch5khat}. For each case, we applied the 5-fold CV of \eqref{eq:ch5khat} to select ${\{}\tau_{1u}, \tau_{1v}, \tau_{2u}, \tau_{2v}{\}}$ among $21$ values of $\tau_{1u}$ and $\tau_{1v}$  (including $0$ and the other 20 values equally spaced on the log scale from $10^{-2}$ to $10$) and $11$ values of $\tau_{2u}$ and $\tau_{2v}$ (including $0$ and the other 10 values equally spaced on the log scale from $10^{-3}$ to $1$).

Figures~\ref{fig:est_u_d1} and \ref{fig:est_v_d1} show the estimates of $u_k(\cdot)$ and $v_k(\cdot)$, respectively, for the four methods based on three different combinations of singular values. Each case contains four estimated functions based on four randomly generated datasets.
Not surprisingly, the MCA estimates considering no spatial structure are very noisy, particularly when the signal-to-noise ratio is small.
Adding only the smoothness penalties (i.e., $\tau_{2u}=\tau_{2v}=0$) reduces noise, but introduces some bias. On the other hand, adding only the sparseness penalties (i.e, $\tau_{1u}=\tau_{1v}=0$) does not reduce much noise, despite that the estimated $\{u_k(\cdot)\}$ and $\{v_k(\cdot)\}$ are forced to be zeros at some locations. Our SpatMCA estimates generally reproduce the targets with little noise for all cases even for the small signal-to-noise ratio, indicating the effectiveness of regularization.

\begin{figure}\centering
	$\hat{u}_1(\cdot)$ based on $K=1$ for $(d_1,d_2)=(1,0)$\\
	\includegraphics[scale=0.39]{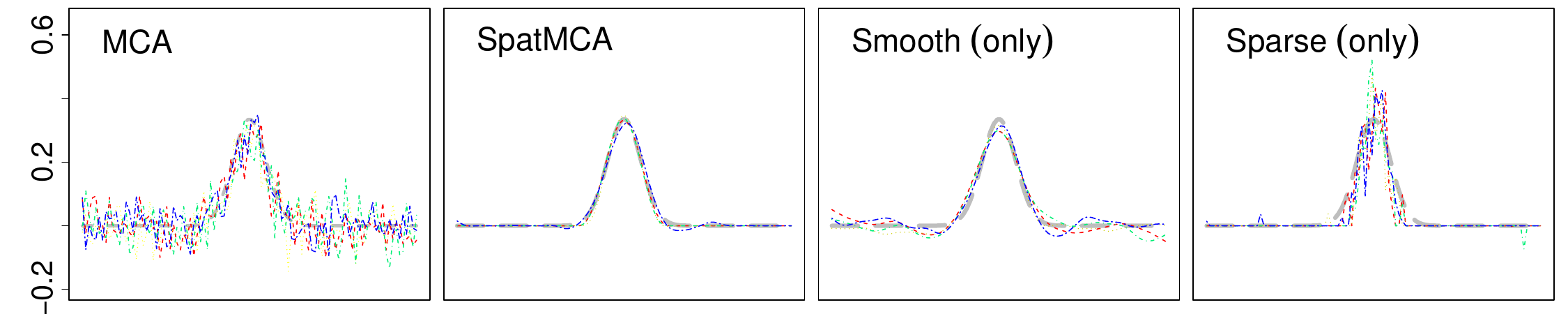}
	$\hat{u}_1(\cdot)$ based on $K=1$ for $(d_1,d_2)=(0.5,0)$\\
	\includegraphics[scale=0.39]{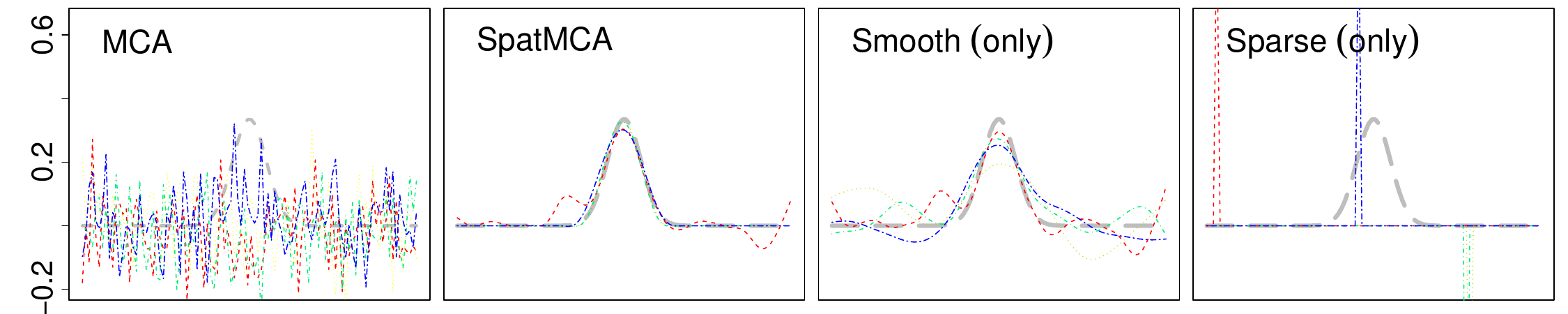}
	$\hat{u}_1(\cdot)$ based on $K=2$ for $(d_1,d_2)=(1,0.7)$\\
	\includegraphics[scale=0.39]{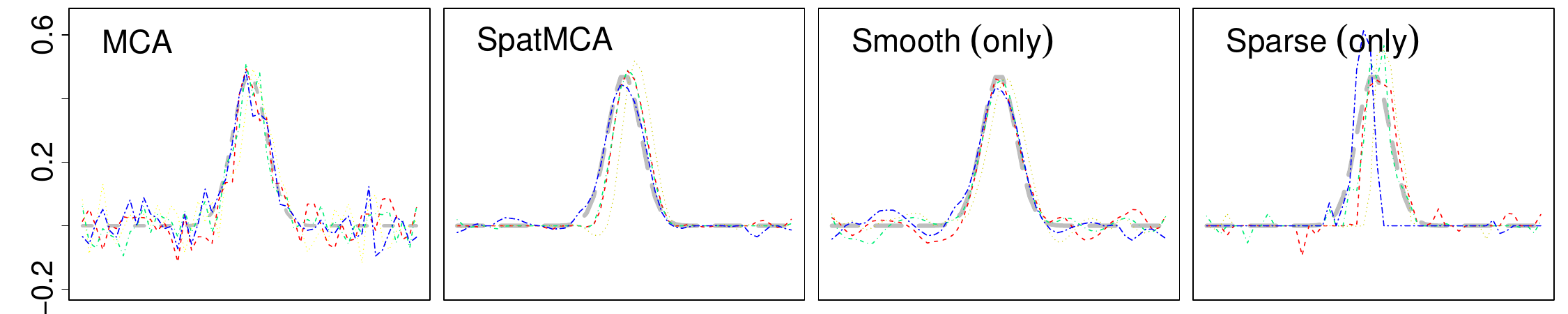}
	$\hat{u}_2(\cdot)$ based on $K=2$ for $(d_1,d_2)=(1,0.7)$\\
	\includegraphics[scale=0.39]{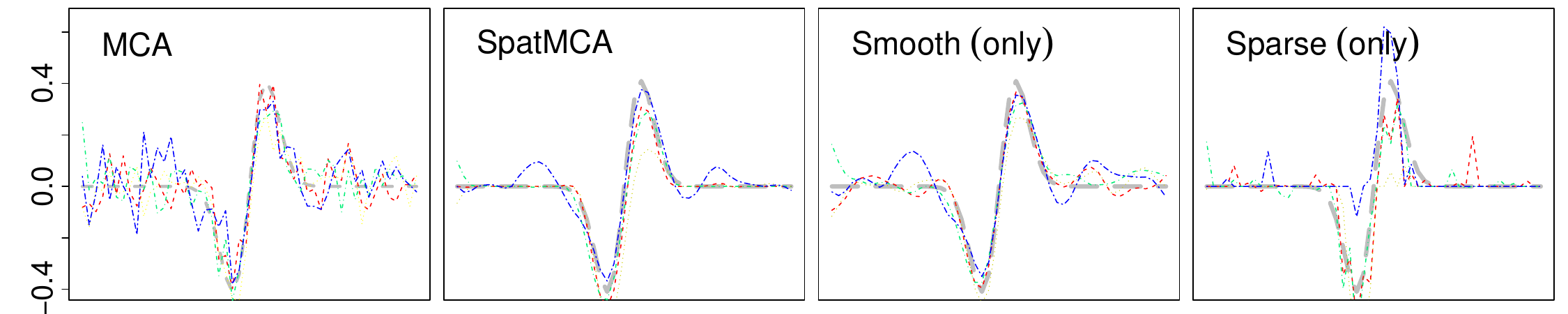}
	\caption{Estimates of $u_1(\cdot)$ and $u_2(\cdot)$ obtained from various methods based on data generated from three different combinations of singular values. Each panel consists of four estimates (in four different line types) corresponding to four randomly generated datasets, where the dash gray lines are the true $u_1(\cdot)$ and $u_2(\cdot)$.}
	\label{fig:est_u_d1}
\end{figure}	

\begin{figure}\centering
	$\hat{v}_1(\cdot)$ based on $K=1$ for $ (d_1,d_2)=(1,0)$\\
	\includegraphics[scale=0.39]{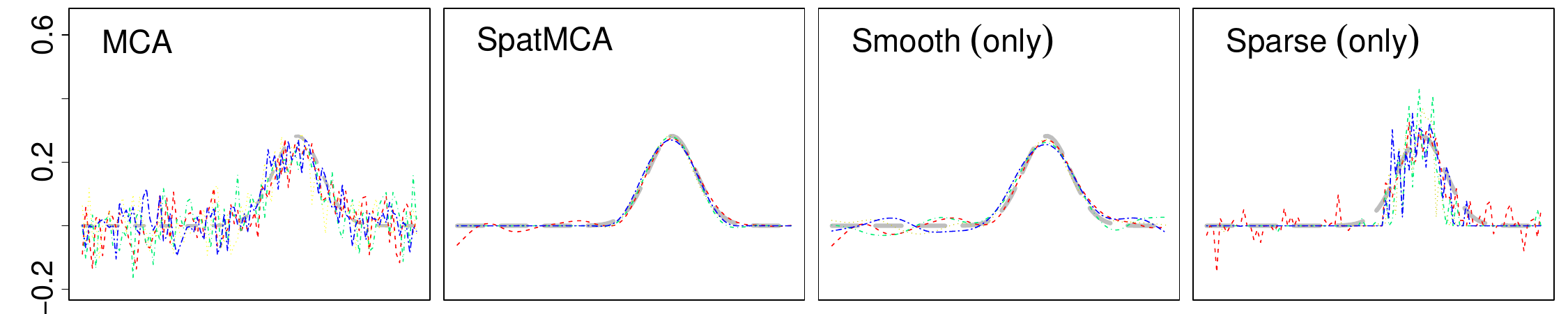}\\
	$\hat{v}_1(\cdot)$ based on $K=1$ for $ (d_1,d_2)=(0.5,0)$\\
	\includegraphics[scale=0.39]{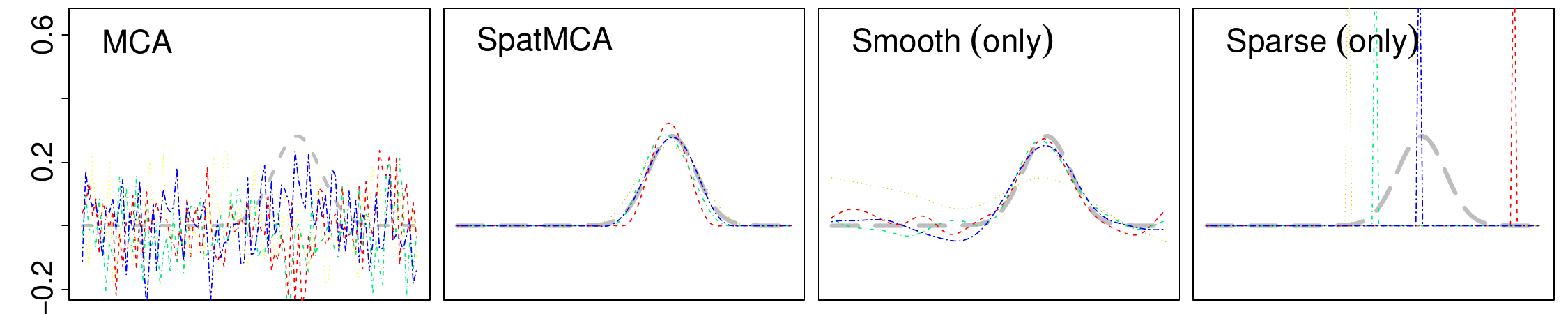}\\
	$\hat{v}_1(\cdot)$ based on $K=2$ for $ (d_1,d_2)=(1,0.7)$\\
	\includegraphics[scale=0.39]{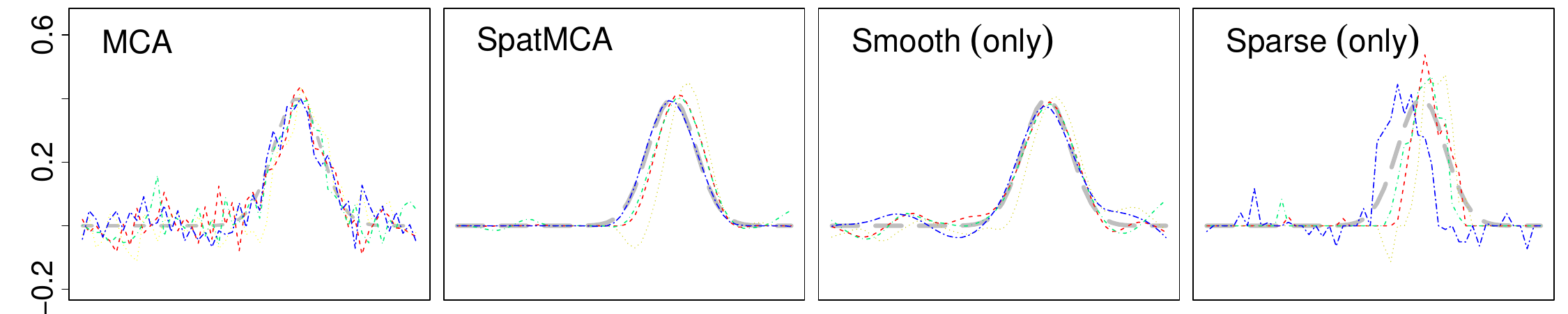}\\
	$\hat{v}_2(\cdot)$ based on $K=2$ for $  (d_1,d_2)=(1,0.7)$\\
	\includegraphics[scale=0.39]{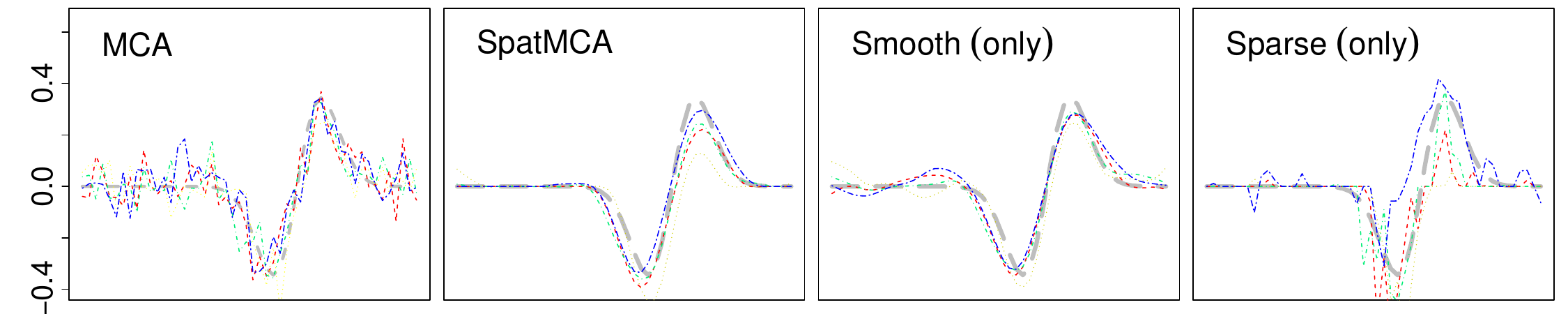}
	\caption{Estimates of $v_1(\cdot)$ and $v_2(\cdot)$ obtained from various methods based on data generated from three different combinations of singular values. Each panel consists of four estimates (in four different line types) corresponding to four randomly generated datasets, where the dash gray lines are the true $v_1(\cdot)$ and $v_2(\cdot)$.}
	\label{fig:est_v_d1}
\end{figure}	
The cross-covariance function estimates for the four methods based on a randomly generated dataset are shown in Figure~\ref{fig:cov_d1}. The proposed SpatMCA can be seen to perform better than the other methods for all cases. Figure~\ref{fig:box_d1_loss} shows boxplots of the four methods in terms of the loss function \eqref{eq:loss1} based on $50$ simulation replicates, which further confirms the superiority of SpatMCA.

\begin{figure}\centering
	$(d_1,d_2)=(1,0)$
	\includegraphics[scale=0.42]{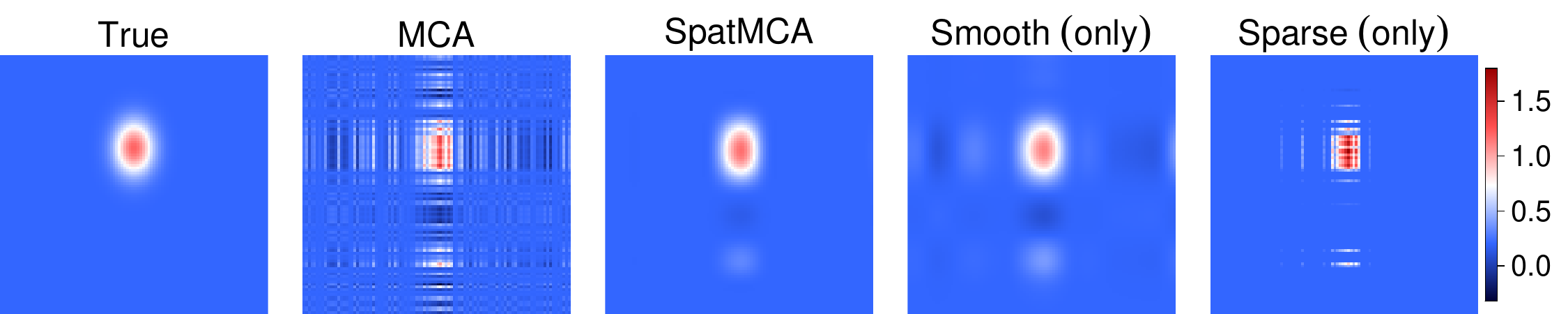}
	$(d_1,d_2)=(0.5)$
	\includegraphics[scale=0.42]{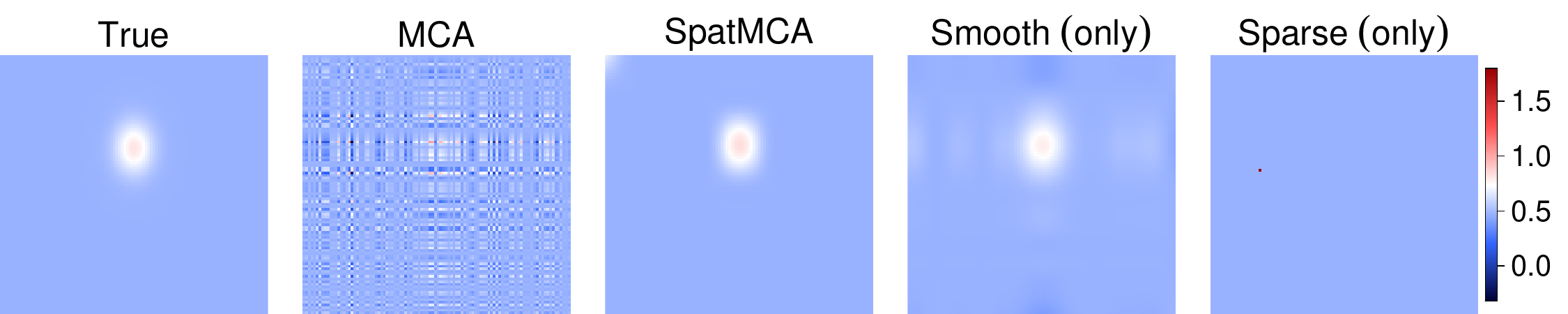}
	$(d_1,d_2)=(1,0.7)$
	\includegraphics[scale=0.42]{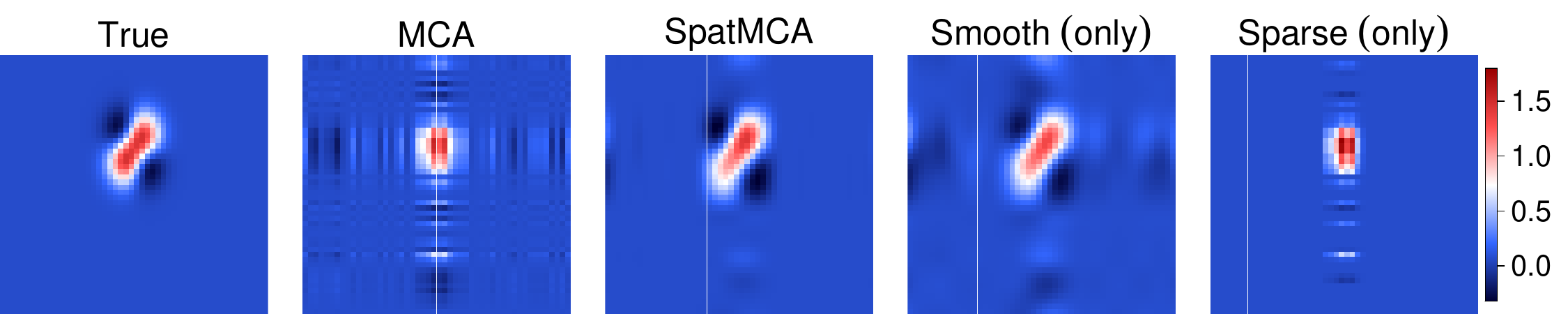}
	\caption{True cross-covariance functions and their estimates obtained from various methods with the rank $\hat{K}$ selected by CV for three different combinations of {singular values}.}
	\label{fig:cov_d1}
\end{figure}

\begin{figure}
	\begin{tabular}{ccc}
		{{$(d_1,d_2)=(1,0)$}, $K=1$}&{{$(d_1,d_2)=(0.5,0)$}, $K=1$}&{{$(d_1,d_2)=(1,0.7)$}, $K=1$}\\
		\includegraphics[scale=0.12]{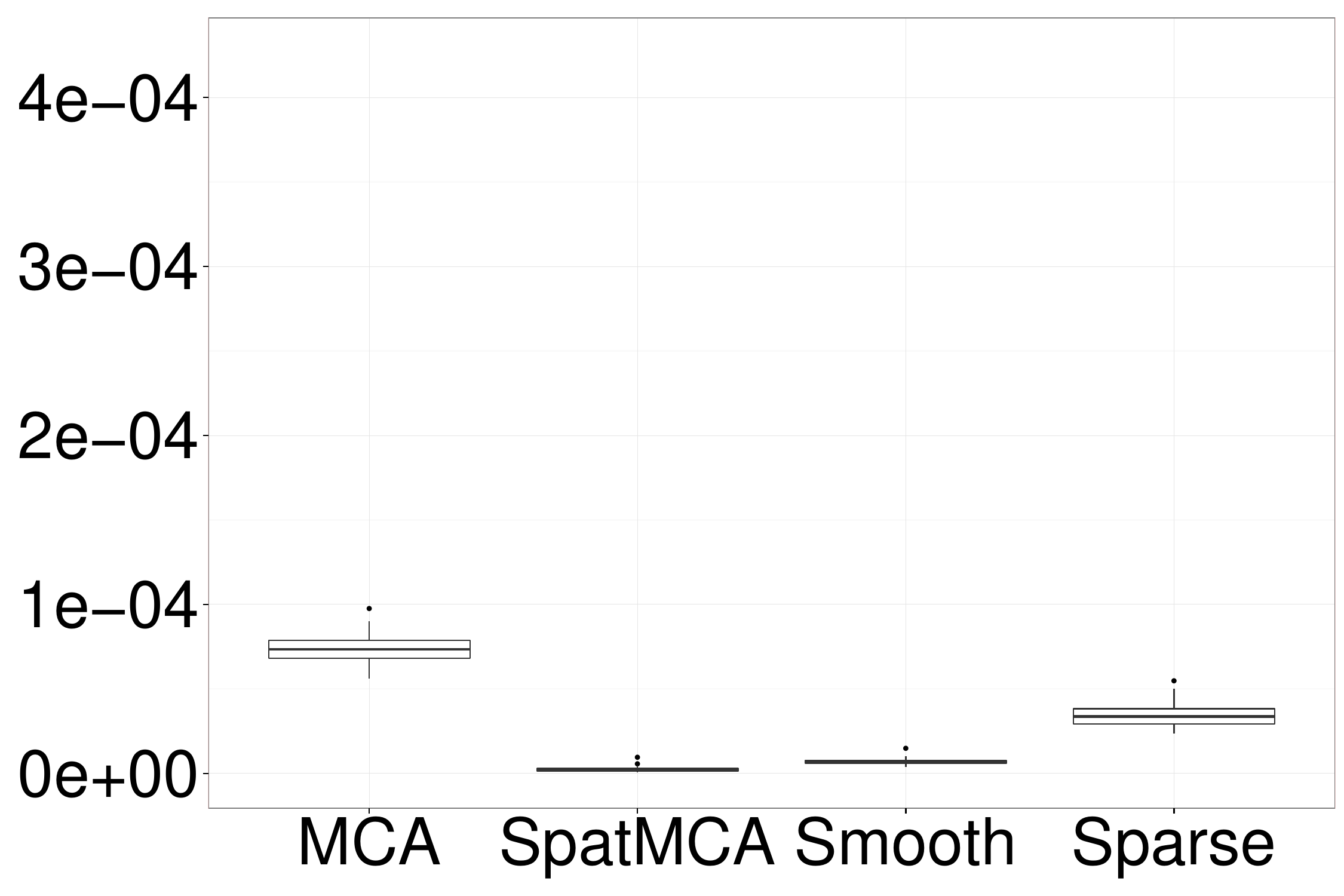}\hspace{4pt}&
		\includegraphics[scale=0.12]{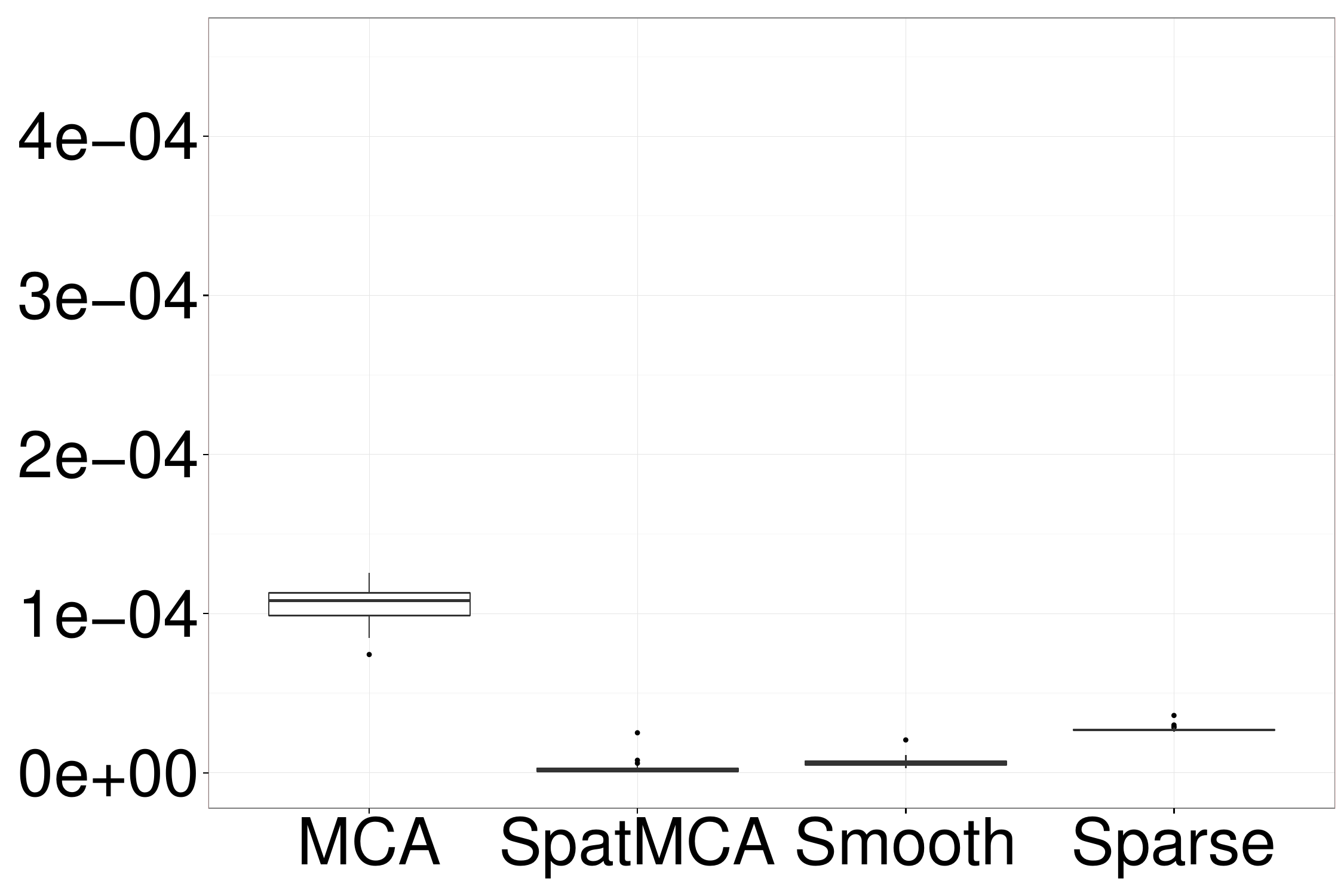}\hspace{4pt}&
		\includegraphics[scale=0.12]{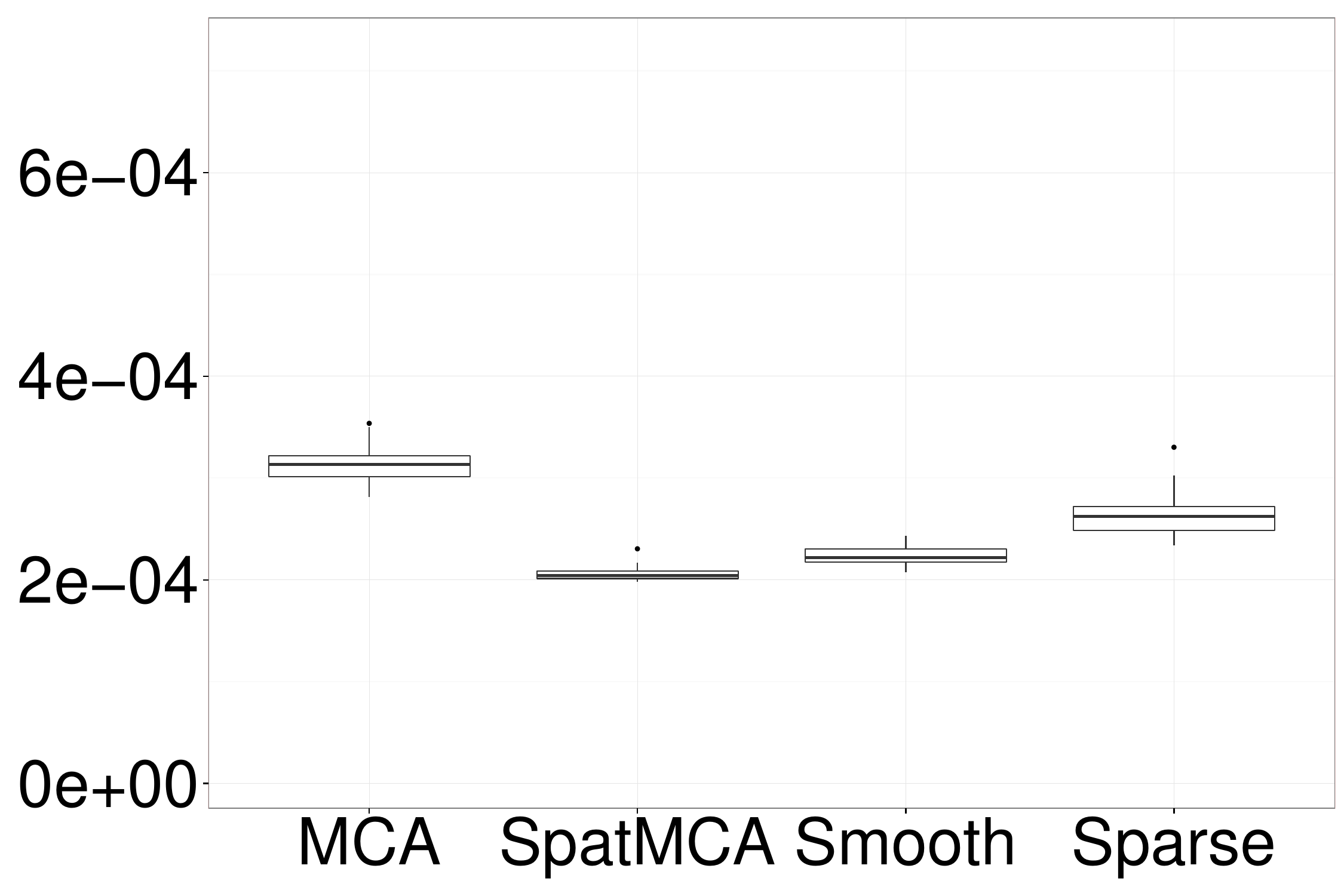}\hspace{4pt}\\
		{{$(d_1,d_2)=(1,0)$}, $K=2$}&{{$(d_1,d_2)=(0.5,0)$}, $K=2$}&{{$(d_1,d_2)=(1,0.7)$}, $K=2$}\\
		\includegraphics[scale=0.12]{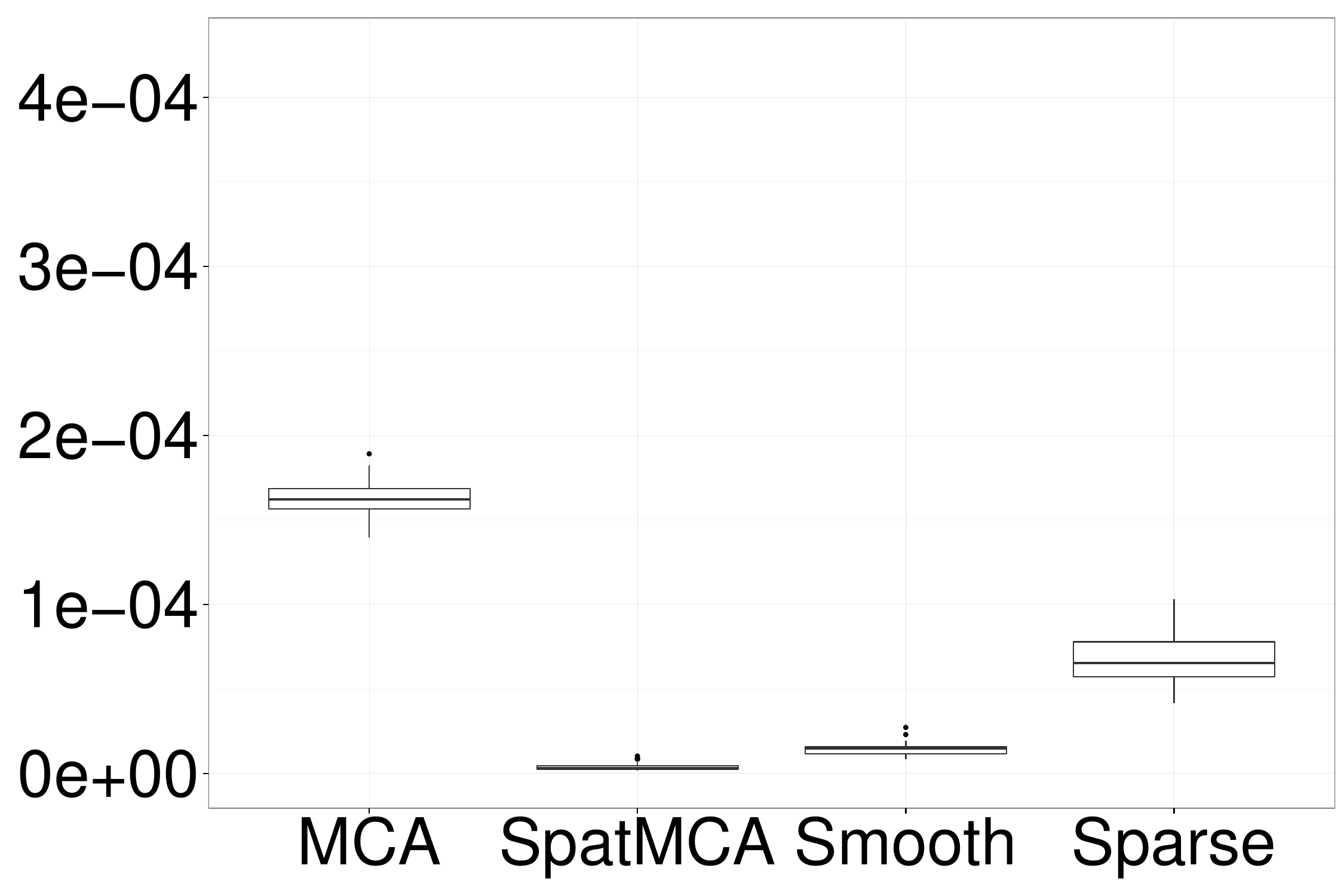}\hspace{4pt}&
		\includegraphics[scale=0.12]{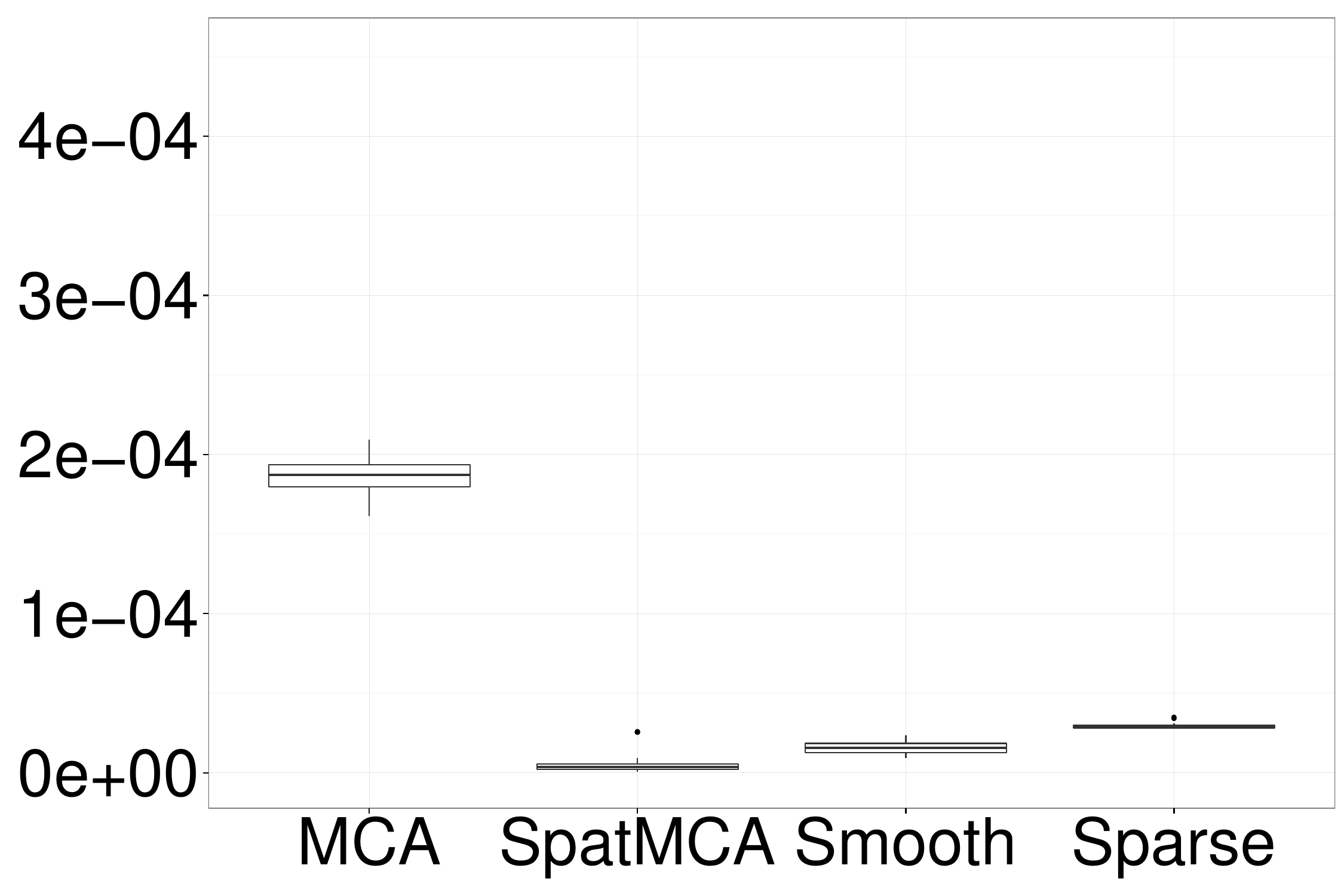}\hspace{4pt}&
		\includegraphics[scale=0.12]{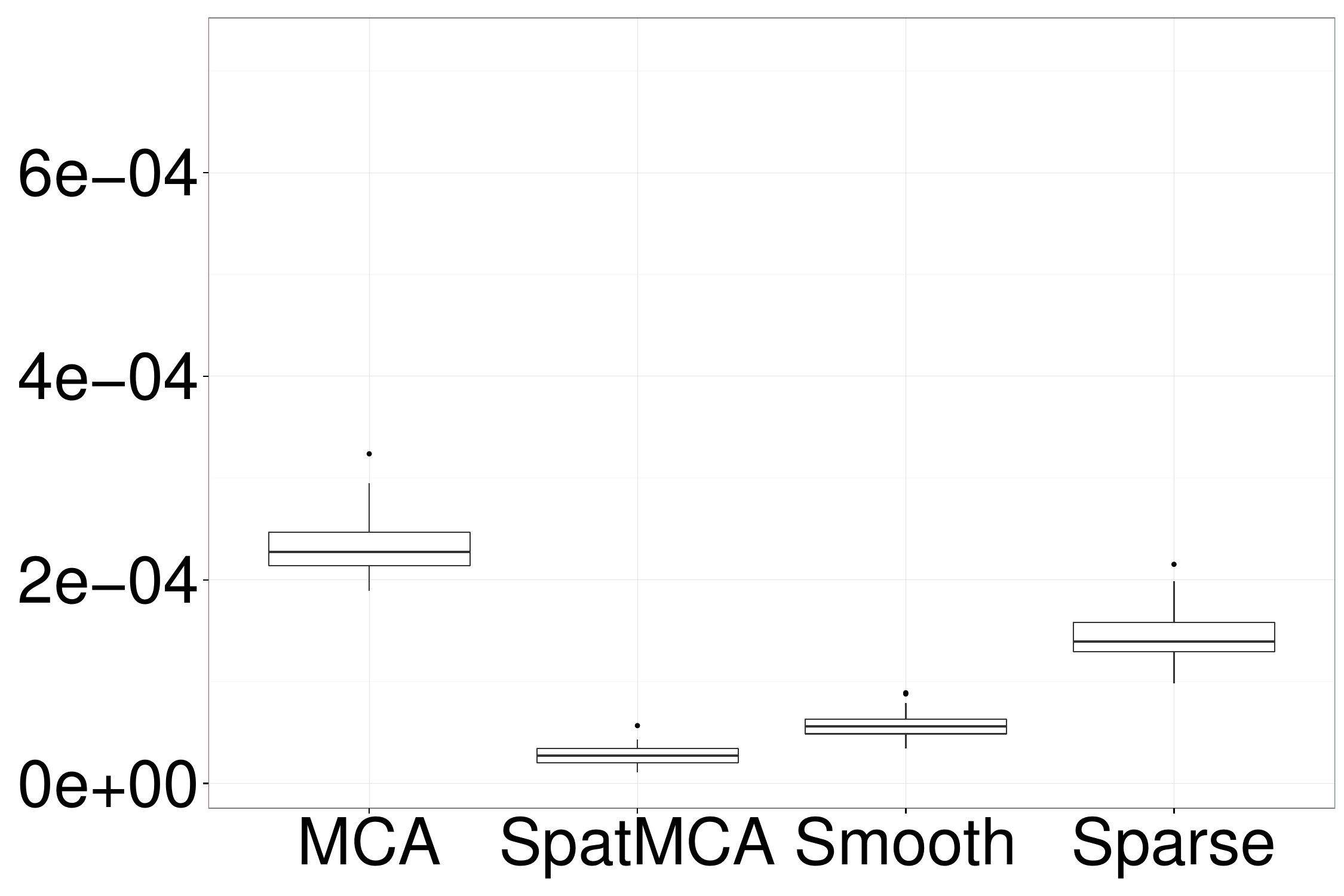}\hspace{4pt}\\
		{{$(d_1,d_2)=(1,0)$}, $K=5$}&{{$(d_1,d_2)=(0.5,0)$}, $K=5$}&{{$(d_1,d_2)=(1,0.7)$}, $K=5$}\\
		\includegraphics[scale=0.12]{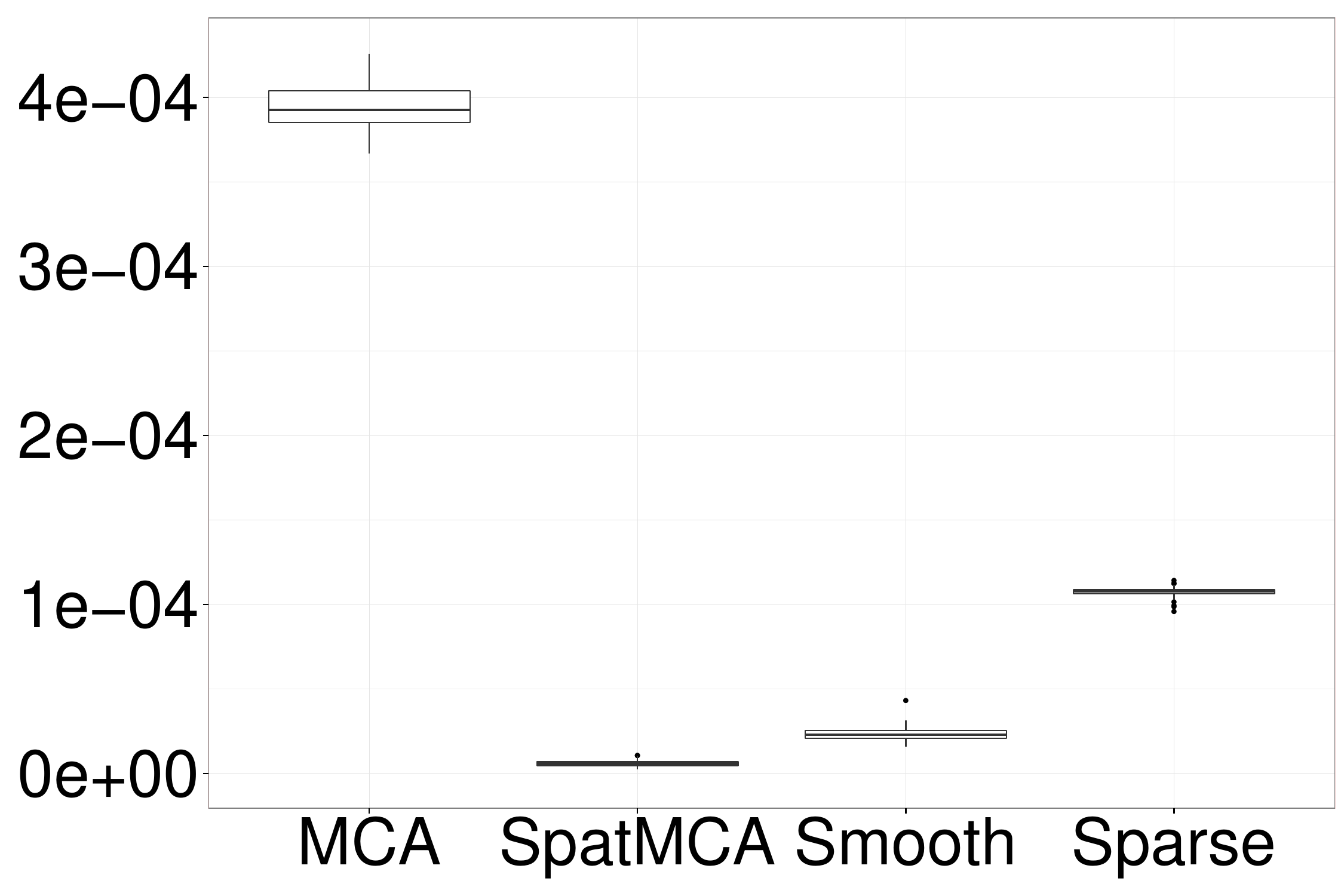}\hspace{4pt}&
		\includegraphics[scale=0.12]{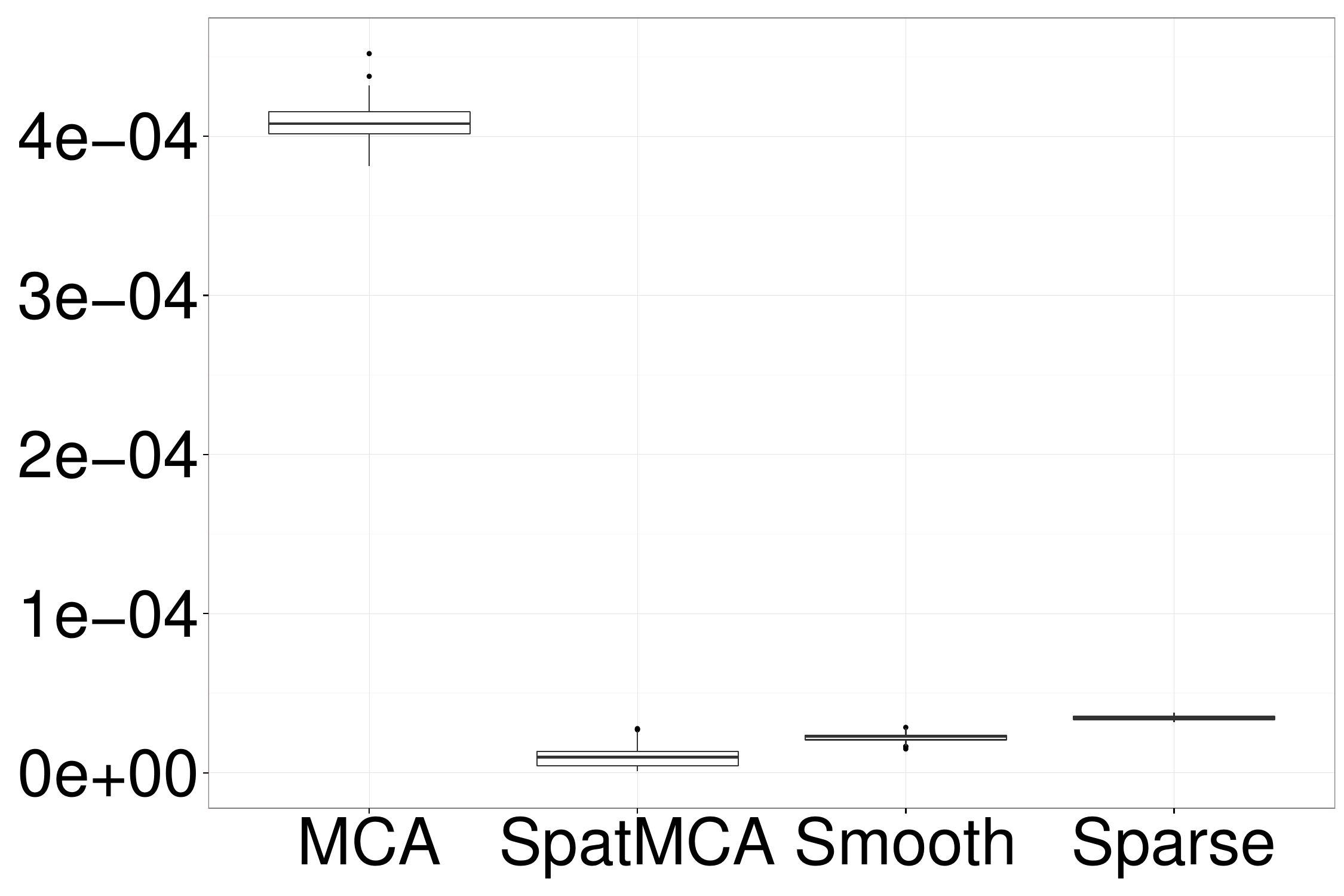}\hspace{4pt}&
		\includegraphics[scale=0.12]{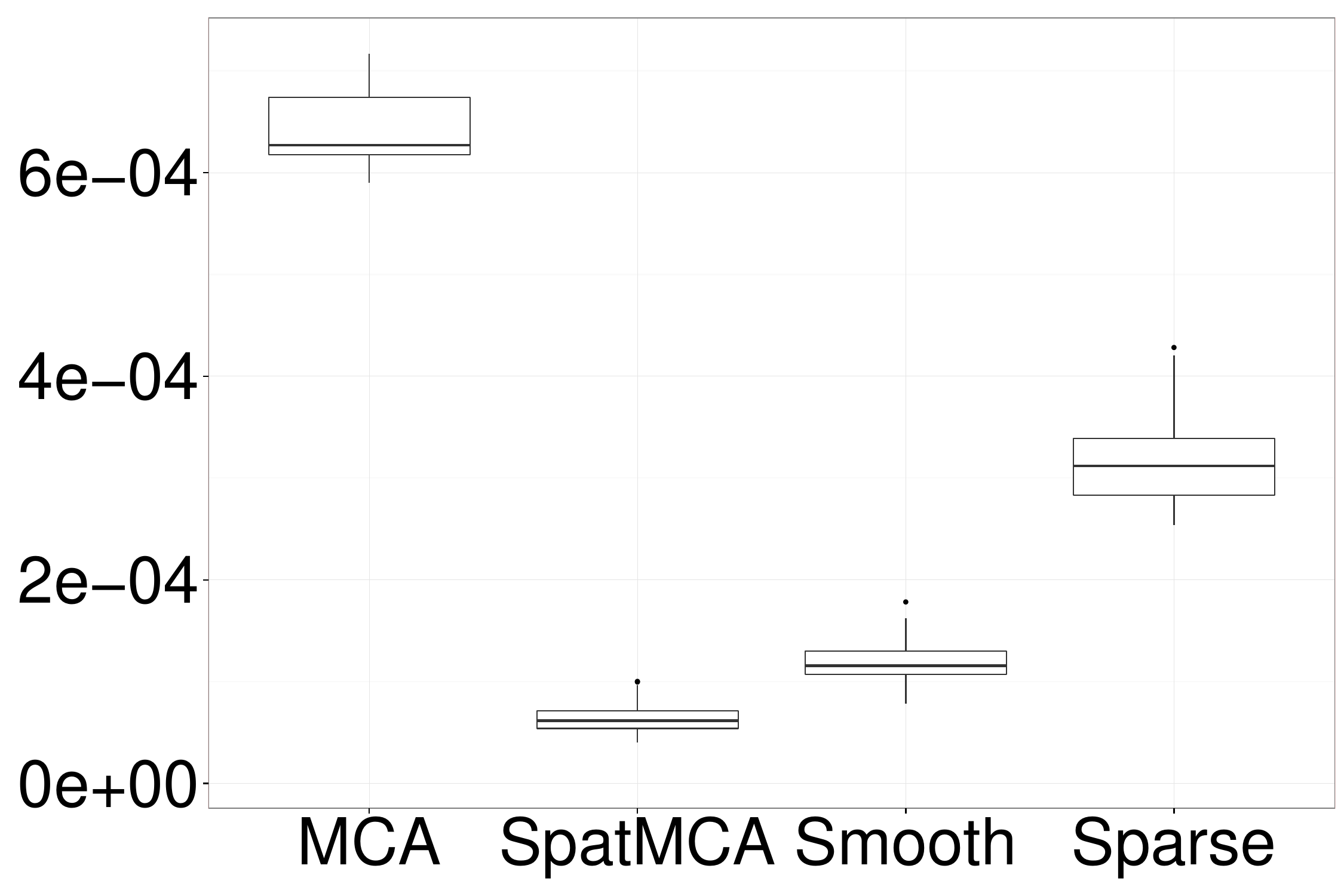}\hspace{4pt}\\
		{{$(d_1,d_2)=(1,0)$}, $K=\hat{K}$}&{{$(d_1,d_2)=(0.5,0)$}, $K=\hat{K}$}&{{$(d_1,d_2)=(1,0.7)$}, $K=\hat{K}$}\\
		\includegraphics[scale=0.12]{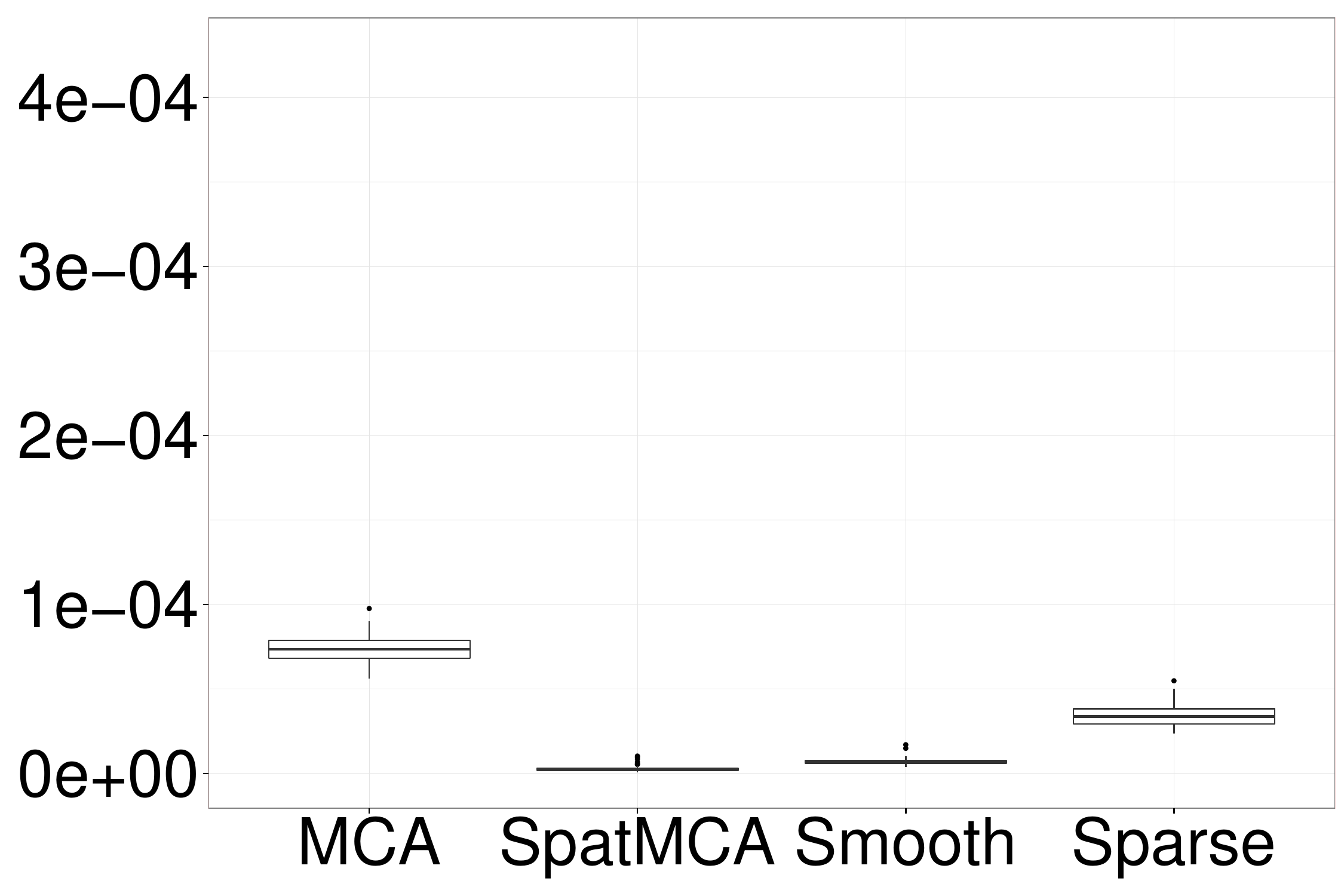}\hspace{4pt}&
		\includegraphics[scale=0.12]{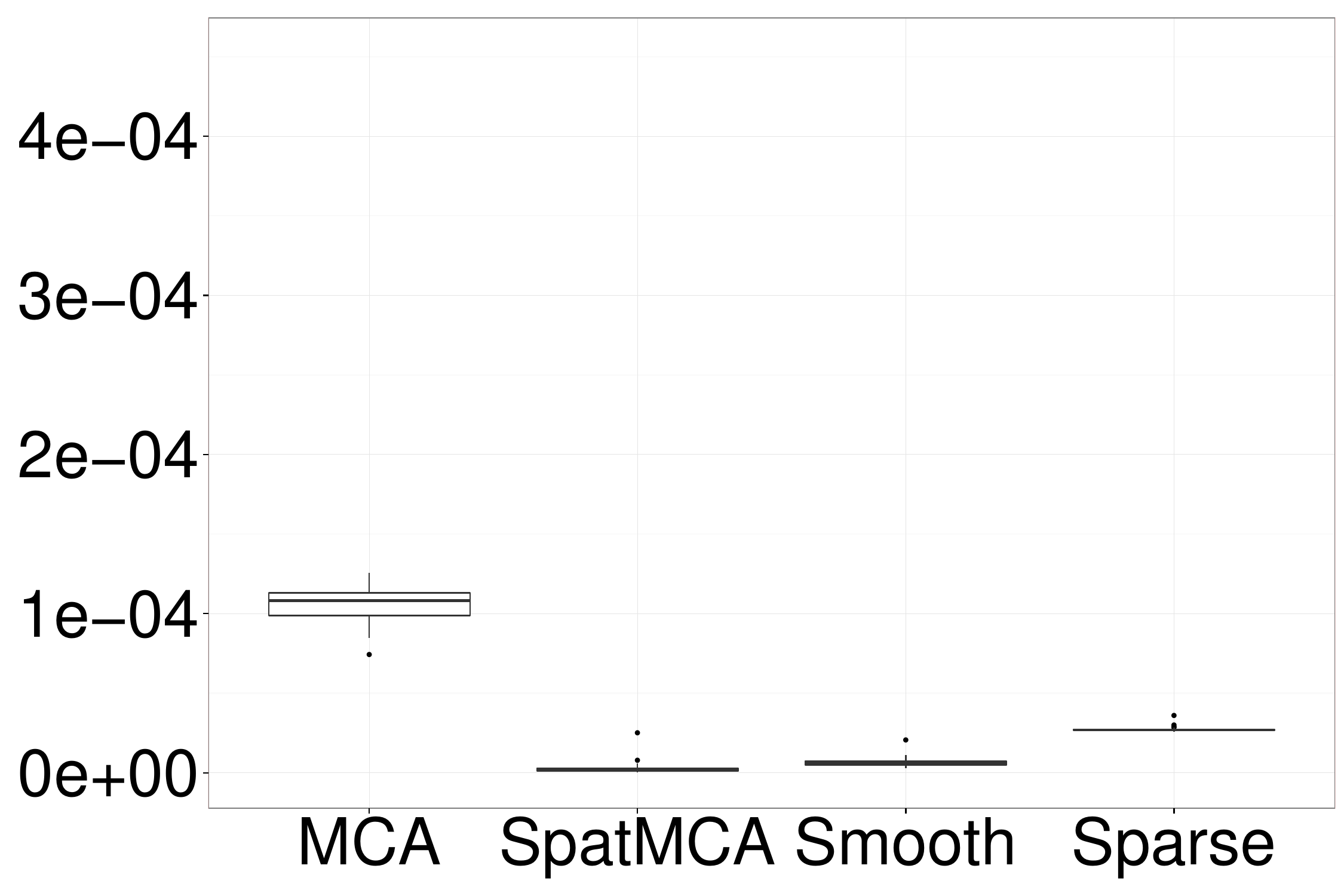}\hspace{4pt}&
		\includegraphics[scale=0.12]{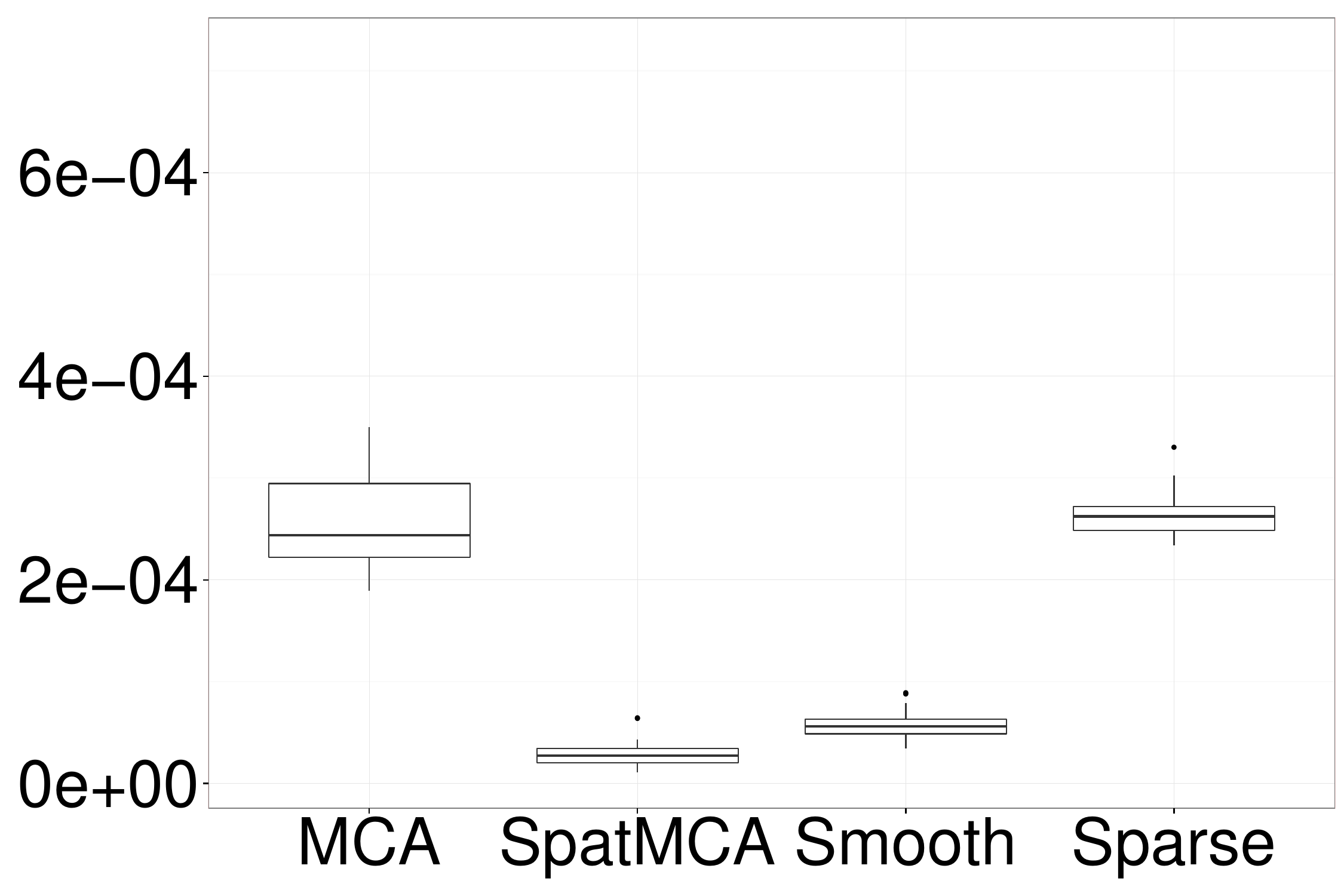}\hspace{4pt}
	\end{tabular}
	\caption{Boxplots of average squared prediction errors of (\ref{eq:loss1}) for various methods in the one-dimensional simulation experiment based on 50 simulation replicates.}
	\label{fig:box_d1_loss}
\end{figure}

\subsection{A Two-Dimensional Experiment}\label{sec:2d}

For a two-dimensional experiment, we generated data according to (\ref{eq:ch5measurement}) with $K=2$, $n=5,000$, $p_1=25^2$, $p_2=20^2$,  \[\left(\begin{array}{c}
\bm{\eta}_{1i} \\
\bm{\eta}_{2i} \\
\end{array}\right)\sim N\left(\bm{0}, \left(\begin{array}{cc}
\bm{I} & \bm{U}\mathrm{diag}(d_1,d_2)\bm{V}' \\
\bm{V}\mathrm{diag}(d_1,d_2)\bm{U}' & \bm{I} \\
\end{array}\right)\right),\] 
$\bm{\epsilon}_{ji}\sim N(\bm{0},\bm{I})$ for $j=1,2$, $(\bm{s}_{11}, \dots,\bm{s}_{1p_1})$ equally spaced in $[-5,5]^2$, and $(\bm{s}_{21}, \dots,\bm{s}_{2p_2})$ equally spaced in $[-7,7]^2$.
Here $u_1(\cdot)$, $v_1(\cdot)$, $u_2(\cdot)$ {and} $v_2(\cdot)$ are given by \eqref{eq:u1_sim}, \eqref{eq:v1_sim}, \eqref{eq:u2_sim} and \eqref{eq:v2_sim} with $d=2$, respectively. We considered three pairs of $(d_1,d_2)\in \{(1,0), (0.5,0), (1,0.7)\}$, and applied the proposed SpatMCA with $K=\{1,2,5\}$ and $\hat{K}$ selected by \eqref{eq:ch5khat}, resulting in $12$ different combinations. Similar to the previous subsection, we applied {the} 5-fold CV of \eqref{eq:ch5khat} to select ${\{}\tau_{1u}, \tau_{1v}, \tau_{2u}, \tau_{2v}{\}}$ among $21$ values of $\tau_{1u}$ and $\tau_{1v}$  (including $0$ and the other 20 values equally spaced on the log scale from $10^{-2}$ to $10$) and $11$ values of $\tau_{2u}$ and $\tau_{2v}$ (including $0$ and the other 10 values equally spaced on the log scale from $10^{-3}$ to $1$). 

Figures~\ref{fig:est_u_d2} and \ref{fig:est_v_d2} show the estimates of $u_k(\cdot)$ and $v_k(\cdot)$, respectively, for the four methods based on randomly selected data generated from three different combinations of singular values. Figure~\ref{fig:box_d2_loss} shows the performance of the four methods in terms of the loss function \eqref{eq:loss1} based on 50 simulation replicates. Similar to the one-dimensional example, SpatMCA outperforms all the other methods in all cases.

\begin{figure}\centering
	{$\hat{u}_1(\cdot)$ based on {$K=1$ for} $(d_1,d_2)=(1,0)$
		\includegraphics[scale=0.39]{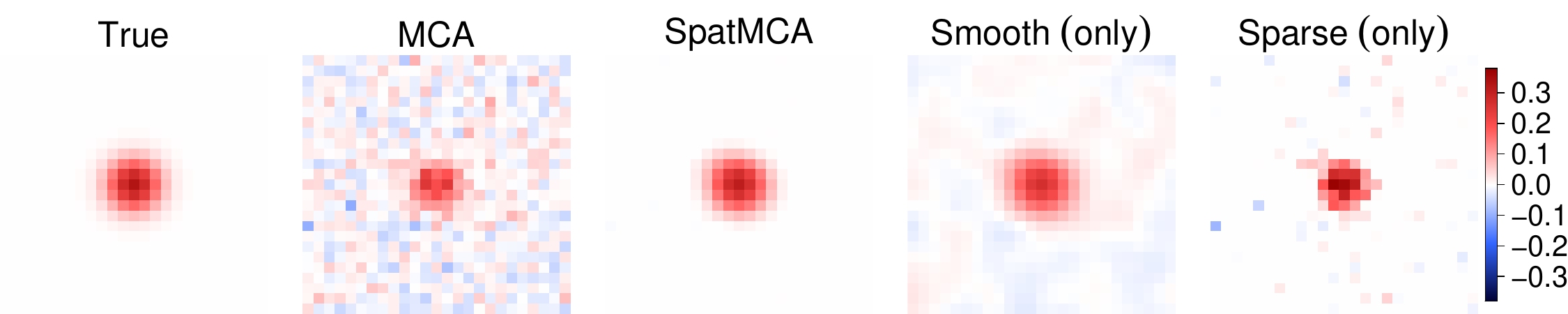}
		$\hat{u}_1(\cdot)$ based on {$K=1$ for} $(d_1,d_2)=(0.5,0)$
		\includegraphics[scale=0.39]{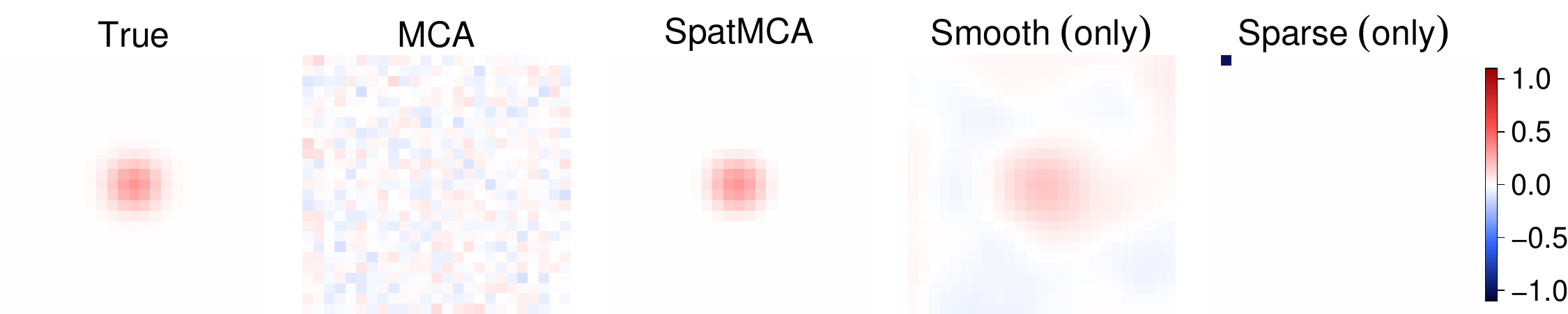}
		$\hat{u}_1(\cdot)$ based on {$K=2$ for} $(d_1,d_2)=(1,0.7)$
		\includegraphics[scale=0.39]{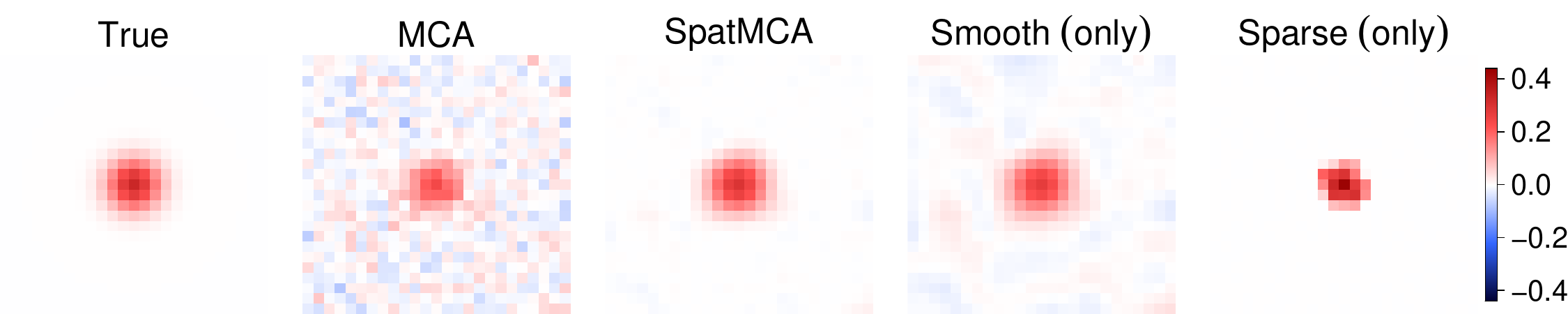}
		$\hat{u}_2(\cdot)$ based on {$K=2$ for} $(d_1,d_2)=(1,0.7)$}
	\includegraphics[scale=0.39]{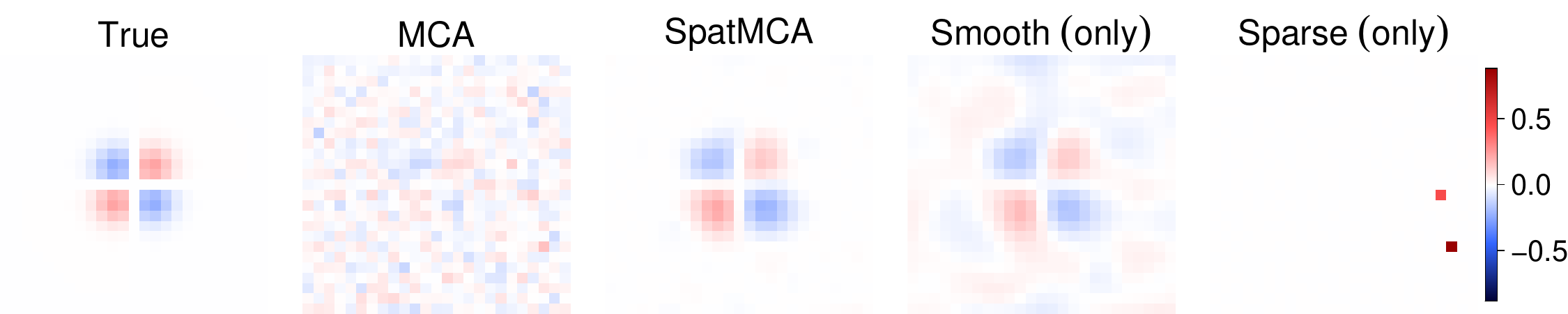}
	\caption{Estimates of $u_1(\cdot)$ and $u_2(\cdot)$ obtained from various methods based on three different combinations of singular values.}
	\label{fig:est_u_d2}
\end{figure}	

\begin{figure}\centering
	$\hat{v}_1(\cdot)$ based on {$K=1$ for} $(d_1,d_2)=(1,0)$
	\includegraphics[scale=0.39]{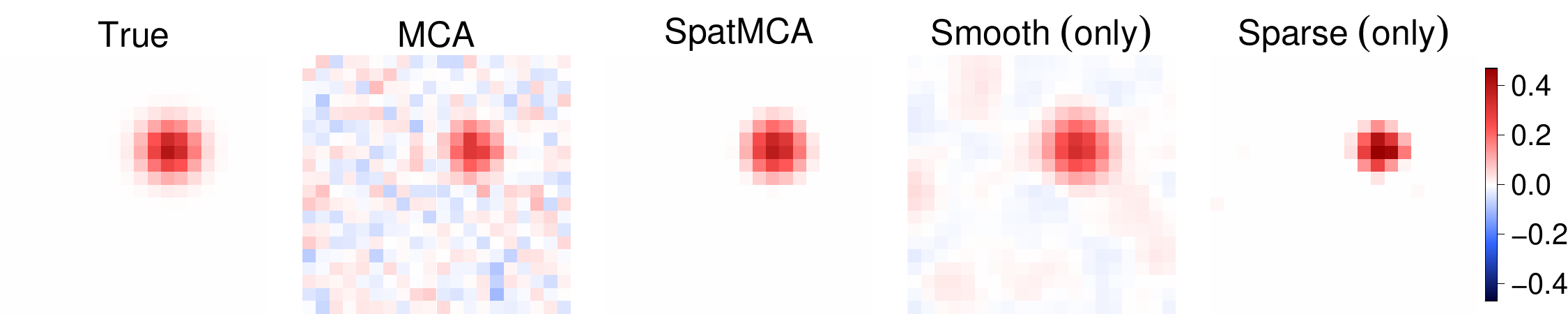}\\
	$\hat{v}_1(\cdot)$ based on {$K=1$ for} $(d_1,d_2)=(0.5,0)$
	\includegraphics[scale=0.39]{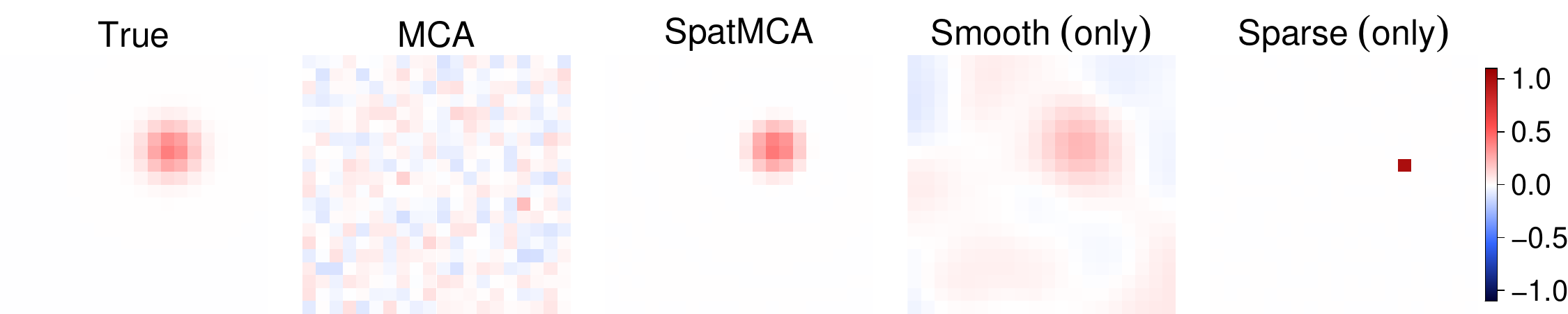}\\
	$\hat{v}_1(\cdot)$ based on {$K=2$ for} $(d_1,d_2)=(1,0.7)$
	\includegraphics[scale=0.39]{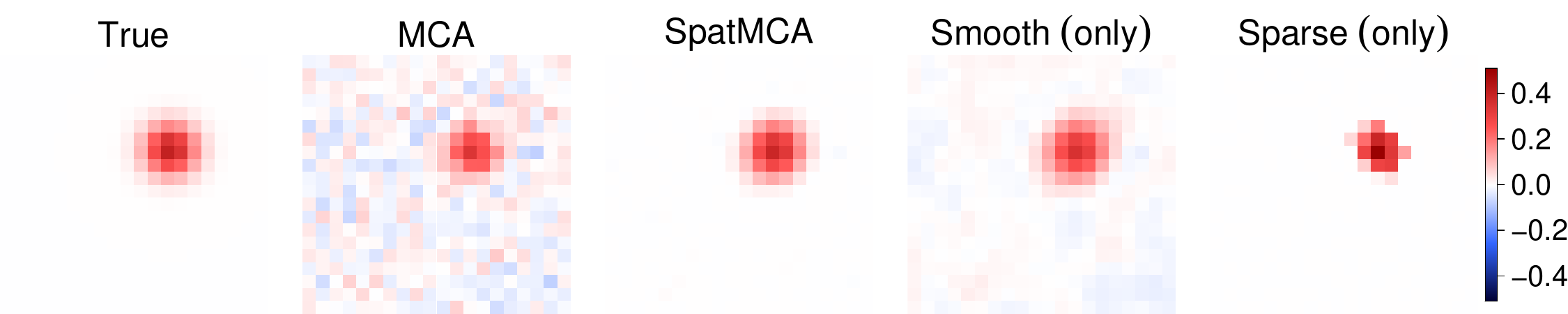}\\
	$\hat{v}_2(\cdot)$ based on {$K=2$ for} $(d_1,d_2)=(1,0.7)$
	\includegraphics[scale=0.39]{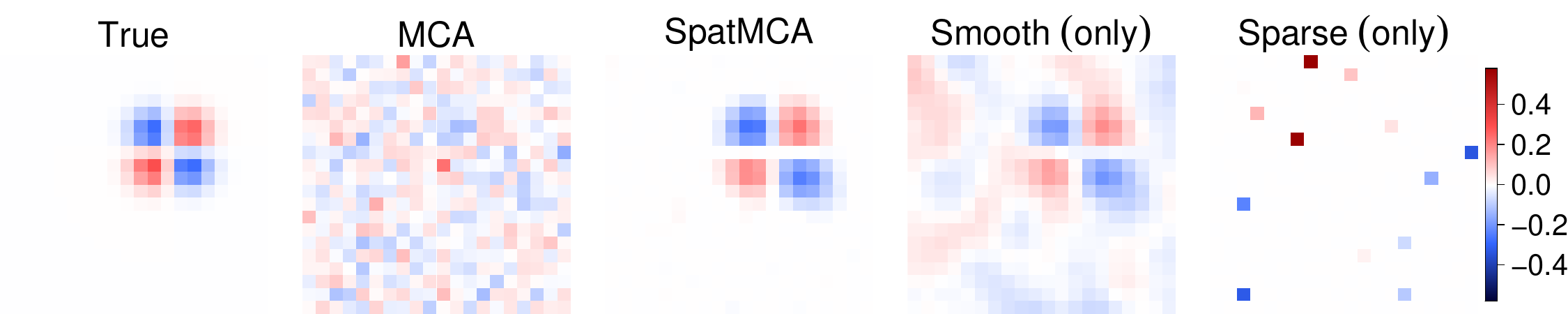}
	\caption{Estimates of $v_1(\cdot)$ and $v_2(\cdot)$ obtained from various methods based on three different combinations of singular values.}
	\label{fig:est_v_d2}
\end{figure}	

\begin{figure}
	\begin{tabular}{ccc}
		{{$(d_1,d_2)=(1,0)$}, $K=1$}&{{$(d_1,d_2)=(0.5,0)$}, $K=1$}&{{$(d_1,d_2)=(1,0.7)$}, $K=1$}\\
		\includegraphics[scale=0.12]{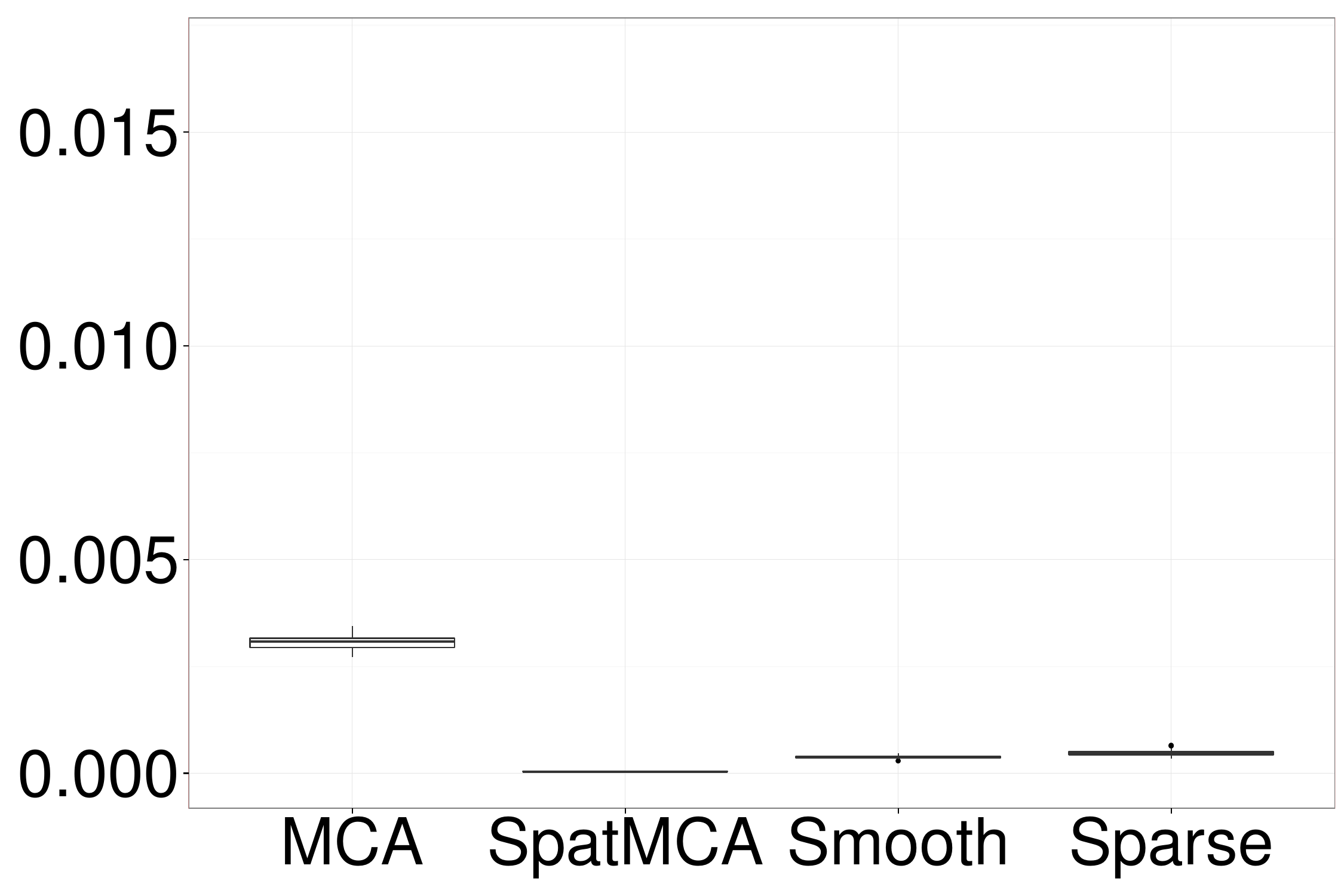}\hspace{4pt}&
		\includegraphics[scale=0.12]{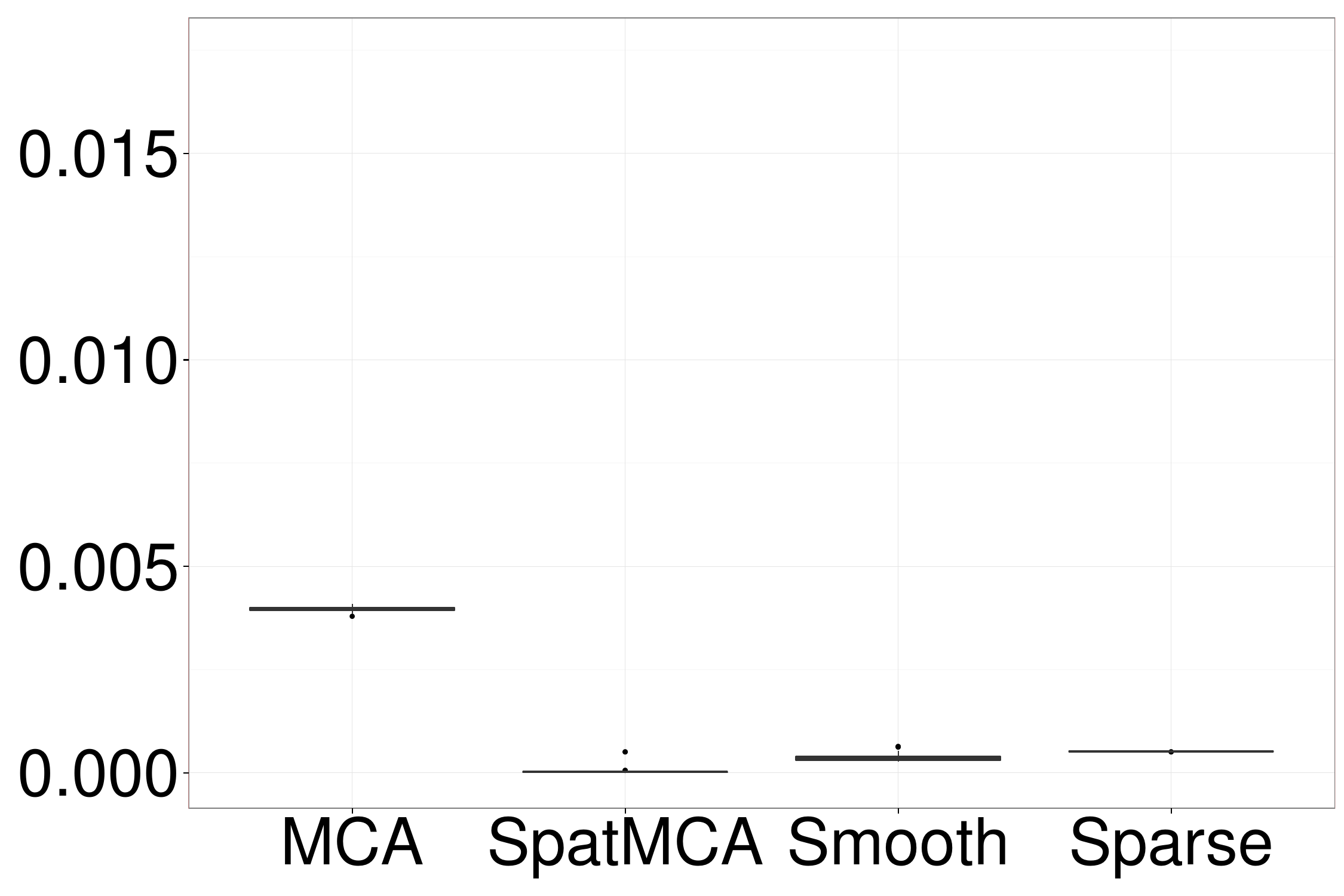}\hspace{4pt}&
		\includegraphics[scale=0.12]{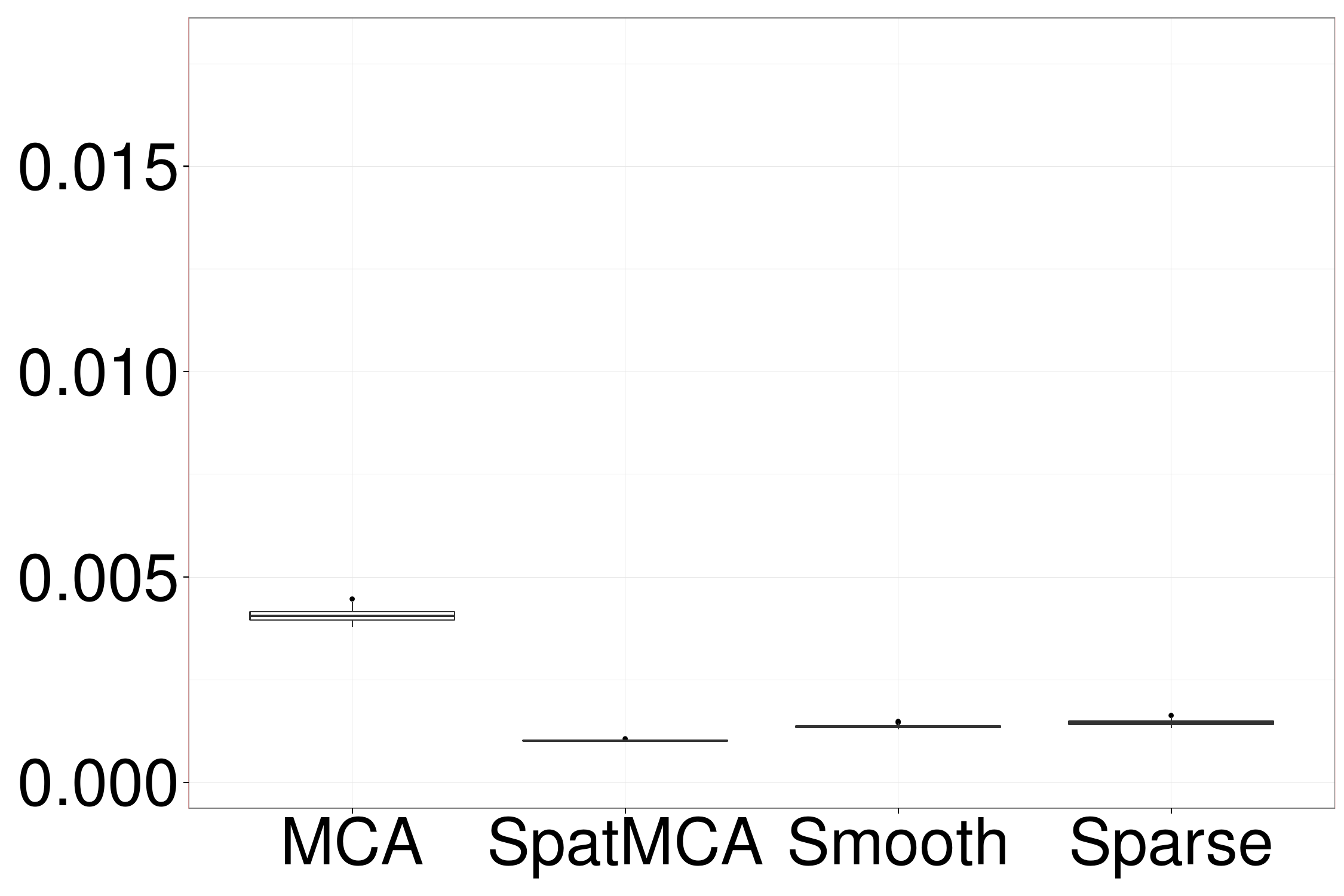}\hspace{4pt}\\
		{{$(d_1,d_2)=(1,0)$}, $K=2$}&{{$(d_1,d_2)=(0.5,0)$}, $K=2$}&{{$(d_1,d_2)=(1,0.7)$}, $K=2$}\\
		\includegraphics[scale=0.12]{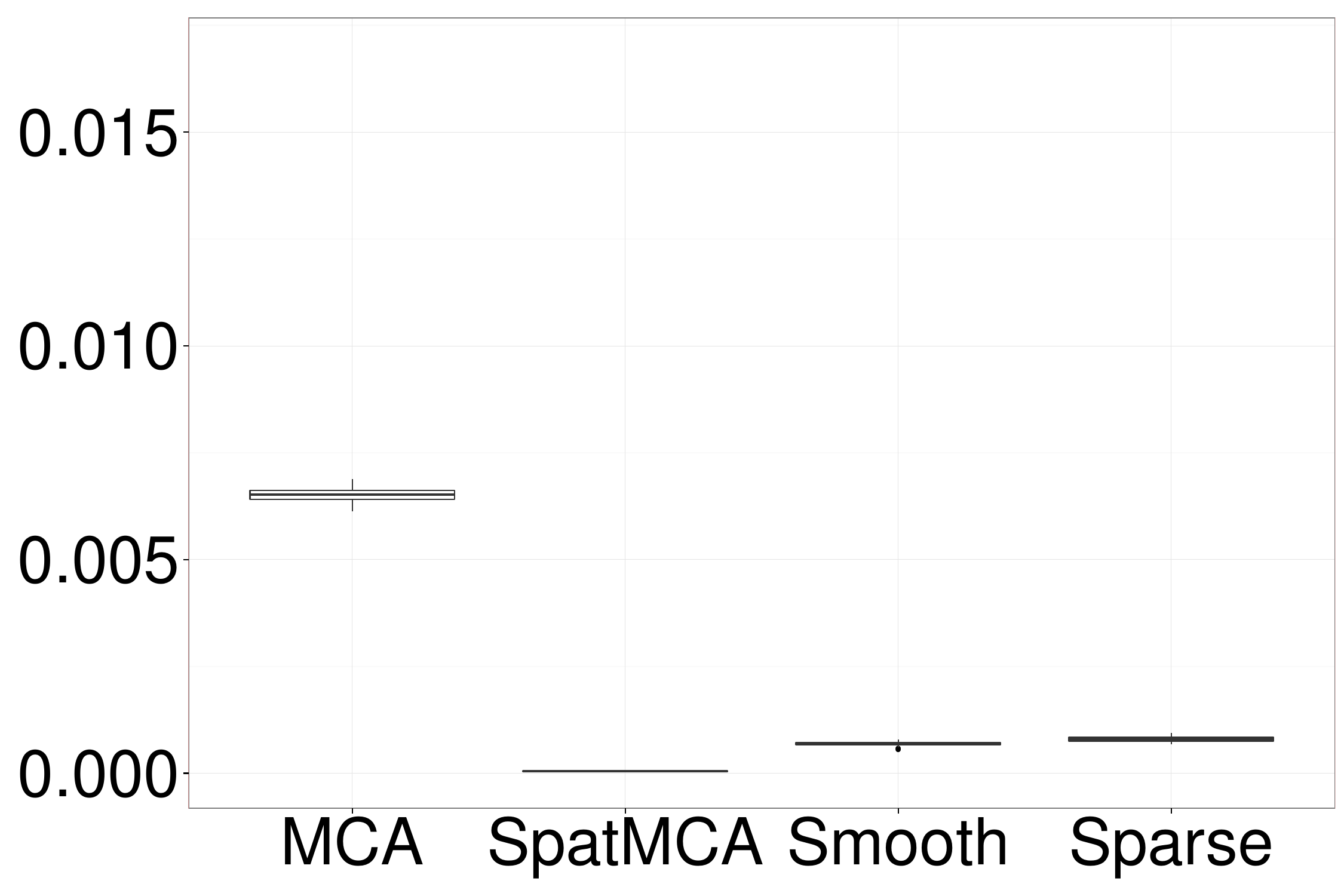}\hspace{4pt}&
		\includegraphics[scale=0.12]{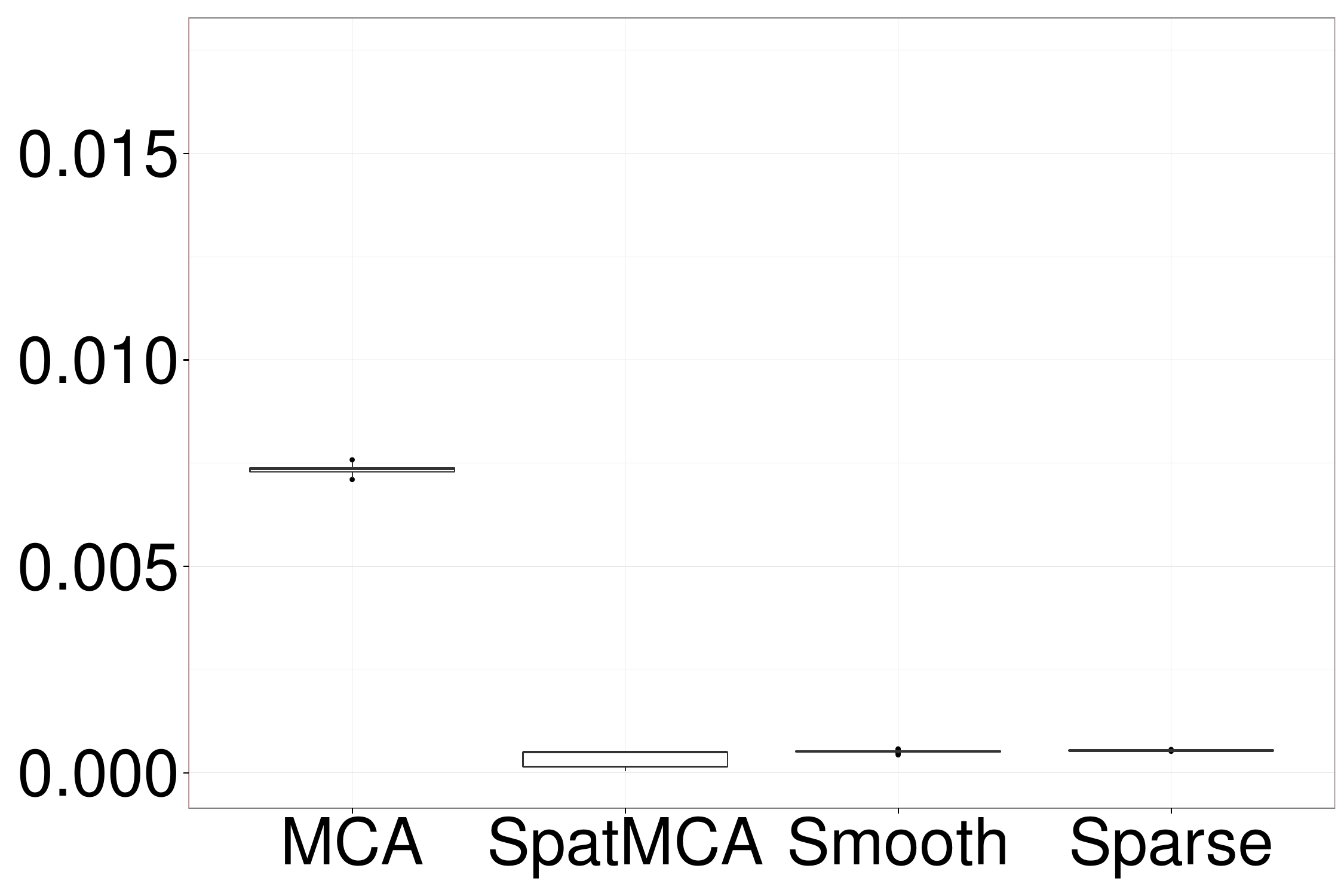}\hspace{4pt}&
		\includegraphics[scale=0.12]{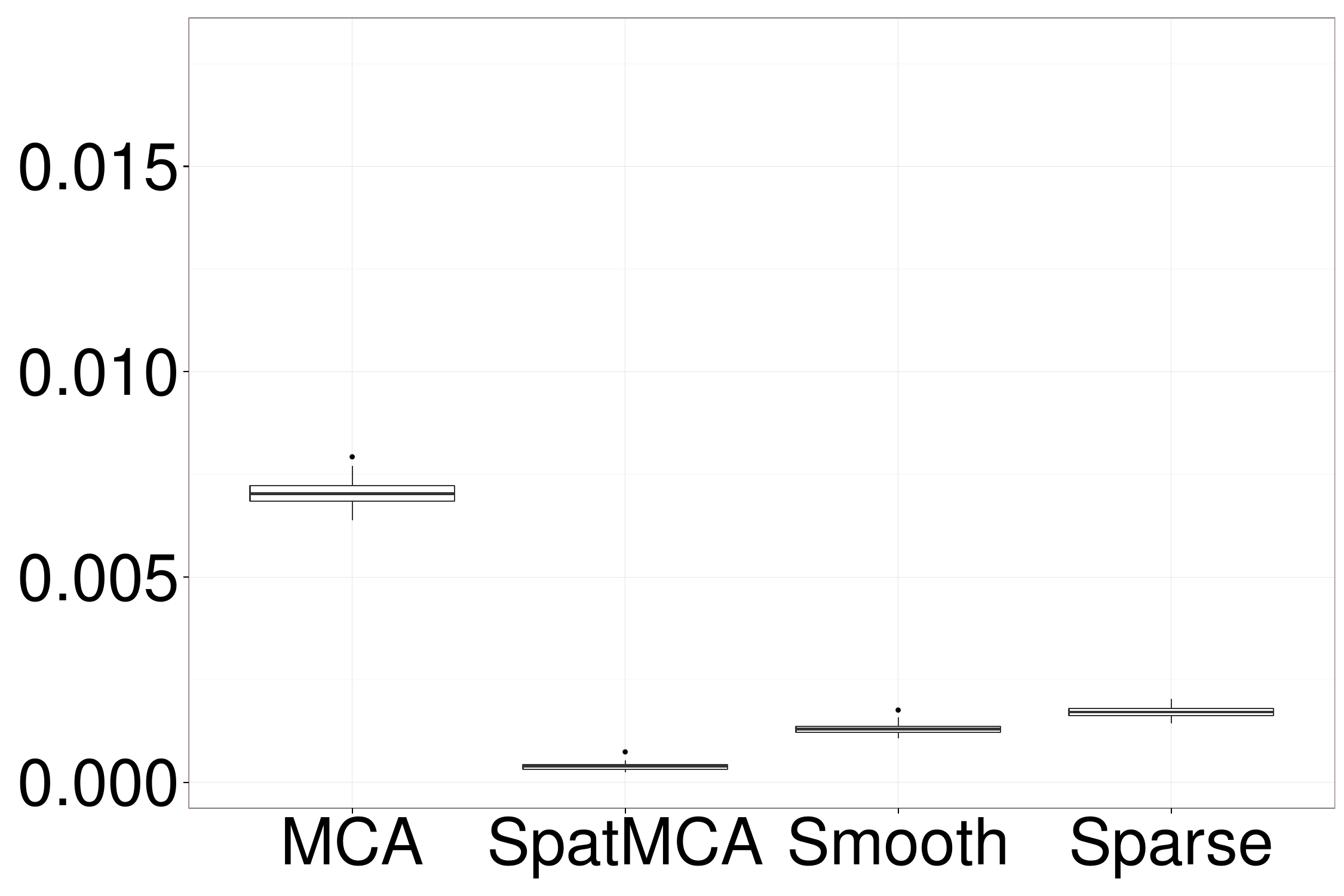}\hspace{4pt}\\
		{{$(d_1,d_2)=(1,0)$}, $K=5$}&{{$(d_1,d_2)=(0.5,0)$}, $K=5$}&{{$(d_1,d_2)=(1,0.7)$}, $K=5$}\\
		\includegraphics[scale=0.12]{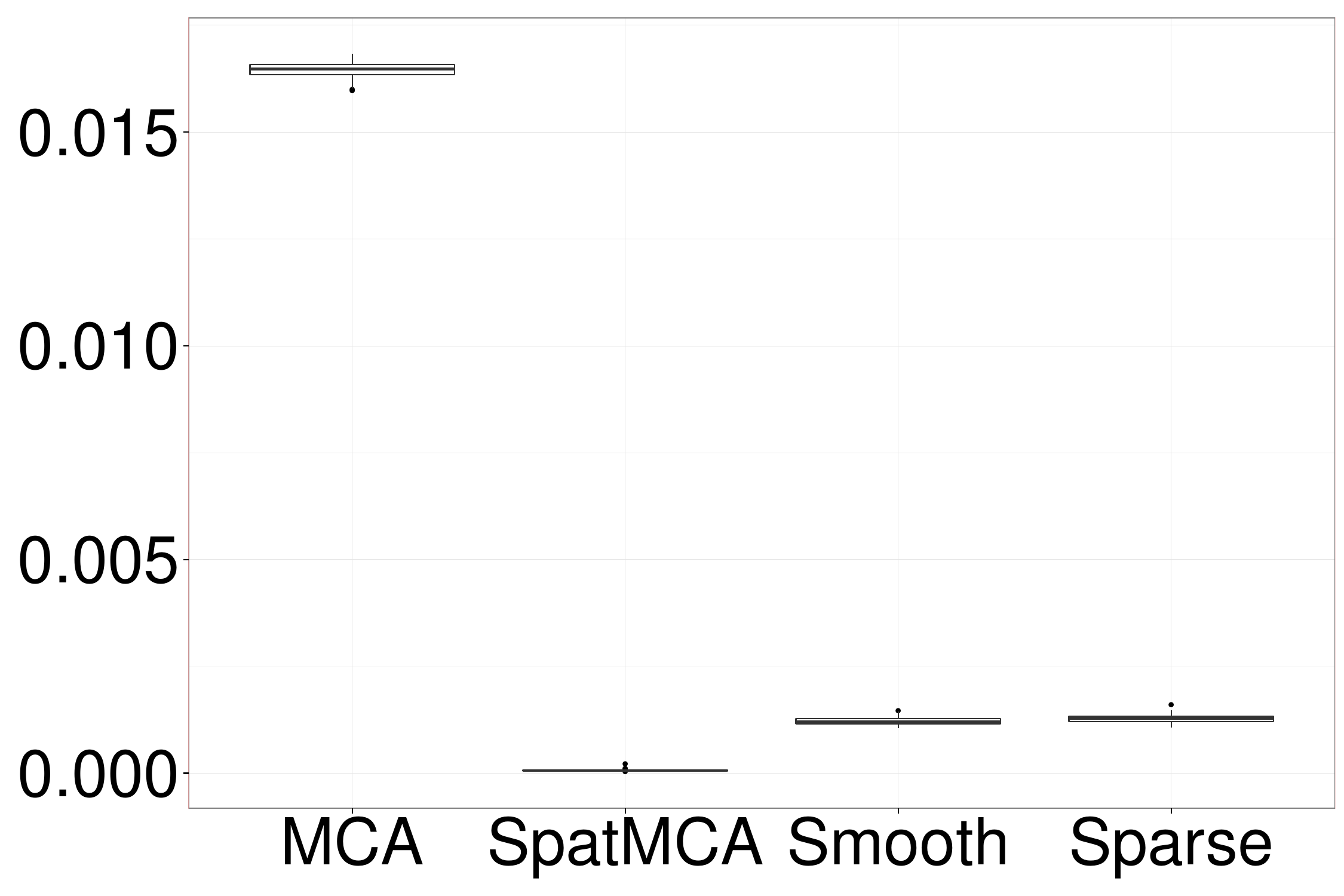}\hspace{4pt}&
		\includegraphics[scale=0.12]{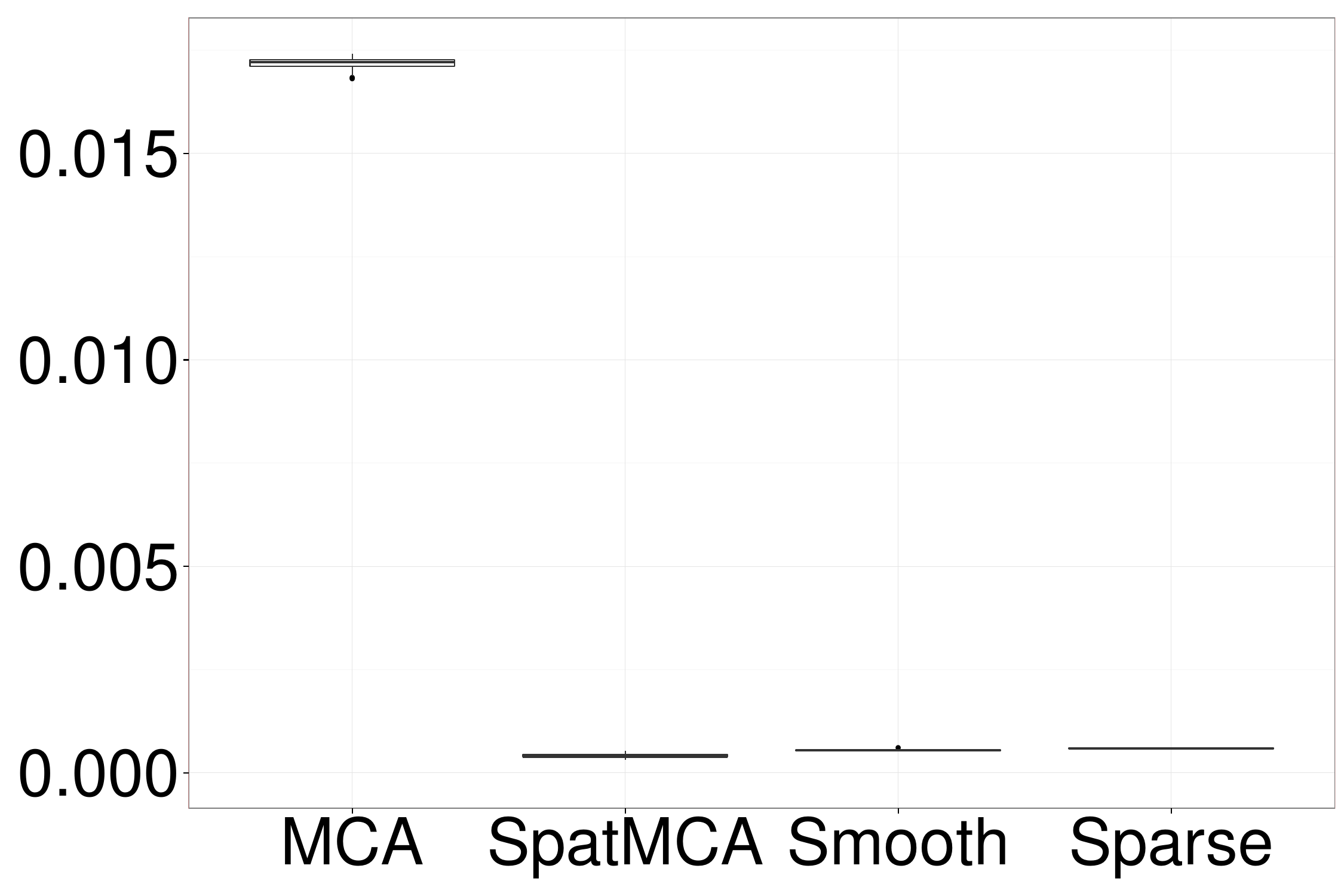}\hspace{4pt}&
		\includegraphics[scale=0.12]{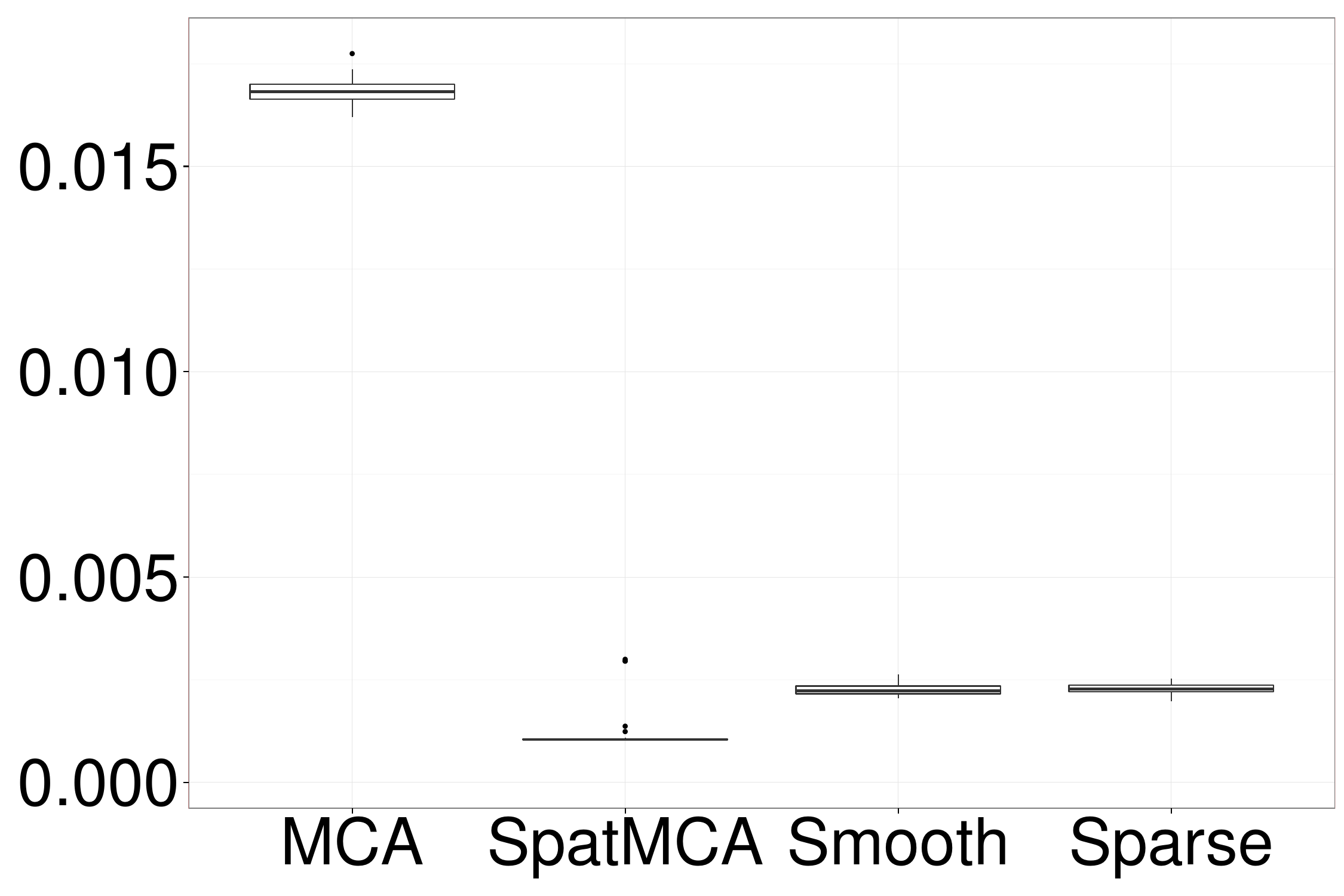}\hspace{4pt}\\
		{{$(d_1,d_2)=(1,0)$}, $K=\hat{K}$}&{{$(d_1,d_2)=(0.5,0)$}, $K=\hat{K}$}&{{$(d_1,d_2)=(1,0.7)$}, $K=\hat{K}$}\\
		\includegraphics[scale=0.12]{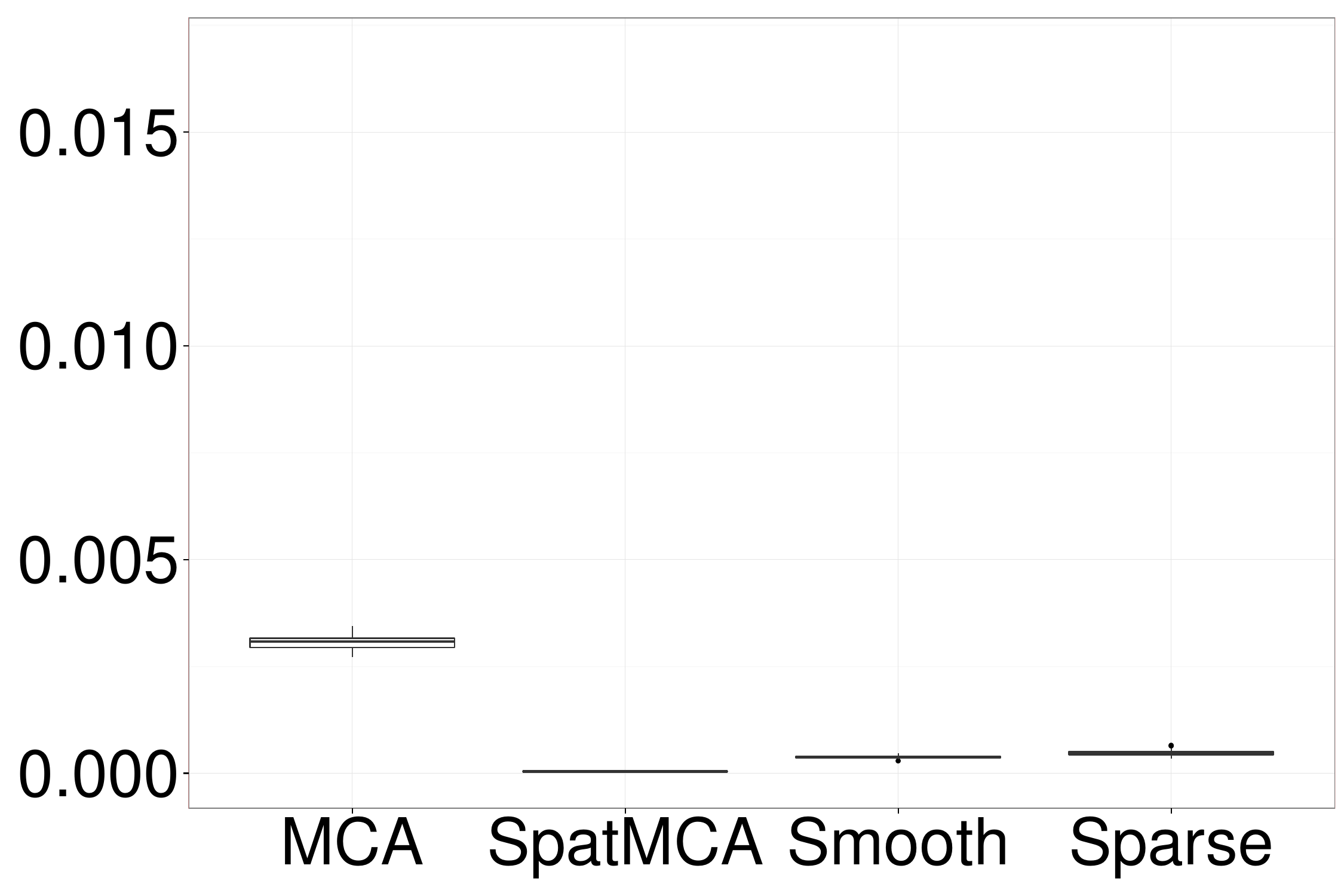}\hspace{4pt}&
		\includegraphics[scale=0.12]{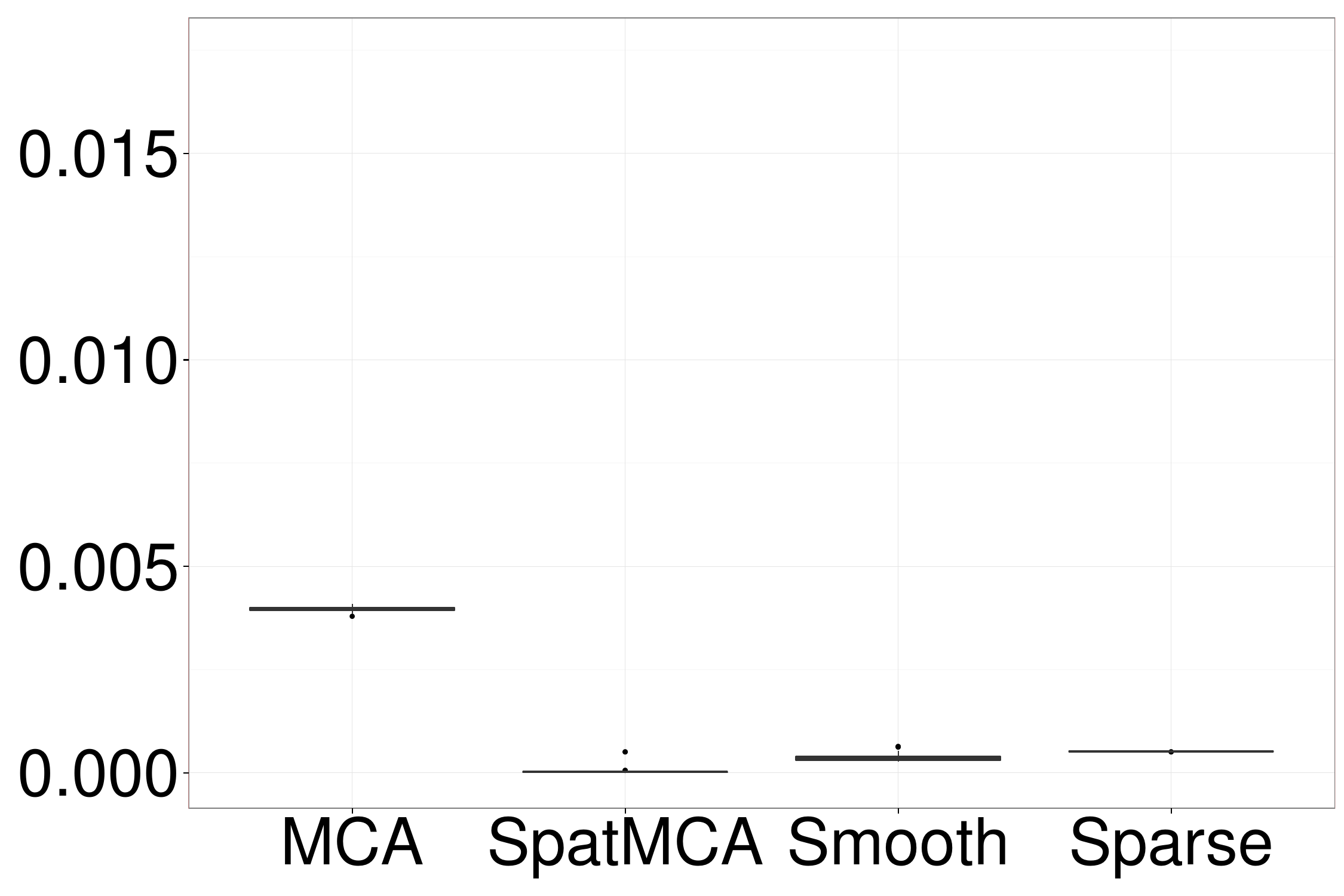}\hspace{4pt}&
		\includegraphics[scale=0.12]{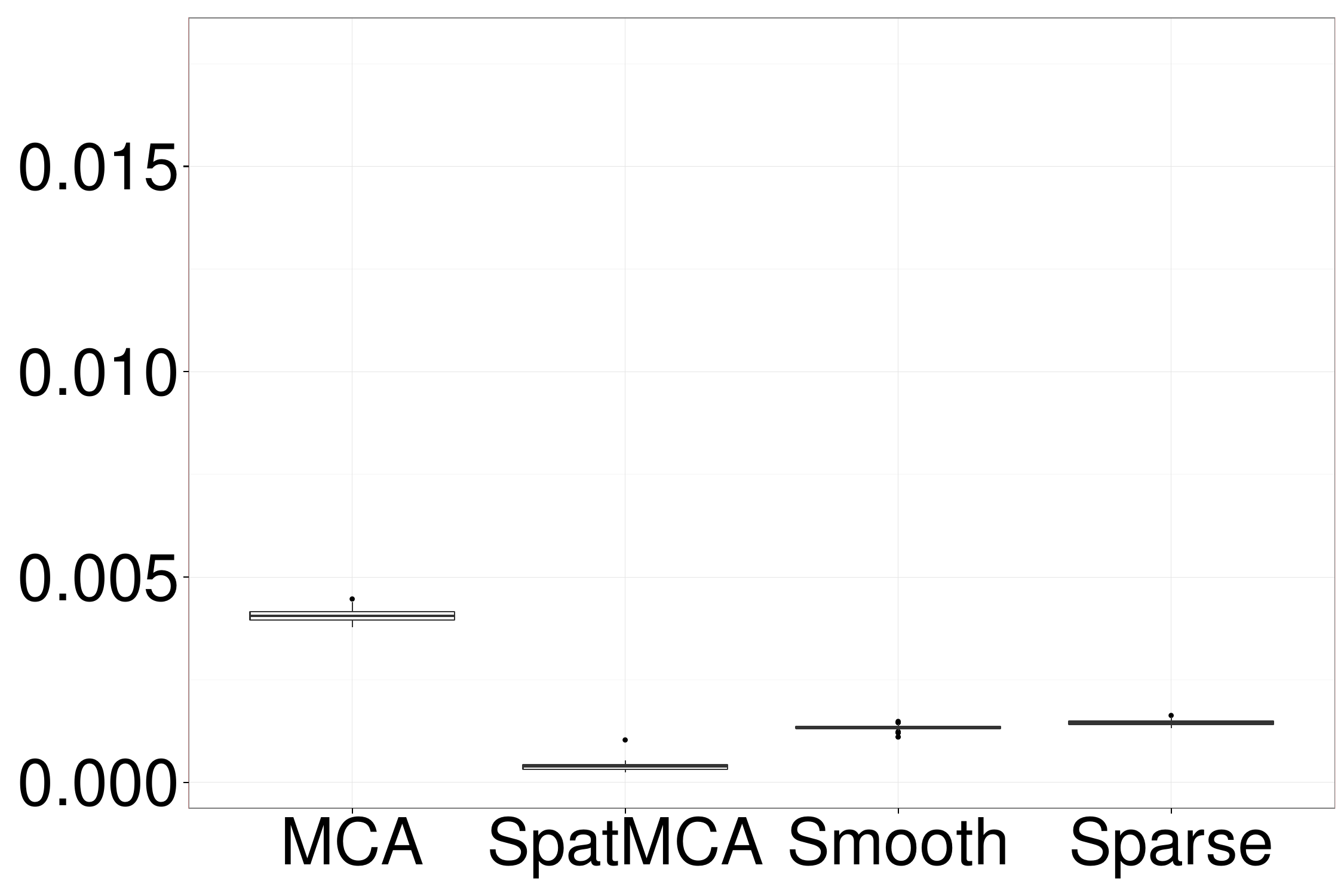}\hspace{4pt}
	\end{tabular}
	\caption{Boxplots of average squared prediction errors of (\ref{eq:loss1}) for various methods in the {two}-dimensional simulation experiment based on 50 simulation replicates.}
	\label{fig:box_d2_loss}
\end{figure}

\subsection{An Application to Sea Surface Temperature and Precipitation Datasets}
We applied the proposed SpatMCA and MCA to investigate how precipitations in eastern Africa are affected by SSTs in the Indian Ocean and compared the differences between the two methods. The SST data are monthly averages (in degree Celsius) provided by the Met Office Marine Data Bank (available at \texttt{http://www.metoffice.gov.uk/hadobs/hadisst/}). The precipitation data are monthly averages (in  mm) provided by the Earth System Research Laboratory, Physical Science Division of the National Oceanic and Atmospheric Administration (available at \texttt{http://www.esrl.noaa.gov/psd/}). Both datasets are on $1$ degree latitude by 1 degree longitude equiangular grid cells. As in \citet{omondi2013influence}, we considered a region of the Indian Ocean between latitudes $20^{\circ}$N and $30^{\circ}$S and between longitudes $20^{\circ}$E and $120^{\circ}$E for the SST dataset, and we chose a region of eastern Africa between $6^{\circ}$N and $12^{\circ}$S and between longitudes $20^{\circ}$E and $42^{\circ}$E for the precipitation dataset. We used the data observed from January 2011 to December 2015. Let $\bm{\eta}_{1i}$ and $\bm{\eta}_{2i}$ be the vectors of \eqref{eq:ch5measurement} corresponding to SST in the Indian Ocean and precipitation in eastern Africa. In this example, $p_1 = 3,591$, $p_2=255$, and $n=60$.

First, the SST data and the precipitation data were detrended by subtracting their individual average for a given cell and a given month. Then, the data were randomly split into two parts as the training data and  the validation data. We applied SpatMCA to the training data with $K$ selected by $\hat{K}$ of $(\ref{eq:ch5khat})$, where 21 values of $\tau_{1u}$ and $\tau_{1v}$ (including $0$ and the other 20 values equally spaced on the log scale from $10^{-1}$ to $10^6$) and 21 values of $\tau_{2u}$ and $\tau_{2v}$ (including $0$ and the other $20$ values equally spaced on the log scale from $10^{-3}$ to $0.5$) were selected by using 5-fold CV of \eqref{eq:cv_max1} and \eqref{eq:cv_max2}. 

The best CV values with respect to $K$ for both methods are shown in Figure~\ref{fig:cv_real}. Clearly, both methods selected $\hat{K} = 1$. Figure~\ref{fig:mca_real} shows the first dominant coupled patterns of  SST and precipitation obtained from SpatMCA and MCA. While both methods produce similar patterns, the SST pattern obtained by MCA is much noisier. Figure~\ref{fig:mca_ts} shows {two} time series of the first maximum covariance variables, $\{\hat{\bm{u}}'_1\bm{Y}_{11},\dots,\hat{\bm{u}}'_1\bm{Y}_{1n}\}$ and $\{\hat{\bm{v}}'_1\bm{Y}_{21},\dots,\hat{\bm{v}}'_1\bm{Y}_{2n}\}$, which are the projections of the training data $(\bm{Y}_{j1},\dots,\bm{Y}_{jn})$ for $j=1,2$, onto $\hat{\bm{u}}_1$ and $\hat{\bm{v}}_1$, respectively. As shown in the figure, the first {maximum covariance variables of} SST and precipitation  are highly correlated. Indeed, the Pearson correlation coefficient between {these} two series is 0.59 for SpatMCA and 0.63 for MCA, showing the {importance} of these patterns. 

We further used the validation data to compare the performance between SpatMCA and MCA in terms of the average squared error (ASE),
${\mathrm{ASE}} = \frac{1}{p_1p_2}\|\bm{S}^{v}_{12} - \hat{\bm{\Sigma}}_{12}\|^2_F$, where $\bm{S}^{v}_{12}$ is the sample cross-covariance matrix of the validation data, 
and $\hat{\bm{\Sigma}}_{12}$ is a generic estimate of $\bm{\Sigma}_{12}$. The resulting ASE for MCA is $2.59\times 10^{-3}$, which is larger than $2.25\times 10^{-3}$ for SpatMCA.
Figure~\ref{fig:mse} shows the ASEs with respect to $K$ for both SpatMCA and MCA, which further demonstrate the superiority of SpatMCA over MCA.
\begin{figure}
	\centering
	\includegraphics[height=5cm, width=6 cm]{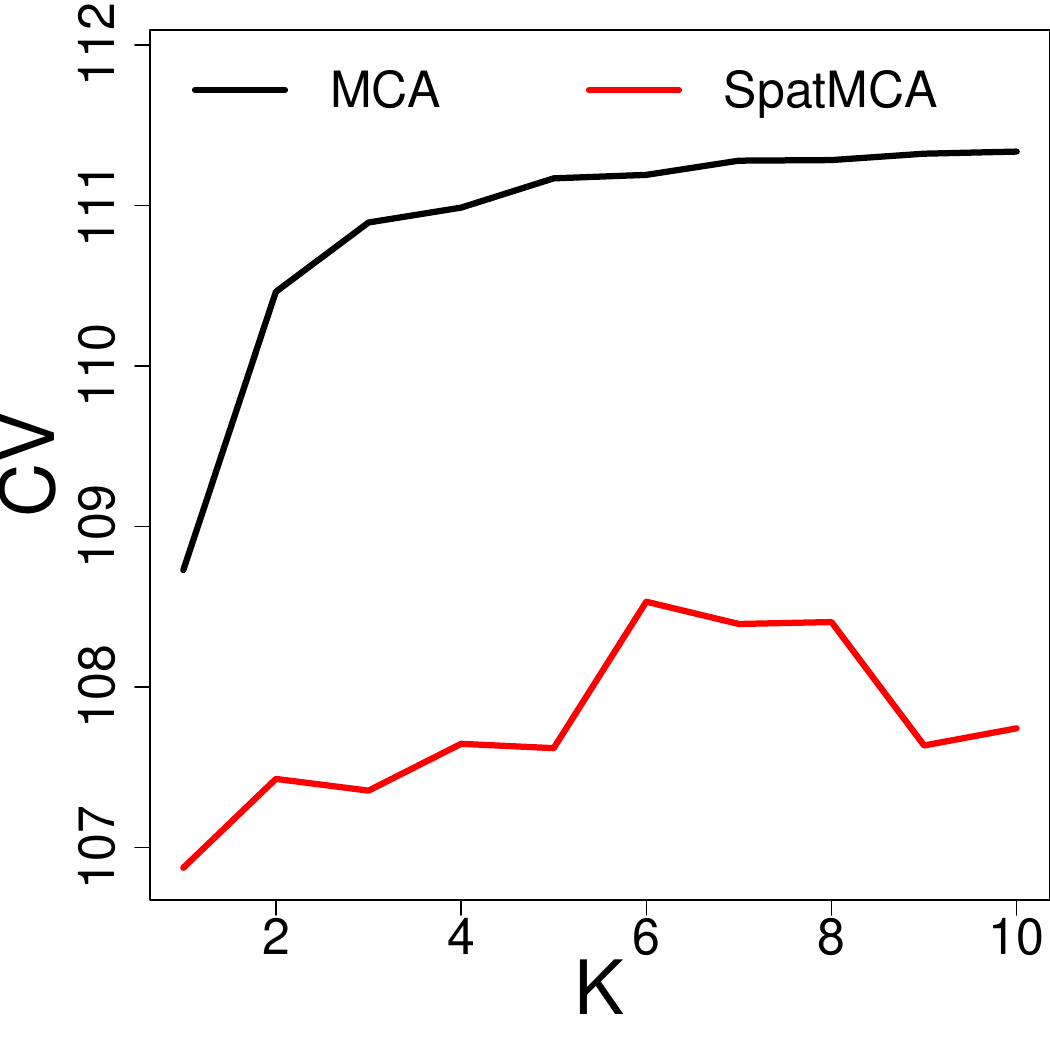}
	\caption{CV values with respect to $K$ for SpatMCA and MCA.}
	\label{fig:cv_real}
\end{figure}

\begin{figure}[tbp]\centering
$\hat{u}_1(\cdot)$ from MCA\hspace{0.28\textwidth}$\hat{u}_1(\cdot)$ from SpatMCA\\[0.3em]
\begin{tabular}{@{}c@{\hspace{0.02\textwidth}}c@{}}
\includegraphics[width=0.48\textwidth]{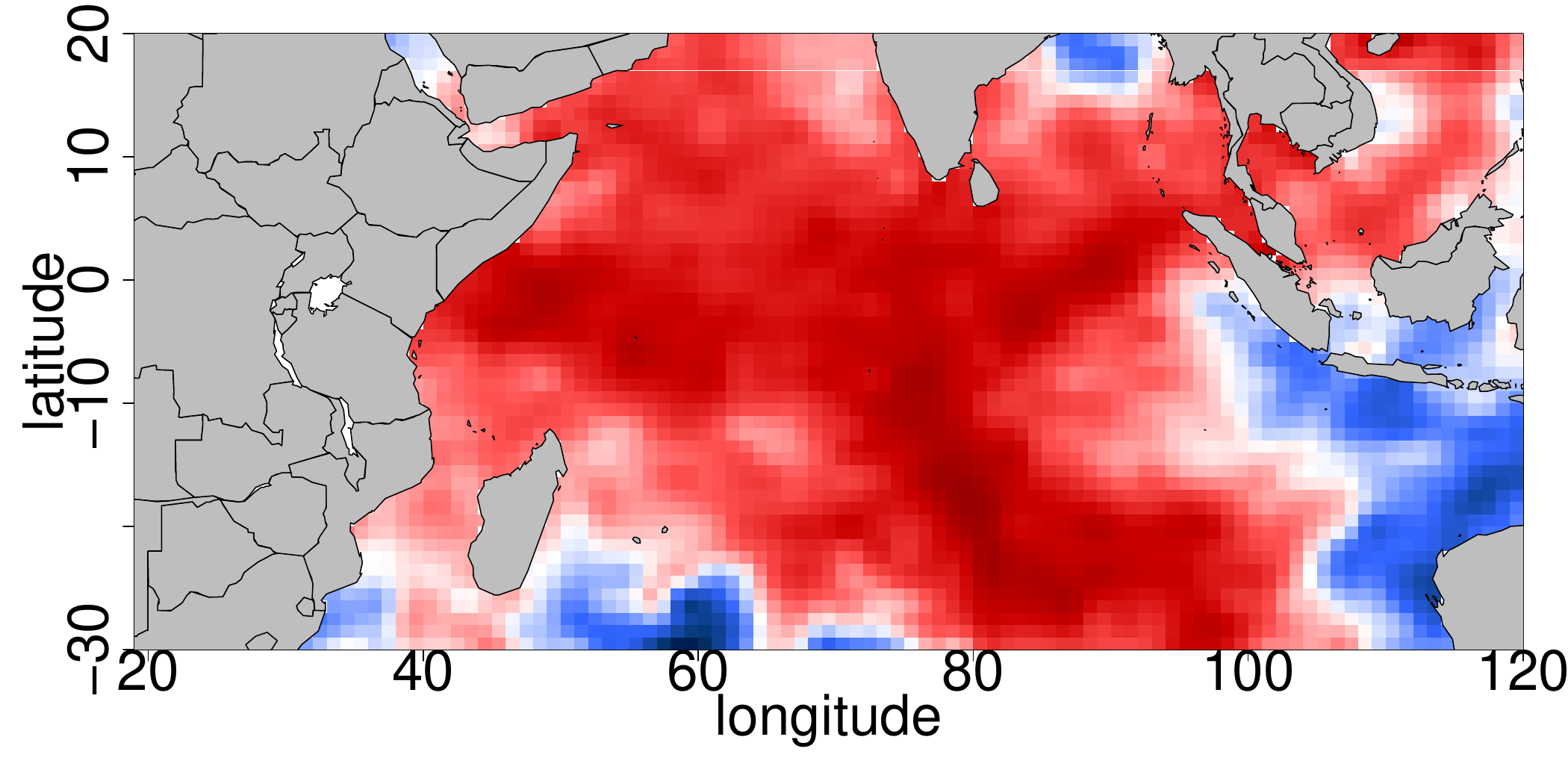} &
\includegraphics[width=0.48\textwidth]{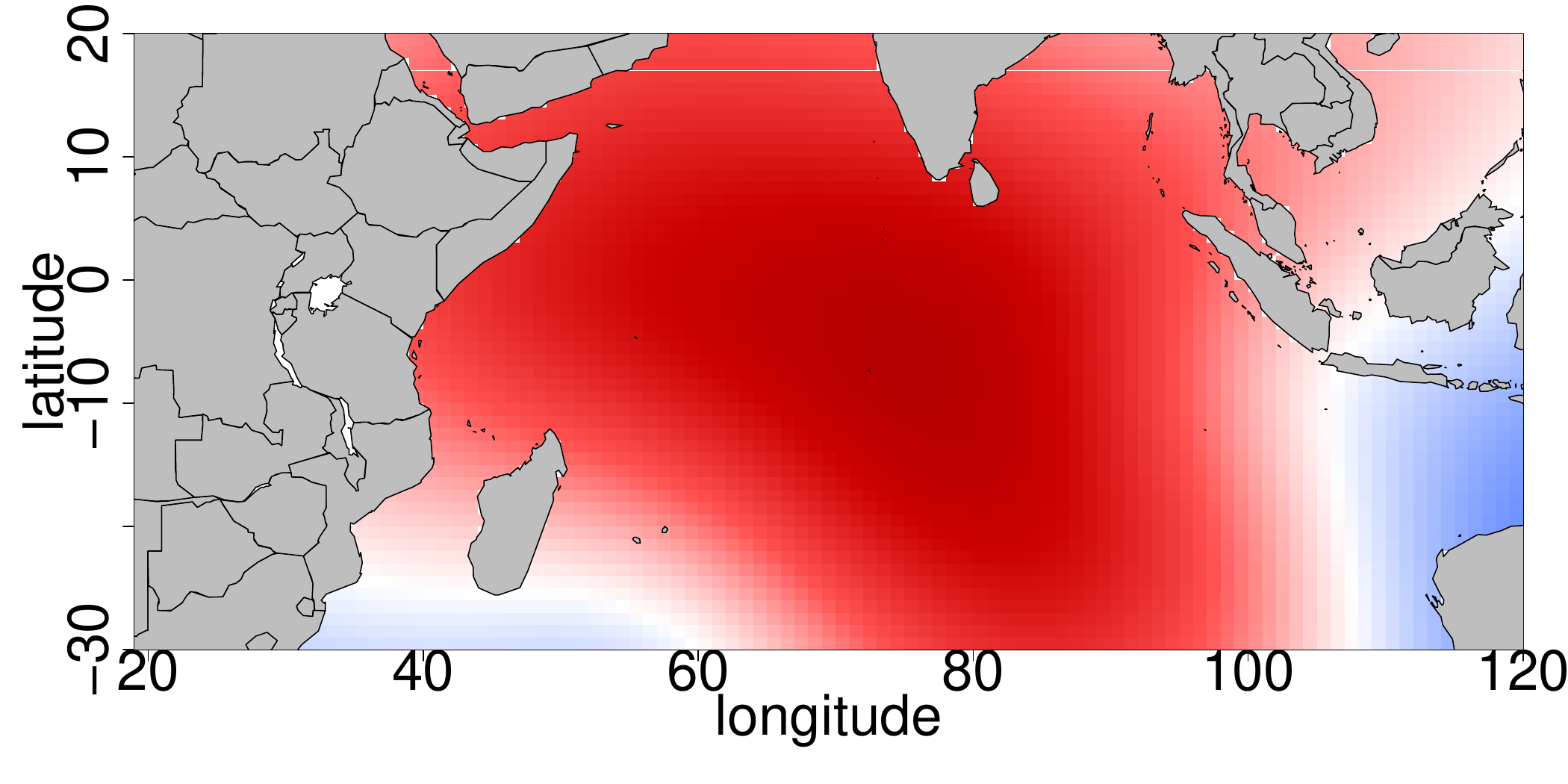} \\[0.2em]
\multicolumn{2}{c}{\includegraphics[width=0.72\textwidth]{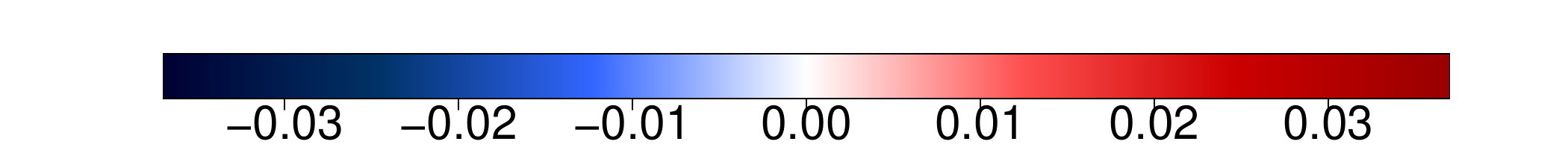}} \\
\end{tabular}\\[0.8em]
$\hat{v}_1(\cdot)$ from MCA\hspace{0.28\textwidth}$\hat{v}_1(\cdot)$ from SpatMCA\\[0.3em]
\begin{tabular}{@{}c@{\hspace{0.02\textwidth}}c@{}}
\includegraphics[width=0.48\textwidth]{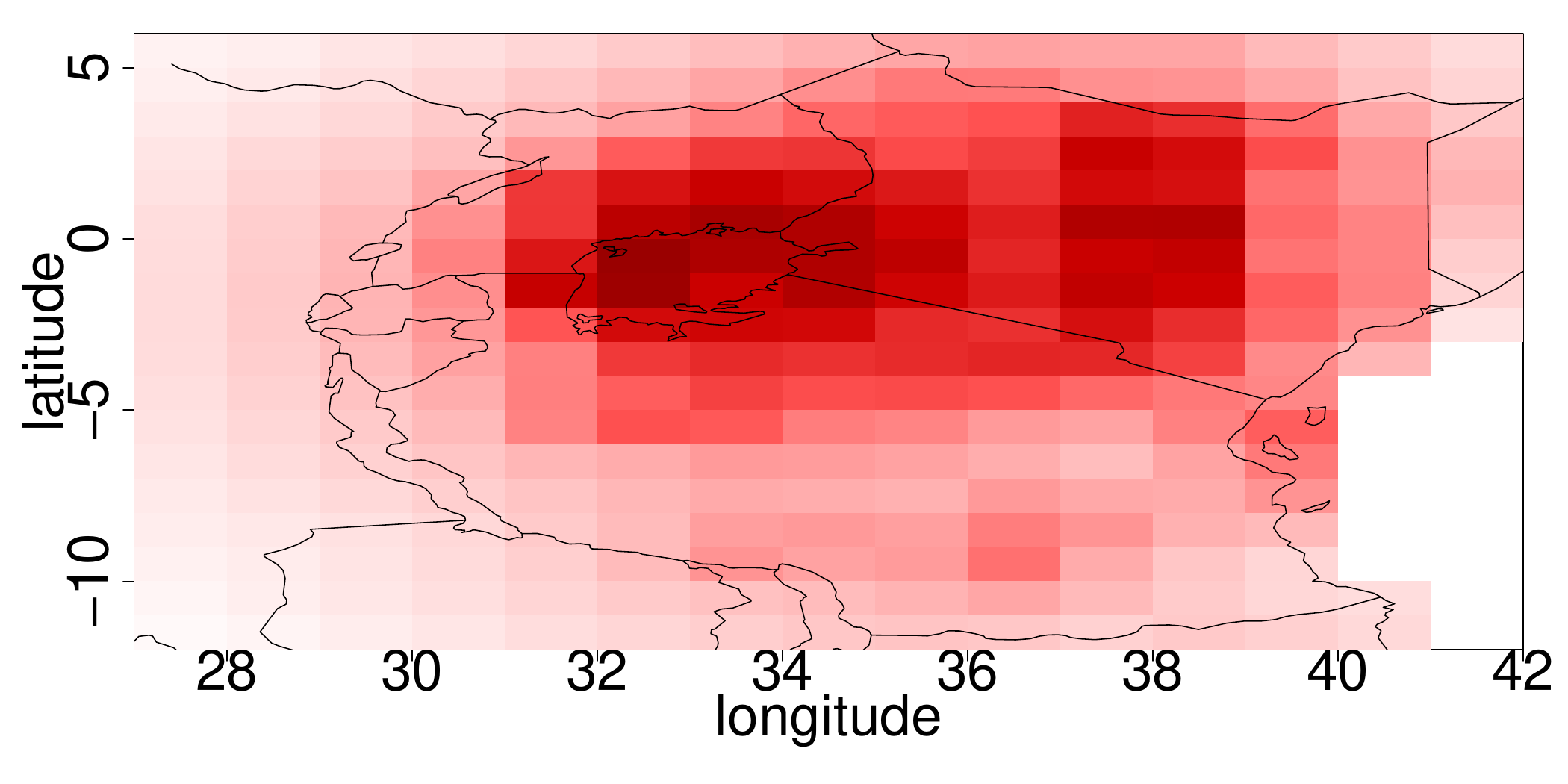} &
\includegraphics[width=0.48\textwidth]{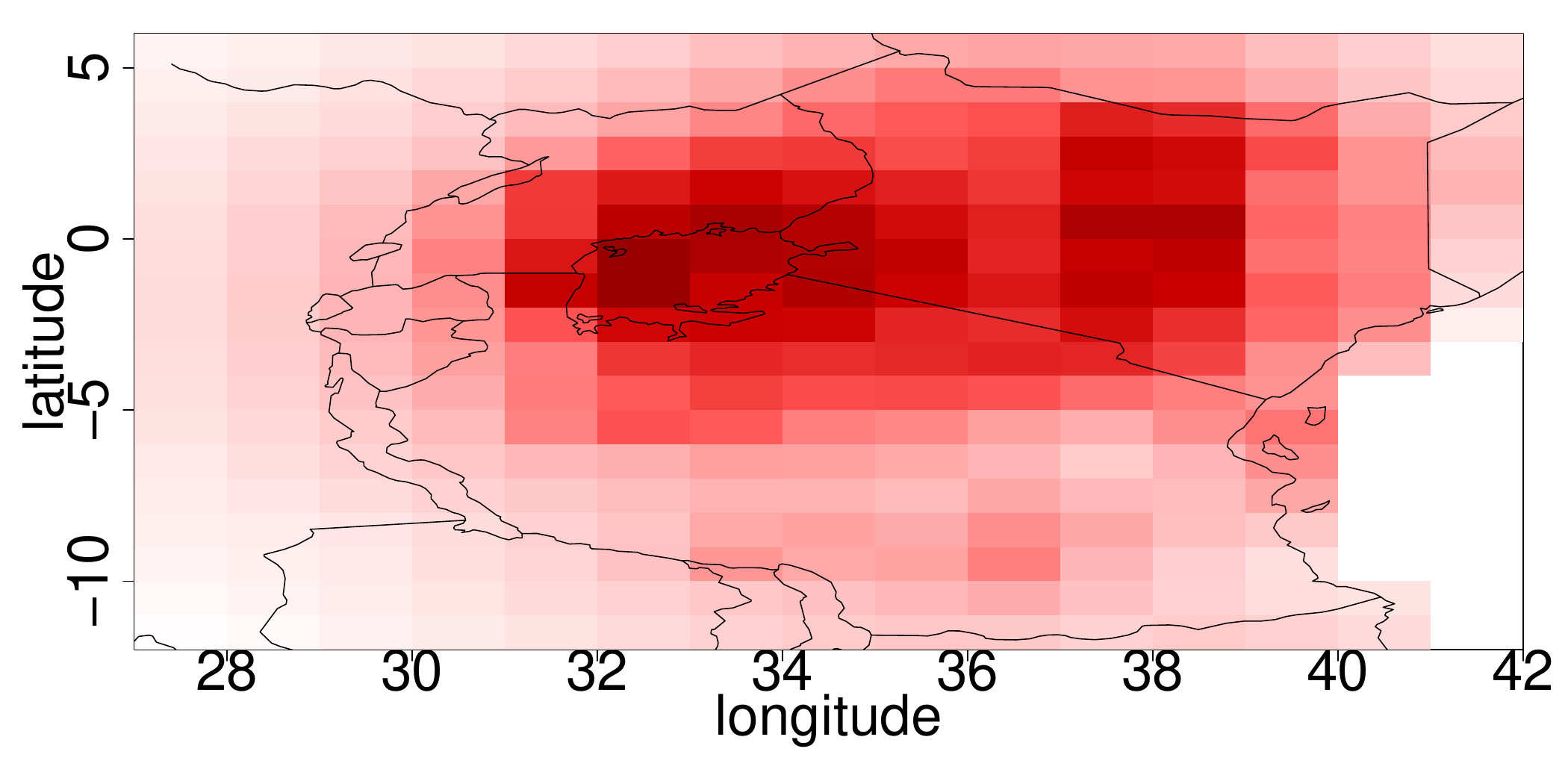} \\[0.2em]
\multicolumn{2}{c}{\includegraphics[width=0.72\textwidth]{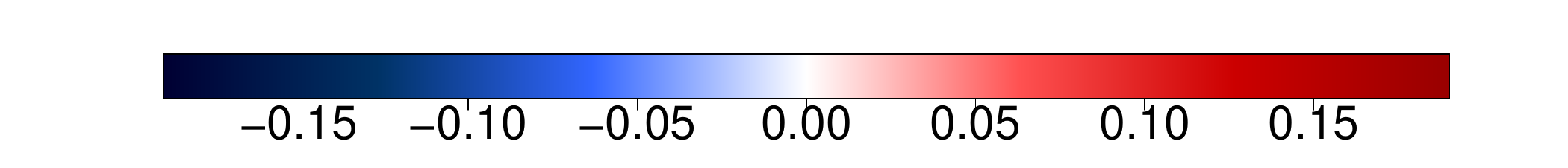}} \\
\end{tabular}
\caption{{Estimated first coupled} patterns of SST ($\hat{u}_1(\cdot)$) and precipitation ($\hat{v}_1(\cdot)$) from MCA and SpatMCA, where the gray regions correspond to the land with no SST data.} 
\label{fig:mca_real}
\end{figure}
\begin{figure}\centering
	
	\begin{tabular}{cc}
		MCA& SpatMCA\\
		\includegraphics[height=6cm, width=7 cm]{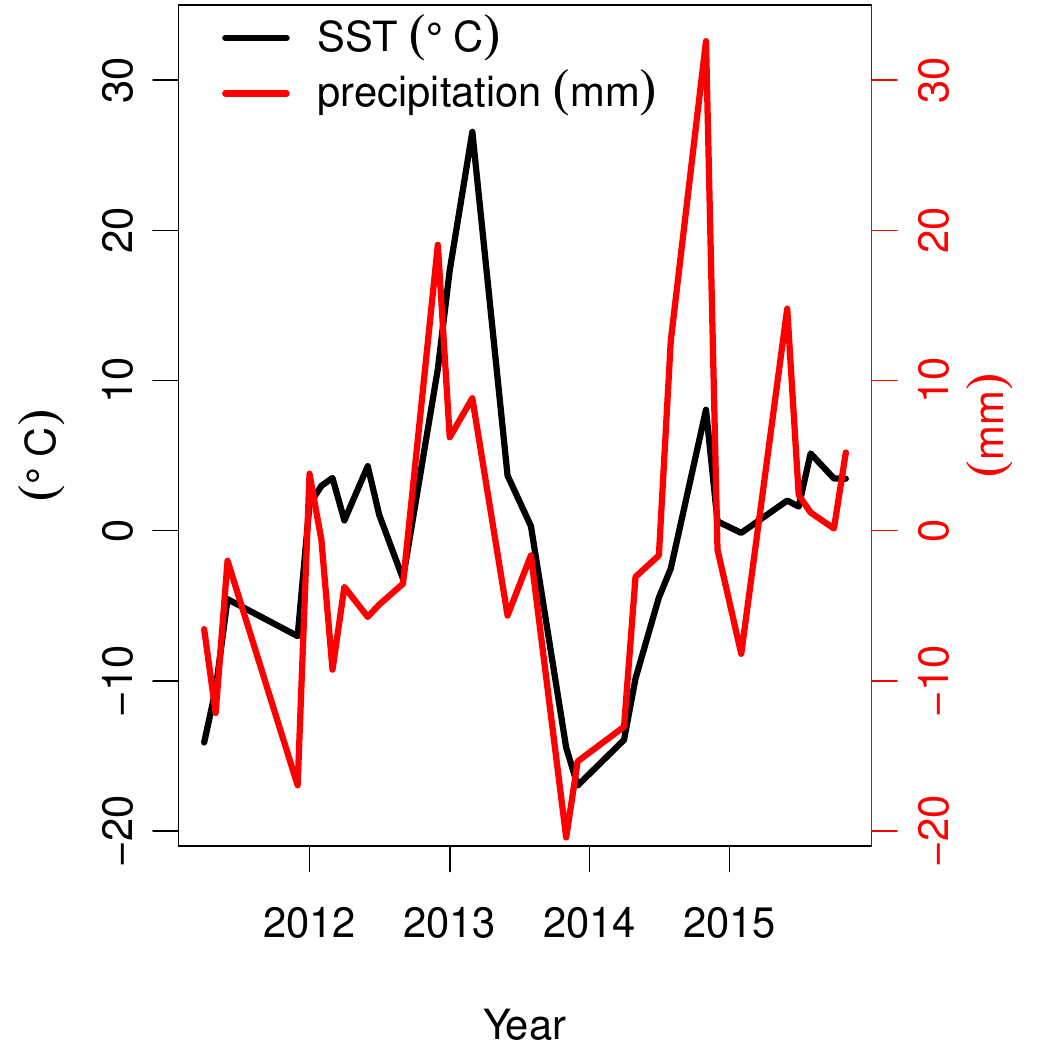}&
		\includegraphics[height=6cm, width=7 cm]{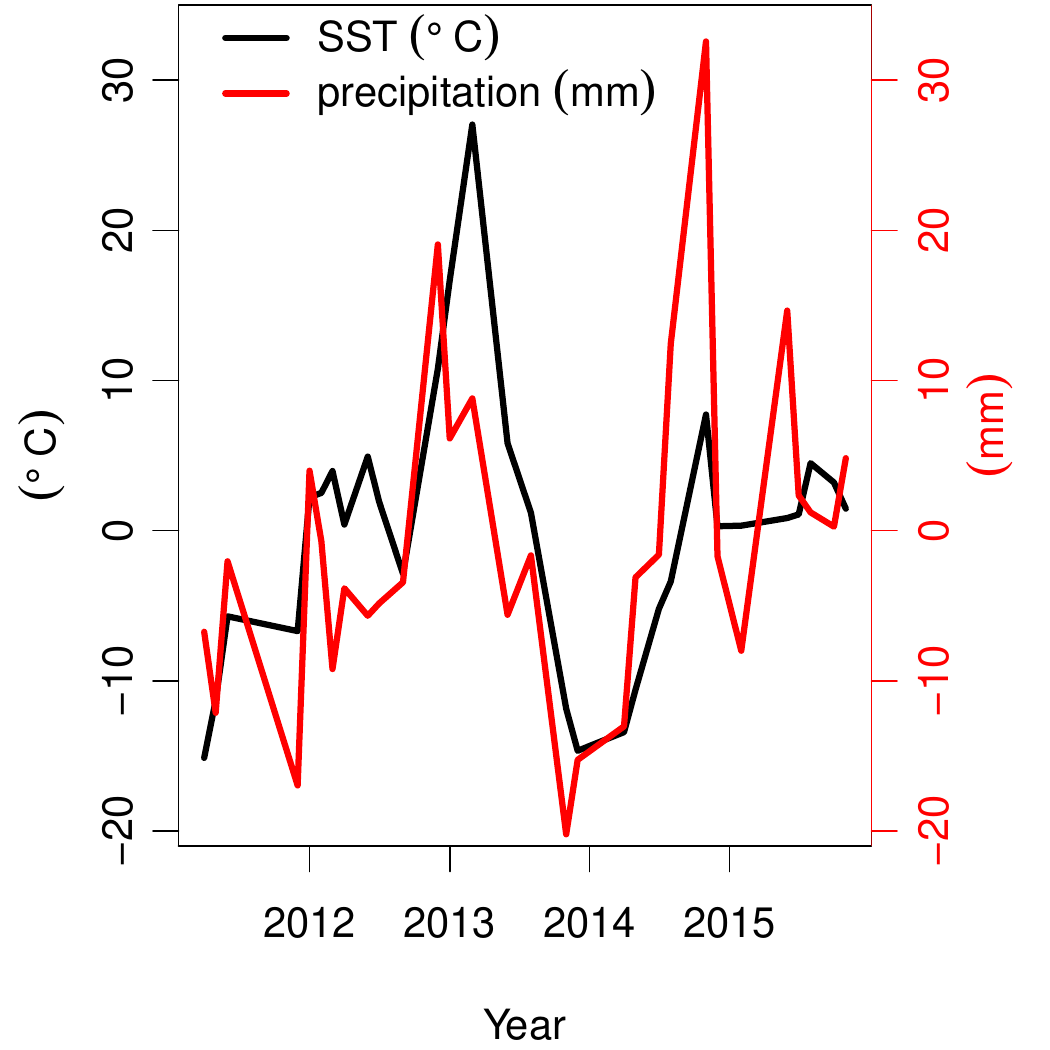}
	\end{tabular}
	\caption{{Time} series of the first  maximum covariance {variables of} SST and precipitation {obtained from MCA and SpatMCA}.}
	\label{fig:mca_ts}
\end{figure}

\begin{figure}
	\centering
	\includegraphics[height=5cm, width=6 cm]{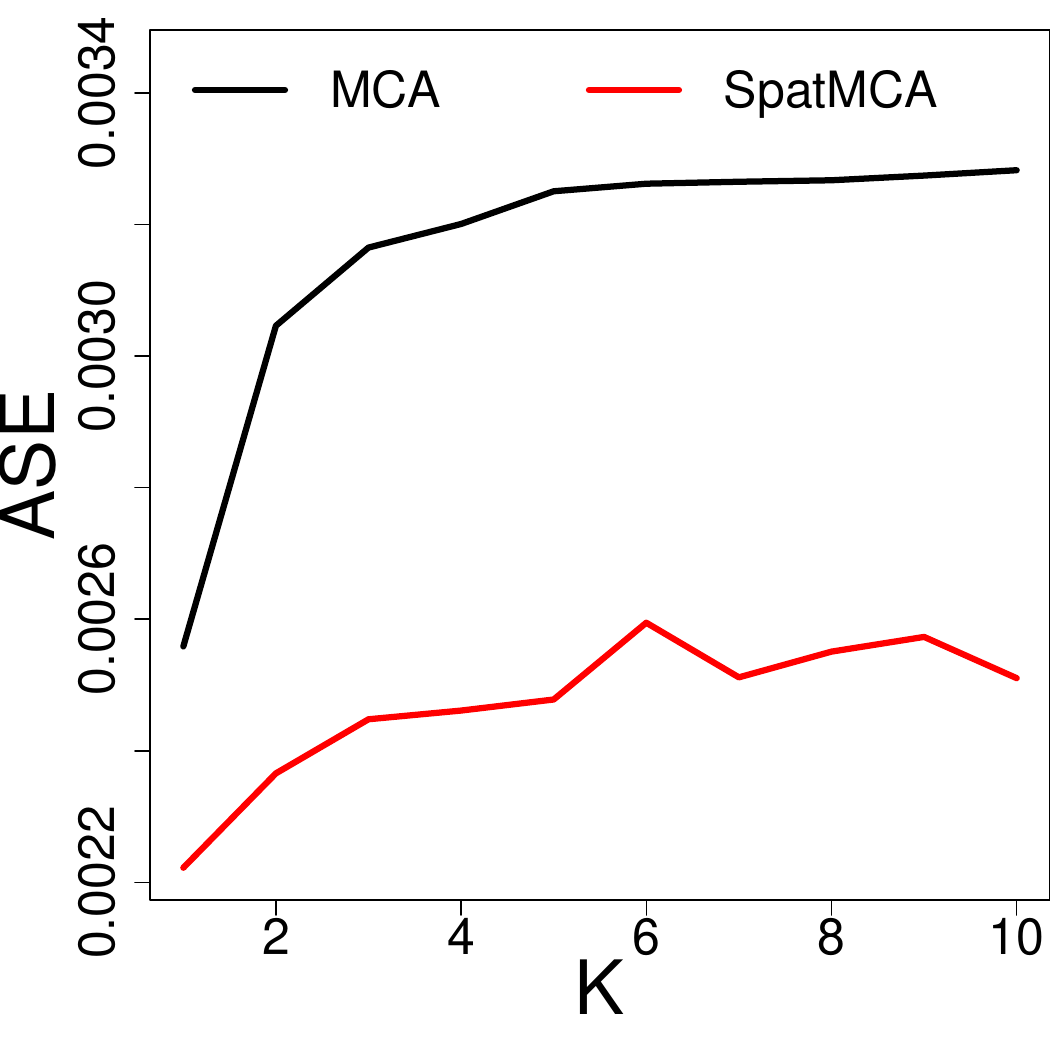}\\
	\caption{Average squared errors of cross-covariance matrix estimates with respect to $K$ for SpatMCA and MCA.}
	\label{fig:mse}
\end{figure}


\chapter{Summary and Further {Developments}}\label{ch:ch6}


\section{Summary}
This thesis was inspired by analyzing the variabilities of climate variables through spatial patterns. {We discovered} that {using} {a} {combination} of smoothness and sparseness penalties can lead to desired patterns that are not only smooth but also localized. The proposed SpatPCA and SpatMCA estimates can effectively enhance the interpretability {of dominant patterns}, even when the signal-to-noise ratio is low. In particular,  SpatPCA can be applied to nonstationary covariance function estimation and spatial prediction {via} fix-rank kriging. {Similarly,} SpatMCA {is} also applicable to nonstationary cross-covariance function estimation, and spatial prediction {via} aggregation cokriging \citep{furrer2011aggregation}. In addition, the {two proposed} ADMM algorithms are efficient and easy to implement. 
\section{Future {Developments}}
\subsection{Spatio-Temporal Models}
In {Section \ref{sec:ch4proposal}, we consider a purely spatial model with repeated measurements by assuming $\bm{\xi}_i\sim(\bm{0}, \bm{\Lambda})$ in \eqref{eq:ch4measurement} to be uncorrelated.} {In practice, it is of interest} to incorporate {both} spatial and temporal dependence structures simultaneously. For example, we may consider {the} following {first order autoregressive (AR) model on the latent variables $\{\bm{\xi}_t\}$:}
\begin{align*}
\bm{Y}_t &= \bm{\Phi}\bm{\xi}_t + \bm{\epsilon}_t,\\
\bm{\xi}_t &= \bm{\Delta}\bm{\xi}_{t-1} + \bm{\nu}_t,
\end{align*}
for $ t=1,\dots,n$, where $\bm{\Delta} = \mathrm{diag}(\delta_1,\dots,\delta_K)$ consists of the AR coefficients with $|\delta_k| <1$ for $k=1,\dots, K$, and $\bm{\nu}_t \sim (\bm{0}, \bm{\Lambda}-\bm{\Delta}\bm{\Lambda}\bm{\Delta}')$ is uncorrelated with $\bm{\xi}_{{t-1}}$ {and $\bm{\epsilon}_t$}.

{We may first estimate $\bm{\Phi}$ {by} the SpatPCA estimate $\hat{\bm{\Phi}}$ in \eqref{eq:ch4objective2}. {Applying} $\mathrm{var}(\bm{Y}_t) = \bm{\Phi}\bm{\Lambda}\bm{\Phi}'+\sigma^2\bm{I}$ and $\mathrm{cov}(\bm{Y}_{t-1},\bm{Y}_{t})= \bm{\Phi}\bm{\Lambda}\bm{\Delta}'\bm{\Phi}'$, we can {then} estimate $\bm{\Lambda}$, $\sigma^2$, and $\bm{\Delta}$ by minimizing}
 	\[
 	\|\bm{S}_0 - \hat{\bm{\Phi}}\bm{\Lambda}\hat{\bm{\Phi}}'-\sigma^2\bm{I}\|^2_F + {\|\bm{S}_1 - \hat{\bm{\Phi}}\bm{\Lambda}{\bm{\Delta}}'\hat{\bm{\Phi}}'\|^2_F},
 	\]
 	where $\bm{S}_0 = \frac{1}{n}\sum_{t=1}^n \bm{Y}'_t\bm{Y}_t$ and {$\bm{S}_1 = \frac{1}{n-1}\sum_{t=2}^{n} \bm{Y}'_{t-1}\bm{Y}_{t}$}. The best linear prediction of {\{$\bm{\xi}_t$\}} may be {computed} using {the} Kalman filter \citep[e.g.,][]{fixedrank}.



\subsection{Testing the Significance of Spatial Patterns}
{It is interesting to know how many patterns $\hat{\varphi}_k(\cdot)$ estimated from \eqref{eq:ch4basis} are significant. This can be done by sequentially testing $\lambda^*_k = 0$ for $k=1,\dots, K$, where $\lambda^*_1\geq\lambda^*_2\geq\cdots\geq\lambda^*_K$ are the first $K$ eigenvalues of  $C_\eta(\cdot,\cdot)$ in \eqref{eq:ch4eta}. A bootstrap test developed by \citet{benko2009common} for testing the equality of eigenvalues can be adapted here.}

\subsection{Testing the Significance of Coupled Patterns}
{Recall from} Section \ref{sec:ch5proposal} {that} the cross-covariance function between $\eta_1(\cdot)$ and  $\eta_2(\cdot)$ is written as $C_{12}(\bm{s}_1,\bm{s}_2) = \sum_{k=1}^\infty d_ku_k(\bm{s}_1)v_k(\bm{s}_2)$ with $d_1\geq d_2\geq\cdots$. It is of interest to {know how many pairs of patterns $(\hat{u}_k(\cdot), \hat{v}_k(\cdot))$  estimated from \eqref{eq:basis_u} and \eqref{eq:basis_v} are significant.} Some methods for testing the significance of canonical correlation coefficients{, using permutation tests \citep{witten2009penalized,witten2009application} and an asymptotic test \citep{yang2015independence}, can be adapted here.}

{For example,  the permutation test of \citet{witten2009penalized} {can be applied for} testing $H_0: d_1=0$ based on the test statistic $\hat{d}_1$ of \eqref{eq:estimate_d} with the corresponding tuning parameters selected by CV of  \eqref{eq:cv_max1} and \eqref{eq:cv_max2}. Alternatively, we can {apply} another approach of {the} permutation test proposed by \citet{witten2009application}, which is more computationally efficient and does not require CV to select the tuning parameters.}

\subsection{Asymptotic Properties}
We {plan} to develop {some} asymptotic properties of the SpatPCA {estimates} $\{\hat{{\varphi}}_k(\cdot)\}$ of \eqref{eq:ch4objective} and SpatMCA {estimates} $\{(\hat{u}_k(\cdot),\hat{v}_k(\cdot))\}$ {of} \eqref{eq:obj1} as follows:
\begin{enumerate} 
 \item {Consistency}: (i) $\hat{{\varphi}}_k(\cdot) {\to} \varphi_k(\cdot)$ {in probability} as $n,p\to \infty$; (ii) {$(\hat{u}_k(\cdot),\hat{v}_k(\cdot)) {\to} (u_k(\cdot),v_k(\cdot))$ in probability} as $n,p_1, {p_2} \to \infty,$ for $k=1,\dots, K$.
 \item Asymptotic distributions {and convergence rates} of $\hat{{\varphi}}_k(\cdot)$, $\hat{u}_k(\cdot)$ and $\hat{v}_k(\cdot)$  as $p,n \to \infty${, for $k=1,\dots, K$.}
   \end{enumerate}
   {Some asymptotic results developed by \citet{hall2006properties}  for functional PCA may be adapted here.}
  \subsection{Dimension Reduction}
  {Our SpatPCA approach provides a {low-dimensional} representation {of the underlying process}. It {can be} applied {not only} for image reconstruction based on noisy and missing data {via} spatial prediction {of} \eqref{eq:ch4loss_sim}, but also for image classification using the first few (estimated) PCs.}


\appendix


\chapter{Proof of Proposition 1}\label{app:appendix_A}

\begin{proof}
First, we prove \eqref{eq:ch4Lambda.hat}. From Corollary 1 of \citet{regularized_covariance}, the minimizer of $h(\bm{\Lambda},\sigma^2)$ given $\sigma^2$ is
\begin{equation}
\hat{\bm{\Lambda}}(\sigma^2)=\hat{\bm{V}} \mathrm{diag}\big((\hat{d}_1-\sigma^2-\gamma)_+,\dots,(\hat{d}_{K}-\sigma^2-\gamma)_+\big)\hat{\bm{V}}'.
\label{eq:Lambda.hat2}
\end{equation}
Hence \eqref{eq:ch4Lambda.hat} is obtained.

Next, we prove \eqref{eq:ch4sigma.hat}. Rewrite the objective function of \eqref{eq:ch4covariance.estimate} as:
\begin{align}
h(\bm{\Lambda}, \sigma^2)
=&~ \frac{1}{2}\|\hat{\bm{\Phi}}\hat{\bm{\Phi}}'\bm{S}\hat{\bm{\Phi}}\hat{\bm{\Phi}}' - \hat{\bm{\Phi}}\bm{\Lambda}\hat{\bm{\Phi}}'
-\sigma^2\bm{I}_p\|^2_F+\frac{1}{2}\|\bm{S} -\hat{\bm{\Phi}}\hat{\bm{\Phi}}'\bm{S}\hat{\bm{\Phi}}\hat{\bm{\Phi}}'\|^2_F \nonumber\\
&~ +\sigma^2\mathrm{tr}( \hat{\bm{\Phi}}\hat{\bm{\Phi}}'\bm{S}\hat{\bm{\Phi}}\hat{\bm{\Phi}}'-\bm{S})+\gamma\|\hat{\bm{\Phi}}
\bm{\Lambda}\hat{\bm{\Phi}}'\|_*.
\label{eq:appendix_obj}
\end{align}

\noindent From \eqref{eq:Lambda.hat2} and \eqref{eq:appendix_obj}, we have
\begin{align*}
h(\hat{\bm{\Lambda}}(\sigma^2), \sigma^2)
=&~ \frac{1}{2}\|\hat{\bm{\Phi}}\hat{\bm{\Phi}}'\bm{S}\hat{\bm{\Phi}}\hat{\bm{\Phi}}' - \hat{\bm{\Phi}}\hat{\bm{\Lambda}}(\sigma^2)\hat{\bm{\Phi}}'
-\sigma^2\bm{I}_p\|^2_F+\gamma\|\hat{\bm{\Phi}}\hat{\bm{\Lambda}}(\sigma^2)\hat{\bm{\Phi}}'\|_*\\
&~+\frac{1}{2}\|\bm{S} -\hat{\bm{\Phi}}\hat{\bm{\Phi}}'\bm{S}\hat{\bm{\Phi}}\hat{\bm{\Phi}}'\|^2_F+
\sigma^2\mathrm{tr}( \hat{\bm{\Phi}}\hat{\bm{\Phi}}'\bm{S}\hat{\bm{\Phi}}\hat{\bm{\Phi}}'-\bm{S})\\
=&~ \frac{1}{2}\sum_{k=1}^{K}\big\{\hat{d}_k^2-(\hat{d}_k-\sigma^2-\gamma)_+^2\big\} +\frac{p}{2}\sigma^4-\sigma^2\mathrm{tr}(\bm{S})
+\frac{1}{2}\|\bm{S} -\hat{\bm{\Phi}}\hat{\bm{\Phi}}'\bm{S}\hat{\bm{\Phi}}\hat{\bm{\Phi}}'\|^2_F.
\end{align*}
\noindent Minimizing $h(\hat{\bm{\Lambda}}(\sigma^2), \sigma^2)$, we obtain
\begin{equation}
\hat{\sigma}^2
=\mathop{\arg\min}_{\sigma^2 \geq 0}\bigg\{p\sigma^4-2\sigma^2\mathrm{tr}(\bm{S})-\sum_{k=1}^{K}
(\hat{d}_k-\sigma^2-\gamma)_+^2\bigg\}.
\label{eq:sigma.hat2}
\end{equation}

\noindent Clearly, if $\hat{d}_1\leq \gamma$, then $\hat{\sigma}^2=\displaystyle\frac{1}{p}\mathrm{tr}(\bm{S})$.
We remain to consider $\hat{d}_1>\gamma$. Let
\[
\hat{L}^*=\max\big\{L:\hat{d}_L-\gamma>\hat{\sigma}^2,\,L=1,\dots,K\big\}.
\]
From \eqref{eq:sigma.hat2}, $\hat{\sigma}^2=\displaystyle\frac{1}{p-\hat{L}^*}
\bigg(\mathrm{tr}(\bm{S})-\sum_{k=1}^{\hat{L}^*} (\hat{d}_k-\gamma)\bigg)$. It suffices to show that $\hat{L}^*=\hat{L}$.
Since $\hat{d}_{\hat{L}^*}-\gamma >\displaystyle\frac{1}{p-\hat{L}^*} \bigg(\mathrm{tr}(\bm{S})-\sum_{k=1}^{\hat{L}^*} (\hat{d}_k-\gamma)\bigg)$,
by the definition of $\hat{L}$, we have $\hat{L}\geq \hat{L}^*$, implying $\hat{d}_{\hat{L}}\geq\hat{d}_{\hat{L}^*}$.
Suppose that $\hat{L}>\hat{L}^*$. It immediately follows from the definition of
$\hat{L}^*$ that $\hat{d}_{\hat{L}}-\gamma\leq\hat{\sigma}^2<\hat{d}_{\hat{L}^*}-\gamma$, which contradicts to $\hat{d}_{\hat{L}}\geq
\hat{d}_{\hat{L}^*}$. Therefore, $\hat{L}=\hat{L}^*$. This completes the proof.
\end{proof}


\bibliographystyle{plainnat}
\bibliography{references}

@article{Shen,
author = {Shen, Haipeng and Huang, Jianhua Z.},
  title = {Sparse Principal Component Analysis via Regularized Low Rank Matrix Approximation},
  journal = {Journal of Multivariate Analysis},
  year = {2008},
  volume = {99},
  pages = {1015--1034}
}

@article{admm,
author = {Boyd, Stephen and Parikh, Neal and Chu, Eric and Peleato, Borja and Eckstein, Jonathan},
  title = {Distributed Optimization and Statistical Learning via the Alternating Direction Method of Multipliers},
  journal = {Foundations and Trends in Machine Learning},
  year = {2011},
  volume = {3},
  pages = {1--122}
}

@article{admm_2,
author = {Gabay, Daniel and Mercier, Bertrand},
  title = {A Dual Algorithm for the Solution of Nonlinear Variational Problems via Finite Element Approximation},
  journal = {Computers and Mathematics with Applications},
  year = {1976},
  volume = {2},
  pages = {17--40}
}

@article{akaike1974new,
author = {Akaike, Hirotugu},
  title = {A New Look at the Statistical Model Identification},
  journal = {IEEE Transactions on Automatic Control},
  year = {1974},
  volume = {19},
  pages = {716--723}
}

@article{azaiez2015karhunen,
Author = {Aza{\"\i}ez, M and Belgacem, F Ben},
	Journal = {Computer Methods in Applied Mechanics and Engineering},
	Pages = {57--72},
	Title = {Karhunen--Lo{\`e}ve's truncation error for bivariate functions},
	Volume = {290},
	Year = {2015}
}

@article{bayes,
author = {Kang, Emily L. and Cressie, Noel},
  title = {Bayesian Inference for the Spatial Random Effects Model},
  journal = {Journal of the American Statistical Association},
  year = {2011},
  volume = {106},
  pages = {972--983}
}

@article{benko2009common,
author = {Benko, Michal and H{\"a}rdle, Wolfgang and Kneip, Alois},
  title = {Common Functional Principal Components},
  journal = {The Annals of Statistics},
  year = {2009},
  volume = {37},
  pages = {1--34}
}

@book{bfda,
author = {Ramsay, J. O. and Silverman, B. W.},
  title = {Functional Data Analysis},
  publisher = {Springer},
  year = {2005},
  edition = {2}
}

@article{bretherton1992intercomparison,
Author = {Bretherton, Christopher S and Smith, Catherine and Wallace, John M},
	Journal = {Journal of climate},
	Number = {6},
	Pages = {541--560},
	Title = {An intercomparison of methods for finding coupled patterns in climate data},
	Volume = {5},
	Year = {1992}
}

@article{chu,
Author = {Chu, Jung-Eun and Ha, Kyung-Ja and Lee, June-Yi and Wang, Bin and Kim, Baek-Min and Chung, Christine E.},
	Journal = {Climate Dynamics},
	Number = {1-2},
	Pages = {535--551},
	Title = {Future Change of the {Indian Ocean} Basin-Wide and Dipole Modes in the {CMIP5}},
	Volume = {43},
	Year = {2014}
}

@article{coord,
author = {Friedman, Jerome and Hastie, Trevor and Tibshirani, Robert},
  title = {Regularization Paths for Generalized Linear Models via Coordinate Descent},
  journal = {Journal of Statistical Software},
  year = {2010},
  volume = {33},
  pages = {1--22}
}

@article{craven1978smoothing,
author = {Craven, Peter and Wahba, Grace},
  title = {Smoothing Noisy Data with Spline Functions},
  journal = {Numerische Mathematik},
  year = {1978},
  volume = {31},
  pages = {377--403}
}

@article{efron2004least,
author = {Efron, Bradley and Hastie, Trevor and Johnstone, Iain and Tibshirani, Robert},
  title = {Least Angle Regression},
  journal = {The Annals of Statistics},
  year = {2004},
  volume = {32},
  pages = {407--499}
}

@article{eofreview,
author = {Hannachi, A. and Jolliffe, I. T. and Stephenson, D. B.},
  title = {Empirical Orthogonal Functions and Related Techniques in Atmospheric Science: A Review},
  journal = {International Journal of Climatology},
  year = {2007},
  volume = {27},
  pages = {1119--1152}
}

@article{fan2001variable,
Author = {Fan, Jianqing and Li, Runze},
	Journal = {Journal of the American statistical Association},
	Number = {456},
	Pages = {1348--1360},
	Publisher = {Taylor \&amp; Francis},
	Title = {Variable selection via nonconcave penalized likelihood and its oracle properties},
	Volume = {96},
	Year = {2001}
}

@article{fda,
author = {Huang, Jianhua Z. and Shen, Haipeng and Buja, Andreas},
  title = {Functional Principal Components Analysis via Penalized Rank One Approximation},
  journal = {Electronic Journal of Statistics},
  year = {2008},
  volume = {2},
  pages = {678--695}
}

@article{fixedrank,
author = {Cressie, Noel and Johannesson, Gardar},
  title = {Fixed Rank Kriging for Very Large Spatial Data Sets},
  journal = {Journal of the Royal Statistical Society: Series B},
  year = {2008},
  volume = {70},
  pages = {209--226}
}

@article{fpca_long,
author = {Yao, Fang and M{\"u}ller, Hans-Georg and Wang, Jane-Ling},
  title = {Functional Data Analysis for Sparse Longitudinal Data},
  journal = {Journal of the American Statistical Association},
  year = {2005},
  volume = {100},
  pages = {577--590}
}

@article{furrer2011aggregation,
author = {Furrer, Reinhard and Genton, Marc G.},
  title = {Aggregation-Cokriging for Highly Multivariate Spatial Data},
  journal = {Biometrika},
  year = {2011},
  volume = {98},
  number = {3},
  pages = {615--631}
}

@article{fused,
author = {Guo, Jian and James, Gareth and Levina, Elizaveta and Michailidis, George and Zhu, Ji},
  title = {Principal Component Analysis with Sparse Fused Loadings},
  journal = {Journal of Computational and Graphical Statistics},
  year = {2010},
  volume = {19},
  pages = {930--946}
}

@article{hall2006properties,
author = {Hall, Peter and M{\"u}ller, Hans-Georg and Wang, Jane-Ling},
  title = {Properties of Principal Component Methods for Functional and Longitudinal Data Analysis},
  journal = {The Annals of Statistics},
  year = {2006},
  volume = {34},
  pages = {1493--1517}
}

@article{hong,
author = {Hong, Zhaoping and Lian, Heng},
  title = {Sparse-Smooth Regularized Singular Value Decomposition},
  journal = {Journal of Multivariate Analysis},
  year = {2013},
  volume = {117},
  pages = {163--174}
}

@article{hotelling1933analysis,
Author = {Hotelling, Harold},
	Journal = {Journal of educational psychology},
	Number = {6},
	Pages = {417},
	Publisher = {Warwick \&amp; York},
	Title = {Analysis of a complex of statistical variables into principal components.},
	Volume = {24},
	Year = {1933}
}

@article{jolliffe1987rotation,
author = {Jolliffe, I. T.},
  title = {Rotation of Principal Components: Some Comments},
  journal = {Journal of Climatology},
  year = {1987},
  volume = {7},
  pages = {507--510}
}

@book{jolliffe2002principal,
author = {Jolliffe, I. T.},
  title = {Principal Component Analysis},
  publisher = {Springer},
  year = {2002},
  edition = {2}
}

@article{karhunen,
Author = {Kari Karhunen},
	Journal = {Annales Academi{\ae} Scientiarum Fennic{\ae} Series A},
	Pages = {1-79},
	Title = {{\"U}ber lineare methoden in der Wahrscheinlichkeitsrechnung},
	Volume = {37},
	Year = {1947}
}

@article{lasso,
author = {Tibshirani, Robert},
  title = {Regression Shrinkage and Selection via the Lasso},
  journal = {Journal of the Royal Statistical Society: Series B},
  year = {1996},
  volume = {58},
  pages = {267--288}
}

@article{jolliffe2002eofs,
author = {Jolliffe, I. T. and Uddin, M. and Vines, S. K.},
  title = {Simplified {EOFs}--Three Alternatives to Rotation},
  journal = {Climate Research},
  year = {2002},
  volume = {20},
  pages = {271--279}
}

@article{lee2011sparse,
Author = {Lee, Woojoo and Lee, Donghwan and Lee, Youngjo and Pawitan, Yudi},
	Journal = {Statistical Applications in Genetics and Molecular Biology},
	Number = {1},
	Title = {Sparse canonical covariance analysis for high-throughput data},
	Volume = {10},
	Year = {2011}
}

@book{loeve,
Author = {Michel Lo{\`e}ve},
	Journal = {1978},
	Publisher = {Springer-Verlag, New York},
	Title = {Probability theory},
	Year = {1978}
}

@book{vonstorch1999analysis,
author = {Von Storch, Hans and Navarra, Antonio},
  title = {Analysis of Climate Variability},
  publisher = {Springer-Verlag},
  address = {Berlin Heidelberg},
  year = {1999}
}

@book{siedler2013ocean,
author = {Siedler, Gerold and Griffies, Stephen M. and Gould, John and Church, John A.},
  title = {Ocean Circulation and Climate: A 21st Century Perspective},
  publisher = {Academic Press},
  year = {2013},
  volume = {103}
}

@article{morioka2012subtropical,
Author = {Morioka, Yushi and Tozuka, Tomoki and Masson, Sebastien and Terray, Pascal and Luo, Jing-Jia and Yamagata, Toshio},
	Journal = {Journal of Climate},
	Number = {12},
	Pages = {4029--4047},
	Title = {Subtropical dipole modes simulated in a coupled general circulation model},
	Volume = {25},
	Year = {2012}
}

@book{nonparametric,
author = {Green, P. J. and Silverman, B. W.},
  title = {Nonparametric Regression and Generalized Linear Models: A Roughness Penalty Approach},
  publisher = {Chapman and Hall},
  year = {1994}
}

@article{omondi2013influence,
Author = {Omondi, P and Awange, JL and Ogallo, LA and Ininda, J and Forootan, E},
	Journal = {Advances in Water Resources},
	Pages = {161--180},
	Publisher = {Elsevier},
	Title = {The influence of low frequency sea surface temperature modes on delineated decadal rainfall zones in Eastern Africa region},
	Volume = {54},
	Year = {2013}
}

@article{pearson1901liii,
Author = {Pearson, Karl},
	Journal = {The London, Edinburgh, and Dublin Philosophical Magazine and Journal of Science},
	Number = {11},
	Pages = {559--572},
	Publisher = {Taylor \&amp; Francis},
	Title = {LIII. On lines and planes of closest fit to systems of points in space},
	Volume = {2},
	Year = {1901}
}

@article{reason2002sensitivity,
Author = {Reason, CJC},
	Journal = {International Journal of Climatology},
	Number = {4},
	Pages = {377--393},
	Publisher = {Wiley Online Library},
	Title = {Sensitivity of the southern African circulation to dipole sea-surface temperature patterns in the south Indian Ocean},
	Volume = {22},
	Year = {2002}
}

@article{regularized_covariance,
author = {Tzeng, ShengLi and Huang, Hsin-Cheng},
  title = {Non-stationary Multivariate Spatial Covariance Estimation via Low-rank Regularization},
  journal = {Statistica Sinica},
  year = {2015},
  volume = {26},
  pages = {151--172}
}

@article{richman1987rotation,
author = {Richman, Michael B.},
  title = {Rotation of Principal Components: A Reply},
  journal = {Journal of Climatology},
  year = {1987},
  volume = {7},
  pages = {511--520}
}

@article{rotate,
author = {Richman, Michael B.},
  title = {Rotation of Principal Components},
  journal = {Journal of Climatology},
  year = {1986},
  volume = {6},
  pages = {293--335}
}

@article{salim2005modelling,
Author = {Salim, Agus and Pawitan, Yudi and Bond, K},
	Journal = {Journal of the Royal Statistical Society, Series C},
	Number = {3},
	Pages = {555-573},
	Title = {Modelling association between two irregularly observed spatiotemporal processes by using maximum covariance analysis},
	Volume = {54},
	Year = {2005}
}

@article{salim2007model,
Author = {Salim, Agus and Pawitan, Yudi},
	Journal = {Journal of agricultural, biological, and environmental statistics},
	Number = {1},
	Pages = {1--24},
	Title = {Model-based maximum covariance analysis for irregularly observed climatological data},
	Volume = {12},
	Year = {2007}
}

@article{schwarz1978estimating,
author = {Schwarz, Gideon},
  title = {Estimating the Dimension of a Model},
  journal = {The Annals of Statistics},
  year = {1978},
  volume = {6},
  pages = {461--464}
}

@article{spatial_pca,
author = {Dem{\v{s}}ar, Ur{\v{s}}ka and Harris, Paul and Brunsdon, Chris and Fotheringham, A. Stewart and McLoone, Sean},
  title = {Principal Component Analysis on Spatial Data: An Overview},
  journal = {Annals of the Association of American Geographers},
  year = {2013},
  volume = {103},
  pages = {106--128}
}

@article{spatpca,
author = {Wang, Wen-Ting and Huang, Hsin-Cheng},
  title = {Regularized Principal Component Analysis for Spatial Data},
  journal = {Journal of Computational and Graphical Statistics},
  year = {2017},
  volume = {26},
  pages = {14--25}
}

@article{spca,
author = {Zou, Hui and Hastie, Trevor and Tibshirani, Robert},
  title = {Sparse Principal Component Analysis},
  journal = {Journal of Computational and Graphical Statistics},
  year = {2006},
  volume = {15},
  pages = {265--286}
}

@article{spca2,
author = {d'Aspremont, Alexandre and Bach, Francis and El Ghaoui, Laurent},
  title = {Optimal Solutions for Sparse Principal Component Analysis},
  journal = {Journal of Machine Learning Research},
  year = {2008},
  volume = {9},
  pages = {1269--1294}
}

@article{spca3,
author = {Lu, Zhaosong and Zhang, Yong},
  title = {An Augmented Lagrangian Approach for Sparse Principal Component Analysis},
  journal = {Mathematical Programming Computation},
  year = {2012},
  volume = {4},
  pages = {151--188}
}

@article{sst,
author = {Deser, Clara and Alexander, Michael A. and Xie, Shang-Ping and Phillips, Adam S.},
  title = {Sea Surface Temperature Variability: Patterns and Mechanisms},
  journal = {Annual Review of Marine Science},
  year = {2009},
  volume = {2},
  pages = {115--143}
}

@article{stein1981estimation,
author = {Stein, Charles M.},
  title = {Estimation of the Mean of a Multivariate Normal Distribution},
  journal = {The Annals of Statistics},
  year = {1981},
  volume = {9},
  pages = {1135--1151}
}

@incollection{svd_trace,
Author = {Lee, Namgil and Cichocki, Andrzej},
	Booktitle = {Advances in Neural Networks--ISNN 2014},
	Pages = {121--130},
	Publisher = {Springer, Switzerland},
	Title = {Big Data Matrix Singular Value Decomposition Based on Low-Rank Tensor Train Decomposition},
	Year = {2014}
}

@article{tucker1958inter,
Author = {Tucker, Ledyard R},
	Journal = {Psychometrika},
	Number = {2},
	Pages = {111--136},
	Title = {An inter-battery method of factor analysis},
	Volume = {23},
	Year = {1958}
}

@article{witten2009application,
Author = {Witten, Daniela M and Tibshirani, Robert J},
	Journal = {Statistical applications in genetics and molecular biology},
	Number = {1},
	Pages = {1-27},
	Title = {Extensions of Sparse Canonical Correlation Analysis with Applications to Genomic Data},
	Volume = {8},
	Year = {2009}
}

@article{witten2009penalized,
Author = {Witten, Daniela M and Tibshirani, Robert and Hastie, Trevor},
	Journal = {Biostatistics},
	Number = {3},
	Pages = {515--534},
	Title = {A penalized matrix decomposition, with applications to sparse principal components and canonical correlation analysis},
	Volume = {10},
	Year = {2009}
}

@article{yang2015independence,
author = {Yang, Yanrong and Pan, Guangming},
  title = {Independence Test for High Dimensional Data Based on Regularized Canonical Correlation Coefficients},
  journal = {The Annals of Statistics},
  year = {2015},
  volume = {43},
  number = {2},
  pages = {467--500}
}

@article{zou2007degrees,
author = {Zou, Hui and Hastie, Trevor and Tibshirani, Robert},
  title = {On the ``Degrees of Freedom'' of the Lasso},
  journal = {The Annals of Statistics},
  year = {2007},
  volume = {35},
  pages = {2173--2192}
}

\end{document}